\documentclass[a4paper,11pt]{article}
\pdfoutput=1 

\input cyracc.def

\usepackage{jheppub} 

\usepackage[T1]{fontenc} 

\usepackage{lmodern}

\usepackage{bbm}
\usepackage{xcolor}
\usepackage{mathtools}
\usepackage{mathdots}
\usepackage{amsmath}
\usepackage{bm}
\usepackage{amsfonts}
\usepackage{dsfont}
\usepackage{amsthm}
\usepackage{enumitem}
\usepackage{booktabs}

\definecolor{cnblau}{RGB}{0,80,147}

\newcommand{\zs}{{\underline z}}

\newcommand{\dd}{\mathrm{d}}
\newcommand{\bq}{\bar{l}}

\newenvironment{gleichung}{\begin{equation}\begin{aligned}}{\end{aligned}\end{equation}\\ \noindent}
\newenvironment{gleichung*}{\begin{equation*}\begin{aligned}}{\end{aligned}\end{equation*}}

\newcommand{\eps}[0]{\epsilon}
\newcommand{\cL}[0]{\mathcal{L}}

\newcommand{\ord}[0]{\mathcal{O}}
\newcommand{\cI}[0]{\mathcal{I}}
\def\beq{\begin{equation}}
\def\eeq{\end{equation}}
\def\bsp#1\esp{\begin{split}#1\end{split}}
\newcommand{\uz}{\underline{z}}
\newcommand{\ux}{\underline{x}}
\newcommand{\uk}{\underline{k}}
\newcommand{\uI}{\underline{I}}
\newcommand{\uJ}{\underline{J}}
\newcommand{\uN}{\underline{N}}
\newcommand{\uL}{\underline{L}}
\newcommand{\unu}{\underline{\nu}}
\newcommand{\uPi}{\underline{\Pi}}
\newcommand{\ualpha}{\underline{\alpha}}
\newcommand{\rd}{\textrm{d}}

\newtheorem{prop}{Proposition}

\title{Feynman Integrals in Dimensional Regularization and Extensions of Calabi-Yau Motives }

\author[ab]{Kilian B\"onisch}
\author[a]{Claude Duhr}
\author[a]{Fabian Fischbach}
\author[acd]{Albrecht Klemm}
\author[a]{Christoph Nega}

\affiliation[a]{Bethe Center for Theoretical Physics, Universit\"at Bonn, D-53115, Germany}
\affiliation[b]{Max-Planck-Institut f\"ur Mathematik, Bonn, D-53111, Germany}
\affiliation[c]{Hausdorff Center for Mathematics, Universit\"at Bonn, D-53115, Germany}
\affiliation[d]{Institute for Theoretical Studies, ETH Z\"urich, CH-8092 Zurich,  Switzerland}

\emailAdd{kilian@mpim-bonn.mpg.de}
\emailAdd{cduhr@uni-bonn.de}
\emailAdd{fischbach@physik.uni-bonn.de}
\emailAdd{aklemm@th.physik.uni-bonn.de}
\emailAdd{cnega@th.physik.uni-bonn.de}

\abstract{
We provide a comprehensive summary of concepts from Calabi-Yau motives relevant to the computation of multi-loop Feynman integrals. From this we derive several consequences for multi-loop integrals in general, and we illustrate them on the example of multi-loop banana integrals. For example, we show how Griffiths transversality, known from the theory of variation of mixed Hodge structures, leads quite generically to a set of quadratic relations among maximal cut integrals associated to Calabi-Yau motives. These quadratic relations then naturally lead to a compact expression for $l$-loop banana integrals in $D=2$ dimensions in terms of an integral over a period of a Calabi-Yau $(l-1)$-fold. This new integral representation generalizes in a natural way the known representations for $l\le 3$ involving logarithms with square root arguments and iterated integrals of Eisenstein series. In a second part, we show how  the results obtained by some of the authors in earlier work can be extended to dimensional regularization. We present a method to obtain the differential equations for banana integrals with an arbitrary number of loops in dimensional regularization without the need to solve integration-by-parts relations. We also present a compact formula for the leading asymptotics of banana integrals with an arbitrary number of loops in the large momentum limit. This generalizes the novel $\widehat{\Gamma}$-class introduced by some of the authors to dimensional regularization and provides a convenient boundary condition to solve the differential equations for the banana integrals. As an application, we present for the first time numerical results for equal-mass banana integrals with up to four loops and up to second order in the dimensional regulator.
}

\begin{document}

\rightline{
BONN-TH-2021-05
}

\maketitle

\flushbottom


\section{Introduction}
\label{sec:intro}

Multi-loop Feynman integrals are the cornerstone of perturbative quantum field theory, and they are the main tool to make precise predictions for collider experiments in particle physics. Having efficient methods for their computation is therefore not only important in order to explore the mathematical structure of quantum field theories, but also to compare theory and experiment. Over the last decade, it has become clear that there is a connection between Feynman integrals and certain branches of mathematics. In particular, Feynman integrals are (families of) periods~\cite{Bogner:2007mn} in the sense of Kontsevich and Zagier~\cite{MR1852188}, depending parametrically on the external kinematic data, e.g., the masses and momenta of all the particles involved. They then satisfy linear systems or first-order differential equations in the external kinematic data~\cite{Kotikov:1991hm,Kotikov:1991pm,Kotikov1991,Gehrmann2000,Henn:2013pwa}, or equivalently (inhomogeneous) linear differential equations of higher order. These differential equations are reminiscent of the first order 
differential system that encodes the Gauss-Manin connection and the Picard-Fuchs equations describing equivalently the variation of mixed Hodge structures attached to families of algebraic varieties~\cite{MullerStach:2012mp,Vanhove:2014wqa}. A closely related observation concerning  the geometrical interpretation of 
Feynman integrals was made already in ref.~\cite{MR1011353}, where linear differential equations were identified with generalized  hypergeometric equations, which became later known more generally  
as \emph{Gel\cprime fand-Kapranov-Zelevinsk\u{\i} (GKZ) systems}.   

 The simplest examples of periods that show up at low loop orders are the so-called \emph{multiple polylogarithms} (also known as hyperlogarithms), which were first introduced in the works of Poincar\'e, Kummer and Lappo-Danilevsky~\cite{Kummer,Lappo:1927} and have recently reappeared in both mathematics~\cite{GoncharovMixedTate,Goncharov:1998kja,Brown:2011ik} and physics~\cite{Remiddi:1999ew,Gehrmann:2000zt,Ablinger:2011te}. Large classes of phenomenologically-important integrals can be expressed in terms of them, and several efficient numerical techniques exist for their computation~\cite{Gehrmann:2001jv,Gehrmann:2001pz,Vollinga:2004sn,Buehler:2011ev,Frellesvig:2016ske,Ablinger:2018sat,Naterop:2019xaf}. An important ingredient in the success of multiple polylogarithms to compute Feynman integrals lies in the fact that their mathematical and algebraic properties are well understood (see, e.g., ref.~\cite{Duhr:2014woa} for a review).
 
It has been known for several decades that starting from two loops not all Feynman integrals can be expressed in terms of multiple polylogarithms~\cite{Broadhurst:1987ei,Bauberger:1994by,Bauberger:1994hx,Laporta:2004rb,Kniehl:2005bc,Aglietti:2007as,Czakon:2008ii,Brown:2010bw,Muller-Stach:2011qkg,CaronHuot:2012ab,Huang:2013kh,Brown:2013hda,Nandan:2013ip}, but no complete analytic results were known. In a landmark paper~\cite{Bloch:2013tra}, Bloch and Vanhove showed that the so-called two-loop sunrise graph with three massive propagators can be expressed in terms of an elliptic dilogarithm (see also refs.~\cite{Adams:2013nia,Adams:2014vja,Adams:2015gva,Adams:2016xah,Ablinger:2017bjx}), which is a special case of the multiple elliptic polylogarithms defined in refs.~\cite{MR1265553,LevinRacinet,BrownLevin}. In the case where the values of the three masses are equal and non-zero, one can also express the result in terms of iterated integrals of Eisenstein series~\cite{Adams:2017ejb,Broedel:2018iwv}, which have recently been introduced in pure mathematics, cf.,~e.g., refs.~\cite{ManinModular,Brown:mmv,Matthes:QuasiModular}. By now it is clear that elliptic polylogarithms and iterated integrals of modular forms are relevant for large classes of multi-loop Feynman integrals~\cite{Broedel:2017kkb,Broedel:2017siw,Adams:2017ejb,Broedel:2018iwv,Broedel:2018qkq,Adams:2018yfj,Adams:2018bsn,Adams:2018kez,Broedel:2019hyg,Broedel:2019kmn,Duhr:2019rrs,Bogner:2019lfa,Abreu:2019fgk,Campert:2020yur,Walden:2020odh,Bezuglov:2020ywm,Weinzierl:2020fyx,Kristensson:2021ani}. This has led to several studies on the properties of these functions from a physics perspective, including methods for their numerical evaluation~\cite{Ablinger:2017bjx,Bogner:2017vim,Walden:2020odh}.

Period integrals associated to elliptic and modular curves, however, are still not sufficient to capture the full breath of Feynman integrals, even for small loop numbers. In particular, there are several infinite families of Feynman graphs~\cite{Vanhove:2018mto,Klemm:2019dbm,Bonisch:2020qmm,Bourjaily:2018ycu,Bourjaily:2018yfy,Bourjaily:2019hmc}, most prominently the  so-called \emph{banana graphs}~\cite{Vanhove:2018mto,Klemm:2019dbm,Bonisch:2020qmm} and 
\emph{train-track graphs}~\cite{Bourjaily:2018ycu}, where the associated geometry at three loops involves a K3 surface, and more generally Calabi-Yau $(l-1)$-folds at $l$ loops. The geometry associated to the three-loop banana graph in $D=2$ space-time dimensions was first studied in ref.~\cite{Bloch:2014qca,MR3780269}. In the case of four equal non-zero masses, this K3 surface is elliptically fibered by the same family of elliptic curves governing the structure in the two-loop case~\cite{verrill1996,Primo:2017ipr}. As a consequence, the three-loop banana graph in $D=2$ space-time dimensions can be expressed in terms of the same class of iterated integrals of Eisenstein series that appear in the two-loop sunrise graph~\cite{Bloch:2014qca,Broedel:2019kmn}. For different masses or higher loops, no analytic representation of the answer in terms of this class of functions is expected to exist. In order to make progress in our understanding of $l$-loop Feynman integrals, it is therefore important to develop new mathematical techniques 
to tackle those period integrals that are defined by integrating the unique holomorphic $(l-1,0)$-Calabi-Yau-form $\Omega_l$ over cycles or chains that arise  in 
embeddings of families of Calabi-Yau  $(l-1)$-folds. The latter integrals are the  natural generalization of the elliptic 
integrals for $l=2$ or K3 periods for $l=3$ to all loop orders. The $l$-loop banana graphs provide the simplest sequence of Feynman 
integrals that can be understood  as integrals on higher-dimensional  Calabi-Yau manifolds  using  
techniques that have been developed in the context of mirror symmetry~\cite{Klemm:2019dbm,Bonisch:2020qmm}.
 
Important progress in identifying  the suitable families $M_{l-1}$ of Calabi-Yau $(l-1)$-folds, understanding their associated differential systems, and matching the solutions of these systems 
to the banana integrals in $D=2$ dimensions using the boundary behavior at the point of maximal unipotent monodromy  (MUM-point, see sections~\ref{sec:deqs} and~\ref{sec:CY}) was 
made by some of the authors in refs.~\cite{Klemm:2019dbm,Bonisch:2020qmm}. 
A  key observation was that the homogeneous solutions of the Picard-Fuchs system describe period integrals of the same integrand $\Omega_l$, 
but with the integration domain replaced by closed cycles, i.e., integration domains without boundary, in the homology of a family of $M_{l-1}$. 
In particular, the  maximal cut integrals of the banana graphs\footnote{See also refs.~\cite{Primo:2016ebd,Frellesvig:2017aai,Primo:2017ipr,Bosma:2017ens} for a discussion of the maximal cut integrals and homogeneous differential equations.} 
can be identified with periods over cycles in the integral middle homology $H_{l-1}(M_{l-1},\mathbb{Z})$  of the Calabi-Yau family $M_{l-1}$. The corresponding 
particular linear combination of the homogeneous solutions can be identified in terms of a Frobenius basis for all closed periods at the MUM-point. 
Similarly, the full Feynman integral is geometrically identified as an  integral of $\Omega_l$ over a special geometrical chain  and given  
as a linear combination of the homogeneous solutions and a special inhomogeneous solution. Using mirror symmetry the leading logarithmic behavior 
in the large momentum regime was determined in ref.~\cite{Bonisch:2020qmm}  using  
the $\widehat \Gamma$-class formalism applied to the large volume regime of the geometry $W_{l-1}$ that is mirror dual to $M_{l-1}$.  
The physical relevant parts of the variation of the Hodge structure  of $M_{l-1}$ and its complexified K\"ahler 
structure are restricted sectors  that are exchanged under mirror symmetry. These sectors of $M_{l-1}$ can 
be identified with the corresponding sectors of $W_{l-1}$ as explained in Section \ref{ssec:cutsbanana}. Using these indentifications, the maximal cut solution that corresponds to the 
imaginary part of the banana integral by the optical theorem, is determined  in particular by the $\widehat \Gamma$-class of the mirror $W_{l-1}$,  that is well-known in the context of topological string 
theory on Calabi-Yau manifolds \cite{Hosono:1994ax,MR3965409}. The full banana integral in $D=2$ dimension was  determined in ref.~\cite{Bonisch:2020qmm}  by a novel $\widehat \Gamma$-class evaluation on the $l$-dimensional Fano ambient space $F_l$, a degree $(1,...,1)$ Fano hypersurface  in $(\mathbb{P}^1)^{l+1}$ in which $W_{l-1}\subset F_l$ 
is embedded.  This  proposed $\widehat \Gamma$-class evaluation in $F_l$, was also subsequently mathematically confirmed in ref.~\cite{Iritani:2020qyh}. 

The results of refs.~\cite{Klemm:2019dbm,Bonisch:2020qmm} show that techniques from geometry offer a promising direction to evaluate more general classes of Feynman integrals, even those with many loops and depending on many scales (see also refs.~\cite{delaCruz:2019skx,Klausen:2019hrg,NasrollahpPeriodsFeynmanDiagrams2016}). At the same time, these techniques are often not known or used by the Feynman integral community, mostly due to a lack of knowledge of the mathematics involved. In particular, it is not clear how these techniques are related to, or how they can be combined with, state-of-the-art techniques developed for the computation of Feynman integrals. The latter are often based on solving differential equations whose solutions involve so-called iterated integrals~\cite{ChenSymbol}, a very general class of functions of which the aforementioned (elliptic) polylogarithms and iterated integrals of modular forms are specific examples. Which classes of iterated integrals arise from Feynman integrals associated to Calabi-Yau geometries (or more generally Calabi-Yau motives), and what are their properties, is still unexplored.
Moreover, the techniques from geometry of refs.~\cite{Klemm:2019dbm,Bonisch:2020qmm} are not yet fully satisfactory from  physics perspective, because they only apply to banana integrals in strictly $D=2$ space-time dimensions, where all integrals are finite for non-zero values of the masses. Phenomenologically-interesting Feynman integrals, however, are usually divergent, and the divergences are regularized using dimensional regularization. The integrals are then considered in $D=D_0-2\epsilon$ dimensions ($D_0$ an even integer), and the divergences in the limit $\epsilon\to0$ show up as poles in the Laurent expansion in the dimensional regulator $\epsilon$. Understanding how to extend the techniques and results of refs.~\cite{Klemm:2019dbm,Bonisch:2020qmm} to dimensionally-regulated integrals is an important step if one wants to apply techniques from geometry to interesting Feynman integrals that require regularization (see also ref.~\cite{Klausen:2019hrg}).  
The aim of this paper is to address some of the aforementioned issues. 

First, we provide a thorough and in depth summary and review of the mathematics underlying Calabi-Yau geometries that seem to play an important role for Feynman integral computations. We provide explicit examples of how abstract concepts from geometry are related to Feynman integrals. From these concepts, we derive several new consequences for Feynman integrals (or more precisely, their maximal cuts), like how the Griffiths transversality on the geometry side leads to the existence of quadratic relations among maximal cut integrals on the physics side, or how to motivate the expected transcendental weight of banana integrals by studying the monodromy group of the integral. The main goal is to present abstract geometrical concepts in a way that directly connects them to Feynman integrals, and we expect  that these mathematical concepts will play a increasingly important role for Feynman integrals in the future.

Second, we show how the quadratic relations among maximal cuts can be used to obtain an expression for all master integrals for the equal-mass banana integrals in $D=2$ dimensions with an arbitrary number of loops in terms of an (iterated) integral over a Calabi-Yau period. This representation is the direct analogue of the well-known representation in terms of polylogarithms and iterated Eisenstein integrals known for low loop orders, and it shows which classes of iterated integrals arise from Feynman integrals associated to Calabi-Yau motives. We then use this representation to speculate how the concept of \emph{transcendental weight} known from physics (cf., e.g.,~refs.~\cite{Kotikov:2001sc,Kotikov:2002ab,Kotikov:2004er,Kotikov:2007cy}) extends to the situation of higher-loop banana integrals associated to Calabi-Yau $(l-1)$-folds.

Third, we extend the results of refs.~\cite{Klemm:2019dbm,Bonisch:2020qmm} to obtain results for all $l$-loop banana graphs for even values of $D_0$ to 
arbitrary order in the dimensional regulator $\epsilon$. Starting from a Mellin-Barnes integral representation valid for any banana integral, we show how one can derive the set of differential equations satisfied by the master integrals in dimensional regularization (the so-called Picard-Fuchs ideal), and also a convenient boundary condition. The latter generalizes the novel $\widehat{\Gamma}$-class of ref.~\cite{Bonisch:2020qmm} to dimensional regularization. We also show that in the context of dimensional regularization, it is possible to give an interpretation of the structure of the solution space of the Picard-Fuchs ideal of ref.~\cite{Bonisch:2020qmm} in terms of a basis of (non-maximal) cut integrals. We also present an alternative and efficient method to derive the Picard-Fuchs operators in the equal-mass case from a Bessel representation, and we use this method to present for the first time numerical results for four-loop equal-mass banana integrals to higher order in the dimensional regulator.

This paper is organized as follows: 
In section~\ref{sec:deqs} we define our notations and conventions and give a brief summary of banana integrals and the differential equations that they satisfy. In section~\ref{sec:CY} we provide a detailed review of the main properties of Calabi-Yau motives that are relevant for Feynman integrals. In section~\ref{sec:banint} we use Griffiths transversality to derive quadratic relations among the maximal cuts of the equal-mass banana integrals in $D=2$ dimensions, and we use those relations to obtain a representation of $l$-loop banana integrals in $D=2$ dimensions as integrals over periods of Calabi-Yau $(l-1)$-folds. In section~\ref{sec:Bananadimreg} we show how the results of ref.~\cite{Bonisch:2020qmm} can be extended to dimensional regularization, and we present a way to derive the Picard-Fuchs ideal and the initial condition for all banana integrals with an arbitrary number of loops. In section~\ref{subsec:equalmasseps} we present an alternative way to derive the Picard-Fuchs operator in the equal-mass case, and we present numerical results for banana integrals with up to four loops and up to $\mathcal{O}(\eps^2)$ in dimensional regularization. In section~\ref{sec:conclusion} we draw our conclusions. We include several appendices with technical material omitted throughout the main text.

\section{Differential equations for Feynman integrals}
\label{sec:deqs}

\subsection{Families of Feynman integrals and master integrals}
\label{sec:fam}

The main focus of this paper are families of $l$-loop Feynman integrals, defined by
\beq\label{eq:feyn_def}
I_{\unu}(\ux;D)	\coloneqq	
\int\left(\prod_{r=1}^l\frac{\rd^Dk_r}{i\pi^{D/2}}\right)\,\left( \prod_{j=1}^p\frac{1}{D_j^{\nu_j}} \right) \,,
\eeq
with $D_j = q_j^2 -m_j^2+i0$. Here $D$ denotes the space-time dimension and the exponents $\unu=(\nu_j)_{1\le j\le p}$ define a point on the integer lattice $\mathbb{Z}^p$, and $|\underline\nu|=\sum_{j=1}^{p}\nu_j$. The propagator masses $m_j^2$ are positive real numbers and the momenta $q_j$ flowing through the propagators are linear combinations of the loop momenta $k_r$ and the external momenta $p_1,\ldots,p_E$, which are constraint to sum up to zero by momentum conservation. By Lorentz invariance the integral only depends on the propagator masses and the dot products between the external momenta. We refer to these collectively as the \emph{scales} $x_k$, and we collect them into the vector ${\ux} = (x_k)_{1\le k\le N}$. By dimensional analysis, the only non-trivial functional dependence is through the ratios \beq\label{eq:scaled_scales}
z_k\coloneqq x_{k+1}/x_1\,, \qquad1\le k<N\,.
\eeq

It is well known that not all the integrals in this family are independent. We can use integration-by-parts (IBP) relations to write every member of this family as a linear combination of a certain set of basis elements, conventionally referred to as \emph{master integrals}~\cite{Chetyrkin:1981qh,Tkachov:1981wb}. The basis of master integrals is known to be always finite~\cite{Smirnov:2010hn,Bitoun:2017nre,Mastrolia:2018uzb}.
In the following it will be useful to group the members of the family into \emph{sectors}, i.e., integrals that share the same set of denominators in the integrand in eq.~\eqref{eq:feyn_def} (though the denominators may be raised to different powers). More precisely, consider the map $\vartheta:\mathbb{Z}^p\to \{0,1\}^p$ which sends $\unu=(\nu_j)_{1\le j\le p}$ to $\vartheta(\unu) = (\theta(\nu_j))_{1\le j\le p}$, where $\theta(m)$ denotes the Heaviside step function:
\beq\label{eq:heaviside}
\theta(m) = \left\{\begin{array}{ll}1\,, &\textrm{ if } m>0\,,\\0\,, &\textrm{ if } m\le0\,.
\end{array}\right.
\eeq
We say that $I_{\unu}(\ux;D)$ and $I_{\unu'}(\ux;D)$ belong to the same sector if $\vartheta(\unu)=\vartheta(\unu')$. There is a natural partial order on sectors, given by $\vartheta(\unu)\le \vartheta(\unu')$ if and only if $\theta(\nu'_i)-\theta(\nu_i)\ge0$, for all $1\le i\le p$.

We work in dimensional regularization, and each member of this family is interpreted as a Laurent series in the dimensional regularization parameter $\eps = (D_0-D)/2$, with $D_0$ a positive integer, cf., e.g., ref.~\cite{speer}. For algebraic values of the scales $\ux$, the Laurent coefficients are periods~\cite{Bogner:2007mn} in the sense of  Kontsevich and Zagier~\cite{MR1852188}. This motivates the use of techniques from 
algebraic geometry to compute Feynman integrals. One of the main goals of this paper is to study how some methods from geometry to compute periods can be used to compute multi-loop Feynman integrals in dimensional regularization. Our recurrent example will be a special class of $l$-loop Feynman integrals in $D=2-2\eps$ with at most $p=l+1$ propagators, known as \emph{banana integrals} (see figure~\ref{figbanana}), and the propagators are given by
\begin{figure}[!t]
	\centering
	\includegraphics[width=0.6\textwidth]{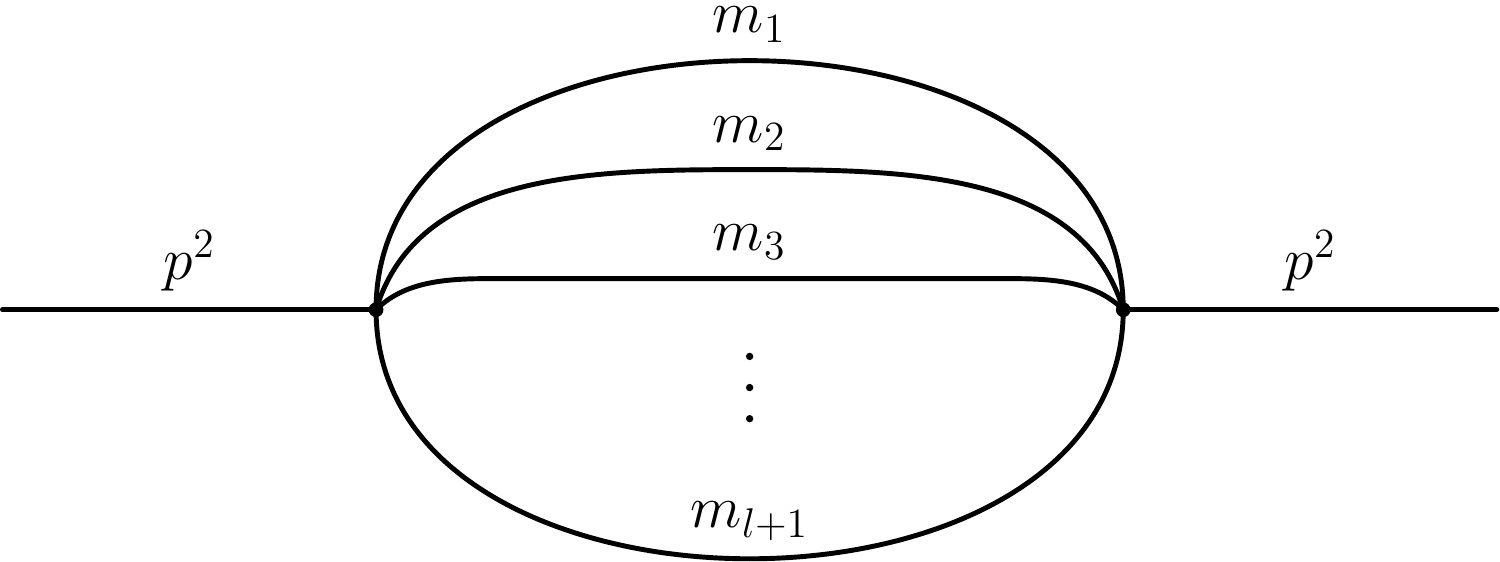}
     \caption{The $l$-loop banana graph with external momentum $p$ and internal masses $m_i$.}
     \label{figbanana}
\end{figure}
\beq\bsp
D_j = k_j^2-m_j^2\,,\qquad 1\le j\le l\,,\\
D_{l+1} = (k_1+\ldots+k_l-p)^2-m_{l+1}^2\,.
\esp\eeq
The integrals depend on the scales $\ux = (p^2,m_1^2,\ldots,m_{l+1}^2)$, and so $\uz = (m_1^2/p^2,\ldots,m_{l+1}^2/p^2)$. In total, there are $2^{l+1}-1$ master integrals, distributed among $l+2$ sectors. There are $l+1$ sectors of the form $\vartheta(\unu) = (1,\ldots,1,0,1\ldots1)$, and every integral in such a sector is proportional to an $l$-loop tadpole integral. For $1\leq i \leq l+1$ we define:
\begin{align}
J_{l,i}(\uz;\eps) 		&\,= 			\frac{(-1)^{l+1}}{\Gamma(1+l\epsilon)}(p^2)^{l\epsilon}\epsilon^l~I_{1,\ldots,1,0,1,\ldots,1}(\ux;2-2\eps) = -\frac{\Gamma(1+\epsilon)^l}{\Gamma(1+l\epsilon)}  \prod_{\substack{j=1\\j\neq i}}^{l+1}z_j^{-\epsilon}\,.
\end{align}
The sector $(1,\ldots,1)$ adds $2^{l+1}-l-2$ master integrals, one for each $\uk\in\{0,1\}^{l+1}$ with $1\le |\uk|\le l-1$ and $|\uk|=\sum_{j=1}^{l+1}k_j\,$:
\begin{equation}
\begin{aligned}\label{eq:banana_MIs}
J_{l,\underline0}({\uz};\eps) 	&\,= 			\frac{(-1)^{l+1}}{\Gamma(1+l\epsilon)}(p^2)^{1+l\epsilon}~I_{1,\ldots,1}(\ux;2-2\eps)	\\
J_{l,\uk}({\uz};\eps) 	&\,= 			(1+2\epsilon)\cdots(1+|{\uk}|\epsilon)~\partial^{\uk}_{\uz}J_{l,\underline0}({\uz};\eps) 	\,,
\end{aligned}
\end{equation}
with $\partial^{\uk}_{\uz}\eqqcolon\prod_{i=1}^{l+1}\partial_{z_i}^{k_i}$. We have explicitly checked for the first few loop orders that these integrals form a basis of master integrals. 
Moreover, it matches with the results of ref.~\cite{Kalmykov:2016lxx}. 
Note that the number $M$ of master integrals may change discontinuously in the limit where some scales vanish or become equal. In the equal-mass case, i.e., $m_i^2\eqqcolon m^2$ for $1\le i\le l+1$, the symmetry implies that there are only $l+1$ master integrals, which can be chosen as:
\beq\bsp
J_{l,0}(z;\epsilon)		&\,=	\frac{(-1)^{l+1}}{\Gamma(1+l\epsilon)}(m^2)^{l\epsilon}\epsilon^l~I_{1,\hdots,1,0}(p^2,m^2;2-2\epsilon)	=	-\frac{\Gamma(1+\epsilon)^l}{\Gamma(1+l\epsilon)}\,,	\\
J_{l,1}(z;\epsilon)		&\,=	\frac{(-1)^{l+1}}{\Gamma(1+l\epsilon)}(m^2)^{1+l\epsilon}~I_{1,\hdots,1}(p^2,m^2;2-2\epsilon)\,,			\\
J_{l,k}(z;\eps) 	&\,= 			(1+2\epsilon)\cdots(1+k\epsilon)~\partial_z^{k-1}J_{l,1}(z;\epsilon)\,,	\qquad \text{for } 2\le k\le l\,,
\label{eq:equal_mass_MIs}
\esp\eeq
where we defined $z\coloneqq \frac{m^2}{p^2}$.

We note that the number of master integrals changes also discontinuously when $\epsilon$ takes special values. In particular, in the  
generic-mass case, for $\eps=0$ we have only $2^{l+1}-\binom{l+2}{\lfloor\frac{l+2}2\rfloor}$ independent 
master integrals instead of $2^{l+1}-(l+1)-1$ in the sector $(1,\ldots,1)$. We note that this corresponds to the even primitive vertical cohomology $H^{k,k}_\text{vert}(W^\text{CI}_{l-1})$ for $k=0,\ldots,l-1$ of $W^\text{CI}_{l-1}$ given in eq.~\eqref{CICY}, or the horizontal middle  cohomology $H_{\text{hor}}^{l-1}(M^\text{CI}_{l-1})$ of its mirror $M^\text{CI}_{l-1}$ (see section~\ref{sec:CY} for the descriptions of these
(co)homology groups). Similar to eq.~\eqref{eq:banana_MIs}, in the latter picture the 
derivatives with respect to the $z_i$ for $i=1,\ldots,l+1$ generate  the cohomology groups in $H_\text{hor}^{l-1-k,k}(M^\text{CI}_{l-1})$, $k=0,\ldots,l-1$, see also eq.~\eqref{eq:transversality}. 
However, keeping in mind the linear dependencies of these derivatives in the cohomology of $M^\text{CI}_{l-1}$, one finds that there are only
\begin{equation} 
h^{l-1-l,k}_\text{hor}(M^\text{CI}_{l-1})	=	\left\{ \begin{array}{ll} 
								\left(l+1\atop k\right) & {\rm if }\ k\le  \left\lceil \frac{l}{2}\right\rceil-1 \\[1ex]
								\left(l+1\atop l-1-k\right) & {\rm otherwise }
						\end{array}\right. 
\end{equation} 
independent ones~\cite{Bonisch:2020qmm}.

\subsection{Gauss-Manin-type differential equations}
\label{subsec:f-odeqs}

Let us collect the master integrals into a vector ${{\uI}}({\ux};\eps)\coloneqq(I_1({\ux};\eps),\ldots,I_M({\ux};\eps))^T$. The master integrals then satisfy a system of first-order linear differential equations~\cite{Kotikov:1990kg,Kotikov:1991hm,Kotikov:1991pm,Gehrmann:1999as,Henn:2013pwa}:
\beq
\rd{\uI}({\ux};\eps) = {\bf A}({\ux};\eps)\,{\uI}({\ux};\eps)\,,
\label{eq:firstorder}
\eeq
where $\rd=\sum_{k=1}^N\rd x_k\,\partial_{x_k}$ is the total differential and ${\bf A}({\ux};\eps)$ is a matrix of rational one-forms. We will refer to this system of first-order differential equations as the \emph{Gauss-Manin system} for the family of integrals. Indeed, from a geometrical point of view, we can interpret the $M$-dimensional vector space spanned by the family of Feynman integrals as a rank $M$ vector bundle over the base defined by the scales $\ux$. On this vector bundle there exists a flat connection called the \emph{Gauss-Manin connection}, and the matrix ${\bf A}({\ux};\eps)$ is the corresponding connection one-form. 

The basis of master integrals is not unique. If ${\bf M}({\ux};\eps)$ is an invertible matrix, we can define a new basis ${\uJ}({\uz};\eps)$ by ${\uI}({\ux};\eps) = {\bf M}({\ux};\eps){\uJ}({\uz};\eps)$, and we have
\beq\label{eq:JDEQ}
\rd{\uJ}({\uz};\eps) = \widetilde{{\bf A}}({\uz};\eps){\uJ}({\uz};\eps)\,,
\eeq
with
\beq\label{eq:A_tilde_def}
\widetilde{{\bf A}}({\uz};\eps) = {\bf M}({\ux};\eps)^{-1}\left[{\bf A}({\ux};\eps){\bf M}({\ux};\eps) - \rd{\bf M}({\ux};\eps)\right]\,.
\eeq
Note that we also use the transformation in eq.~\eqref{eq:A_tilde_def} to pass from $\ux$ to $\uz$. 
It is then possible to choose the matrix ${\bf M}({\ux};\eps)$ such that the new differential equation is as simple as possible. It was argued in ref.~\cite{Lee:2019wwn} that it is always possible to change basis to a so-called \emph{$\eps$-regular basis}, where the master integrals ${J}_i({\uz};\eps)$ are finite and non-zero as $\eps\to0$ (see also ref.~\cite{Chetyrkin:2006dh} for a closely related concept). It is easy to see that in this case also the matrix $\widetilde{{\bf A}}({\uz};\eps)$ remains regular as $\eps\to0$, and we define ${\bf A}_0({\uz}) \coloneqq \lim_{\eps\to0}\widetilde{{\bf A}}({\uz};\eps)$. In the following we assume that this limit exists, though we may allow bases that are not necessarily $\eps$-regular.

In the special case where we can find a matrix ${\bf M}({\ux};\eps)$ that is rational in $\eps$ and algebraic in ${\ux}$ such that $\widetilde{{\bf A}}({\uz};\eps) = \eps\,{\bf A}_1({\uz})$, the Gauss-Manin system in eq.~\eqref{eq:JDEQ} is said to be in \emph{canonical form}~\cite{Henn:2013pwa} and can easily be solved in terms of a path-ordered exponential
\beq\label{eq:Pexp}
{\uJ}({\uz};\eps) = \mathbb{P}\exp\left[\eps\int_{{\uz}_0}^{{\uz}}{\bf A}_1({\uz}')\right]{\uJ}({\uz}_0;\eps)\,,
\eeq
where the integral is over a path from the point ${\uz}_0$ to the point ${\uz}$. This representation has the advantage that the path-ordered exponential can be expanded around $\eps=0$, and the expansion can easily be truncated after a few terms. The coefficient of $\eps^k$ will involve iterated integrals over algebraic one-forms. The same conclusion holds if one can find ${\bf M}({\ux};\eps)$ such that ${\bf A}_0({\uz})=0$, even if the system may not be strictly speaking into canonical form.

If there is no algebraic matrix ${\bf M}({\ux};\eps)$ to bring the system in eq.~\eqref{eq:JDEQ} in canonical form, then we need to proceed in a different fashion.\footnote{Though a canonical form may possibly be reached by allowing a transcendental rotation, cf., e.g., refs.~\cite{Adams:2018yfj,Adams:2018bsn,Adams:2018kez,Broedel:2018rwm,Bogner:2019lfa}.} The partial order on the sectors implies that we can always find a basis in which $\widetilde{{\bf A}}(\uz;\eps)$ is block-triangular. We can order the master integrals such that 
\beq\label{eq:sector_masters}
{\uJ}({\uz};\eps) = ({\uJ}_1({\uz};\eps)^T,\ldots, {\uJ}_s({\uz};\eps)^T)^T\,,
\eeq
where the elements of ${\uJ}_r({\uz};\eps)$ share exactly the same propagators, i.e., they belong to the same sector. The master integrals in each sector satisfy an inhomogeneous differential equation of the type
\beq\label{eq:DEG_inhom}
\rd{\uJ}_r({\uz};\eps) = {\bf B}_{r}({\uz};\eps)\,{\uJ}_r({\uz};\eps) + {\uN}_r({\uz};\epsilon)\,,\qquad 1\le r\le s\,.
\eeq
where the inhomogeneity ${\uN}_r({\uz};\epsilon)$ collects contributions from Feynman integrals from lower sectors, which we assume to be known. The associated homogeneous equation, obtained by putting ${\uN}_r({\uz};\eps)$ to zero, is the differential equation satisfied by the maximal cuts of ${\uJ}_r({\uz};\eps)$, defined, loosely speaking, by putting all the propagators in eq.~\eqref{eq:feyn_def} on shell~\cite{Primo:2016ebd,Frellesvig:2017aai,Harley:2017qut}. If the basis $\uJ(\uz;\eps)$ is $\eps$-regular, then so are ${\bf B}_{r}({\uz};\eps)$ and $ {\uN}_r({\uz};\epsilon)$. We define ${\bf B}_{r,0}({\uz}) \coloneqq \lim_{\eps\to0}{\bf B}_{r}({\uz};\eps)$.

Assume that we have found the general solution to the homogeneous equation for $\eps=0$. If ${\uJ}_r({\uz};\eps)$ has $M_r$ elements, this general solution can be conveniently cast in the form of an $M_r\times M_r$  matrix ${\bf W}_r({\uz})$ (called the \emph{Wronskian matrix}):
\beq
\rd{\bf W}_r({\uz}) = {\bf B}_{0,r}({\uz})\,{\bf W}_r({\uz})\,.
\eeq
Since the columns of ${\bf W}_r({\uz})$ form a basis for the solution space, this matrix must have full rank (for generic values of ${\uz}$). Letting 
\beq
{\uL}_r({\uz};\eps) = {\bf W}_r({\uz})^{-1}\,{\uJ}_r({\uz};\eps)\,,
\label{eq:L_to_J}
\eeq
we obtain the equation
\beq\label{eq:DEG_inhom_L}
\rd{\uL}_r({\uz};\eps) =  \widetilde{{\bf B}}_r({\uz};\eps)\,{\uL}_r({\bf z};\eps) + \widetilde{{\uN}}_r({\uz};\eps)\,,
\eeq
with 
\beq\bsp\label{eq:B_tilde_N_tilde}
\widetilde{{\bf B}}_r({\uz};\eps) &\,= {\bf W}_r({\uz})^{-1}\,\left[{\bf B}_{r}({\uz};\eps)-{\bf B}_{r,0}({\uz})\right]{\bf W}_r({\uz})\,,\\
\widetilde{{\uN}}_r({\uz};\eps) &\,= {\bf W}_r({\uz})^{-1}{\uN}_r({\uz};\epsilon)\,.
\esp\eeq
Note that by construction we have $\lim_{\eps\to0}\widetilde{{\bf B}}_r({\uz};\eps) =0$. Hence, we can easily solve the Gauss-Manin system in eq.~\eqref{eq:DEG_inhom_L} order-by-order in $\eps$. Since ${\uL}_r$ and ${\uN}_r$ must be regular at $\eps=0$, they admit a Taylor expansion:
\beq
{\uL}_r({\uz};\eps) = \sum_{k=0}^{\infty}\eps^k\,{\uL}_r^{(k)}({\uz})\textrm{~~~and~~~}
\widetilde{{\uN}}_r({\uz};\eps) = \sum_{k=0}^{\infty}\eps^k\,\widetilde{{\uN}}_r^{(k)}({\uz})\,.
\eeq
In particular, the leading order in $\eps$ leads to the equation:
\beq\label{eq:DEG_inhom_0}
d{\uL}_r^{(0)}({\uz}) =  \widetilde{{\uN}}_r^{(0)}({\uz})\,,
\eeq
which can easily be solved by quadrature:
\beq
{\uL}_r^{(0)}({\uz}) = {\uL}_r^{(0)}({\uz}_0) +\int_{{\uz}_0}^{\uz} \widetilde{{\uN}}_r^{(0)}({\uz}')\,.
\eeq
We can iteratively solve eq.~\eqref{eq:DEG_inhom_L} order by order in $\eps$ by inserting the solution into the expansion. This strategy was successfully applied to several complicated Feynman integrals for which no canonical form can be reached via an algebraic transformation matrix ${\bf M}({\ux};\eps)$, see, e.g., refs.~\cite{Remiddi:2016gno,Primo:2017ipr,vonManteuffel:2017hms,Chen:2017pyi,Kniehl:2019vwr,Lee:2019wwn,Lee:2020obg,Lee:2020mvt}.

\paragraph{The Gauss-Manin system for the equal-mass banana integrals.}

For the equal-mass banana family, we can collect the master integrals from the sector $(1,\ldots,1)$ in eq.~\eqref{eq:equal_mass_MIs} into the vector ${\uJ}_l(z;\eps)= (J_{l,1}(z,\eps),\ldots,J_{l,l}(z,\eps))^T$. At every loop order, this vector satisfies an inhomogeneous differential equation of the form (cf.~eq.~\eqref{eq:DEG_inhom})
\beq\label{eq:banana_DEQ}
\partial_z{\uJ}_{l}(z;\eps) = {\bf B}_{l}(z;\eps)\,{\uJ}_{l}(z;\eps) + {\uN}_{l}(z,\eps)\,,
\eeq
with 
\beq\bsp
{\bf B}_{l}(z;\eps)	&=	{\bf B}_{l,0}(z) + \sum_{k=1}^l {\bf B}_{l,k}(z) \eps^k	\,, \label{1}\\
{\uN}_l(z,\eps) 
&=\left(0,\ldots,0, (-1)^{l+1}(l+1)!\, \frac{z}{z^l\prod_{k\in\Delta^{(l)}}(1-kz)}\frac{\Gamma(1+\epsilon)^l}{\Gamma(1+l\epsilon)}\right)^T \,,
\esp\eeq
where we defined
\begin{equation}
	\Delta^{(l)}	\coloneqq	\bigcup_{j=0}^{\lceil\frac{l-1}{2}\rceil}\left\{(l+1-2j)^2\right\}	~.
\label{singpts}
\end{equation}
In section \ref{sec:diffeqs} we will discuss how to derive eq.~\eqref{1}, and we will show how to obtain the ${\bf B}_{l,k}(z)$ for $k=1,\hdots,l$. 

The Wronskian of the system is given by the maximal cuts of equal-mass banana integrals in $D=2$ dimensions. They can be defined by replacing the integration contour in eq.~\eqref{eq:feyn_def} by a contour $\Gamma$ that encircles the poles of the propagators $D_j=0$, $1\le j\le l+1$:
\beq\label{eq:cut_def}
\textrm{Cut}_{\Gamma}I_{\unu}(p^2,m^2;2) \coloneqq \oint_{\Gamma}\left(\prod_{j=1}^l\frac{d^2k_j}{D_j^{\nu_j}}\right)\frac{1}{D_{l+1}^{\nu_{l+1}}}\,.
\eeq
Let us introduce the following notation for the maximal cuts of the master integrals in eq.~\eqref{eq:equal_mass_MIs}:
\beq\bsp\label{eq:equal-mass_cuts}
J_{l,1}^{\Gamma}(z) &\,= \frac{(-1)^{l+1}}{\Gamma(1+l\eps)}\,(m^2)^{1+l\eps}\,\textrm{Cut}_{\Gamma}I_{1,\ldots,1}(p^2,m^2;2)\,,\\
J_{l,k}^{\Gamma}(z) &\,= (1+2\eps)\ldots(1+k\eps)\,\partial_z^{k-1}J_{l,1}^{\Gamma}(z)\,,\quad 2\le k\le l\,.
\esp\eeq
The vector $\uJ_l^{\Gamma}(z) \coloneqq (J_{l,1}^{\Gamma}(z),\ldots,J_{l,l}^{\Gamma}(z))^T$ satisfies the homogeneous version of the differential equation~\eqref{eq:banana_DEQ} for $\eps=0$:
\beq\label{eq:homogeneous_equation}
\partial_z\uJ_l^{\Gamma}(z) = {\bf B}_{l,0}(z)\uJ_l^{\Gamma}(z)\,.
\eeq
If we fix a basis of independent integration contours $\Gamma_1,\ldots,\Gamma_l$, then the Wronskian of the differential equation can be identified with the matrix of maximal cuts for the $l$-loop equal-mass banana graphs:
\beq
{\bf W}_l(z) \coloneqq \left(\uJ_l^{\Gamma_1}(z),\ldots, \uJ_l^{\Gamma_l}(z)\right)\,.
\eeq
Any other maximal cut contour can then be written in this basis, e.g., $\Gamma=\sum_{i=1}^{l}\alpha_i^{\Gamma}\,\Gamma_i$, and we have
\beq\label{eq:max_cut_Gamma}
\uJ_l^{\Gamma}(z) = {\bf W}_l(z)\ualpha^{\Gamma}\,,\qquad \ualpha^{\Gamma} = (\alpha_1^{\Gamma}\,\ldots,\alpha_l^{\Gamma})^T\,.
\eeq
For $l=1$, there is only one maximal cut (up to normalization), and it is simply an algebraic function.
The matrix of maximal cuts was evaluated for $l=2$ in ref.~\cite{Laporta:2004rb} in terms of complete elliptic integrals of the first and second kind:
\beq\bsp\label{eq:complete_elliptic}
\textrm{K}(z) &\,= \int_{0}^1\frac{dx}{\sqrt{(1-x^2)(1-zx^2)}}\,,\\
\textrm{E}(z) &\,= \int_{0}^1dx\,\sqrt{\frac{1-zx^2}{1-x^2}}\,.
\esp\eeq
For $l=3$, the solution can be expressed in terms of products of two elliptic integrals of the first and/or second kind~\cite{Primo:2017ipr}. For $l\ge 4$, it can be extremely challenging to write down an explicit basis of contours $\Gamma_j$ or to evaluate the corresponding cut integrals in eq.~\eqref{eq:equal-mass_cuts}. In particular, for $l\ge 4$ no representation of the maximal cuts in terms of known classical transcendental functions like elliptic integrals is known to exist.

It follows from the optical theorem (see, e.g., refs.~\cite{Landau:1960jol,Cutkosky:1960sp,tHooft:1973wag,Remiddi:1981hn,Veltman:1994wz}) that the banana integral develops a non-zero imaginary part for $p^2 > (l+1)^2m^2$ (or equivalently $0<z<1/(l+1)^2$) and the value of the imaginary part is proportional to a maximal cut integral.\footnote{Note that for Feynman integrals with more propagators, the imaginary part is given by a combination of non-maximal cuts.} In other words, there is a maximal cut contour $\Gamma_{\textrm{Im}}$ such that
\beq\label{eq:optical_thm}
\textrm{Im } \uJ_l(z) \sim \uJ_l^{\Gamma_{\textrm{Im}}}(z)\,,\qquad \textrm{for } 0<z<\frac{1}{(l+1)^2}\,.
\eeq

From the previous discussion it is clear that in order to solve eq.~\eqref{eq:banana_DEQ} order by order in $\eps$ in terms of iterated integrals, we first need to understand the Wronskian matrix of the associated homogeneous solution. From refs.~\cite{Bloch:2014qca,Bloch:2016izu,Vanhove:2014wqa,Vanhove:2018mto,Bonisch:2020qmm,Klemm:2019dbm} it is known that the geometry associated to the maximal cuts at $l$ loops is a Calabi-Yau $(l-1)$-fold. Understanding this geometry in some detail is needed in order to evaluate the (iterated) integrals that arise from eq.~\eqref{eq:DEG_inhom}. One of the goals of this paper is to show how the geometry associated to the $l$-loop banana integrals provides new methods to solve the integrals. In section~\ref{sec:banint} we will study in some detail what these Calabi-Yau $(l-1)$-folds can teach us about the solutions of the $l$-loop equal-mass banana integrals.

\subsection{Picard-Fuchs-type differential equations}
\label{subsec:PF}

Instead of solving the system of first-order differential equations for the vector $\underline I(\ux;\epsilon)$ of master integrals, it is also possible to consider an inhomogeneous higher-order differential equation satisfied by each master integral:
\beq\label{eq:higher_DEQ}
\cL_{k,\eps}\, I_k(\ux;D) = R_k(\ux;\eps)\,,
\eeq
where the inhomogeneity $R_k(\ux;\eps)$ is related to master integrals from lower sectors, and the differential operator $\cL_{k,\eps}$ has the form
\beq
\cL_{k,\eps} = \sum_{j_1,\ldots,j_m\ge 0}Q_{k,j_1\ldots j_m}(\ux;\eps)\,\partial_{x_1}^{j_1}\ldots \partial_{x_m}^{j_m}\,,
\eeq
where $Q_{k,j_1\ldots j_m}(\ux;\eps)$ are polynomials in $\ux$ and $\eps$. The operator $\cL\coloneqq\cL_{k,\eps=0}$ will annihilate the maximal cuts for $\eps=0$. Geometrically, these higher-order equations are also known as Picard-Fuchs differential equations. They describe the periods of algebraic varieties. This will be explained further in section~\ref{sec:CY}. The higher-order differential equations can for example be obtained by decoupling the first-order Gauss-Manin system. However, this may not be the only way to obtain them. In the case of the banana integrals, we will see that it is easier to derive the decoupled higher-order differential equations directly, without passing through the coupled first-order system. 

In the remainder of this section we review some general strategies to solve homogeneous linear higher-order differential equations (inhomogeneous equations can be brought into homogeneous form by acting with a suitable differential operator). The material in this section is well known in the literature (see, e.g., refs.~\cite{MR0010757,Orszac,Yoshida}), but we review it here because it will play an important role to understand the properties of the banana integrals, as studied by some of the authors in refs.~\cite{Klemm:2019dbm,Bonisch:2020qmm}. We start by reviewing in some detail the case of a single variable $z$ (which corresponds to the case of Feynman integrals depending on 2 scales), and briefly comment on the multi-variate generalization at the end.

\paragraph{One-parameter Picard-Fuchs-type differential equations.}
\label{para:OneParameterPF}

Consider a differential equation of the form
\beq\label{defdiffeq}
\mathcal L f(z) = 0 \qquad \textrm{with} \qquad \mathcal L=q_n(z)\partial_z^n+q_{n-1}(z)\partial_z^{n-1} + \hdots + q_0(z)\,,\qquad q_n(z) \neq 0\,,
\eeq
where the $q_i(z)$ are polynomials, and we assume that the $q_i(z)$ do not have any common zero. The leading coefficient $q_n(z)\eqqcolon\textrm{Disc}(\cL)$ is called the \emph{discriminant}. 
It will often be convenient to write the differential operator in the equivalent form
\beq
\cL = \tilde{q}_n(z)\theta^n+ \tilde{q}_{n-1}(z)\theta^{n-1} + \hdots + \tilde{q}_0(z) \qquad \text{with } \theta = \theta_z \coloneqq z\,\partial_z\,.
\eeq
One can relate both forms simply by the relations
\begin{align}\label{eq:d_to_theta}
	\theta^n	=	\sum_{i=1}^n	s_2(n,k) z^k\partial_z^k	\qquad\text{or}\qquad		z^n\partial_z^n	=	\prod_{j=0}^{n-1}(\theta-j)~,
\end{align}
with the Stirling numbers of second kind $s_2(n,k)=\frac1{k!}\sum_{i=0}^k(-1)^i \binom{k}{i} (k-i)^n$. In particular, we have $q_n(z)= z^n\tilde{q}_n(x)$.

This equation has $n$ independent solutions $f_i(z)$ for $1\leq i \leq n$. The solution space $\textrm{Sol}(\cL)$ is the $\mathbb{C}$-vector space generated by the $f_i(z)$.  
 Let $p_i(z) \coloneqq q_i(z)/q_n(z)$, $0\le i<n$. We want to understand the singularities of the solutions. We say that the differential equation \eqref{defdiffeq} has an \emph{ordinary point} at $z=z_0$ if the coefficient functions $p_i(z)$ are analytic in a neighbourhood of $z_0$ for all $0\leq i< n$. A point $z_0$ is called \emph{regular singular point} if the $(z-z_0)^{n-i}p_i(z)$ are analytic in a neighbourhood of $z_0$. An \emph{irregular singular point} is neither an ordinary nor a regular singular point. Note that all singular points $z_0\neq\infty$ are zeroes of the discriminant, $\left.\textrm{Disc}(\cL)\right|_{z=z_0} = q_n(z_0) = 0$. For the point at infinity $z_0=\infty$ one introduces the variable $t=1/z$ and makes the analysis around $t=0$. A differential equation without irregular singular points is called a \emph{Fuchsian} differential equation. Feynman integrals are expected to have only regular singularities, and no irregular singularities can appear. We therefore do not distinguish between regular and irregular singularities from now on.

Let us now briefly review how one can obtain a basis for the solution space using the well-known \emph{Frobenius method}. The goal will be to construct for every point $z_0\in\mathbb{C}$ $n$ linearly independent local solutions. Each local solution will be given in terms of power series convergent up to the nearest singularity. These local solutions can be analytically continued to multivalued global solutions over the whole parameter space. In the following we assume without loss of generality $z_0=0$ (if not, we perform a variable substitution $z\rightarrow z^\prime=z-z_0$ or $z^\prime=1/z$). Our starting point is the indicial equation
\begin{equation}
	\tilde{q}_n(0) \alpha^n + \tilde{q}_{n-1}(0) \alpha^{n-1} 	+ \hdots	+	\tilde{q}_0(0) \alpha	=	0\,,
\label{indicialeq}
\end{equation}
The solutions of eq.~\eqref{indicialeq} are called the \emph{indicials} or \emph{local exponents} at $z_0=0$.

We now discuss the structure of the solution space close to an ordinary or regular-singular point $z_0$.
If $z_0=0$ is an ordinary point, then there are $n$ different solutions $\alpha_1,\ldots,\alpha_n$ to eq.~\eqref{indicialeq}. The $n$-dimensional solution space is then spanned by:
\begin{equation}
	z^{\alpha_i} \Sigma_{i,0}(z)	=	z^{\alpha_i} \sum_{k=0}^\infty a_{i,k}z^k\,,	\qquad a_{i,0}\neq0\,, \qquad1\leq i \leq n~,
\label{solutionspower}
\end{equation}
where the $\Sigma_{i,0}(z)$ are power series around $z_0=0$ with non-vanishing radius of convergence and normalized according to $\Sigma_{i,0}(0)=1$. The coefficients $a_{i,k}$ can be computed from recurrence relations obtained by applying the operator $\mathcal L$ on the ansatz in eq.~\eqref{solutionspower}.
 
 If $z_0=0$ is a regular-singular point, there are still $n$ independent local solutions, but the solution space contains also solutions other than those in eq.~\eqref{solutionspower}. Again one analyzes the indicial equation in eq.~\eqref{indicialeq}, but now some solutions appear with multiplicities. Let us sort them as $(\alpha_1, \hdots, \alpha_1, \alpha_2,$ $ \hdots, \alpha_2, \hdots, \alpha_m, \hdots, \alpha_m)$. For all indicials $\alpha_1, \hdots, \alpha_r$ such that $\alpha_i-\alpha_j\notin\mathbb Z$ for pairwise distinct $i,j$, one gets $r$ power series-type solutions as in eq.~\eqref{solutionspower}. The missing $n-r$ solutions contain powers of $\log (z)$ and are constructed by the following procedure. For an indicial $\alpha_i\in\{ \alpha_1, \hdots, \alpha_r\}$ appearing with multiplicity $s$, one has $s-1$ different logarithmic solutions containing up to $s$ powers of $\log(z)$. They are given by
\begin{equation}
	z^{\alpha_i} \sum_{j=0}^k \frac1{(k-j)!}\,\log^{k-j}(z)\,\Sigma_{i,j}(z)	\quad\text{for } 0\leq k\leq s-1~,
\label{logsols}
\end{equation}
where again $\Sigma_{i,j}(z)$ are power series convergent until the nearest singularity, normalized such that $\Sigma_{i,j}(z)=\delta_{j0}+\mathcal O(z)$ for $j\geq 1$. For indicials $\alpha_i$ and $\alpha_k$ such that $\alpha_i-\alpha_k \in \mathbb Z$, one has to check case by case whether one obtains a power series-type solution as in eq.~\eqref{solutionspower} or a logarithmic solution as in eq.~\eqref{logsols}. 

For some Fuchsian differential equations there is a special singular point where all indicials are equal. Close to such a point the solution space can be characterized by an increasing hierarchical structure of logarithmic solutions, i.e., there exists a power series-type solution $\varpi_0$, a single logarithmic solution $\varpi_1$, and so on, up to $\log^{n-1}(z)$. Such a point is also-called a \emph{point of maximal unipotent monodromy} (MUM-point), and the associated basis $\varpi_0(z),\ldots,\varpi_{n-1}(z)$ is called a \emph{Frobenius basis}.

It may be convenient to collect the information about all singular points and their indicials in the so-called \emph{Riemann $\mathcal P$-symbol}. Let $\{z_1, \hdots, z_s \}$ be the singular points of the $n$-th order operator $\cL$, including possibly also the point at infinity. We denote the indicials for the singular point $z_i\in\{z_1, \hdots, z_s \}$ by $\{ \alpha^{(i)}_1, \hdots , \alpha^{(i)}_n\}$ (some indicials may be equal). The Riemann $\mathcal P$-symbol is then:
\begin{equation}
	\mathcal P		\left\{	\begin{matrix}	%
							z_1 			& z_2 			& \hdots	&	z_s			\\ \hline
							\alpha^{(1)}_1	& \alpha^{(2)}_1	& \hdots	&	\alpha^{(s)}_1	\\
							\vdots		& \vdots			& \ddots	&	\vdots		\\
							\alpha^{(1)}_n	& \alpha^{(2)}_n	& \hdots	&	\alpha^{(s)}_n
						\end{matrix} \right\}~.
\label{Riemannsymbol}
\end{equation}
The sum of all indicials fulfills the so-called \emph{Fuchsian relation}:
\begin{equation}
	\sum_{i=1}^s\sum_{j=1}^n \alpha^{(i)}_j	=	\frac{n(n-1)(s-2)}2~.
\label{fuchsianrelation}
\end{equation}

\paragraph{Equal-mass banana integrals in $D=2$ dimensions from the Frobenius method.}

Let us illustrate the concepts from the previous section on the example of the equal-mass banana integrals in $D=2$ dimensions (see section~\ref{sec:fam}, and in particular eq.~\eqref{eq:equal_mass_MIs}). We only quote here the results, and we refer to ref.~\cite{Bonisch:2020qmm} for details. The maximal cuts of $J_{l,0}(z;0)$ are annihilated by an $l^{\textrm{th}}$-order differential operator $\cL_l$. For low loop order, the explicit form of $\cL_{l}$ (and of its solutions) can be found in refs.~\cite{Laporta:2004rb,Muller-Stach:2011qkg,Bloch:2013tra,Bloch:2014qca,Adams:2014vja,Bloch:2016izu,Remiddi:2016gno,Primo:2017ipr,Broedel:2019kmn}. A procedure to obtain the operators $\cL_{l}$ for arbitrary values of $l$, and to construct their solutions, was presented in ref.~\cite{Bonisch:2020qmm} (see, in particular table 1 of ref.~\cite{Bonisch:2020qmm}). The operator $\cL_l$ is the Picard-Fuchs operator associated to a family of Calabi-Yau $(l-1)$-folds (see section~\ref{sec:CY})  parametrized by $z=m^2/p^2 \in\mathbb{R}$.\footnote{We could also consider complex values $z$, but for physics applications it is sufficient to consider $z$ real.}

In general, the differential operator $\cL_{l}$ has regular singular points at
\beq\label{eq:singularities}
z\in\{0,\infty\}\cup \bigcup_{j=0}^{\lceil\frac{l-1}{2}\rceil}\left\{\frac{1}{(l+1-2j)^2}\right\}\,,
\eeq
and the discriminant of $\cL_{l}$ is~\cite{Bonisch:2020qmm}
\beq\label{eq:discriminant}
\text{Disc}(\cL_{l}) = (-z)^l\,\prod_{k\in\Delta^{(l)}}(1-kz)~,
\eeq
with $\Delta^{(l)}$ defined in eq.~\eqref{singpts}. The discriminant is up to a factor of $z^l$ the coefficient of the highest $\theta$-derivative in the operator $\mathcal L_l$.
 
The Riemann $\mathcal P$-symbol has to be computed for each loop order separately, but there are some common features. The point at $z=0$ is a MUM-point with indicials $\alpha_i=1$ for $1\leq i \leq l$ for all loop orders. The closest singularity to $z=0$ is at $z=1/(l+1)^2$ and corresponds to the physical threshold at $p^2=(l+1)^2m^2$. Close to this threshold there is an additional single logarithmic solution for $l$ even and a square root solution for $l$ odd, which corresponds to a half integer indicial. At the other singularities except infinity the situation is similar. At $z=\infty$ the situation is more complicated, and one obtains more and higher logarithmic solutions depending on the loop order $l$. For example, for $l=2,3,4$, one finds
\begin{equation}
	\mathcal P_2		\left\{	\begin{matrix}	%
							0 	& \frac1{9}			& 1		&	\infty			\\ \hline
							1	& 0					& 0		&	0			\\
							1	& 0					& 0		&	0
						\end{matrix} \right\}~,\quad \mathcal P_3		\left\{	\begin{matrix}	%
							0 	& \frac1{16}			& \frac14	&	\infty			\\ \hline
							1	& 0					& 0		&	0			\\
							1	& \frac12				& \frac12	&	0			\\
							1	& 1					& 1		&	0
						\end{matrix} \right\}~,\quad	 \mathcal P_4		%
					\left\{	\begin{matrix}	%
							0 	& \frac1{25}			& \frac19	& 1		&	\infty			\\ \hline
							1	& 0					& 0		& 0		&	0			\\
							1	& 1					& 1		& 1		&	0			\\
							1	& 1					& 1		& 1		&	1			\\
							1	& 2					& 2		& 2		&	1
						\end{matrix} \right\}~.
\label{Riemann34}
\end{equation}
One can easily check that all Riemann $\mathcal{P}$-symbols satisfy the Fuchsian relation in eq.~\eqref{fuchsianrelation}. For more details we refer again to ref.~\cite{Bonisch:2020qmm}.

The point $z=0$ is a MUM-point, and we have the Frobenius basis\footnote{Note that we use here a different normalization of the periods compared to ref.~\cite{Bonisch:2020qmm}.} $\varpi_{l,0}(z), \ldots, \varpi_{l,l-1}(z)$ such that
\beq
\label{eq:frobenius}
	\varpi_{l,k}(z)	=	\sum_{j=0}^{k}	\frac 1{(k-j)!}\log^{k-j}(z)\, \Sigma_{l,j}(z)~,
\eeq
where the $\Sigma_{l,k}(z)$ are holomorphic in a neighbourhood of the MUM-point $z=0$, normalized such that
$\Sigma_{l,k}(z) = \delta_{k0}\,z + \ord(z^2)$. For $k=0$, we have~\cite{Bonisch:2020qmm}
\beq
\varpi_{l,0}(z) = \Sigma_{l,0}(z) = \sum_{k_1,\ldots,k_{l+1}\ge 0}\binom{\scriptstyle|k|}{\scriptstyle k_1,\ldots,k_{l+1}}^2\,z^{|k|+1}\,,
\eeq
with $|k| = k_1+\ldots+k_{l+1}$, and we have introduced the multinomial coefficient $\binom{\scriptstyle k}{\scriptstyle k_1,\ldots,k_{l+1}} = \frac{k!}{k_1!\cdots k_{l+1}!}\,$. Later on it will be useful to package the elements of the Frobenius basis into a vector:
\beq\label{eq:Pi_vec}
\uPi_l(z) \coloneqq \left(\varpi_{l,0}(z) ,\ldots, \varpi_{l,l-1}(z) \right)^T\,.
\eeq

The power-series representations for the $\Sigma_{l,k}(z)$ have a finite radius of convergence. The singularities are all located on the real axis (or at infinity), and we need to carefully analytically continue the functions across each of the singularities. The branch of the logarithm must be chosen such as to comply with the standard $i0$-prescription from physics, i.e., we need to pick the branch according to
\beq
\log(z-z_0) = \log|z-z_0| - i\pi\,\theta(z_0-z)	\qquad \text{for }z_0\in\mathbb{R}\,,
\eeq
with $\theta(z)$ defined in eq.~\eqref{eq:heaviside}.
The Frobenius basis was described in detail in ref.~\cite{Bonisch:2020qmm}. An algorithm\footnote{The code is available on \url{http://www.th.physik.uni-bonn.de/Groups/Klemm/data.php} in the supplementary material of the paper ''Analytic Structure of all Loop Banana Integrals``.} for their evaluation was implemented by some of us into {\tt PariGP}. It allows to evaluate the periods for high values of $l$ and for all positive real values of $z$ in a fast and efficient way by computing series expansions around any singular point.

\paragraph{Multi-parameter Picard-Fuchs operators.}
\label{para:MultiParameterPF}

Let us conclude by making some brief comments about how the Frobenius method generalizes to the multi-parameter case. 
In the multi-parameter case one has a set of differential operators $\mathcal D=\{ \mathcal \cL_1, \hdots,\mathcal \cL_r \}$, and we are looking for functions $f(\uz)$ that are simultaneously annihilated by all elements in $\mathcal{D}$. The solution space $\textrm{Sol}(\mathcal D)$ is the $\mathbb{C}$-linear span of all common solutions, i.e.,
\begin{equation}
	\text{Sol}(\mathcal D)	\coloneqq	\{ f(\uz) 	|	\mathcal \cL_i f(\uz)=0	 \text{ for all operators }	\mathcal \cL_i \in \mathcal D	\}~.
\label{solspacemulti}
\end{equation}
The set $\mathcal D$ actually generates a (left-)ideal of differential operators. Indeed, if $\mathcal \cL_i\in\mathcal D$ and $f\in \text{Sol}(\mathcal D)$, we have  $\tilde{\mathcal \cL}\mathcal \cL_if(\uz)=0$, for every differential operator $\tilde{\mathcal \cL}$.

It is possible to generalize the Frobenius method to the multi-variate case. 
Close to an ordinary point (in the sense of section \ref{para:OneParameterPF}), one again finds a basis of local solutions in terms of generalized power series-type solutions:
\begin{equation}	
	\left(\prod_{i=1}^mz_i^{\alpha_i}\right) \sum_{j_1,\hdots, j_m\ge 0} a_{j_1,\hdots,j_m}z_1^{j_1} \cdots z_m^{j_m}\,,
\label{multipower}
\end{equation}
with indicials $(\alpha_1, \hdots, \alpha_m)$.
At singular points also multi-variate logarithmic solutions can show up, and we have
\begin{equation}	
	\left(\prod_{i=1}^mz_i^{\alpha_i}\right) \sum_{\substack{j_1,\hdots, j_m\ge0\\ k_1,\hdots, k_m\ge0}} a_{j_1,\hdots,j_m,k_1,\hdots,k_m}\log^{j_1}(z_1)\cdots\log^{j_m}(z_m)z_1^{k_1} \cdots z_m^{k_m}~,
\label{multilog}
\end{equation}
where $j_1+\hdots+j_m$ depends on the multiplicity and the differences of the local indicials. 
Similarly to the one-parameter case, the local basis can be analytically continued to a global solution. In the multi-parameter case, however, this is a much harder problem, and may require blow ups at certain singular points, see, e.g., refs.~\cite{Candelas:1993dm,Candelas:1994hw}. As a side remark we note that if the singularity 
is too `bad' due to a crossing of many singular loci, choosing a good set of coordinates $(z_1, \hdots, z_m)$ can be important. It may happen that with the wrong choice of coordinates the Frobenius method does not produce all expected solutions. For a more thorough discussion we refer again to refs.~\cite{Candelas:1993dm,Candelas:1994hw}.

There may be different ways to choose the set of differential operators, or more precisely, how to choose a representation of the differential ideal generated by $\mathcal D$.  There can be different sets of operators, e.g., $\mathcal{D}=\{\cL_1, \hdots, \cL_s\}$ and $\mathcal{D}'=\{\cL'_1, \hdots, \cL'_{s'}\}$, which generate the same ideal, and thus they have the same solution space, i.e.,
\begin{equation}
	\text{Sol}(\mathcal D) = \text{Sol}(\mathcal D')\,.
\label{solspaces}
\end{equation}
Note that the sets $\mathcal{D}$ and $\mathcal{D}'$ can have different lengths, $s\neq s'$, and also the degrees of the operators can be different. Sometimes even a single but complicated operator is enough to generate the complete ideal. A clever choice of how to represent the ideal can have an impact on how complicated it is to find all the solutions. In particular, the higher-order differential operators obtained by decoupling the Gauss-Manin-type system is only one possible way to choose a set $\mathcal{D}$ that generates the ideal of differential operators; other, equivalent, choices are possible, and may lead to simplifications. We will exploit this freedom in later sections to obtain the differential equations satisfied by banana graphs at high loop orders.

Let us conclude with a comment. In the case of differential operators in one variable, it is always possible to choose a single 
differential operator that generates the differential ideal completely. More precisely, in appendix~\ref{app:pid_one_var} 
we show that the ring of linear differential operators in one variable with rational coefficients is a principle left-ideal domain, 
i.e., every differential operator in this ideal is of the form ${\cL}\cL_0$, for some distinguished differential operator $\cL_0$.

\section{Calabi-Yau motives and  Feynman graphs}
\label{sec:CY}

The first-order homogeneous differential equations discussed in the previous section
correspond quite generally to the Gauss-Manin connection for periods of families of 
algebraic varieties. In examples such as the sequence of $l$-loop banana or train-track graphs~\cite{Bourjaily:2018ycu,Vergu:2020uur}, families of Calabi-Yau $(l-1)$-folds (or 
more generally Calabi-Yau motives) play an important role, see also 
refs.~\cite{Bourjaily:2018yfy,Bourjaily:2019hmc} for additional examples. One might even 
speculate that this is a quite general feature. In all of these examples, Feynman integrals, or to be more precise their maximal cuts, are identified with Calabi-Yau periods, 
and the complex parameters of the Calabi-Yau families correspond to the dimensionless ratios of  
scales $\uz$ in eq.~\eqref{eq:scaled_scales}.  

\begin{table}[h]
\begin{center}
\begin{tabular}{|c|c|c| } 
 \hline\hline
& \textbf{${\bm{l=(n+1)}}$-loop banana} & \textbf{Calabi-Yau (CY) geometry}  \\ 
& \textbf{integrals in $\bm{D=2}$ dimensions} &\\
 \hline\hline
1 & Maximal cut integrals & $(n,0)$-form periods of CY  \\  
   & in $D=2$ dimensions  & manifolds or  CY motives \\
   && \\ [-5 mm]  
  \hline 
   && \\ [-5 mm]  
2 & Dimensionless ratios  $z_i = {m_i}^2/p^2$   & Unobstructed complex moduli of $M_n$, or  \\
&& equivalently  K\"ahler moduli of the mirror $W_n$   \\
   && \\ [-5 mm]  
  \hline 
   && \\ [-5 mm]     
3 & Integrand-basis for maximal cuts of     & Middle (hyper) cohomology $H^n(M_n)$ of  \\  
   & master integrals in $D=2$ dimensions & $M_n$ \\
   && \\ [-5 mm]  
  \hline 
   && \\ [-5 mm]  
4 & Quadratic relations among    & Quadratic relations from  \\  
 & maximal cut integrals    &  Griffiths transversality \\  
   && \\ [-5mm]
\hline 
   && \\ [-5 mm]  
5 & Integration-by-parts (IBP) reduction & Griffiths reduction method   \\ 
   && \\ [-5 mm]  
  \hline 
   && \\ [-5 mm]  
6 & Complete set of differential & Homogeneous Picard-Fuchs       \\  
 & operators annihilating a given&  differential ideal  (PFI) /    \\  
 &  maximal cut in $D=2$ dimensions & Gauss-Manin (GM) connection   \\ 
&& \\ [-5 mm]  
  \hline 
   && \\ [-5 mm] 
 7& (Non-)maximal cut contours & (Relative) homology of CY\\
  &  &  geometry $H_n(M_n)$ ($H_{n+1}(F_{n+1},\partial\sigma_{n+1})$)\\
 && \\ [-5 mm]  
  \hline 
   && \\ [-5 mm] 
8 &  Contributions from subtopologies & Extensions of the PFI    \\
   &to the differential equations&or the GM connection  \\
  \hline 
   && \\ [-5 mm]  
 9 & Full banana integrals  & Chain integrals in CY geometry  or  \\  
   & in $D=2$ dimensions & extensions of Calabi-Yau motive \\
 && \\ [-5 mm]  
   \hline 
 && \\ [-5 mm]     
  10 & Degenerate kinematics  & Critical divisors   \\  
  &(e.g., $m_i^2=0$ or $p^2/m_i^2\rightarrow 0$) &of the moduli space\\  
   \hline 
 && \\ [-5 mm]    
 11 &  Large-momentum regime & Point of maximal unipotent  \\
   & $p^2\gg m_i^2$ &  monodromy \& $\widehat\Gamma $-classes of $W_n$\\
 && \\ [-5 mm]  
   \hline 
 && \\ [-5 mm]    
 12 & General logarithmic degenerations & Limiting mixed Hodge structure \\
 && from monodromy weight  filtration \\
  && \\ [-5 mm]  
   \hline 
 && \\ [-5 mm]    
 13 &  Analytic structure and& Monodromy  of the CY motive \\
   &   analytic continuation &  and its extension \\
&& \\ [-5 mm]  
   \hline 
 && \\ [-5 mm]  
 14 &  Special values of the integrals & Reducibility of Galois action\\
  & for special values of the $z_i$&  \& $L$-function values  \\
  && \\ [-5 mm]  
    \hline 
 && \\ [-5 mm]  
15 & (Generalized?) modularity of  & Global $\mathrm{O}({\bf\Sigma}, \mathbb{Z})$-monodromy, integrality\\ 
 & Feynman integrals &  of mirror map \& instantons expansion\\
\hline \hline
\end{tabular}
\end{center}
\caption{Relationship between concepts related to $l$-loop banana integrals and  
geometric structures  related to  Calabi-Yau geometries.}
\label{tabledictionary}
\end{table}

The goal of this section is to review the most important features of Calabi-Yau $n$-folds, with a particular  
emphasis on their moduli spaces, their description by periods and variations of mixed Hodge 
structures, their discrete symmetries that allow to focus on sub-motives of the Hodge structures 
and in particular mirror symmetry. The arguably simplest examples of families of Calabi-Yau 
manifolds correspond to $n=1$ and are elliptic curves, also-called Calabi-Yau one-folds. For example, they can be realized by the affine Weierstrass equation:
\begin{equation}
y^2=4 x^3- g_2(z) x -g_3(z) \,,
\label{eq:ellipticcurve} 
\end{equation}  
in a patch\footnote{The homogeneous cubic equation $P_E=w y^2-4 x^3+ g_2(z) x w^2 + g_3(z)w^3=0$ with 
$[w:x:y]$ projective  coordinates of the complex projective $\mathbb{P}^2$ yields the global description 
of the elliptic curve.}  $w=1$ of $\mathbb P^2$. We note that families of elliptic curves are in 
homogeneous coordinates  the vanishing locus of cubics, i.e., of degree-three hypersurfaces, in $\mathbb{P}^2$. Generically such a cubic has ten 
cubic monomials in $x,y,w$ whose coefficients can vary to define a different member of the  complex 
family. However, these ten  variations are not independent, because nine combinations  can be undone by 
 PSL$(3,\mathbb{C})$ reparametrizations  and projective scalings of the ambient space coordinates (recall that the ambient space is the complex projective space $\mathbb{P}^2$). 
The latter are used to bring the curve into Weierstrass form, which depends only on one  complex structure parameter $z$. There are still 
discrete redundancies in this parametrization which can be removed by considering the $j$ invariant
$j(z)=g_2^3(z)/(g_2^3(z)-27 g_3^2(z))$. 

Many concepts familiar from the theory of elliptic curves and their periods, which are {elliptic 
integrals} depending on the complex deformation parameter $z$, generalize in a natural way to $n\ge 1$.  For example 
the simplest $n$-dimensional Calabi-Yau manifolds $M_n$ are degree $n+2$ hypersurfaces  in $\mathbb{P}^{n+1}$ with 
$\left(2n +3\atop n+2\right) -(n+2)^2=h^1(M_n,TM_n)=h^{n-1,1}(M_n)$ independent complex deformations,
see, e.g., eq.~\eqref{eq:isom}. Also  in the  
favorite representations of the  $l$-loop banana integrals in eq.~\eqref{MCicy} we arrive at the physical  and geometrical 
parameters in eq.~\eqref{eq:cicyT} by eliminating the reparametrization of the ambient space.  Unfortunately, no complete set of geometrical 
invariants are known for general Calabi-Yau $n$-folds.        

The geometrical and mathematical  concepts in this section appeared in the mathematical literature cited below. 
We collect here those that have been useful in the analysis of the banana graphs and that seem most likely to  play an important role 
for more general Feynman integrals in the future. This leads to a dictionary between the techniques used to compute 
banana integrals  and geometrical concepts in Calabi-Yau geometries that will pave the way to systematic higher-loop calculations. In table \ref{tabledictionary} we have summarized these concepts in the context of banana integrals. It is interesting to speculate if and in which form this dictionary can be extended to other Feynman integrals.
An overview of mathematical properties of Calabi-Yau $n$-folds can be found in the recent review in ref.~\cite{MR1963559}, 
a review from the point of view of string theory compactifications and mirror symmetry is ref.~\cite{MR3965409}, 
and a mathematical review of mirror symmetry can be found in ref.~\cite{MR1677117}. We  illustrate 
these mathematical concepts with the examples of the Weierstrass family in eq.~\eqref{eq:ellipticcurve}, the Legendre family and the  banana integrals.

\subsection{Calabi-Yau $n$-folds and their complex structure moduli spaces ${\cal M}_\text{cs}$}  
\label{subsec:CYmanifolds}

\emph{Calabi-Yau $n$-folds}  $M_n$ are complex $n$-dimensional \emph{K\"ahler manifolds}. They are equipped with a  
K\"ahler form $\omega$  of Hodge-type  $(1,1)$ that resides in the cohomology group $H^{1,1}(M_n,\mathbb{Z})$. The extra condition of 
being Calabi-Yau implies the existence of a \emph{non-trivial holomorphic $(n,0)$-form $\Omega$} spanning $H^{n,0}(M_n,\mathbb{C})$. 
In the case of the family of elliptic curves in eq.~\eqref{eq:ellipticcurve}, the former is the volume form on the elliptic 
curve and the latter generalizes the familiar holomorphic $(1,0)$-form ${\rm d} x/y$.      
The two forms $\Omega$ and $\omega$ are so characteristic for the Calabi-Yau manifold that one often refers to the triple 
$(M_n,\Omega,\omega)$ as a Calabi-Yau manifold. They are related via {the unique  volume form} $\omega^n/n!=(-1)^{n(n-1)/2} (i/2)^n\Omega\wedge \bar \Omega$. 
One can show that the existence of the form $\Omega$ (which is unique up to a phase) is equivalent to the fact that the 
holonomy group is $\mathrm{SU}(n)$, from which it follows that the first Chern class is trivial,  $c_1(M_n)=0$. This in turn implies by the famous Theorem  
of Yau that a Ricci-flat K\"ahler  metric $g_{i\bar \jmath}$ exists on every Calabi-Yau manifold~\cite{MR480350},\footnote{Calabi constructed some 
of these metrics  $g_{i\bar \jmath}$ explicitly for non-compact Calabi-Yau manifolds. For elliptic curves ($n=1$), Ricci-flatness implies 
flatness of the metric, $g_{1\bar 1}=\text{const}$. For higher-dimensional compact Calabi-Yau manifolds starting with the complex  $\mathrm{K3}$ surfaces ($n=2$) 
the Ricci-flat metric is not known explicitly.} i.e., $R_{i\bar \jmath}(g_{i\bar \jmath})=0$. We want the holonomy group to be 
the full $\mathrm{SU}(n)$ group, which implies that the cohomology groups $H^{k,0}(M_n,\mathbb{C})$ vanish 
unless $k=0$ or $k=n$, in which case its dimension is one. This is due to the fact that $\mathrm{SU}(n)$ acts canonically on these forms and the 
only invariant representations available are the trivial and the totally antisymmetric one.     

Another important property of Calabi-Yau manifolds $M_n$ is that their \emph{complex structure moduli space} ${\cal M}_\text{cs}$ has particularly nice and 
simple structures. The first-order deformations of a complex manifold are given by (finitely) many linearly independent elements 
in the cohomology group  $H^1(M_n,TM_n)$~\cite{MR2109686}. For Calabi-Yau manifolds  this space is isomorphic 
to the space of harmonic $(n-1,1)$-forms  
\beq
\label{eq:isom} 
H^1(M_n,TM_n)\cong  H^{n-1,1}(M_n) \,,
\eeq  
where the isomorphism is simply provided by contracting the elements in  $H^1(M_n,TM_n)$   with $\Omega$. First-order 
complex structure deformations can in general be globally obstructed  by higher-order obstructions,
which generically depend on the position in the complex moduli space. One can think of these obstructions
as higher terms in a potential  $W$ that obstructs the movement of a particle in a specific direction. 
Tian~\cite{MR915841} and Todorov~\cite{MR1027500} have proven the important fact that for Calabi-Yau 
varieties $M_n$ the complex  $h^{n-1,1}$-dimensional moduli space ${\cal M}_\text{cs}$ 
of complex structure deformations is globally unobstructed, i.e., in the picture with the potential, 
one has  $W(\uz)\equiv 0$ on ${\cal M}_\text{cs}$.

\paragraph{Complex families of Calabi-Yau $n$-folds.}
\label{para:CYfamilies}

It is natural to consider \emph{complex families of Calabi-Yau $n$-fold}s ${\cal M}_n$ with projection $\pi: {\cal M}_n\rightarrow {\cal M}_\text{cs}$ over the complex 
moduli space ${\cal M}_\text{cs}$, i.e., at each point $\underline{z}_0\in {\cal M}_\text{cs}$ one
has as a fiber a Calabi-Yau $n$-fold $\pi^{-1}({\underline z_0})= M_{n}^{\underline z_{0}}$ with a fixed complex structure.  In this picture 
one understands easily  that ${\cal M}_\text{cs}$  can have special  so-called \emph{critical} boundary components\footnote{Sometimes they are called singular 
components. However, since a component can itself be singular, we call them critical. As explained in section \ref{sssec:boundary} the moduli space can be 
compactified and resolved to ${\overline  {\cal M}}_\text{cs}$ so that these components become critical divisors with normal crossings in ${\overline  {\cal M}}_\text{cs}$.} . 
At these loci, the manifold itself, i.e., the fiber of the family, becomes singular. 
For example,  at a point in the one-dimensional complex moduli space of an elliptic curve a cycle $S^1$ might 
shrink to a point, and  the elliptic curve develops  a nodal singularity. More generally, a \emph{nodal singularity} corresponds to an $S^n$ shrinking. This is 
the most generic type of singularity for $n$-dimensional Calabi-Yau manifolds. The  corresponding 
critical  boundary of ${\cal M}_\text{cs}$ is called a \emph{conifold divisor} in ${\overline  {\cal M}}_\text{cs}$. However, $n$-dimensional Calabi-Yau manifolds can acquire a
much greater variety of more interesting singularities, which are only classified up to $n=2$ by the celebrated ADE-type classification 
of canonical  surface singularities. Families of higher-dimensional Calabi-Yau manifolds have in general a higher-dimensional 
moduli spaces, ${\rm dim}_\mathbb{C}({\cal M}_\text{cs})= h_{n-1,1}$. The divisors at which the fiber $M_n$ is singular will intersect in higher co-dimensional sub-loci 
in ${\overline  {\cal M}}_\text{cs}$. This produces over the intersection locus an  even more singular Calabi-Yau $n$-fold fibre. Generically, 
these critical divisors are given by the vanishing locus of algebraic equations, $\Delta_i({\underline z})=0$, $i=1,\ldots,r$. This locus can itself have singularities and non-generic intersections. There are mathematical techniques suggesting that within Calabi-Yau moduli spaces these singularities can be resolved by a finite sequence of blow ups to divisors with 
normal crossings~\cite{MR0199184}.  

One of the most important properties of a family of complex manifolds is the monodromy that its  periods or 
its homology groups undergo if one encircles the critical divisors in a normal crossing model
of  the moduli space $\overline {{\cal M}}_\text{cs}$. The local monodromy reflects the nature of the singularity of 
the fibres  and the degeneration of the periods. The global monodromy often restricts the class of functions
that can be periods. For example, for elliptic families these are weight-one modular forms for a congruence 
subgroup of $\mathrm{SL}(2,\mathbb{Z})$, determined by the global monodromy of the family. 
We will discuss  these concepts further in  section~\ref{sssec:boundary}.

\subsection{Geometric structures in the complex moduli space and period integrals} 
\label{subsec:complexmoduliandperiods}

Next, we discuss some features of the structures in the complex moduli space, in particular  
period integrals. We focus on the concepts that we expect to be most relevant for the application to Feynman integrals. 
In  section~\ref{sssec:bulk} we focus on the  interior of the moduli space, commonly called the \emph{bulk}. We give the
conceptual explanations, together with many references, that underly the most useful tools developed over a long 
period of time in mathematics, for example the Gauss-Manin connection, the  Picard-Fuchs differential ideal and the 
Griffiths transversality. As we will show in section~\ref{sssect:Griffithstransversality}, the latter leads straightforwardly to quadratic relations between maximal cuts of Feynman integrals.  In section \ref{sssec:boundary} we review 
concepts relevant to describe the possible degenerations of the geometry and the Feynman integrals at the critical divisors.

The mathematical properties of periods on compact smooth K\"ahler manifolds are captured by a so-called 
\emph{variation of a pure Hodge structure}, characterized by its \emph{decreasing Hodge filtration} in eq.~\eqref{eq:decreasing}. In a 
series of spectacular papers~\cite{MR0441965,MR498551,MR498552}, Deligne generalized that notion to the 
variation of \emph{mixed Hodge structures}  to include open smooth, complete singular and general varieties. In addition to  the decreasing Hodge filtration,
one defines in these situations  a second increasing filtration, the so-called \emph{monodromy weight filtration}. For many applications to Feynman integrals the
generalization to open manifolds  is essential, as their integration domain is open.  In our discussion, however, we 
concentrate on the case of complete singular spaces, as our main point is the study of the generalization of elliptic 
periods  to Calabi-Yau periods, because they appear as the maximal cuts of (at least) the banana integrals.  At the same time, our 
discussion is a prerequisite for the final step to describe the analytic structure 
of the banana integrals in terms of the generalized $\widehat \Gamma$-class. A standard reference  on mixed Hodge structures  
is the book of Peters and Steenbrink \cite{MR2393625}; some applications to mirror symmetry are 
discussed in the book by Cox and Katz~\cite{MR1677117}.

\subsubsection{On the bulk of  ${\cal M}_\text{cs}$} 
\label{sssec:bulk}    
 Away from the critical  divisors, i.e., in the bulk,  ${\cal M}_\text{cs}$ is a nicely behaved globally-defined 
K\"ahler  manifold of dimension $h^{n-1,1}(M_n)$, where the real K\"ahler potential $K(z,\bar z)$ is given  by    
\begin{equation} 
 e^{-K(\zs, \bar \zs)}=  i^{n^2} \int_{M_n} \Omega(\zs) \wedge \bar \Omega({\bar \zs})= i^{n^2} \uPi^\dagger(\zs)\,{\bf \Sigma}\, \uPi(\zs) \ .
 \label{eq:Kaehlerpotential}
 \end{equation}
For the last equal sign in eq.~\eqref{eq:Kaehlerpotential}, in particular the definition of the vector of period functions 
$\uPi(\zs) \coloneqq (\Pi_i(\zs))_{1\le i\le b_n}$ (with $b_k \coloneqq \dim H_k(M_n,\mathbb{Z})$ 
the Betti numbers) and the intersection pairing ${\bf \Sigma}$, see below.  Note that there is a ``gauge freedom'' 
to rescale $\Omega \rightarrow e^{f(\zs)}\Omega(\zs)$ with $f(\zs)$ a holomorphic function, under which  the K\"ahler potential   $K(\zs, \bar \zs)$ 
undergoes a K\"ahler gauge transformation
\begin{equation} 
K(\zs, \bar \zs) \rightarrow  K(\zs, \bar \zs)-f(\zs)-\bar f(\bar \zs) \ . 
\label{eq:kaehlergauge} 
\end{equation}     
 Most  geometrical structures on Calabi-Yau manifolds, e.g., the K\"ahler metric, are invariant under the gauge transformation in eq.~\eqref{eq:kaehlergauge}, 
but the Feynman integral is taken in a specific K\"ahler gauge.\footnote{This means concretely that the holomophic $(n,0)$-form 
could be modified by an $e^{f(z)}$ factor, but we make  an explicit choice in eqs.~\eqref{Griffith residue form} and~\eqref{eq:Omcicy}.}    
In eq. \eqref{eq:Kaehlerpotential} one understands the form $\Omega(\zs)$ to depend on the complex structure parameters  $\zs$ so  
that the  $\Omega(\zs_0)$ is  of type $(n,0)$ exactly for the complex structure defined by $\zs=\zs_0$. 
For each point $\zs_0$, the fibre $M^{\zs_0}_n$ over it enjoys a Hodge decomposition, in particular, of its middle dimensional cohomology:\footnote{We suppress 
the dependence on the point $\zs_0$ to ease the notation.}
\beq\label{eq:Hodge_decomposition}
H^n(M_n,\mathbb{C})=\bigoplus_{p+q=n} H^{p,q}(M_n)\textrm{~~~with~~~}\overline{ H^{p,q}(M_n)}=H^{q,p}(M_n)\,,
\eeq
and a canonical polarization, i.e., for  $\omega^{p,q}\in  H^{p,q}(M_n)$ and $\omega^{r,s}\in  H^{r,s}(M_n)$ with $p+q=r+s=n$, one has
\beq
\begin{array}{rl}
\displaystyle{\int_{M_n}\omega^{p,q}  \wedge  \omega^{r,s}}  &=0, \qquad  {\rm unless} \quad p=s\ \  {\rm and} \ q=r\, ,  \\ [ 3mm]
\displaystyle{i^{p-q} \int_{M_n}\omega^{p,q}  \wedge  {\overline {\omega^{p,q}}}}& >0, \qquad {\rm for} \ \ \omega^{p,q}\neq 0 \ .
\end{array}
\label{eq:type} 
\eeq

The \emph{periods} of $M_n$ are pairings between the middle homology and the middle cohomology. They carry information about how the Hodge structure varies in the family.  
 One fixes an integral topological  basis $\Gamma_i$, $1\le i\le b_n$ for the middle homology  
 $H_n(M_n,\mathbb{Z})$. The choice of this basis is topological, and it does not depend on the complex structure. In particular, the intersection pairing in this fixed topological basis 
 $\Sigma_{ij}=\Gamma_i\cap \Gamma_j$ is given by an integer $b_n\times b_n$-matrix.\footnote{Here $\cap$ denotes the standard intersection pairing on cycles.} If $n$ is odd, then it is skew-symmetric 
 and can be chosen  to be the standard symplectic pairing ${\bf \Sigma}=\left(\begin{array}{rr} 0 & {\mathbbm 1}\\ -{\mathbbm 1}&0 \end{array}\right)$. 
Instead, if $n$ is even, ${\bf \Sigma}$ is symmetric, and general lattice arguments restrict its form considerably. For example, for  a $\mathrm{K3}$ surface, which is the only topological type of 
 Calabi-Yau  geometry in two complex dimensions ($n=2$), the $(22\times 22)$-dimensional integral matrix ${\bf \Sigma}_{\mathrm{K3}}$ has to have  signature $(19,3)$ and 
 has to be even and self-dual. This only leaves the unique possibility ${\bf \Sigma}_{\mathrm{K3}}=E_8(-1)^{ \oplus 2}  \oplus   \left(\begin{array}{rr} 0 & 1\\ 1&0 \end{array}\right)^{\oplus 3}$, where 
 $E_8(-1)$ denotes the negative of the Cartan matrix of the Lie algebra of $E_8$.

 Concretely the \emph{periods} are the pairing $\Pi: H_n(M_n,\mathbb{Z})\times H^n(M_n,\mathbb{C})\rightarrow \mathbb{C}$ given by the integrals 
 \beq
 \Pi_{ij}=\int_{\Gamma_i} \hat\Gamma ^j\,,
 \label{eq:periods}
 \eeq
 with $\hat\Gamma ^j$ some basis of  $H^n(M_n,\mathbb{C})$. One can fix a basis $\gamma_j\in H^n(M_n,\mathbb{C})$ with 
 $\int_{\Gamma_i} \gamma^j=\delta_i^j$ and  $\int_{M_n} \gamma^i \wedge \gamma^j=\Sigma^{ij}=\Sigma_{ij}$. Then one 
 expands $\Omega(\zs)\in H^n(M_n,\mathbb{C}) $ in terms of the period functions $\Pi_i(\zs)$ in this fixed cohomology  basis:
 \begin{equation} 
 \Omega(\uz)=\sum_i \left(\int_{\Gamma_i} \Omega(\uz) \right) \gamma^i= \sum_{i} \Pi_i(\uz) \gamma^i \,.
 \label{eq:perioddef}   
 \end{equation}  

As an example, for the elliptic curve $E$ with modulus $z$ like in eq.~\eqref{eq:ellipticcurve},
one may choose a symplectic basis   $S^1_a,S^1_b \in  H_1(E,\mathbb{Z})$ , with $S^1_a\cap S^1_b=-S^1_b\cap S^1_a=1$ and $S^1_a\cap S^1_a =S^1_b \cap S^1_b=0$ 
in integral homology,
  and a dual symplectic basis $\alpha, \beta\in H^1(E,\mathbb{Z}) $ with  $\int_E \alpha \wedge \beta=-\int_E \beta \wedge \alpha=1$ 
and $\int_E \alpha \wedge \alpha=\int_E \beta \wedge \beta=0$ in integral cohomology. Then $\Pi_a(z)=\int_{S^1_a} {\rm d}x/y $, 
$\Pi_b(z)=\int_{S^1_b} {\rm d}x/y$ are the well-known elliptic integrals, and one can evaluate eq.~\eqref{eq:Kaehlerpotential} in terms 
of these periods. The elliptic periods can in turn be evaluated in terms of complete elliptic integral of the first kind. 
For the two-loop banana and train-track integrals (also known as the sunrise and the elliptic double-box integrals, respectively), the maximal cuts evaluate to the 
periods of an elliptic curve, cf., e.g., refs.~\cite{Laporta:2004rb,Muller-Stach:2011qkg,MullerStach:2012mp,Bourjaily:2017bsb}. 
Introducing the parameter $\tau=\Pi_b(z)/\Pi_a(z)$ on the upper half-plane, and keeping 
in mind  that  $\Pi_a(z)$ (and $\Pi_b(z)$) are holomorphic,  we can evaluate from eq.~\eqref{eq:Kaehlerpotential}, with $G_{\tau\bar \tau}=\partial_\tau \bar \partial_{\bar \tau} K= 1/(2 {\rm Im}(\tau))^2$,
the famous parabolic metric on the upper half-plane (the Teichm\"uller space of $E$). Moreover, using   $\tau$ as complex structure 
variable one can express the periods of the elliptic curve as modular forms of weight one  in $\tau$~\cite{MR2409678}. 
The Calabi-Yau periods generalize this relationship between periods and maximal cuts to higher-loop banana and train-track 
integrals, cf., e.g., refs.~\cite{Bloch:2014qca,Bloch:2016izu,Klemm:2019dbm,Bonisch:2020qmm,Bourjaily:2018ycu}. 

\paragraph{The variation of the Hodge structure  on the middle cohomology.}
\label{para:variationHodgestructure}

As pointed out, for a given complex structure specified, say, by $\zs_0$, the form $\Omega(\zs_0)$ is of Hodge type $(n,0)$. At $\zs_0$ 
one can, due to (\ref{eq:Hodge_decomposition}), chose  a basis $\tilde \gamma^i_{p,q}$ of specific Hodge type $(p,q)$, $p+q=n$, and express them with constant coefficients  in terms of the 
topologically basis $\gamma^i$ at $\uz_0$, cf. eq.~\eqref{eq:perioddef}.  If we vary the complex structure, 
say $\zs_0\to \zs_0+\eps$, then in particular the $n$-form $\Omega(\zs_0+\epsilon)$ gets admixtures of forms 
of other types compared to the original complex structure at $\zs_0$. The period functions $\Pi_i(\uz)$ 
describe this \emph{variation of Hodge structures}, cf. eq.~\eqref{eq:perioddef}.   One could have 
studied  analogs of eq.~\eqref{eq:perioddef} for all forms of Hodge type  $\tilde \gamma^i_{p,q}$ but for Calabi-Yau 
 manifolds  the period functions $\Pi_i(\zs)$ over the unique holomorphic form $\Omega(\zs_0)$  play 
 a special role, and many seemingly more general questions follow from them. 
 
 Mathematically, one captures this variation of the Hodge decomposition of the middle cohomology $H^n$ in 
 eq.~\eqref{eq:Hodge_decomposition}\footnote{We will often suppress the dependence of the the cohomology 
 groups on the manifold $M_n$, i.e., we simply write $H^m$ for $H^m(M_n)$.} in terms of the so-called  \emph{Hodge 
 filtration of  weight $m=n$} with\footnote{A good and more thorough treatment of the variation of the Hodge structure and the limiting 
mixed Hodge structure  can be found in ref.~\cite{MR2393625} and in relation with mirror symmetry in ref.~\cite{MR1677117}.}
\begin{equation} 
F^p H^m=\bigoplus_{l\ge p} H^{l,m-l} ~,
\end{equation}
such that
\beq
H^m=F^0H^m\supset F^1H^m\supset \cdots \supset F^{m}H ^m\supset F^{m+1}H^m=0\,.
\label{eq:decreasing} 
\eeq
One can recover the Hodge decomposition in eq.~\eqref{eq:Hodge_decomposition} from the Hodge filtration $F^\bullet$ via the relations:
\beq
F^pH^m\oplus \overline{F^{m-p+1}H^m}=H^m \qquad\text{and} \qquad H^{p,m-p}=F^p H^m\cap \overline{{ F^{m-p}}H^m}\ .
\label{eq:recoverHodge1}
\eeq
The Hodge cohomology groups come from what is more generally known as the  \emph{associated graded complex} of the filtered complex
\beq
{\rm Gr}_F^p H^m = F^pH^m/F^{p+1}H^m\cong H^{p,m-p}\ .
\label{eq:recoverHodge2}
\eeq
Unlike the $H^{p,q}$, the $F^p H^m$ vary holomorphically with the complex structure and 
fit into locally free constant sheaves $\mathcal{F}^p$ over ${\cal M}_\text{cs}$,  with the inclusions ${\cal F}^p\subset {\cal F}^{p-1}$.   
This  defines a \emph{decreasing varying Hodge filtration} of the Hodge bundle $ {\cal H}$ for the family
\beq
{\cal H}={\cal F}^0\supset {\cal F}^1\supset \cdots \supset{\cal  F}^{n}\supset {\cal F}^{n+1}=0\ ,  
\label{eq:decreasingsheafs} 
\eeq
where we suppressed the dependence on $H^m$ since we understand that we work always  in  the middle cohomology $m=n$.   
By construction ${\cal F}^0=R^n\pi_*\mathbb{C}\otimes {\cal O}_{{\cal M}_\text{cs}}$ contains a local system, namely the locally constant sheaf $R^n\pi_*\mathbb{C}$. 
This defines a flat connection called the Gauss-Manin  connection 
\begin{equation} 
\nabla:{\cal F}^0\mapsto {\cal F}^0 \otimes \Omega^1_{{\cal M}_\text{cs}} \,.
\label{eq:GaussManin}
\end{equation}  
The flat sections of the Gauss-Manin coincide with the local system $R^n\pi_*\mathbb{C}$, i.e., for $f$ a  holomorphic function on ${\cal M}_\text{cs}$ 
and $s$ a flat section of $R^n\pi_*\mathbb{C}$, one defines $\nabla$ by $\nabla (s\otimes f)=s\otimes \rd f$.   As before, one 
chooses in ${\cal H}$ the locally  constant subsheaf ${\cal H}_{\mathbb{C}}=R^n \pi_* \mathbb{C}$, and within that an integer 
subsheaf ${\cal H}_{\mathbb{Z}}=R^n \pi_* \mathbb{Z}$. The quadruple $({\cal H},\nabla, {\cal H}_\mathbb{Z},{\cal F}^\bullet)$ is called 
a variation of pure Hodge structures.   One of the most important features is 
the \emph{Griffiths transversality} of $\nabla$:
\begin{equation}
\nabla {\cal F}^p\subset {\cal F}^{p-1}\otimes \Omega^1_{{\cal M}_\text{cs}}~.
\label{eq:trans} 
\end{equation} 
A modern language proof of eq.~\eqref{eq:trans} can be found in ref.~\cite{MR2451566}.
In every  coordinate  system that  corresponds to a local trivialization of ${\cal H}$, the connection $\nabla$ 
is  the normal derivative,  and for $\Omega(\uz)\in {\cal F}^n$  one gets in particular~\cite{MR717607} 
\beq 
\partial^{\underline k}_{\underline z} \Omega({\underline z})\in {\cal F}^{n-|k|}~,
\label{eq:transversality}
\eeq     
with $\partial^{\underline k}_{\underline z} \coloneqq\partial_{z_1}^{k_1}\ldots \partial_{z_r}^{k_r} $ and $|k|=\sum_{i=1}^{h^{n-1,1}} k_i$. 
Let us note in passing that one can define a non-holomorphic connection that allows 
one to kill the starting terms from ${\cal F}^{n-|k|}$ in eq.~\eqref{eq:transversality}. For example,  we have $\partial_{z_k}\Omega (\zs)\in {\cal F}^{n-1}={\cal  H}^{n,0}\oplus  {\cal  H}^{n-1,1}$, 
but one checks from eqs.~\eqref{eq:Kaehlerpotential} and~\eqref{eq:bryantgriffith} that the application of 
$D_k=\partial_{z_k}+ (\partial_{z_k} K)$  yields $D_k\Omega({\underline z})\in {\cal  H}^{n-1,1}$. Writing 
down higher iterations of the non-holomorphic connection with this  property is more involved, but possible,
and  is known as \emph{special K\"ahler geometry} for $n=3$, and generalizes to higher dimensions, 
see, e.g., ref.~\cite{MR3965409}.     

Since eq.~\eqref{eq:transversality} is  a cohomological inclusion into the finite-dimensional space ${\cal H}$,
there must be, up to exact terms, linear relations between the derivatives.  The coefficients of these linear relations
turn out to be rational functions in the complex moduli $\uz$. The relations form a finitely-generated differential ideal called the \emph{Picard-Fuchs  differential ideal}. One can find  
its generators ${\cal L}_i$, $i=1,\ldots, r$   concretely, for example by the Griffiths reduction method pioneered 
in ref.~\cite{MR0260733}, or for Calabi-Yau spaces embedded in toric varieties by the \emph{Gel\cprime fand-Kapranov-Zelevinsk\u{\i} 
(GKZ) systems}~\cite{MR1011353}, see also refs.~\cite{MR1316509,Hosono:1994ax}.  A  complete Picard-Fuchs differential ideal 
is equivalent to the flat Gauss-Manin connection. In fact, it often allows one to construct more easily the global flat sections
(and other important structures, like the $n$-points couplings~\cite{MR717607}, see section~\ref{sssect:Griffithstransversality}).  
Since exact terms vanish when integrated over any closed cycle, the period functions are annihilated by the ${\cal L}_i$:
\beq 
{\cal L}_i\Pi_*({\underline z})=0 	\qquad\text{for } i=1,\ldots ,r\ . 
\label{eq:solutions}
\eeq  
Here the index $*$ refers to period integrals in any basis. The differential ideal  is complete if it has $b_n$ independent solutions near any $\uz_0\in{\cal M}_\text{cs}$.

\paragraph{Gauss-Manin connection from Griffiths reduction on elliptic families.}
\label{para:GaussManin}

Let us illustrate the previous concepts on the example of the
Legendre family ${\cal E}_{\textrm{Leg}}$ of elliptic curves\footnote{It can be transformed into Weierstrass form with 
$g_2=1/27(2-3 z-3 z^2+2 z^3)$ and $g_3=(2^\frac{2}{3}/3)(1-z+z^2)$.} 
\begin{equation} 
y^2=x(x-1)(x-z)\,.
\end{equation} 
Let us start by deriving the single Picard-Fuchs operator for the Legendre family
using the Griffiths reduction method. First we  write the elliptic curve as the vanishing locus of the 
homogeneous cubic  $P=w y^2- x(x-w)(x-z w)$ in  the projective space $\mathbb{P}^2$. The  
holomorphic form $\rd x/y$, which is valid in the  patch $w=1$ of $\mathbb{P}^2$, can be written 
more symmetrically in the Griffiths residue form:\footnote{A preferred K\"ahler gauge choice of    
$\Omega$ is to replace $\mu_l$ with $a_0 \mu_l$ in eq.~\eqref{Griffith residue form}, where $a_0$ is the coefficient of the symmetric perturbation $x_1\cdots x_{l+1}$ in $P$ 
if the latter exists, see ref.~\cite{MR1115626}. The so-defined $\Omega$ posseses additional symmetries which can be used to derive the Gel\cprime fand-Kapranov-Zelevinsk\u{\i} 
differential ideal for the Calabi-Yau periods in preferred complex structure  coordinates, see refs.~\cite{MR3965409,Klemm:2019dbm}.}
\beq\label{Griffith residue form}
\Omega(z)=\frac{1}{2 \pi i}\oint_\gamma \frac{\mu_2}{P}\,, 
\eeq
where $\gamma$ is a cycle 
around $P=0$ and the measure is the standard volume form on $\mathbb{P}^2$. Since 
this definition of the $(l-1,0)$-form $\Omega$ is valid for hypersurfaces  in $\mathbb{P}^l$, we report the general volume form on $\mathbb{P}^l$:   
\begin{gleichung}
\mu_{l} = \sum_{k=1}^{l+1} (-1)^ {k+1} x_k \ \dd x_1 \wedge  \ldots \wedge {\widehat {\dd x_k} } \wedge  \ldots \wedge \dd x_{l+1}~ ,
\label{measure}
\end{gleichung}
where the hat indicates the omission of one differential. Note that  $\mu_l/P$ is well-defined, because, e.g. for $l=2$, 
it is invariant under non-zero $\lambda\in \mathbb{C}^*$ scalings $(w,x,y)\mapsto  (\lambda w,\lambda x,\lambda y)$ in  $\mathbb{P}^2$. If we set  $w=1$ and evaluate the $\rd y$ integration using 
the residue at $P=0$, we get $1/2 \oint_\gamma \rd x\wedge \rd y/P=\rd x/y$, i.e., we recover (up to an overall constant), the familiar form of the holomorphic one-form on an elliptic curve. Note that  the residue form in eq.~\eqref{Griffith residue form} generalizes naturally to Calabi-Yau 
hypersurfaces $P=0$ in  weighted projective spaces or toric varieties as ambient space in any dimensions.  The residue form of $\Omega(\zs)$  
in eq.~\eqref{Griffith residue form} also generalizes  to complete intersection Calabi-Yau spaces  $P_1=P_2=\ldots =P_r=0$ in products of projective spaces 
in $\otimes_{i=1}^m \mathbb{P}_{(i)}^{n_i}$. The latter are Calabi-Yau $(\sum_{i=1}^m n_i)-r)$-folds, if the sum of the degrees $d_{kl}$ of the homogeneous polynomial 
$P_k$ w.r.t. to the homogeneous coordinates of $\mathbb{P}_{(l)}^{n_l}$ fullfills $\sum_{k=1}^r {d_{kl}}=n_l+1$, $\forall\,  l=1,\ldots,m$.  
Let  $\gamma_k$  be a loop encircling $P_k=0$, then the holomorphic $(n,0)$-form is given as 
\beq
\Omega(\uz)=\frac{1}{(2\pi i)^r}\oint_{\gamma_1}\ldots \oint_{\gamma_r}\frac{\wedge_{i=1}^{m} \mu_{n_i}}{P_1\cdots P_r}\ ,
\label{eq:Omcicy}
\eeq   
see ref.~\cite{MR0260733} for projective spaces and ref.~\cite{MR3965409}  for an overview of the 
measure $\mu$  for more general ambient spaces. Both  the hypersurface-- and  the complete intersection representation 
will be useful for the maximal cut of the banana integrals, as we explain at the end of this section. 

Now, let $\partial_k$ be the derivative w.r.t. the $k^{\textrm{th}}$ coordinate (in $\underline x$)
of the projective space. Then $\partial_k(q({\underline x},z) \mu_2/P^r)$ is exact if and only if $q({\underline x},z)$ 
is chosen so that the expression is scale-invariant~\cite{MR0260733}.  In other words, under the integration over $\gamma$ and 
closed cycles $\Gamma$ of $E$,   the following partial integration formula holds:
\beq
 \frac{r   q({\underline x},z) \partial_k P}{P^{r+1}}\mu_2= \frac{\partial_k q({\underline x},z)}{P^r}  \mu_2\,.
\label{eq:partialintegration}
\eeq 
With the notations $\Pi(z)=\int_\Gamma\Omega(z) =\int_\Gamma \oint_\gamma \frac{\mu_2}{P}$ and $(\cdot)'=\partial_z(\cdot)$,
one gets for the Legendre curve after a short calculation:
\begin{equation}
\begin{aligned}
	\Pi''(z)	&=	\int_{\Gamma}\frac{2(w^2 x-w x^2)^2}{P^3}\,\mu_2	\\
			&=	-\frac13 \int_{\Gamma} \frac{2(w^4-2w^3x)\partial_wP+2(w^2x^2-2w^3x)\partial_xP+(2w^2xy-w^3y)\partial_yP}{P^3}\,\mu_2	\\
			&=	-\frac{1-2z}{z(1-z)}\,\Pi'(z)+\frac1{2z}\int_{\Gamma}\left(\frac{w\partial_wP}{1-z}+w\partial_xP-\frac{y\partial_yP}{2(1-z)}\right)\frac{\mu_2}{P^2}	\\
			&=	-\frac{1-2z}{z(1-z)}\,\Pi'(z)+\frac1{4z(1-z)}\,\Pi(z)~.
\label{eq:firstorderform}
\end{aligned}
\end{equation}
To get the second equality, i.e., to write  $-2 (w^2x\!-\!w x^2)^2=\sum_k q_k({\underline x},z) \partial_k P$, one might use the \emph{Gr\"obner basis algorithm}~\cite{MR3330490}. The third equality is obtained 
using the partial integration rule for the integrand, which reduces the power of $P$ in the denominator. The resulting numerator  is expressed (using again if necessary  the Gr\"obner basis algorithm) in 
the form  $q(z) \partial_z P + \sum_k q_k({\underline x},z) \partial_k P$. This allows one to identify the  $\Pi'(z)$ term and to use partial integration on the rest.  With the logarithmic derivative $\theta=z\,\partial_z$,  we get 
the following differential operator generating the Picard-Fuchs ideal
\begin{equation}
{\cal L}_{\text{Leg}}= z^2 (1-z)\partial_z^2+z(1-2 z)\partial_z-\frac{z}{4} = \theta^2-z \left(\theta  +\frac{1}{2}\right)^2\ ,
\label{eq:difflegendre} 
\end{equation} 
and the Picard-Fuchs differential equation for the Legendre family ${\cal E}_{\textrm{Leg}}$ is
\beq
{\cal L} _{\text{Leg}}\Pi(z)=0\,.
\label{eq:Legendre_PF}
\eeq
It clearly has $b_1(E)=2$ solutions.  
 Our short computation illustrates the general 
strategy based on eq.~\eqref{eq:transversality}. In particular,  $\partial_z \Omega(z)$ generates an 
 expression that can be written, up to exact terms, in terms of the holomorphic $(1,0)$-form $\Omega(z)$ itself  
 and a differential of the second kind, which is a meromorphic one-form with vanishing residues. These span 
 ${\cal H}({\cal E})$, in agreement with eq.~\eqref{eq:transversality}. Taking the second derivative $\partial^2_z \Omega(z)$ closes 
 on this space, and produces up to exact terms a linear relation, leading to the Picard-Fuchs operator in eq.~\eqref{eq:difflegendre}. At this point we stress an important point: arguing that 
 these meromorphic differentials, which have a well-defined pairing with homology,  equivalently encode the Hodge 
 decomposition in eqs.~\eqref{eq:recoverHodge1} and~\eqref{eq:recoverHodge2} at a particular point in ${\cal M}_\text{cs}$ is a non-trivial step. 
By taking derivatives of $\Omega$  with respect to the  complex moduli one of course never literally obtains for example  $\bar \Omega$. 
 Rather, one has to use  \emph{hypercohomology theory} to show that the homology spanned by the emerging 
 meromorphic differential is cohomologically  equivalent to the Hodge cohomology.  This application of  
 hypercohomology  is described in ref.~\cite{MR1288523}, see also ref.~\cite{MR756842} for a simplified physical explanation.  
 Essential for us is  the  general  fact that among the derivatives of $\Omega(\zs)$ w.r.t to the moduli $\uz$ one can find a 
 basis of ${\cal H}$ modulo the relation in eq.~\eqref{eq:partialintegration}.\footnote{More precisely,  the horizontal subspace  ${\cal H}_\text{hor}({\cal M},\mathbb{C})$  of 
 ${\cal H}({\cal M},\mathbb{C})$. This distinction becomes important for higher-dimensional Calabi-Yau manifolds.}
 Integrals  over these cohomology elements correspond  to independent master integrals for the maximal cuts in some integer
 dimension. Geometrically,  they can easily be made explicit. In particular, for hypersurfaces in an $n$-dimensional (weighted) projective ambient space, let $P({\underline{x}})$  be of  (weighted) degree $d$ in the coordinates  $x_1,\ldots, x_{n+1}$, and $\mathrm J(P)$ is the Jacobian ideal, i.e., the polynomial ideal generated by the partial derivatives $\partial_{x_i}P$ at a smooth point of $P$. Then we can consider a finitely generated ring of deformations of $P$: 
 \beq
 \label{eq:ring} 
 {\cal R}(P) =\frac{\mathbb{C}[x_1,\ldots,x_{n+1}]}{{\rm J}(P)}.
 \eeq     
 This ring is graded  with degrees $d_k=k d$ with $k=0,\ldots, n$.  Note that the dimension of ${\cal R}(P)$ and its graded pieces are constant for 
 smooth loci in ${\cal M}_\text{cs}$, and they do of course not depend on reparametrizations of the ambient space. Hence, 
 the simplest smooth configuration of $P_0$ at a smooth reference point, say  ${\underline z}_s$, can be used to study generic 
 properties of ${\cal R}(P)$. Using this freedom, one can choose a monomial basis as its generators. Let us denote 
 generators of degree $d_k$ by $m^{(k)}_i(x)$, $i=1,\ldots,\#_k$. Then one can represent  a basis of differentials as:
 \beq
\label{Griffith residue form2}
\Omega^{(k)}_i= \frac{1}{2 \pi i}\oint_\gamma \frac{ m^{(k)}_i(\underline{x}) \mu_{l+1}}{P^k}\,, \quad i=1,\ldots, \#_k \ .
\eeq
Exept for $\Omega^{(0)}_1=\Omega$,  these differentials are meromorphic, but hypercohomology allows one 
to relate $\Omega^{(k)}_i$ to the forms generating $F^{n-k}H^n(M)$ for $i=1,\dots, h_{\text{hor}}^{n-k,k}=\#_k$.  Via eq.~\eqref{eq:isom} the complex structure  
deformations or local coordinates of ${\cal M}_\text{cs}$ are naturally identified with the deformations $\uz$ in 
\beq
P=P_0+ \sum_{i=1}^{h^{2,1}_\text{hor}} z_i m_i^{(1)}({\underline x})\ .             
 \eeq
For example for the generic cubic in $\mathbb{P}^2$ (see eq.~\eqref{eq:Pbanana}), one  can choose $m^{(0)}_1(\underline{x})=1$ and $m^{(1)}_1(\underline{x})=x y w $ (the former gives rise to the differential of the first kind on the elliptic curve, while the latter represents a differential of the second kind).  
The above concepts generalize to complete intersections. In particular, eq.~\eqref{eq:partialintegration} becomes  for complete 
 intersections~\cite{MR0260733}
\begin{equation}
\sum_{k\ne j}{m_k\over m_j-1}{P_j\over P_k}
{q\partial_{x_i} P_k\over \prod_{l=1}^r P_l^{m_l}}\mu\ =
{1\over m_j-1} {P_j \partial_{x_i} q\over \prod_{l=1}^r P_l^{m_l}}\mu -{q \partial_{x_i} P_j\over \prod_{l=1}^r P_l^{m_l}}\mu \ ,
\label{partialintegrationcicy}
\end{equation}
where $r$ is defined as in eq.~\eqref{eq:Omcicy} and  $q(\underline{x})$ are polynomials of the appropriate degree to ensure the scale invariances and $\mu$ 
is a suitable generalization of the measure in eqs.~\eqref{measure} and \eqref{eq:Omcicy} for the ambient space under consideration.  An example for the application of 
 the Griffiths reduction method for an elliptic curve represented as a complete intersection  can be found in  ref.~\cite{MR1287096}.

Note that eq.~\eqref{eq:Legendre_PF} is the famous Gauss hypergeometric system with solution $\Pi(z)={}_2F_1\left( \frac{1}{2},\frac{1}{2},1;z\right)$. More generally, Calabi-Yau  manifolds embedded canonically into toric ambient spaces have Picard-Fuchs differential ideals given by resonant  GKZ systems~\cite{MR1011353,MR1316509,Hosono:1994ax}, 
 which is one canonical extension of hypergeometric systems to the multi-parameter case.  The canonical embedding 
 of the  Legendre curve in this sense would be by a complete intersection of  two quadrics in $\mathbb{P}^3$, which 
 we denote  by the short form $E= \left(\mathbb{P}^3||2,2\right)$. The residue form  of the holomorphic form $\Omega$ is given in eq.~\eqref{eq:Omcicy}, 
 and the Griffiths reduction method generalizes using  eq.~\eqref{partialintegrationcicy} to complete intersections~\cite{MR0260733}. The explicit  reduction for the case  at hand
 can be found in ref.~\cite{Klemm:1993jj}, where it  is compared for Calabi-Yau $n$-folds with the generalized hypergeometric GKZ  systems. 
 
 Let us now illustrate how the Picard-Fuchs equation~\eqref{eq:Legendre_PF} can be used to obtain the Gauss-Manin connection $\nabla$ for the Legendre family $\mathcal E_\text{Leg}$. If we define $\underline{I}(z) =(\Pi(z),\partial_z \Pi(z))^T$ and $\underline{I}^\theta(z)=(\Pi(z),\theta \Pi(z))^T$, then the Picard-Fuchs equation $\mathcal{L}_{\textrm{Leg}}\Pi(z)=0$ can be cast in the form:
 \begin{equation}
 \partial_z \underline{I}(z)={\bf A}(z)\underline{I}(z) \quad {\rm or} \quad \theta \underline{I}^\theta(z)={\bf A}^{\theta}(z)\underline{I}^\theta(z)   
 \label{eq:firstorderform}
 \end{equation}  
 with
  \begin{equation}
{\bf A}(z)=\left(\begin{array}{cc} 0 &1\\  \frac{1}{4 z (1-z)}&\frac{2 z-1}{z (1-z)} \end{array}\right) \quad {\rm or} \quad  {\bf A}^\theta(z)=\left(\begin{array}{cc} 0 &1\\  \frac{z}{4 (1-z)}&\frac{z}{1-z} \end{array}\right) \ .    
 \label{eq:firstorderform_A}
 \end{equation}     
We see that the Gauss-Manin connection in eq.~\eqref{eq:firstorderform} is indeed very reminiscent of the homogeneous system of first-order differential equations satisfied by the maximal cuts of the master integrals for $\eps=0$, see for example eq.~\eqref{eq:homogeneous_equation}.
 
 As a second example we discuss the maximal cuts of the two-loop banana integrals. 
 At two loops, the geometry associated to the banana graph is a family $\mathcal{E}_{\textrm{ban}_2}$ of elliptic curves, which can be described as the vanishing locus in $\mathbb{P}^2$ of the polynomial (with $1/z=p^2/m^2$) 
 \begin{equation} 
 P= (3-1/z)w x y+w^2 x+w^2 y+w x^2+w y^2+x^2 y+x y^2\ .
 \label{eq:Pbanana}
 \end{equation} 
 Note that this polynomial is closely related to the second Symanzik polynomial of the banana graph, cf. eq.~\eqref{eq:banana_F_pol}.
 We can derive the Picard-Fuchs differential equation as before, but 
 let us point out a slightly more universal way.  Every elliptic curve can be brought into the Weierstrass form in eq.~\eqref{eq:ellipticcurve} by Nagell's algorithm.\footnote{See ref.~\cite{MR3965409}, appendix 1,
 for explicit formulas that bring cubics in $\mathbb{P}^2$  (used here), bi-quadrics in $\mathbb{P}^1\times \mathbb{P}^1$ 
 and all quartics in  $\mathbb{P}(1,1,2)$ into Weierstrass form.} To apply the Griffiths reduction method  to  eq.~\eqref{eq:ellipticcurve}, 
 we introduce the notations:
 \beq\bsp\label{eq:Legendre_shorthands}
 \Delta=g_2^3-27 g_3^2\,,\quad &\, \beta=3 g_2' g_3-2 g_2 g_3'\,,\\
  \rho=3\beta/(2 \Delta)\,, \quad &\,
 \kappa=\log(\Delta)\,,\quad  \alpha=\log(\rho)\,.
 \esp\eeq
 After a short computation we arrive at the following general expression for the Gauss-Manin connection:\footnote{Often, besides  $\frac{\rd x}{y}$, the second basis element 
 $\frac{x \rd x}{y}$ is used as  meromorphic differential of the second kind, i.e., one with no non-vanishing residues. We note the relation 
 $\partial_z   \frac{ \rd x}{y}=-\frac{1}{12} \kappa' \frac{\rd x}{y}-\rho \frac{x \rd x}{y} +\frac{\rd}{\rd x}\left(\frac{18 \beta x^2-\Delta' x -3 g_2 \beta}{6 y \Delta}\right) \rd x$.}     
 \begin{equation}
 \partial_z \underline{I}(z)=\left(\begin{array}{cc} 0 &1\\  \displaystyle{\frac{1}{12}\left( \kappa' \alpha'-\kappa''-\rho^2 g_2 +\frac{(\kappa')^2}{12}\right)}& \alpha' \end{array}\right) \underline{I}(z) \ .    
 \end{equation}         
 For eq.~\eqref{eq:Pbanana} one gets $g_2=(3 z-1) (3 z^3-3 z^2+9 z-1)$ and $g_3=\sqrt{3}(1-6z-3 z^2 ) (9 z^4-36 z^3+30 z^2-12 z+1)/9$. 
 After factoring out a trivial factor, one arrives at the Picard-Fuchs differential operator:
 \beq\bsp
 {\cal L}_{2}&\,=  (1-9 z)(1-z) z^2\,\partial_z^2 -z(1-9z^2) \partial_z+(1-3z)\\
 &\,	=	 (1 - 9 z) (1 - z) \theta^2-(2-10z)\theta+1-3z \ .
 \label{eq:diffbanana}
\esp\eeq
 This homogeneous differential operator annihilates the maximal cut of the two-loop banana integral.
 
 Let us conclude this discussion by an important comment. While the Picard-Fuchs operator ${\cal L}_{2}$ annihilates the maximal cuts of the two-loop banana integrals, it does not annihilate the complete Feynman integral. Instead, ${\cal L}_{2}$ needs to be supplemented by an inhomogeneous term related to the tadpole master integrals (which vanish on the maximal cut). Alternatively, the corresponding integration domain is not closed in the elliptic curve defined by eq.~\eqref{eq:Pbanana}. The corresponding 
 period integral is then called a \emph{relative period}.  This adds one inhomogeneous solution to the $b^\text{hor}_n(M)$ solutions of the Picard-Fuchs differential ideal. The 
 precise  linear combination of homogeneous  and inhomogeneous solutions that correspond to the Feynman 
 integral or the relative period has still to be determined, which happens in  generality for the 
 banana integrals in eq.~\eqref{eq:fullFeynamasymptotic}. Mathematically, adding boundary terms to a Calabi-Yau 
 Picard-Fuchs  differential ideal that describes a variation of Hodge structures  is known as an 
 \emph{extension} of the latter.   
 
 \subsubsection{Cuts of banana integrals as periods of  Calabi-Yau motives}
 \label{ssec:cutsbanana}
 A convenient starting point to relate the banana inte.grals to Calabi-Yau geometries is to consider their Feynman parameter representation in $D=2-2\eps$ dimensions:
\begin{equation}
	 I_{\unu}(p^2,\underline m^2;D)	=	\int_{\sigma_{l}}\left(\prod_{k=1}^{l+1}x_k^{\delta_k}\right)\dfrac{\mathcal{U}^{\omega- \frac{D}{2}}}{\mathcal{F}(p^2,\underline m^2)^{\omega}}\mu_{l}\,,
\label{bananageneralgeneric}
\end{equation}
where we defined $\nu_i = 1+\delta_i$ and $\omega \coloneqq \sum_{i=1}^{l+1} \nu_i - \frac{l D}{2}-1+l\epsilon + \sum_i \delta_i$. The two Symanzik polynomials for the banana graph are given by:
\begin{align}
\mathcal{U} &= \left( \prod_{i=1}^{l+1}x_i \right) \left(\sum_{i=1}^{l+1}\frac{1}{x_i} \right) = \sum_{i=1}^{l+1}\prod_{\substack{
   j=1 \\
   j\neq i
  }}^{l+1} x_j\,, \\
  \label{eq:banana_F_pol}
\mathcal{F}(p^2,\underline m^2) &= -p^2 \left( \prod_{i=1}^{l+1} x_i\right) +  \left( \sum_{i=1}^{l+1}m_i^2 x_i \right) \, \mathcal{U} \,.
\end{align}
The edge variables $x_i$ form a set of homogeneous coordinates for the projective space $\mathbb{P}^l$, and the $l$-real-dimensional 
integration domain $\sigma_l$ is defined as:
 \begin{equation}
 \sigma_l=\{[ x_1:\ldots: x_{l+1} ] \in \mathbb{P}^l \, | \, x_i \in \mathbb{R}_{\ge 0} \ \text{for all} \ 1 \leq i \leq l+1 \}  ~. 
\label{sigmal}
\end{equation}  

Let us briefly illustrate how we can identify the Calabi-Yau geometry associated to the $l$-loop banana graphs. Following ref.~\cite{Vanhove:2018mto}, a maximal cut of the banana integral in $D=2$ dimensions can be obtained by replacing the integration contour in eq.~\eqref{sigmal} by the $l$-dimensional torus 
\beq
T^l \coloneqq \{[ x_1:\ldots: x_{l+1} ] \in \mathbb{P}^l \, | \, |x_i|=1 \ \text{for all} \ 1 \leq i \leq l+1 \}\,.
\eeq
Using the notation introduced in eq.~\eqref{eq:equal-mass_cuts}, there is a cycle $\Gamma_T$ (in loop momentum space) which corresponds to the cycle $T^{l}$ in Feynman parameter space, so that  
\begin{equation}\label{eq:max_cut_feyn_par}
	J_{l,\underline 0}^{\Gamma_T}(\uz;0)=	\int_{T^l} \frac{\mu_l}{\mathcal{F}(1,\uz)}=
	\int_{T^{l-1}} \oint_\gamma \frac{\mu_l}{\mathcal{F}(1,\uz)}=	2\pi i\int_{T^{l-1}} \Omega_{l-1}(\uz)~.
\end{equation}
Now we note that one $S^1$ cycle  of $T^l$ in eq.~\eqref{eq:max_cut_feyn_par} can be identified with $\gamma$ in eq.~\eqref{Griffith residue form},  so that the 
resulting form $\Omega_{l-1}(\uz)$ can be identified as the holomorphic $(l-1,0)$-form of the  Calabi-Yau hypersurface, defined as the vanishing locus 
\begin{equation}
	M^\mathrm{HS}_{l-1}	=	\{ \ux \in \mathbb{P}^l|\mathcal F(1,\uz;\ux)	=	0	\}	~.
\label{eq:HSgeom}
\end{equation}
This hypersurface geometry is singular in the physical sub-slice. Nevertheless  all $l$ residua of the $S^1$ contours 
in eq.~\eqref{eq:max_cut_feyn_par} can be formally performed and  one gets 
\beq
J_{l,\underline 0}^{\Gamma_T}(\uz;0)=(2\pi i)^l\sum_{n=0}^\infty \sum_{|k|=n}\left(n\atop {k_1\ldots k_{l+1}}\right)^2 \ \prod_{i=1}^{l+1} z_i^{k_i} \ ,  
\label{cutMHS}
\eeq    
with $|k|=\sum_{i=1}^{l+1} k_i$.

The period integral in eq.~\eqref{eq:max_cut_feyn_par} is only one possible maximal cut of $J_{l,\underline 0}(\uz;0)$, and this  
period can easily be evaluated explicitly (cf.~eq.~\eqref{cutMHS}), even in the singular geometry in eq.~\eqref{eq:HSgeom}.  However, as we have seen in section~\ref{sec:deqs}, we 
would like to know the complete Wronskian ${\bf W}(\zs)$, or equivalently the vector of period functions $\uPi(\zs)$, preferrably in an integral basis that allows us to identify the periods with the maximal cuts. For this it is 
convenient to have a smooth model. For example, the geometric invariants that enter the  $\widehat \Gamma$-class evaluation, can be properly  
defined only for a smooth model. In ref.~\cite{Klemm:2019dbm} the smooth model was obtained in eq.~\eqref{eq:HSgeom} by considering the  
 reflexive Newton polytope associated to eq.~\eqref{eq:banana_F_pol}. The smooth model is a deformation of the singular configuration 
in eq.~\eqref{eq:banana_F_pol}, and its toric ambient space as well as its mirror manifold  could then be obtained using Batyrev's toric mirror construction.   For 
example, for the two-loop case ($l=2$), this yields the elliptic curve studied in ref.~\cite{MR3780269}. The three-loop integral geometry with 
$\chi(M_2^\text{HS})=24$ is  a K3  with 9 complex structure  deformations and a Picard group of rank 11. For the four-loop integral one gets a Calabi-Yau threefold with 
Euler number $\chi(M_3^\text{HS})=20$, $h_{21}=16$ and $h_{11}=26$, while for the five-loop integral the Calabi-Yau fourfold has 
$\chi(M_4^\text{HS})=540$, $h_{31}=25$ and $h_{11}=57$, $h_{21} =0$ and $h_{22}=422$.  It is easy to see that the complex structure 
deformations grow like $h_{l-2,1}=l^2$. Even though these complex deformations grow much faster than the $l+1$ physical 
parameters $\zs$, the GKZ system  induced from the toric ambient space and the large radius coordinates made it  possible to construct from 
the toric description of eq.~\eqref{cutMHS} the Picard-Fuchs differential ideal  in ref.~\cite{Klemm:2019dbm} for three and 
four loops,  and for general  masses in ref.~\cite{Bonisch:2020qmm}. This was done by restricting the GKZ system 
to the physical sub-slice and determining the inhomogeneous  term 
by an ansatz confirmed by extensive numerical integration of eq.~\eqref{bananageneralgeneric}. Clearly, this method becomes  
more cumbersome at higher loops, because  the physical parameter space becomes an ever  tinier subspace in the 
complex moduli space of the higher-dimensional  hypersurfaces in eq.~\eqref{eq:HSgeom}, 
with an even more redundant middle homology and cohomology. 

A very elegant way to circumvent this problem was proposed in ref.~\cite{Bonisch:2020qmm}. One considers the complete intersection of two polynomials of degree 
$(1,\ldots,1)$ in 
\begin{equation}
	\mathbb{P}_{l+1}	\coloneqq	\bigotimes_{i=1}^{l+1} \mathbb{P}_{(i)}^1	\, , 
\label{Pl+1} 
\end{equation}
i.e., we have
\begin{equation}
M^\text{CI}_{l-1}\!=\!\Biggl\{\Bigl(w_1^{(i)}:w_2^{(i)}\Bigr)\in \mathbb{P}_{(i)}^1,\forall i  \biggr|P_1\!\coloneqq\!\sum_{i=1}^{l+1} a^{(i)} w_1^{(i)} + b^{(i)} w_2^{(i)} =\sum_{i=1}^{l+1} c^{(i)} w_1^{(i)} + d^{(i)} w_2^{(i)}\!\eqqcolon\!P_2 \!=\!0\Biggr\}\,.
\label{MCicy} 
\end{equation}    
Such a complete intersection manifold in a product of manifolds is denoted in short as
\begin{equation} 
M^\mathrm{CI}_{l-1}=
\left(\begin{array}{c} \mathbb{P}^1_1\\ \vdots \\      \mathbb{P}^1_{l+1}\end{array}\right|\!\! \left|  \left. \begin{array}{cc} 1&1 \\ \vdots &\vdots \\ 1& 1\end{array}\right\} \ l+1\right) \ \ \ \subset \ \ \
\left(\begin{array}{c} \mathbb{P}^1_1\\ \vdots \\      \mathbb{P}^1_{l+1}\end{array}\right|\!\! \left|  \left. \begin{array}{c} 1\\ \vdots  \\1\end{array}\right\} \ l+1\right)=F_l\subset \mathbb{P}_{l+1}~.
\label{CICY} 
\end{equation} 
To be a transversal complete intersection, i.e., $\mathrm dP_1\wedge \mathrm dP_2\neq 0$ when $P_1=P_2=0$, we have to have  
{\footnotesize ${\rm det}\left(\begin{array}{cc} a^{(i)}& b^{(i)}\\ c^{(i)}& d^{(i)}\end{array}\right)\neq 0$} for all $i=1,\ldots,l+1$. Moreover, on every 
$\mathbb{P}^1_{(i)}$ there is a natural $\mathrm{SL}_i(2,\mathbb{C})$ action  which allows one to eliminate three of the four deformation 
parameters $a^{(i)},b^{(i)},c^{(i)}$ and $d^{(i)}$. With the choices 
\begin{align} 
\nonumber a^{(i)}&=-\frac{m_i^2}{p^2}=-z_i,  \quad d^{(i)}=x, \quad i=1,\ldots, l+1,\\
\label{eq:birationalmap} b^{(1)}&=\frac{x}{w_2^{(1)}},\quad c^{(1)}=\frac{1}{w_1^{(1)}},\\
\nonumber b^{(i)}&=c^{(i)}=0,\quad i=2,\ldots, l+1,  
\end{align}
we can construct a birational  map from the smooth geometry in eq.~\eqref{MCicy} to the singular hypersurface 
geometry in eq.~\eqref{eq:HSgeom}. Solving for $P_1=0$ one  gets $x=\sum_{i=1}^{l+1} \frac{m_i^2}{p^2} w_1^{(i)}$, 
while $P_2$ becomes $P_2=1-x\sum_{i=1}^{l+1} w_2^{(i)}$. Passing to toric coordinates 
$w_1^{(i)}=W_i$ and $w_2^{(i)}=1/W_i$, for $i=1,\ldots, l+1$, we arrive at 
\begin{equation}  
P_2=p^2-\left(\sum_{i=1}^{l+1} m_i^2\, W_i\right)\left(\sum_{i=1}^{l+1}\frac{1}{W_i}\right)\,  ,
\label{eq:singbananatoric}
\end{equation} 
which is eq.~\eqref{eq:HSgeom} written in toric coordinates. Having found a birational map from eq.~\eqref{MCicy} to  eq.~\eqref{eq:HSgeom}, it is also  
interesting to establish that  $J_{l,\underline 0}^{\Gamma_T}(\uz;0)$ comes out correctly in the representation in eq.~\eqref{MCicy}. Using eq.~\eqref{eq:Omcicy} 
and setting $x=1$, we write the latter as
\begin{align}
\label{eq:cicyT}
&\frac{1}{(2\pi i)^l}\,J_{l,\underline 0}^{\Gamma_T}(\uz;0)	=	\frac{1}{(2\pi i)^{l+1}}\oint_{\gamma_1} \oint_{\gamma_2} \frac{\mu}{P_1 P_2}	\\ \nonumber
&\hspace{5pt}	=	\frac{1}{(2\pi i)^{l+1}}\oint_{W_1=0}\!\!\!\! \cdots\oint_{W_{l+1}=0} \mu{\left(1-\sum_{i=1}^{l+1} z_i\, W_i\right)^{-1}\left(1-\sum_{i=1}^{l+1}
\frac{1}{W_i}\right)^{-1}} \frac{\rd W_1}{W_1}\wedge \ldots \wedge \frac{\rd W_{l+1}}{W_{l+1}}	\\ \nonumber
&\hspace{5pt}	=	\frac{1}{(2\pi i)^{l+1}}\oint_{W_1=0}\!\!\!\! \! \cdots\oint_{W_{l+1}=0}\!\!\! \mu\!\!\! \sum_{|m|=m \atop {|n|=n}}\!\!
\left(m\atop {m_1\ldots m_{l+1}}\right) 
\left(n \atop {n_1\ldots n_{l+1}}\right) 
\prod_{i=1}^{l+1} (z_i\, W_i)^{m_i} 
\left( \frac{1}{W_i}\right)^{n_i}	\\ \nonumber
&\hspace{5pt}= \sum_{n=0}^\infty \sum_{|k|=n}\left(n\atop {k_1\ldots k_{l+1}}\right)^2 \ \prod_{i=1}^{l+1} z_i^{k_i}\ . 
\end{align}
We see that we get exactly the same expression for the torus period as in eq.~\eqref{cutMHS} for the hypersurface. 
Here $\mu=1/2^{l+1}\mu^{(1)}_1\wedge \ldots \wedge \mu^{(l+1)}_1$, where $\mu^{(i)}_1$ is the standard measure
for $\mathbb{P}^1_i$ from eq.~\eqref{measure}. Note that under the identification $w_1^{(i)}=W_i$ 
and $w_2^{(i)}=1/W_i$ the measure becomes $\mu = \frac{\rd W_1}{W_1}\wedge \ldots \wedge \frac{\rd W_{l+1}}{W_{l+1}}$. The last identity is obtained by
performing all the residues.

There are at least four reasons why the smooth complete intersection representation in eq.~\eqref{MCicy}  is superior to the hypersurface 
representation in eq.~\eqref{eq:HSgeom}. First, after using the  $\mathrm{SL}_i(2,\mathbb{C})$ symmetries, it contains exactly the right number 
of deformations which are very easily identifiable with the physical parameters $\uz$. Second, and even more importantly, 
eq.~\eqref{MCicy} defines a natural closed submotive (see below) $H^\text{hor}_{l-1}(M^\text{CI}_{l-1})$ and  $H_\text{hor}^{l-1}(M^\text{CI}_{l-1},\mathbb{Z})$
of the total cohomology and homology of  $M^\text{CI}_{l-1}$.\footnote{Which similar to the cohomology and homology  of $M^\text{HS}_{l-1}$ is much 
bigger then the desired  physical  sub-motive.}  In particular $H_\text{hor}^{l-1}(M^\text{CI}_{l-1})$ is generated by taking derivatives 
of eq.~\eqref{eq:Omcicy} w.r.t. to the independent deformation parameters  in eq.~\eqref{CICY} modulo eq.~\eqref{partialintegrationcicy}.
Among this cohomology group we can identify the integrands of the master integrals in $D=2$ dimensions, 
while among the dual homology group   $H^\text{hor}_{l-1}(M^\text{CI}_{l-1})$ we find a basis of different maximal cut  contours. Third, it allows one to extract the Picard-Fuchs differential ideal with 
general  masses straightforwardly as a simpler GKZ system only in the physical masses, 
as  pioneered for these cases in ref.~\cite{Hosono:1994ax}. The last point  is that, according to ref.~\cite{Hosono:1994ax}, mirror symmetry maps 
the horizontal middle cohomology $H_\text{hor}^{l-1}(M^\text{CI}_{l-1})$ to the vertical cohomology  $H_\text{vert}^{k,k}(W^\text{CI}_{l-1})$, i.e., the one  
that is inherited from the ambient space, and the corresponding middle homology $H_\text{hor}^{l-1}(M^\text{CI}_{l-1})$ 
to  the even homology $H^\text{vert}_\text{even}(W^\text{CY}_{l-1})$, that is obtained by restricting the Chow group of the 
ambient space,  on the same manifold. If we restrict ourselves to these vertical-- and horizontal subspaces of homology and cohomology parametrized by the physical subspace of the  
moduli spaces and denote the restriction by the superscript $\text{res}$  then the restricted  complete intersection geometries for the 
banana graphs are self mirrors 
\beq
M_{l-1}^\text{CI, res}=W_{l-1}^\text{CI, res}\ .
\label{eq:simplemirror}
\eeq
The latter fact allowed some of us in ref.~\cite{Bonisch:2020qmm} to find the full banana integral in $D=2$ dimensions using the $\widehat \Gamma$-classes of the mirror geometry $W_{l-1}^\text{CI, res}$ 
in the large volume regions for the  full physical parameter space and to specify the exact branching behavior of the Feynman 
integral  at the conifold divisors. 

The  important lesson to draw from the two geometrical representations for the banana 
graphs is that there is no such thing as an {unique Calabi-Yau geometry} (or its extension) associated to a Feynman graph. 
To underline the point, we note that $M^\text{HS}_{l-1}$ and $M^\text{CY}_{l-1}$ 
have different  topologies, e.g., the Euler numbers for $M^\text{HS}_{l-1}$ are 
\beq
	\chi(M^\text{HS}_{3})	=	20\,,\qquad  \chi(M^\text{HS}_{4})=540\,,\qquad  \ldots \,, 
\eeq
while for $M^\text{CI}_{l-1}$ they are
\beq
	\chi(M^\text{CI}_{3})	=	-80\,,\qquad \chi(M^\text{CI}_{4})=720\,,\qquad \ldots\,.
\eeq
Therefore, one cannot find a smooth map relating these 
geometries. Rather, one must focus on finding the uniquely defined \emph{family of Calabi-Yau motives}, preferably  in 
the easiest geometrical  setting. 

One of the goals of this paper is to generalize the results of ref.~\cite{Bonisch:2020qmm} to include 
all higher-order terms in the dimensional regulator $\eps$. We find that also in this context that the motive given by eq.~\eqref{MCicy} 
is also more natural. In particular, the homology of the ambient spaces $F_l$ and specially $\mathbb P_{l+1}$ 
seem to play an important role section \ref{subsec:commentsnumbersolutions}. 
So again the motive defined by eq.~\eqref{eq:HSgeom} is less suited to understand that generalization.

\paragraph{Families of Calabi-Yau motives.}
\label{para:CYmotives}

A \emph{family of Calabi-Yau motives} of rank $r$ and weight $n$ can be characterised by its Picard-Fuchs differential ideal and 
an intersection form ${\bf \Sigma}$. The Picard-Fuchs differential ideal has to have the properties that its 
complete set of $r$  solutions $\uPi(\zs)$  fulfill the Griffiths transversality conditions in eq.~\eqref{eq:bryantgriffith} 
and form an irreducible integral representation of the global  monodromy group $\Gamma_{\mathcal{M}_n}$, which is a subgroup of 
${O}({\bf \Sigma},\mathbb{Z})$ defined in eq.~\eqref{eq:defOS}. Depending on wether ${\bf \Sigma}$ 
is even or odd, we speak of even or odd Calabi-Yau motives. These occur in the simplest geometrical setting 
from the middle cohomomology of  $n$-dimensional Calabi-Yau spaces $M_n$ of even or odd complex dimension $n$. 
For one-parameter families the Griffiths transversality conditions can be recasted as the condition in eq.~\eqref{eq:selfadjointness}. 
In  particular, for the case of the three-fold, ${\bf\Sigma}$ is canonical (see the discussion in section~\ref{sssec:bulk}), and  several hundreds of abstract 
families of Calabi-Yau motives  have been found using the 
additional assumption that the motive has a MUM degeneration, see section \ref{sssec:boundary}
in ref.~\cite{MR3822913}. 

Since for higher-dimensional families such abstract  studies of the possible motives have not yet 
been carried out, we describe below two strategies that allow to obtain constructively desired sub-motives of a Calabi-Yau mixed Hodge structure, that are relevant for the banana integrals. A {family of Calabi-Yau motives}  can be characterized as invariant piece or sub-motive of a 
Calabi-Yau mixed Hodge structure under a discrete symmetry or as the sub-slice  
of the complex structure moduli space that  is inherited after a singular transition from 
the singular configuration. 

Let us give first an example of the latter type. This is explained 
in ref.~\cite{MR2164402}, however, only for Calabi-Yau three-folds. The example of ref.~\cite{MR2164402} has direct bearing on realizations of 
the four-loop banana graph motive, and we expect it to generalize to all loops. 
The authors of ref.~\cite{MR2164402} consider the following transitions\footnote{We follow the notation of~\cite{MR2164402}. The subscript  $u$ 
denotes collectively the $16$ complex moduli of the $X^{16,26}_u$ family and subscript $a$ the five moduli of the singular family with $30$ nodes.}  
\begin{equation}
X^{16,26}_u \rightarrow X^\text{sing}_a \rightarrow \hat X^{5,45}\ .
\end{equation} 
Here $X^{16,26}_u$  is our deformed space $M^\text{HS}_{3}$ with $\chi(M^\text{HS}_{3})=20$. The superscripts are $h_{21}$ and $h_{11}$, respectively. $X^\text{sing}_a$ is the singular five-parameter 
space in the physical slice written in torus variables, and  $\hat  X^{5,45}$ is the small resolution of $X^\text{sing}_a$. 
The observation is that $X^\text{sing}_a$ has $30$ nodes (conifolds) in ten two 
planes $S_{ij}=\{W_i=W_j=0\}$, where $S^3$-spheres are shrinking. These conical singularities can be 
resolved small projectively by replacing the singular loci by $\mathbb{P}^1$ blow ups. Each blow 
up adds $\chi(\mathbb{P}^1)=2$ to the Euler characteristic, see  \cite{MR930270} for 
an early application of this technique to construct  Calabi-Yau spaces with positive Euler number. The five-dimensional moduli space of 
$X^\text{sing}$ gets identified with the five-dimensional moduli space of $X^{5,45}$.  It is conceivable that 
$X^{5,45}$ is the mirror of $M_3^\text{CI}$, but in any case in view of the map in eq.~\eqref{eq:birationalmap}, 
it is clear that the rank ten motive defined by the middle cohomology of  $X^{5,45}$ is the same 
as the one of the horizontal cohomology of the smooth manifold in eq.~\eqref{CICY}. Hence, we can 
see the construction of our geometrical realization in eq.~\eqref{MCicy} for the massive banana integral 
as a variant of the transition method, which generalizes in the simplest possible way to all 
dimensions and has  the  additional advantage that the simple mirror identification in eq.~\eqref{eq:simplemirror}
 is available. 

The second  common strategy to characterize a sub-family of motives  within the Hodge structure of  a 
Calabi-Yau manifold $M_n$  is to find a discrete  symmetry group $G$ that acts\footnote{We note in passing 
that Calabi-Yau spaces have no continuous symmetries with this property. The latter would 
correspond to non-trivial holomorphic vector fields which are not present as $h_{1,0}(M_n)=0$.} on 
special configurations of $M_n$ and leaves the holomorphic $(n,0)$-form invariant. One can then focus 
on invariant subspaces of the Hodge structure defined on $M_n/G$. The action on the cohomology can 
be determined from an explicit realization of the latter, e.g., the one in eq.~\eqref{Griffith residue form2}. 
For example, in many constructions of  mirror manifolds one considers the invariant part of the 
cohomology under a maximal phase symmetry group $G$. To obtain the mirror, one 
resolves the quotient singularities of $M_n/G$,  but to construct the differential equation of the 
sub-motive it is only necessary that the invariant family is either smooth or that the generic singularities of 
the family  can be resolved.  Finding all symmetries $G$ is not an easy task, as their detection  
depends very much on how we represent $M_n$.  Starting,  e.g., from eq.~\eqref{eq:HSgeom}, it is not easy to find  
a discrete  symmetry $G$  that restricts the motive of  $M^\text{HS}_{l-1}$ to the motive of  the 
physical sub-slice. The physical sub-slice,  on the other hand,  exhibits further phase symmetries besides the obvious $S_{l+1}$ permutation 
symmetry. Together they can be used to restrict to the equal-mass sub-slice further, 
see the end of section~\ref{para:4lequalmass} for a discussion of that point. 
While the symmetry action of $G$ on $M_n$ that induces a splitting of the Hodge structure into irreducible 
representations of the Galois group might be hard to find, an additional tool to detect it is
to consider  splittings of the Hasse-Weil zeta function using a $p$-adic analysis, as it was 
employed in simple Calabi-Yau cases~\cite{candelas2019one,adek}. 
This strategy might be adaptable  to families of motives.     

It is important to note that by definition the properties of Calabi-Yau Hodge structures that we discuss in 
sections  \ref{sssect:Griffithstransversality}, \ref{sssec:boundary} and \ref{subsection: MirrorSymmetry} 
apply with necessary qualifiers to the families of Calabi-Yau motives. We will comment more on the question of  the 
universal applicability of the latter to Feynman integrals in section \ref{subs:integralsmotives}.

\subsubsection{Quadratic  relations from Griffiths transversality}
\label{sssect:Griffithstransversality}
In this section we explore an important consequence of Calabi-Yau geometries for the maximal cuts of Feynman integrals. More precisely, we will show that the Calabi-Yau geometry leads to quadratic relations among the maximal cuts (in integer dimensions). 
We limit the exposition here to the mathematical background, and we will describe the resulting relations explicitly in the context of the equal-mass banana integrals in section~\ref{sec:quad_rel_max_cut}.

Our starting point is the Griffiths transversality in eq.~\eqref{eq:transversality}. Combing eq.~\eqref{eq:transversality} with the first polarization condition in eq.~\eqref{eq:type} and considerations of type, 
one gets as a generalization of the observations of Bryant and 
Griffiths~\cite{MR717607}  for Calabi-Yau manifolds in any dimension $n$:
\begin{equation} 
\uPi(\zs)^T\,{\bf \Sigma}\, \partial^{\underline k}_{\underline z}\uPi(\zs) = \int_{M_n} \Omega \wedge  \partial^{\underline k}_{\underline z}  \Omega
= \left\{\begin{array}{ll} 0 & \ \ {\rm for}\  0 \le r <n \\ 
         C_{\underline k} ({\underline z})& \ \ {\rm for}\  |k|=n
        \end{array}\right. \ , 
\label{eq:bryantgriffith} 
\end{equation}
where the $C_{\underline k} ({\underline z})$ are rational functions in the complex structure parameters.  For the first equality
 in eq.~\eqref{eq:bryantgriffith}, we used eq.~\eqref{eq:perioddef} and the properties of the integer  basis described earlier. The second equality follows very generally from eqs.~\eqref{eq:transversality} and~\eqref{eq:type}. We point out 
that even in an arbitrary local basis $\tilde \uPi(\zs)$  corresponding to an (implicit)  choice of a basis of cycles $\tilde \Gamma^i\in H_n(M_n,\mathbb{C})$ (obtained, for example, as independent local solutions of the Picard-Fuchs differential ideal),
one can find a ${\bf \tilde \Sigma}$ and write down the corresponding
relations $\tilde{\uPi}(\zs)^T\,{\bf \tilde\Sigma}\, \partial^{\underline k}_{\underline z}\tilde\uPi(\zs)$ among the solutions very explicitly.    

The quadratic relations in eq.~\eqref{eq:bryantgriffith} have important implications for Feynman integrals: Since the vector of periods $\uPi(\zs)$ describes the maximal cuts, the relations in eq.~\eqref{eq:bryantgriffith} can equally be interpreted as a set of quadratic relations among the maximal cuts! Note that these relations are not obvious from a purely physical view, e.g., the momentum-space representation of the Feynman integrals. We will describe these relations among maximal cuts explicitly for the equal-mass banana integrals in section~\ref{sec:quad_rel_max_cut}. Here we only mention that for the equal-mass banana graphs  one finds $l(l+1)/2$ quadratic relations for $l-1=n$ even and $l(l-1)/2$ for  $n$ odd. 
The reason for this difference is that in the latter case the intersection form ${\bf \Sigma}$ is antisymmetric, so symmetric quadratic relations are  trivially fulfilled.

\paragraph{The  $n$-point (Yukawa) couplings and self-adjoint operators.}
\label{para:Yukawacoupling}

In order to understand the quadratic relations in eq.~\eqref{eq:bryantgriffith} and to write them down explicitly, we need to know the functions
$C_{\underline k} ({\underline z})$, sometimes referred to as the Yukawa $n$-point couplings. They can be obtained from the rational coefficients in front of the derivatives 
in the Picard-Fuchs differential operators, if and only if the latter generate the Picard-Fuchs ideal completely, see ref.~\cite{MR3965409} for details.  

Let us illustrate this on the example for the Legendre family $\mathcal E_\text{Leg}$ of elliptic curves. Taking the derivative of $C_1(z) = \uPi(z)^{T}\,{\bf \Sigma}\,\partial_z\uPi(z)$, and using the Picard-Fuchs equation in eq.~\eqref{eq:Legendre_PF} to write $\partial_z^2\uPi(z)$ as a linear combination of  $\uPi(z)$ and $\partial_z\uPi(z)$,
we obtain  the  differential equation $\frac{\partial_z C_1(z)}{C_1(z)}=\alpha'(z)$, with $\alpha(z)$ given in eq.~\eqref{eq:Legendre_shorthands}. Hence $C_1(z)=c \rho(z)$, where $c$ is an integration constant 
which  can be fixed in the integral symplectic basis of the periods to be $c=1$. For the Legendre curve we then find $C_1^\text{Leg}(z)=\frac{1}{z (1-z)}$.  

More generally, if the  Picard-Fuchs differential ideal is generated by a single differential  operator (as it is the case 
for one-parameter families, see appendix~\ref{app:pid_one_var}) with normalization such that  
\begin{equation} 
{\cal L}^{(n+1)}=\partial_z^{n+1}+ \sum_{i=0}^{n} a_i(z) \,\partial_z^i\,,
\end{equation}
then the Yukawa coupling fulfills the differential equation
\beq\label{eq:Yukawa_equation}
\frac{\partial_z C_n(z)}{C_n(z)}=\frac{2}{n+1}a_n(z)\,. 
\eeq
One can define the \emph{adjoint 
differential  operator} \cite{MR1579749}
\begin{equation} 
{\cal L}^{* (n+1)}=\sum_{i=0}^{n+1}\left(-\partial_z \right)^i a_i(z) \,.
\end{equation} 
An operator is called \emph{essentially self-adjoint} if 
\begin{equation} 
{\cal L}^{* (n+1)}A(z)=(-1)^{n+1}A(z){\cal L}^{(n+1)}\ , 
\label{eq:selfadjointness}
\end{equation}       
where $A(z)$ satisfies the differential relation $\frac{\partial_zA(z)}{A(z)}=\frac{2}{n+1} a_n(z)$. Note that $A(z)$ is up to a multiplicative constant given by the Yukawa coupling $C_n(z)$. It was noticed  in the search for Calabi-Yau operators \cite{MR3822913} that the  self-adjointness
of an abstractly constructed  differential operator with regular singularities implies that the 
solutions admit an even or odd intersection form for $n$ even or odd, respectively, if 
$A(z)$ is an algebraic function.  This gives an easy criterium to decide whether one-parameter specializations of Picard-Fuchs operators  
can come from a Calabi-Yau motive: This can only be the case if eqs.~\eqref{eq:selfadjointness} and \eqref{eq:bryantgriffith} 
are fulfilled and in addition the global monodromy is in ${O}({\bf \Sigma},\mathbb{Z})$.  
In other words, imagine that the maximal cut of a Feynman integral depends on a single dimensionless variable (if there are more 
kinematic variables, we may consider a one-parameter slice in the rescaled kinematic space), and that its maximal cut is annihilated 
by some Picard-Fuchs operator ${\cal L}^{(n+1)}$. The previous discussion gives an easy criterion to determine from the 
Picard-Fuchs operator if the Feynman integral is associated with a Calabi-Yau geometry. One can check that this criterion 
is satisfied for all the Picard-Fuchs operators for the maximal cuts  of the banana integrals in $D=2$ dimensions. We will see 
towards the end of this paper that the Picard-Fuchs operators in $D=2-2\eps$ dimensions cease to be self-adjoint.

\subsubsection{Monodromy and limiting mixed Hodge structure on the boundary of  ${\cal M}_\text{cs}$}
\label{sssec:boundary}

In this section we explain the structures related to the boundaries of the moduli space that 
are related to the special monodromies of the periods when we analytically continue them around  
the critical divisors which form the boundaries of ${\cal M}_\text{cs}$.

\paragraph{A normal crossing model for the boundaries of ${\cal M}_\text{cs}$.}
\label{para:Normalcrossing}

By going in a loop $\gamma_{\Delta_k}$ from a base-point $\zs_0$ around a divisor given by $\Delta_k({\underline z})=0$,
induces a monodromy on the period 
integrals,\footnote{See section \ref{sssec:bulk} for a detailed introduction to period integrals on Calabi-Yau varieties.} 
and hence on the maximal cut Feynman integral identified with the latter. These monodromies are 
very characteristic for the singularity that the  fibre $M_{\{\Delta_k=0\}}$  over $\{\Delta_k(\zs)=0\}$ 
acquires. The branching behavior of the periods at the crictical loci is crucial to 
understand the analytic structure of the Feynman integral in all regions of its physical parameters.  
Mirror symmetry suggests that Calabi-Yau $n$-folds have a maximal degenerate singular point in their complex 
moduli space, called point of maximal unipotent monodromy (MUM-point, see section~\ref{subsec:PF}). This point was identified with 
the large momentum regime of the banana integrals  in ref.~\cite{Bonisch:2020qmm} and  used with the monodromy at other 
singularities to clarify the analytic structure of these integrals in $D=2$ dimensions
completely to all loop orders.  

The boundary of the moduli space ${\cal M}_\text{cs}$ refers to the critical divisors at which the 
Calabi-Yau fibre becomes singular.  As we mentioned in section~\ref{subsec:CYmanifolds}, we 
 assume to be able  to compactify and to resolve the moduli space to achieve a situation where
 all  divisors are normal crossing in the compactified moduli space $\overline{ {\cal M}}_\text{cs}$. 
 We refer to ${\cal M}_\text{cs}$ as ${\cal M}_\text{cs}={\overline {\cal M}_\text{cs}}\setminus D$, where $D$ is a divisor
 with normal crossings, $D=\bigcup_k D_k$. 

 Let us explain how we can find the boundary components  $\Delta_k(\zs)=0$. The first method to find them is to identify 
 the sub-loci of ${\cal M}_\text{cs}$ over which the fibre of  the  family becomes singular. For example, for a Calabi-Yau manifold defined 
 by $P_1=\ldots = P_s=0$, we   have to find values in ${\overline {\cal M}}_\text{cs}$ such that  $P_1=\ldots =P_s=0$ and 
 ${\rm d} P_1\wedge \ldots \wedge {\rm d} P_s =0$ admits a solution. Specifically, for the Legendre curve we compactify $\overline{ {\cal M}}_\text{cs}=\mathbb{P}^1$, 
 and by determining the singular (in this case nodal) fibres we find ${\cal M}_\text{cs}=\mathbb{P}^1\setminus \{z=0,1,\infty\}$. For the equal-mass banana 
 graph  in eq.~\eqref{eq:Pbanana} or~\eqref{MCicy} (after setting all masses equal), we get ${\cal M}_\text{cs}=\mathbb{P}^1\setminus \{z=0,1/9,1,\infty\}$. We  can also 
 determine the critical loci from the  Picard-Fuchs differential ideal  $J$ that is generated by differential operators 
 of order  ord$_k$, ${\cal L}^{{\rm ord}_k}_k\left({\underline z},\partial_{\uz}\right) \in 
\mathbb{C}\left[z_1,\ldots,z_m,\partial_{z_1},\ldots, \partial_{z_m}\right]$, $k=1,\ldots,|J|$.
We can replace the $ \partial_{z_i}\mapsto \xi_i$, $i=1,\ldots, m={\rm dim}(M_{cs})$ by formal variables to get
$|J|$ elements in the polynomial ring  in $\underline z$ and $\underline \xi$. Then we consider the smallest differential ideal  that 
characterizes the periods and  restrict to the leading pieces, i.e., to the elements 
$S_k({\underline z},{\underline \xi})\coloneqq\left.{\cal L}^{{\rm ord}_k}_k({\underline z},{\underline \xi})\right|_{\text{deg}(\xi)=\text{ord}_k}$, 
which are homogeneous of order ord$_k$ in the variables $\underline \xi$. One refers to the $S_k$ as the \emph{symbol} of the differential operator ${\cal L}^{{\rm ord}_k}_k$. 
The critical  loci $\tilde \Delta_j({\underline z})=0$  are now given by the resultant of 
the $S_k({\underline z},{\underline \xi})=0$ in the $\uz$ parameters, i.e., resultant$(\{S_k({\underline z},{\underline \xi})=0,\ \forall k\},{\uz})$. 
The resultant characterizes all divisors  $\tilde \Delta_j({\underline z})=0$ for which the system $\{S_k({\underline z},{\underline \xi})=0,\ \forall k\}$ 
has non-trivial solutions, and it contains, in particular the critical divisors of the family. For the Picard-Fuchs ideal generated by a single ordinary 
differential operator, e.g., as in eqs.~\eqref{eq:difflegendre} or \eqref{eq:diffbanana}, 
this amounts to find all zeroes of the coefficients of the highest derivatives for $z\in \mathbb{P}^1$. This {second method} to find the 
boundary components is in general superior as it detects also the apparent singularities, which are not present in our examples as  
can be checked since both methods lead to the same result.   

Let us note that for moduli spaces $\overline {\cal M}_\text{cs}$ of dimension greater than one, non-generic intersections, e.g., tangencies of order $m$ between the $ \{ \tilde \Delta_k(z)=0\}$ or singularities of the $ \{\tilde \Delta_k(z)=0\}$, generally occur. In a procedure that can involve several steps of blow ups, they can be resolved to achieve 
a geometry  of $\overline {\cal M}_\text{cs}$  with only normal crossing divisors.  Adding all the exceptional divisors of the blow ups, 
the critical locus $D\subset \overline {\cal M}_\text{cs}$ is described as the  set of irreducible normal crossing  
divisors $\{D_k\}$, $k=1,\ldots, \#_D$. Normal crossing means that locally we can describe the intersections of components in   
$D$ as  $w_1=0,\cdots, w_p=0$, $p\leq r$  in local coordinates $w_i$, $i=1,\ldots,r={\rm dim}({\cal M}_\text{cs})$.  In this case we say that the 
family  ${\pi}:{{\cal M}_n}\rightarrow  {{\cal M}_\text{cs}} $ can be extended to a family ${\overline \pi}:{\overline {\cal M}_n}\rightarrow  {\overline {\cal M}_\text{cs}}$.
Concrete examples for the blow up  procedure in Calabi-Yau moduli spaces can be found in refs.~\cite{Candelas:1993dm,Candelas:1994hw}.

\paragraph{Local and global monodromies.}
\label{para:Localandglobalmonodromies}
 
We now analyze the monodromies that the vectors  of periods undergo, when we take $\underline z$ around 
a loop $\gamma_{\Delta_i}$ around  the critical  divisor $D$ given by $\Delta_i(z)=0$.  
We illustrate this first on the example of a one-parameter differential operator ${\cal L}$ of order $l$. 
In that case, the compactification of ${\cal M}_\text{cs}$ 
is $\mathbb{P}^1$, and  the divisors are just isolated points, $\Delta_i(z)=z-z_i$, $i=1,\ldots, \# p_{\mathrm{crit}}$. 
We can find a basis for the solution space at $z_0$ using the Frobenius method (see section~\ref{subsec:PF}). 

Let us discuss in some detail the case of the Legendre family of elliptic curves. We have determined its Picard-Fuchs operator $\cL_{\mathrm{Leg}}$ in eq.~\eqref{eq:difflegendre}. We can apply the Frobenius method and solve the indicial equation for each singular point $z_0\in\{0,1,\infty\}$. 
The local exponents at all critical points are summarized in the  Riemann 
$\mathcal P$-symbol:
\begin{equation} 
\mathcal P_\text{Leg}	\left\{	\begin{matrix}	%
							0 	& 1	& \infty		\\ \hline
							0	& 0	& \frac12		\\
							0	& 0	& \frac12	&	
						\end{matrix} \right\}~.
\label{riemannsymbolgeneral}
\end{equation} 
Since the local exponents at each singular point are equal to $\alpha$, say, at each singular point there is a power series solution $\omega(\Delta)=\Delta^\alpha  +{\cal O}(\Delta^{\alpha+1})$ (with $\Delta = z-z_0$) and 
a logarithmic solution  $\hat \omega(\Delta)=\frac{m}{2 \pi i}\omega(\Delta) \log(\Delta) +{\cal O}(\Delta^\alpha)$.
In particular, if 
$\alpha\in \mathbb{Z}$, then the period vector transforms for a positively oriented loop $\gamma_{\Delta_i}$  around $\Delta_i(z)=0$ with a $\textbf{T}_{\gamma_{\Delta_i}}$ monodromy matrix    
\begin{equation} 
\left( \begin{array}{c} \hat \omega \\ \omega \end{array} \right) \mapsto 
\left(\begin{array}{cc} 1&m\\ 0&1 \end{array} \right) \left(\begin{array}{c} \hat \omega\\ \omega \end{array} \right)\eqqcolon \textbf{T}_{\gamma_{\Delta_i}} \left(\begin{array}{c} \hat \omega\\ \omega \end{array} \right)\ . 
\label{eq:elllipticcurvemonodromy} 
\end{equation}
In particular, let  $\omega (\Delta)= \int_{S^1_{\nu}} \rd x/y(z)$ and the dual period is over the dual cycle $\hat \omega(\Delta) =\int_{S^1_{\hat \nu}} \rd x/y(z)$. 
Without loss of generality we can assume that  $S^1_\nu=S^1_a$ and $S^1_{\hat \nu}=m S^1_b$ in an integral symplectic basis $(a,b)$ 
of $H_1(E,\mathbb{Z})$. Then $m$ has to be an integer. Generally, a monodromy  matrix ${\bf T}$ in an integral 
basis has to be integral and has to respect the intersection form, $\textbf{T}^T\,{\bf{\Sigma}}\, \textbf{T}={\bf{\Sigma}}$. 
We denote the group of all  integer matrices that respect the intersection form ${\bf{\Sigma}}$ on the middle cohomology of rank $b_n(M)=r$ (or the 
rank of the Calabi-Yau motive $r$) by
\beq
\mathrm{O}({\bf\Sigma}, \mathbb{Z})=\bigr\{  \textbf{T}\in {\mathrm{GL}}(r,\mathbb{Z})\,\,\bigr |\, \, \textbf{T}^T{\bf{\Sigma}} \textbf{T}={\bf{\Sigma}}\bigl\}\,.
\label{eq:defOS}
\eeq
The subgroup  of  $\mathrm{O}({\bf\Sigma}, \mathbb{Z})$ that is generated by the actual monodromies of the family is denoted by $\Gamma_{{\cal M}_n}$.
For example, in odd dimensions $n$, the intersection pairing ${\bf \Sigma}$ is the standard symplectic pairing, and so $\Gamma_{{\cal M}_n}$ has to be a subgroup of the integral symplectic 
matrices $\mathrm{Sp}(b_n,\mathbb{Z})$. In particular, for elliptic curves ($n=1$), it is a subgroup of $\mathrm{SL}(2,\mathbb{Z})$. Note that  the K\"ahler potential  in eq.~\eqref{eq:Kaehlerpotential} is single-valued under all monodromies.  

A famous theorem of Landman~\cite{MR344248} states that all possible monodromy matrices on an algebraic  $n$-fold have to obey the relation 
\begin{equation} 
(\textbf{T}^k-{\mathbbm{1}})^{n+1}=0\ .
\label{eq:Landman}
\end{equation} 
Here  $k\in \mathbb{N}_0$, implying that the indicial $\alpha$ has to be a rational number. A monodromy matrix $\textbf{T}$ can be unipotent of lower order 
$m<n$, i.e.,  $(\textbf{T}^k-{\mathbbm{1}})^{m+1}=0$. It is clear that $m$  is the size of the biggest Jordan block in $\textbf{T}$.
The  maximal $n$ that can appear is $n=\dim(M)$. It is not too hard to see that the unipotency of 
order $m\leq n$ implies that a period on an $n$-fold cannot degenerate 
worse than with a logarithmic singularity of type $\log(\Delta)^n$. This has an important consequence for Feynman integrals. Assume that we have a maximal cut of a Feynman integral in integer dimensions that degenerates in a dimensionless 
physical parameter $\Delta$ (or, more generally, some polynomial combination thereof) as  $\log(\Delta)^m$.  Then it follows from Landman's theorem that the geometry associated to this integral cannot be an algebraic 
manifold of dimension less than $m$, or a Calabi-Yau motive of weight less than $m$! 

In  the example of the Legendre family one sees  that for $z_0\in\{0,1,\infty\}$ the curve is singular, i.e., $P=0$ and $\mathrm dP=0$ have at least one common solution. 
This happens  at a point on the curve, say $(x,y)=(x_0,y_0)$ (we assume 
that the singularity is not at $w=0$ and use the corresponding  local patch $w=1$). Using local coordinates 
$(x,y)=(x_0+\tilde \epsilon_x,y_0+\tilde \epsilon_y)$ the expansion around this point is given after a linear change 
in the deformation parameters up to  quadratic order  by $P=\epsilon^2_x+\epsilon^2_y=0$.
Allowing in addition small perturbations around the critical point $z=z_0+\mu$ in ${\cal M}_\text{cs}$, the local  singularity becomes  
\beq 
\epsilon^2_x+\epsilon^2_y=\mu^2\ .
\label{eq:S1}
\eeq 
This describes a node, where a $S^1$-cycle  $\nu$  shrinks with $\mu\rightarrow 0$. The $S^1$ vanishing 
cycle $\nu$  can be literally seen by taking the  real slice of the equation \eqref{eq:S1}, which gives the $S^1$ or radius $r^2={\rm Re}(\mu)^2$. 
This can be generalized to higher dimension and the period integral over the $S^n$ can be 
performed perturbatively, see eq. (3.3) in ref.~\cite{MR1115626}. The corresponding critical locus in ${\cal M}_\text{cs}$ is hence a {conifold}, while the singularity in the fibre is a node. The corresponding monodromy follows purely topologically  
from the Picard-Lefshetz formula 
\begin{equation}
W(\Gamma)=\Gamma+(-1)^{(n+2)(n+1)/2} (\Gamma \cap \nu) \nu \,,
\label{eq:PicardLefshetz}
\end{equation}
in any dimension $n$, see ref.~\cite{MR592569} for a clarification regarding  
the signs  in higher dimensions. The formula says that the conifold monodromy  
action $W$ on  any  cycle $\Gamma\in H_n(M,\mathbb{Z})$, which can be 
identified (up to finite multicovering issues in the choice of parametrization $z$ of ${\cal M}_\text{cs}$)
with the monodromy on the periods, depends only on its intersection with the 
vanishing cycle. Together with the self intersection of $n$ spheres  in projective $n$-folds~\cite{MR592569},  
\begin{equation}
S^n \cap S^n=\left\{\begin{array}{ll}
0 &\qquad \quad n\ \ \text{for odd}\,, \\
2(-1)^\frac{n}{2} & \quad \qquad n \ \text{for even}\,,\ 
 \end{array} \right.  
 \label{eq:Snintersection}
\end{equation}
eqs.~\eqref{eq:PicardLefshetz} and~\eqref{eq:Snintersection} give eq.~\eqref{eq:elllipticcurvemonodromy} with 
$S^1_{\hat \nu}\cap S^1_{ \nu}=-m$.  They also imply something completely general for 
the degenerations of Feynman integrals. If the maximal cut integral corresponds to a period of a  
$n$-dimensional  algebraic variety, then the most generic singularity will be a square root cut if $n$ 
is even, and a logarithmic cut if $n$ is odd (cf. the fact that $(l=n+1)$-loop banana integrals have square roots cuts when $l$ is odd and only logarithmic singularities when $l$ is even).  This follows  simply  because  
eqs. \eqref{eq:PicardLefshetz} and \eqref{eq:Snintersection} imply in odd dimensions an infinite-order operation, a so-called 
\emph{symplectic reflection}, and in even dimensions a standard \emph{Euclidean $\mathbb{Z}_2$-reflection}.  

A simple application of this structure is that the logarithmic/square root cut behavior of the solutions to the Picard-Fuchs differential 
at the conifold  detects uniquely the period over the geometric vanishing cycle. As an actual cycle 
the latter might contain information about to the Feynman integral. Let us illustrate this for the banana integrals. 
In section~\ref{sec:deqs} we have argued that as a consequence of the optical theorem, the imaginary part of the banana integral is proportional to a specific maximal cut for $0<z<1/(l+1)^2$, cf.~eq.~\eqref{eq:optical_thm}. The corresponding cycle in eq.~\eqref{eq:optical_thm} must be such that this maximal cut vanishes at the threshold $z=1/(l+1)^2$, which is a conifold divisor. There is a unique period that vanishes at this conifold divisor, and so this period corresponds to the maximal cut that describes the imaginary part of the banana integral above threshold~\cite{Bonisch:2020qmm}. Note that the cycle that describes the imaginary part is different from the cycles $T^l$ or $\Gamma_T$ in eq.~\eqref{eq:max_cut_feyn_par}, and we will comment further on this at the end of the this section.

Another  beautiful example for the even dimensional  situation ($n=2$) are the  monodromies around the $r$ divisors meeting with normal crossing at a 
codim $r$  locus in ${\cal  M}_\text{cs}(\text{K3})$ over which the K3 fibre acquires a rank $r$ ADE singularity  (ADE as a subgroup in $E_8$).
In this case the monodromies in eq.~\eqref{eq:PicardLefshetz} are literally the  Weyl-reflections generating the Weyl group of the corresponding ADE Lie algebra, 
because  the corresponding vanishing $S^2$-spheres intersect according to the negative of the Cartan matrix of the 
Lie algebra, as a direct consequence of the intersection form ${\bf \Sigma}_{\text{K3}}$. 

The analysis of the solutions from the Picard-Fuchs differential ideal or eq.~\eqref{eq:PicardLefshetz}
yields the local monodromies.  To determine the global monodromy group $\Gamma_{{\cal M}_n}$ in 
an integral symplectic basis requires global knowledge of the periods. For the elliptic curve case this  can 
be obtained by analyzing the behavior of the explicit elliptic integrals near the critical points.  
Let the period vector of the Legendre curve be  $(\Pi_b,\Pi_a)=(\int_b \Omega,\int_a \Omega)$, and let $\Pi_b$ be the logarithmic period as in eq.~\eqref{eq:elllipticcurvemonodromy}. Then,  up 
to SL$(2,\mathbb{Z})$ conjugation, the the monodromy group $\Gamma_{{\cal E}_{\text{Leg}}}$ of the Legendre family is generated by the following matrices (we use the notation ${\bf{T}}_{\gamma_{z-a}}\eqqcolon {\bf{T}}_a$, with $a\in\{0,1,\infty\}$):
\begin{equation}
{\textbf{T}}_0=\left(\begin{array}{cc} 1& 2 \\ 0& 1\end{array}\right)  \qquad\text{and} \qquad    {\textbf{T}}_1=\left(\begin{array}{rr} 1& 0 \\ -2& 1\end{array}\right)\,.
\label{eq:monodromy} 
\end{equation}
One can check that these matrices generate the congruence subgroup $\Gamma(2)$ of index 6 in SL$(2,\mathbb{Z})$, and so the monodromy group of the Legendre family is $\Gamma_{{\cal E}_{\text{Leg}}}=\Gamma(2)$. One can also check that the matrices in eq.~\eqref{eq:monodromy} satisfy Landman's theorem in eq.~\eqref{eq:Landman} with $n=k=1$. 
Due to the obvious relation by successively going around all the loops  $\gamma_{\Delta_i}$ for $i=1,2,3$ in $\mathbb{P}^1$ one has 
${\textbf{T}}_0{\textbf{T}}_{1}{\textbf{T}}_{\infty}=\mathbbm{1}$, and therefore {\footnotesize{${\textbf{T}}_{\infty}=\left(\begin{array}{rr} 1& -2 \\ 2& -3\end{array}\right)$}}. We can 
conjugate the basis by $\Pi_b\rightarrow \Pi_b$, $\Pi_a\rightarrow \Pi_a+\Pi_b$ to get 
{\footnotesize{${\textbf{T}}^c_{\infty}=\left(\begin{array}{rr} -1& 2 \\ 0& -1\end{array}\right)$}.} Comparing with eq.~\eqref{eq:Landman},
we see that ${\textbf{T}}^c_{\infty}$ (and thus also ${\textbf{T}}_{\infty}$) satisfy Landman's theorem with $k=2$ and $n=1$.  For the differential 
equation associated to the banana graph one finds the corresponding monodromy group to be $\Gamma_{{\cal E}_{\text{ban}_2}} =\Gamma_1(6)$, cf.~refs.~\cite{Bloch:2013tra,Adams:2017ejb,Frellesvig:2021vdl}. Finding  the integral basis  and the monodromy group  $\Gamma_{{\cal M}_n}$ for families of higher-dimensional Calabi-Yau manifolds with higher-dimensional moduli space can become  
a formidable task.  We comment on some strategies to do this at the very end of this section.  

It follows generally from eq.~\eqref{eq:Landman}  that  ${\textbf{T}}$ can always be 
factored  as ${\textbf{T}}={\textbf{T}}^{(s)} {\textbf{T}}^{(u)}$, where ${\textbf{T}}^{(s)}$  is semi-simple and of finite order and ${\textbf{T}}^{(u)}$ is unipotent, i.e., $({\textbf{T}}^{(u)}-\mathbbm1)^{n+1}=0$. 
For example, all singularities of the fibres in the  Legendre family $\mathcal E_\text{Leg}$ are nodes, and a homologically-different  cycle $S^1_\Gamma$,  
with $\Gamma$ primitive in $H_1(E,\mathbb{Z})$, vanishes at each conifold point. The square root 
cut at $z=\infty$ ($k=2$), as well as the shifts by two at the other points, are due the 
global choice of the parameter $z$.  Locally, one can get rid of the semi-simple piece by choosing different
local variables, e.g., for the Legendre family at $z=\infty$, one can choose $v=\sqrt{w}$ instead of $w=1/z$. Only for elliptic curves the conifold points are also MUM-points.

\paragraph{The limiting mixed Hodge structure.}
\label{para:LimitingMixedHodge}

The general situation of more involved singularities in families with higher-dimensional fibers is 
described by the limiting mixed Hodge structure.  The first statement of Deligne~\cite{alma991033235483005251}
is that the bundle ${\cal F}^0$ on ${\cal M}_\text{cs}$ has a canonical extension  ${\overline {{\cal F}^0}}$ over $\overline{{\cal M}}_\text{cs}$. 
As we have learned from the theorem of Landman applied to the monodromy matrices, the forms the singularities of 
${\overline {{\cal F}^0}}$ are only logarithmic. This is referred also as regular singularities (see section~\ref{subsec:PF}), and allows one to define an extension of the Gauss-Manin connection to
\begin{equation}
{\overline \nabla}:   {\overline {{\cal F}^0}} \rightarrow   {\overline {{\cal F}^0}}\otimes \Omega_{\overline{\cal M}_\text{cs}}^1(\log(D))~ .
\label{eq:GaussManinextend} 
\end{equation}  
Here $\Omega_{\overline {\cal M}_\text{cs}}^1(\log(D))$ are meromorphic one-forms on $\overline{\cal M}_\text{cs}$ which can have the indicated 
logarithmic coefficients. In local coordinates where the divisor $D$ is defined by  $z_1=\cdots= z_p=0$, 
$\Omega_{\overline{\cal M}_\text{cs}}^1(\log(D))$ is generated by $\mathrm dz_1/z_1,\ldots, \mathrm dz_p/z_p,$ $ z_{p+1},\ldots, z_r$. This 
means that in the first-order form of the Picard-Fuchs equation, the entries in the matrices $\mathbf{A}(z)$ and $\mathbf{A}^{\theta}(z)$ in eq.~\eqref{eq:firstorderform_A} for families of any dimension of fibre and base can have only first-order poles at the singular loci! 
Locally, we can model ${\cal M}_\text{cs}$ as products of punctured discs $(\frak{D}^*)^r$ and  $\overline{{\cal M}}_\text{cs}$ as products of  full discs $(\frak{D})^r$. We assume to have changed coordinates such that we got rid of the semi-simple piece, and that going clockwise around  a loop in the $k^{\textrm{th}}$ disc we generate the 
unipotent piece ${\textbf{T}}_k^{(u)}$ of the monodromy. We define
\begin{equation}
{\textbf{N}}_k=-\log({\textbf{T}}_k^{(u)})=-\log(1+[{\textbf{T}}^{(u)}_k-1])=\sum_{l=1}^{\text{max}\bigl|J_{{\textbf{T}}^{(u)}_k}\bigr|} (-1)^{l} [{\textbf{T}}^{(u)}_k-1]^l/l\ .
\label{eq:defNk}
\end{equation}       
It is obvious that the sum is bounded by the maximal size of a Jordan block in ${\textbf{T}}^{(u)}_k$. Now a section $s$
of ${\cal F}^0$ defined on $(\frak{D}^*)^r$  transforms like $s\mapsto {\textbf{T}}_k^{(u)} s$, but one can construct a 
monondromy invariant section  $\overline{s}$, i.e.,  one that is single-valued on $(\frak{D}^*)^r$, by
\begin{equation}
\overline{s}=\exp\left(-\frac{1}{2 \pi i}\sum_{k=1}^r N_k \log(z_k)\right) s\eqqcolon{\cal O}({\underline N})s\,,
\end{equation} 
and extend it canonically over $\frak{D}^r$. This defines a natural  extension  of 
${\overline {{\cal F}^0}}$ over  $\frak{D}^r$. The so-called nilpotent  orbit theorem 
of W. Schmid~\cite{MR382272} guarantees further that the ${\cal F}^p$ extend in a 
canonical way to sub-bundles $\overline {{\cal F}^p}$ of ${\overline {{\cal F}^0}}$ that
extend the fibration in eq.~\eqref{eq:decreasing} to $\frak{D}^r$. In particular, at the origin 
$z=0\in \frak{D}^r$, the  $\overline {{\cal F}^p}$ define the \emph{limiting Hodge  filtration} 
$F^p_{\rm lim}$. The $N_k$ fulfill a transversality $N_k( F^p_{\rm lim})\subset F^{p-1}_{\rm lim}$ 
like the Griffiths transversality in eq.~\eqref{eq:trans}. One can show that the action of the  
extension  of the Gauss-Manin connection ${\overline{ \nabla}}_{\theta_k}$ (with $\theta_k = z_k\, \partial_{z_k}$) 
to   $\frak{D}^r$ becomes proportional to the action of $N_k$ on $F^p_{\rm lim}$ as well as on  
${\rm Gr}^p_{\rm lim}= F^p_{\rm lim} / F^{p+1}_{\rm lim}$ :
\begin{equation} 
{\overline{ \nabla}}_{\theta_k}=-\frac{1}{2 \pi i} N_k\ .
\end{equation}     
Also a section $s_\mathbb{Z}\in {\cal H}_\mathbb{Z}$ of the integer local system ${\cal H}_\mathbb{Z}$ on $\frak{D}^r$
can be extended as ${\overline s}_{\mathbb{Z}} ={\cal O}({\underline N})s_\mathbb{Z} $ to $\frak{D}^r$ and defines 
an integral structure ${\overline s}_{\mathbb{Z}}(0)$ over $z=0$. However, there is a freedom  in the choice of 
 coordinates on the discs $\frak{D}^r$. More precisely, the change of coordinates $\zs\to {\underline{\tilde z}({\underline z})}$ induces a choice   
$\exp\left(2 \pi i \alpha_k (d \tilde z/d z)_k(0) N_k\right)$, referred to as nilpotent orbit,  
in the choice of the integral structure ${\overline s}_{\mathbb{Z}}(0)$.  

The second filtration of the \emph{limiting mixed Hodge structure} at the  boundary is the ascending  
\emph{monodromy weight filtration}:
\begin{equation} 
W_\bullet: W_0\subset W_1\subset \cdots \subset W_{2n-1}\subset W_{2n}=H^n(M_{\underline z},\mathbb{C})\ , 
\label{eq:monodromyweightfiltration}
\end{equation}   
with ${\rm Gr}^W_k=W_p/W_{p-1}$. The spaces $W_\bullet$ are defined by the action of the operator 
\begin{equation}
\begin{array}{rl} 
(\text{i.}) & \qquad\qquad  N(W_k)\subset N(W_{k-2})\ , \\ [2 mm] 
(\text{ii.}) & \qquad\qquad N^k : {\rm Gr}^W_{n+k}\xrightarrow{\sim}   {\rm Gr}^W_{n-k}\  .
\end{array} 
\label{eq:conditions} 
\end{equation}  
The  first few and the last $W_\bullet$ are explicitly given by
\begin{equation} 
\begin{array}{rl}
W_0	&=	{\rm Im}(N^n)\,,\\
W_1	&=	{\rm Im}(N^{n-1}) \cap {\rm Ker}(N)\,,\\
W_2	&=	{\rm Im}(N^{n-2}) \cap {\rm Ker}(N)+{\rm Im}(N^{n-1}) \cap {\rm Ker}(N^2)\,,\\
	&\, \vdots	\\
W_{2n-1}&= {\rm Ker}(N^n)\,.
\end{array}
\end{equation}
A key point is that  $F_{\rm lim}^\bullet$ induces a Hodge structure of pure  weight $k$ on ${\rm Gr}^W_k$~\cite{MR382272}, 
and the triple $(s_{\mathbb{Z}}(0),F_{\rm lim}^\bullet, W_\bullet)$ fits together to define a polarized  limiting mixed
Hodge structure.  Its \emph{limiting Hodge diamond} is given according to Lemma 1.2.8 of ref.~\cite{MR498551} by (see also ref.~\cite{MR664326}) 
\begin{equation} 
H_{\rm lim}^{p,q}=F^p_{\rm lim}\cap W_{p+q} \cap (\bar F^p_{\rm lim}\cap W_{p+q}+\sum_{j\ge 1} \bar F^{q-j}_{\rm lim} \cap   W_{p+q-j-1})\,,
\label{eq:Hlim}
\end{equation}  
with the property that $W_l=\oplus_{p+q\le l } H^{p,q}_{\rm lim}$, $F^p_{\rm lim}=\oplus_{r\ge p} H^{r,s}_{\rm lim}$, $H_{\rm lim}^{p,q}$ 
projects isomorphically to $H^{p,q}({\rm Gr}^W_{p+q})$  and  $H^{p,q}_{\rm lim}={\overline H^{p,q}}_{\rm lim}\ ({\rm  mod}\ W_{p+q-2})$.     
From  the Feynman integral point of view, one would like to have an application of this 
structure  like for the conifold, i.e., one that relates the branching behavior of the periods  
to the singularity type  in the fibre and predicts something  concrete about the integral. Consider a filtration of a complex 
$K^\bullet$, with the standard definition $\rd\cdot \rd=0$ and the cohomology $H^\bullet(K^\bullet)=\oplus_{p\ge 0} H^p(K^\bullet)$, where  
$H^p(K^\bullet)=\frac{Z^p}{\rd K^{p-1}}$ and $Z^p={\rm ker}\{\rd:K^p\rightarrow K^{p+1}\}$ are cycles while $\rd K^{p-1}=B^p\in Z^p$ 
are boundaries. One gets a spectral sequence $E_r=\oplus_{p,q\ge 0} E^{p,q}_r$ with 
$\rd_r:E_r^{p,q}\rightarrow E_r^{p+1,q-r+1}$, $\rd_r\cdot \rd_r=0$  and $H^\bullet(E_r)=E_{r+1}$, and this spectral sequence typically converges, i.e.,  there is 
an $r_0$ so that $E_r=E_{r+1}=\ldots\eqqcolon E_\infty^{p,q} $ for all $r\ge r_0$. One says that  $E_r$ \emph{degenerates} at $r_0$  and \emph{abuts} to $H^\bullet(K^\bullet)$. 
The spectral sequence that comes from the filtration of $K$  is given by 
\begin{equation}
\begin{array}{rl} 
E_0		&=	F^p K^{p+q}/ F^{p+1}K^{p+q}, \\  
E_1^{p,q}	&=	H^{p+q}({\rm Gr}^p K^\bullet),\\ [1 mm]  
E_r^{p,q}	&=	\displaystyle{\frac{\{ c\in F^p K^{p+q}|\rd c \in F^{p+r} K^{p+q+1}\}}{\rd F^{p-r+1} K^{p+q+1}+ F^{p+1} K^{p+q}}},\\ [2mm]
		& \, \vdots	 \\
 E^{p,q}_\infty&={\rm Gr}^p(H^{p,q}(K^\bullet))\ .
 \end{array} 
 \end{equation}    
 A statement that relates logarithmic degenerations generally  to the structure of the singularity of the fibre $M_0=\pi^{-1}(0)$ in 
 $\pi:M\mapsto \frak{D}$, or rather to its resolutions, is as follows: If by a chain of blow ups the singularity 
 of  $M_0$ can be made to a reduced divisor $E$ with normal crossing components $E_i$ for $i=0,\ldots, k$, then 
 the Hodge spectral sequence based on $F^\bullet_{\rm lim}$ degenerates at $_{F_{\rm lim}}E_1^{pq}$, while the \emph{monodromy weight spectral  
 sequence} degenerates at $_WE_2$, and one has~\cite{MR2393625,MR3822913}
 \begin{equation} 
 _WE_1^{p,q}=\bigoplus_{k=0}^{2p} H^{q+2(p-k)}(E[2 p-k],\mathbb{Q})\ ,
 \end{equation} 
 where we define $E[k]$ by the disjoint union as $E[k]=\coprod_{I} E_{i_0} \cap  E_{i_1} \cap \cdots \cap E_{i_k}$, $I=\{i_0<i_1<\cdots <i_k\}$.  In particular, $E^{p,0}_1$ is 
 identified with the  complex $0\rightarrow H^0(D[0])\rightarrow H^0(D[1])\rightarrow \cdots \rightarrow H^0(D[2 n])\rightarrow 0$  and 
 ${\rm Gr}_0^WH^k_{\rm lim}(M_0)=H^k(\Gamma_I)$, where $\Gamma_I$ is the dual intersection complex of $M_0$. What we 
 said about the limiting mixed Hodge  structures and monodromies in this section does not require the Calabi-Yau 
 property $c_1(M)=0$. Calabi-Yau manifolds, however, have generically a MUM-point.

A point of  \emph{maximal unipotent monodromy} or MUM-point fulfills the following conditions:
\begin{enumerate}[label=(\roman*.)]

	\item  The point is defined by  $P=\bigcap_{i=1}^r D_i$ in an $r$ dimensional moduli space ${\cal M}_\text{cs}$  
and all monodromies $T_{\gamma_{D_i}}$  corresponding to the loops around all normal crossing divisors $D_i$ are unipotent.

	\item One has ${\rm dim}(W_0)={\rm dim}(W_1)=1$ and ${\rm dim} (W_2)=1+r$.

	\item For a basis  $g^0,g^1,\ldots, g^r$ of 
$W_2$, with $g^0$ a basis of $W_0$  and $N_k$ defined by eq.~\eqref{eq:defNk}, the $r\times r$ matrix $m_k^j$, $j,k=1,\ldots, r$  
defined by  $N_k g^j=m_k^j g^0$ is invertible.

\end{enumerate} 
These criteria given in ref.~\cite{MR1416348} might be sufficient, but there are easier necessary 
conditions in many cases. For example, for one-parameter Calabi-Yau three-folds with fourth-order differential operators, it is 
sufficient that all four local exponents $\alpha$ are equal~\cite{MR3822913} at $z_0$ to make  $z_0$ a MUM-point. More 
generally, one can characterize the MUM-point by demanding that at $z_0$ the solutions of the leading order symbols of the complete set of generators of the 
differential ideal should  be $(r+1)$-fold degenerate with local exponent $\alpha$, say, and that there is one  normalized holomorphic solution $\varpi_0({\underline z})=z^\alpha+{\cal O}({\underline{z}}^{1+\alpha}) $ 
with this  local exponent and $r$ independent single logarithmic solutions of the form $\varpi_k=\frac{1}{2 \pi i} \varpi_0({\underline z})\log(z_k)+
\Sigma_k({\underline z})$~\cite{MR1316509}, where we have chosen $\Sigma_k({\underline z})=z_k^{\alpha+1}+ {\cal O}({\underline{z}}^{\alpha+2})$. That replaces a.)-c.). 

More generally, let $I_p$ an index set of order $|I_p|=p$ and define the Frobenius basis:  
\begin{equation} 
S_{(p),k}(\zs)={\small{\frac{1}{(2\pi i)^{p} p!}}}\sum_{I_{p}} \kappa_{(p),k}^{i_1,\ldots, i_{p}}  \varpi_0({\underline z})  \log(z_{i_1})\cdots  \log(z_{i_{p}})+ {\cal O}(\uz^{1+\alpha})\, ,
\label{eq:frobenius_S} 
\end{equation}
where ${\cal O}(\uz^{1+\alpha})$ can also include logarithmic terms of total power up to $p-1$. The isomorphism in eq.~\eqref{eq:conditions} (ii.) implies that there is a non-degenerate pairing over $\mathbb{Q}$  
between the solutions $S_{(n-p),k}$  and $S_{(p),k}$ for  $k=1,\ldots,|S_{(p)}(\zs)|= |S_{(n-p)}(\zs)|$  and $p=0,\ldots, n$. 
Here $|S_{(p)}(\zs)|$ denotes the total number of solutions which are of leading order $p$ in $\log(z_i)$.  
The statement is, roughly, that  solutions of degree $n-p$ ($n=l-1$) in 
the logarithms $\log(z_k)$ are dual to solutions of degree $p$ in the logarithms. In particular, the 
unique holomorphic solution with $p=0$, $S_{(0),1}=\varpi_0({\underline z})$, is dual to the unique solution which is of maximal degree $n$ in the logarithms. The 
paring of the other solutions in this Frobenius basis depends on the details of the  intersection numbers 
$\kappa_{(p),k}^{i_1,\ldots, i_{p}}$.  To define this  pairing over  $\mathbb{Z}$ and to get a basis of solutions 
that correspond to period integrals over an integral basis of cycles in $H_n(M_n,\mathbb{Z})$, 
one has to analyze the pairing in eq.~\eqref{eq:pairingG} and the map $\frak{M}^{-1}$ defined by the $\widehat \Gamma$-class, as explained in section \ref{subsection: MirrorSymmetry}.  
The basis change which transforms the Frobenius basis in eq.~\eqref{eq:frobenius_S} to this integer basis is \emph{triangular} with respect to 
the grading by the logarithmic  degree, i.e., it adds to solutions of degree $p$ in the logarithms only solutions of lower degree in the logarithms, with 
coefficients that depend on the global topology of $M_{n}$ and $\zeta$ values.    
  
Note that  the leading symbols of the differential ideal allow one to calculate the $\kappa_{(0),0}^{i_1,\ldots, i_{n}}$ up to a constant. In the case of mirror 
symmetry, these are the classical intersections $\kappa_{(0),0}^{i_1,\ldots, i_{n}}=D_{i_1}\cap \cdots \cap D_{i_n} $ of the basis of 
divisors in  the mirror $W_n$ to $M_n$, and the pairing over $\mathbb{Q}$ can be identified  with the intersection pairing in its Chow ring.  Much of the structure of the logarithmic 
solutions will survive if $n={\rm dim}(M_n)$ is replaced with the size of the Jordan block  for a non-complete degeneration.   For 
example, in ref.~\cite{MR3822913} it is argued that  the limiting mixed Hodge structure of one-parameter Calabi-Yau three-folds can be one 
of the following types:
\begin{itemize}

	\item The generic point $F$ is characterized by generic local exponents.
	\item The conifold point $C$  has local exponents $(a,b,b,c)$ and a single $2\times 2$ Jordan block.
	\item The $K$ point has local exponents $(a,a,b,b)$ and two $2\times 2$ Jordan blocks.
	\item Finally, the MUM-point $M$ has local exponents $(a,a,a,a)$ and a $4\times 4$ Jordan block. 
	
\end{itemize}
Here different characters stand for different rational numbers and the \emph{limiting mixed Hodge diamond} $H^{p,q}_{\rm lim}$ in eq.~\eqref{eq:Hlim} for the different 
degenerations at $F-,C-,K-,M-$ points are depicted below:           
$$
\begin{array}{rl}
F&:\quad  
\begin{array}{ccccccc}
&\phantom{000}&\phantom{000}&\phantom{0}0 \phantom{0}& \phantom{000}&\phantom{000}&\\
&&0&&0&&\\
&0&&0&&0&\\
1&&1&&1&&1\\
&0&&0&&0&\\
&&0&&0&&\\
&&&0&&&\\
\end{array}\qquad \qquad \qquad 
C:\ \quad 
\begin{array}{ccccccc}
&\phantom{000}&\phantom{000}&\phantom{0}0 \phantom{0}& \phantom{000}&\phantom{000}&\\
&&0&&0&&\\
&0&&1&&0&\\
1&&0&&0&&1\\
&0&&1&&0&\\
&&0&&0&&\\
&&&0&&&\\
\end{array}
\end{array}\, \phantom{.}
$$
\begin{equation} 
\begin{array}{rl}
K&:\quad 
\begin{array}{ccccccc}
&\phantom{000}&\phantom{000}&\phantom{0}0 \phantom{0}& \phantom{000}&\phantom{000}&\\
&&0&&0&&\\
&1&&0&&1&\\
0&&0&&0&&0\\
&1&&0&&1&\\
&&0&&0&&\\
&&&0&&&\\
\end{array} \qquad  \qquad \qquad M: \quad
\begin{array}{ccccccc}
&\phantom{000}&\phantom{000}&\phantom{0}1\phantom{0}& \phantom{000}&\phantom{000}&\\
&&0&&0&&\\
&0&&1&&0&\\
0&&0&&0&&0\\
&0&&1&&0&\\
&&0&&0&&\\
&&&1&&&\\
\end{array}\,.
\end{array}
\label{eq:exampleslimitone}
\end{equation} 
For example, the previous considerations allow us to completely classify the singular points of the Calabi-Yau three-fold associated to the four-loop banana integral. The complex moduli space is ${\cal M}_{\textrm{4-loop}}=\mathbb{P}^1\setminus \{z=0,1/25,1/9,1,\infty\}$. The local exponents of the singular points are summarized in the Riemann $\mathcal{P}$-symbol in eq.~\eqref{Riemann34}. In particular, using eq.~\eqref{eq:exampleslimitone}, we immediately see
that $z=0$ is a $M$-point (MUM-point), $z=1/25,1/9,1$ are $C$-points (conifolds), and 
$z=\infty$ is a $K$-point.

The SL$(2,\mathbb{C})$ orbit theorem~\cite{MR382272}  extends the  standard SL$(2,\mathbb{C})$ Lefshetz 
decomposition on polarized K\"ahler structures~\cite{MR1288523} to the limiting Hodge 
structure. Using this one can see for example that a $3\times 3$ Jordan block in the above decomposition
is not possible.  The  generic form of the limiting Hodge structures for multi-parameter 
families has been studied in refs.~\cite{MR3505643,MR4012553}. These works characterize  the types of \emph{limiting  mixed Hodge structures}
that can occur in these cases, and one finds as a consequence which types of critical divisors can intersect. 
A generic feature is that at the MUM-point the horizontal middle
cohomology with the degeneracies $1=h_{n,0},h_{n-1,1}^\text{hor},\ldots, h_{1,n-1}^\text{hor}, h_{0,n}=1$ is 
mapped to the vertical entries of the limiting mixed Hodge structure $1=h_{\rm lim}^{0,0},  
h_{\rm lim}^{1,1},\ldots,  h_{\rm lim}^{n-1,n-1}, h_{\rm lim}^{n,n}=1$ as explained in refs.~\cite{Klemm:1996ts,Bizet:2014uua}. This has clearly 
bearings on the degenerations that can occur in maximal cut Feynman integrals as illustrated at the end of  
section~\ref{subsection: MirrorSymmetry}.  Since the K\"ahler potential is given  in terms of the 
periods in eq.~\eqref{eq:Kaehlerpotential}, the mixed Hodge structure in eq.~\eqref{eq:Hlim}, together  with 
eq.~\eqref{eq:frobenius_S}, determines its leading  logarithmic degeneration. One can therefore determine the leading 
behavior of the Weil-Peterssen metric and  distinguish for example whether the critical divisors are at 
finite or infinite distance from the bulk of the moduli space. This leading behavior  is enough to make statements about the  
swampland distance conjectures for Calabi-Yau three-folds, see refs.~\cite{Blumenhagen:2018nts,Grimm:2018ohb,Joshi:2019nzi} and~\cite{Palti:2019pca} for a review. 
The exact metric has been fixed using the Barnes integral representation and derivatives of the gamma function at the MUM -point before~\cite{MR1115626,Klemm:1992tx,MR1316509,Hosono:1994ax}.  
In ref.~\cite{Joshi:2019nzi}  the Weil-Peterssen metric was determined exactly at the possible degenerations ~\eqref{eq:exampleslimitone} of hypergeometric one-parameter Calabi-Yau three-folds.

\paragraph{The Frobenius basis and the integer basis.}
\label{para:Frobeniusandintegerbasis }

So far, even if all classical constants in the symmetric tensors  $\kappa_{(p),k}$ have been determined, the 
basis in eq.~\eqref{eq:frobenius_S} for the periods is only what is called a \emph{Frobenius basis} by mathematicians. In fact, 
eq.~\eqref{eq:bryantgriffith} extends to the singular locus $z_0$ and gives non-trivial relations  between the intersection numbers.  
A Frobenius basis does not correspond to a basis of cycles $\Gamma$ for $H_{n}(M_n,\mathbb{Q})$, sometimes called a \emph{Betti 
basis} and of course not to a basis  in $H_{n}(M_n,\mathbb{Z})$, which we call an \emph{integral basis}. A maximal cut 
integral does  correspond to an integral over an element in  $H_{n}(M_n,\mathbb{Z})$, and so the latter
has to be found if one is interested in the maximal cuts computed with an integral basis of cycles. A basis transformation from a generic Frobenius basis to a rational or integral basis will involve interesting 
\emph{transcendental  numbers}. 

A pedestrian way to construct an integral basis proceeds by the following method: Resolve the critical loci in the 
moduli space to divisors with normal crossings. After that step, we construct near sufficiently many points $\uz_i$ 
and, in particular, around the intersections of the critical divisors, local Frobenius bases of solutions $\tilde \Pi_{\uz_i}$ of the 
Picard-Fuchs operators. Here, `sufficient' means that the finite regions of convergence of the  $\tilde \Pi_{\uz_i}$ define sufficiently 
many overlapping patches $U_i$ to cover  ${\cal M}_\text{cs}$. Once one has picked such a system of local solutions  
$\tilde \Pi_{\uz_i}$, $i=1,\ldots,s$, one finds a global basis by analytic continuation of the solutions into all patches. 
Between neighboring patches $U_i$ and $U_j$ with $U_j\cap U_j\neq \emptyset$, one can construct numerically 
connection  matrices $C_{ij}$ such that $\tilde \Pi_{\uz_i}=C_{ij} \tilde \Pi_{\uz_j}$ between patches,  eventually in intermediate steps 
to achieve the necessary numerical precision.  In  this globally defined  basis one can construct the 
simultaneous action of all independent generators of the monodromy group $\Gamma_{\cal M}$. 
The latter generate (not freely, but with the so-called Van Kampen relations) the monodromies around all critical divisors. 
A basis change, involving transcendental  entries, makes all these generators simultaneously elements of ${O}({\bf \Sigma},\mathbb{Z})$ and 
leads up to conjugation in  ${O}({\bf\Sigma},\mathbb{Z})$  to the desired integer basis of solutions  $\Pi$.  

This complicated procedure can be much simplified if one knows certain integral geometric cycles a priori. For 
example, at the conifolds the vanishing $S^n$-spheres can be identified, and the corresponding integral 
over $\Omega$ can be perturbatively  performed to low order in the moduli  to get the exact normalization.  
Most information concerning the integral basis can be extracted at the MUM-point. As we have seen in 
eq.~\eqref{eq:max_cut_feyn_par}, the  holomorphic solution at the MUM-point  is an integral over an $n$-torus ${\bf T} = T^n = T^{l-1}$, 
which can be performed by taking residues, cf.~eqs.~\eqref{cutMHS} and \eqref{eq:cicyT}.  
Consider now the unique cycle ${\bf S}$, whose period degenerates at the MUM-point with the 
highest power of the logarithms  i.e. with order  $\log^n$.  The latter is  dual to {\bf T}, whose period 
has no logarithm, according to eq.~\eqref{eq:conditions} (ii),  as explained after eq.~\eqref{eq:frobenius_S}. Since 
both cycles are unique, they are dual with respect to the intersection form ${\bf\Sigma}$. The cycle ${\bf S}=S^n$ 
corresponds  to the $S^n$-sphere that vanishes at the conifold  locus that is nearest to the MUM point 
under consideration. This property of the dual periods is known to hold quite generally and  plays an 
important role in homological mirror symmetry, where the shift monodromy by the maximal degenerating 
cycle of the MUM point at the nearest  conifold is known as Seidel-Thomas  twist~\cite{MR1831820}.
   
It is a very remarkable fact found in  ref.~\cite{Bonisch:2020qmm} that for the banana integral the maximal cut integral that corresponds to the period over this ${\bf S}$ yields  
the imaginary part of the integral above threshold. We have thus been able to identify two distinguished cycles in the Calabi-Yau: The sphere ${\bf S}$ 
provides the imaginary part above threshold. It corresponds in loop momentum space to the maximal cut contour $\Gamma_{\textrm{Im}}$ 
from eq.~\eqref{eq:optical_thm}, and so it has a direct physical interpretation and relevance. The torus ${\bf T}$ considered in ref.~\cite{Vanhove:2018mto}, 
which corresponds in loop momentum space to the maximal cut contour $\Gamma_T$ from eq.~\eqref{eq:max_cut_feyn_par}, 
does not seem to have any known physical interpretation. Its importance, however, lies in the fact that 
it furnishes the unique holomorphic period at the MUM-point. This integral allows one to reconstruct the 
generators of the  Picard-Fuchs differential ideal and plays an important role in understanding  
this ideal (cf. ref.~\cite{Bonisch:2020qmm}, as well as section~\ref{sec:Bananadimreg}). 
Amazingly, these two distinguished cycles are exactly the dual cycles that play a
crucial  role in homological mirror symmetry, as discussed above.

 Finally, let us mention the important observation by Deligne that the mixed Hodge structure becomes a mixed Hodge-Tate 
 structure at the MUM-point~\cite{MR1416353}. In this situation one  expects that the $\widehat \Gamma$-class  
 governs the integral structure  and the \emph{transcendental weight} of the periods. As it turns out, this 
 $\widehat \Gamma$-class (and the closely related Mellin-Barnes representation of the 
 banana integrals) are the most effective analytic tools to find the integral basis and to perform
 some of its analytic continuations, respectively. We explain in section \ref{subsection: MirrorSymmetry} that 
 the  $\widehat\Gamma$-class can also be extended to include the inhomogeneous solutions, i.e., the full Feynman 
 integral. In section \ref{sec:mellin-barnes} we show that the Mellin-Barnes representation 
 of the banana integral allows one to prove this $\widehat\Gamma$-conjecture if one considers the right contours.  
 We  illustrate the construction of the integral basis including inhomogeneous  solutions for the four-loop 
 equal-mass banana graph in section \ref{para:4lequalmass}, where also the generators of the 
 monodromy group are given and the transcendental weights are specified.

\subsection{Mirror symmetry and the $\widehat \Gamma$-conjecture at MUM-points}
\label{subsection: MirrorSymmetry}

In this section we explain how we can determine the integral basis using \emph{mirror symmetry} 
if a pair mirror manifolds $(M_n,W_n)$ is known. This turns out to be easy at the MUM-point. 
Let us recapitulate the main ideas that underly this application. 
   
There is a quite  \emph{remarkable fact}, namely that  Calabi-Yau manifolds come quite 
generically in \emph{mirror pairs} $(M_n,W_n)$.  This can  be understood as the exchange of 
two deformation-- or moduli spaces. It has precise implications on how a 
banana integral can degenerate, e.g., in the large momentum regime. 

So far we have only described the complex structure moduli space  
${\cal M}_\text{cs}$ of the manifold $M_n$. One can depicture the  infinitessimal directions of  this moduli 
space as infinitessimal deformations  $\delta g_{\tilde \imath \tilde \jmath}$  of the Calabi-Yau K\"ahler metric  $g_{i\bar \jmath}$ that 
preserve the Calabi-Yau property namely its Ricci-flatness, $R_{i\bar \jmath}(g_{i\bar \imath}+\delta g_{\tilde \imath \tilde \jmath})=0$. 
While the K\"ahler metric  in a given complex structure has mixed index structure ${i\bar \jmath}$, the deformation 
$\delta g_{\tilde \imath \tilde \jmath}$ can have any  index structure. It is clear that the pure deformations  
correspond to the \emph{complex} complex structure  deformations, which change the meaning of the unbarred and barred indices.  
Morevover, using the Weitzenb\"ock formula~\cite{MR1288523}, one shows that the latter are related to harmonic forms  spanning $H^{1}(M_n,TM_n)$, i.e., to 
complex structure deformations~\cite{MR2109686}. Deformations with mixed  index structure  are identified with  \emph{real} K\"ahler structure deformations. They 
correspond  to a choice of the K\"ahler form $\omega$ as a real linear combination of $h^{1,1}(M_n)$ 
harmonic $(1,1)$-forms. The  classical  K\"ahler moduli space is hence of real dimension $h^{1,1}(M_n)$. It has been 
suggested by type II string theory on $M_n$ (see ref.~\cite{MR3965409})  that  one should complete the choice of  the 
real  K\"ahler form  $\omega$ with the choice of real Neveu-Schwartz B-field $b=b_{i\bar \jmath} d x^i d \bar x^{\bar \jmath}$  
whose equations of motion imply that it is also a harmonic $(1,1)$ to describe a   \emph{complexified} K\"ahler moduli space.  
Let us fix topological curve classes  $C_i$ for $i=1,\ldots, h^{1,1}$ in $H_{2}(M_n)$ dual to a reference basis of $H^{1,1}(M_n)$ 
at $t_0$  on $M_n$.  The independent K\"ahler parameters of the large volume  Calabi-Yau  $n$-fold $M_n$  are identified  
with the complexified  areas   
\begin{equation}
t^i=\int_{C_i}( \omega + i b)\,, \qquad\text{for } i=1,\ldots, h^{1,1}(M_n) \,.
\label{eq:complexKaehler} 
\end{equation}  
These curves\footnote{Recall that the area of a curve $(C_i)$ is given by $\text{area}(C_i)=\int_{C_i}\omega$.} parametrize the 
\emph{complexified K\"ahler moduli space} ${\cal M}_\text{cKs}(M_n)$. 

 The mirror symmetry conjecture states that for a Calabi-Yau $n$-fold $M_n$ there is a Calabi-Yau  $n$-fold $W_n$ so 
 that the structures 
 \beq
 	H^{p,q}(M_n)	\cong	H^{n-p,q}(W_n)        
\label{eq:mirror}     
\eeq
are identified.   On the one hand, as reviewed in section~\ref{subsec:CYmanifolds}, the infinitesimal  complex structure deformations are described by the cohomology groups 
$H^1(M_n,TM_n)\cong H^{n-1,1}(M_n)$. They are unobstructed and the dimension of $\mathcal M_\text{cs}^{h^{n-1,1}}(M_n)$ is  
$h^{n-1,1}$, as indicated. On the other hand, we know from eq.~\eqref{eq:complexKaehler} that the complexified K\"ahler 
moduli space ${\cal M}_\text{cKs}^{h^{1,1}}(M_n)$ has dimension $h^{n-1,1}(M_n)$. So, schematically, mirror symmetry states 
that the structures associated to the following moduli spaces are identified:\footnote{The 
statement also applies to $\mathrm{K3}$ surfaces where $H^{1,1}(M_{\mathrm{K3}})\sim H^{1,1}(W_{\mathrm{K3}})$, and there is anyways 
only a universal  $\mathrm{K3}$, albeit in a more subtle sense. The exchange means in this case  that the 
polarization is changed, so that the role of the transcendental-- and the holomorphic cycles are exchanged. More generally, 
in the symmetric cohomology groups, like $H^{\frac{n}{2},\frac{n}{2}}(X)$ for $n$ even, one can define vertical-- and horizontal 
pieces, that get exchanged.}
\begin{equation}
{\cal M}_\text{cs}^{h^{n-1,1}(M_n)}(M_n) \Longleftrightarrow {\cal M}_\text{cKs}^{h^{1,1}(W_n)}(W_n) \quad\text{and}\quad  {\cal M}_\text{cKs}^{h^{1,1}(M_n)}(M_n) \Longleftrightarrow {\cal M}_\text{cs}^{h^{n-1,1}(W_n)}(W_n)\, . 
\end{equation}
Note that  eq.~\eqref{eq:mirror}  corresponds to a $90$ degree rotation of the Hodge diamond of $M_n$ relative to the one of $W_n$,
in which the  unique $(n,0)$-- and $(0,n)$-forms  on $M_n$ and  the unique $(0,0)$-cohomology  
element  and the unique $(n,n)$-volume form on $W_n$, respectively, are identified. In other words, mirror symmetry exchanges 
the vertical-- and horizontal cohomologies and their associated structures. In particular, it exchanges $H^n_\text{hor}(M_n)$ with  $H_\text{vert}=\oplus_{k=0}^n H^{k,k}_\text{vert}$. 
This is also what we see when comparing the middle cohomology of  $M_n$ with the limiting mixed Hodge structure at the MUM-points  
according to eq.~\eqref{eq:Hlim}, for example when comparing the $F$-point with the $M$-point in eq.~\eqref{eq:exampleslimitone}. Using monodromy 
considerations, the following \emph{mirror map} can be  identified at the MUM-points:
\begin{equation} 
t^k(z)=\frac{S_{(1),k}(\zs)}{S_{(0),0}(\zs)}=\frac{1}{2 \pi i} \left( \log(z_k)+  \frac{\Sigma_k(\zs)}{\varpi_0(\zs)} \right) \qquad\text{for } k=1,\ldots, h^{11}(W_n)=h^{n-1,1}(M_n) \ .
\label{eq:mirrormap}
\end{equation}  
The last ingredient is the \emph{homological mirror symmetry conjecture}  which states the equivalence of the derived  categories on $M_n$ and $W_n$: 
\begin{equation}
	\begin{array}{rl}
		&D^\pi(\text{Fukaya} (M_n))\\
		&\text{the bounded derived Fukaya}\\ 
		&\text{category of } M_n    
	\end{array} 
		\quad \Longleftrightarrow \quad 
	\begin{array}{rl}
		&D^b(\text{Coh}(W_n))\\
		&\text{the bounded derived category}\\ 
		&\text{of coherent sheaves on } W_n  \, ,  
 	\end{array}   
 \end{equation} 
According to the conjecture, mirror symmetry is really supposed be an order two isomorphism $\frak{M}:  D^b(\text{Coh}(W_n)) \rightarrow D^\pi({\rm Fukaya} (M_n))$ 
between these categories respecting all structures. 
The objects in $D^\pi({\rm Fukaya} (M_n))$ are \emph{Lagrangian cycles} supporting local systems. The definition of the 
Lagrangian cycles  $L$ uses the symplectic structure $L|_\omega=0$, as it is familiar from classical mechanics.   
They have real dimension $n$ and can be characterized by their homology classes  $\Gamma$ in $H_n(M_n,\mathbb{Z})$ that 
carry a \emph{mass} $M_\Gamma(\zs)$ given by $M_\Gamma(\zs)=e^{K(\zs)/2} |Z_\Gamma(\zs)|$ related to the period  or charge $\Pi_\Gamma(\zs)$.  The 
objects in $D^b(\text{Coh}(W_n))$ are \emph{coherent holomorphic sheaves}. They are  supported on holomorphic sub-manifolds 
and carry additional bundle structures and can be characterized by their class ${\cal G}$ in  the 
algebraic $K$-theory group $K_\text{alg}^0$. A key point is that, on the one hand, their \emph{charge } $\Pi_{{\cal G}}({\underline{t}})$ can 
be calculated using the $\widehat \Gamma$-class of ${\cal G}$ in the large volume regime in 
terms of classical intersections of divisors and  characteriztic classes on $W_n$, and on the other 
hand, they can be identified at the MUM-points  with the periods of the mirror using the mirror map in eq.~\eqref{eq:mirrormap} as 
$\Pi_{{\cal G}}(\underline{t})=\Pi_{\frak{M}({\cal G})}(\underline{t})$. Here we introduced the convention that $\Pi_{\Gamma}(\underline{t})$ is evaluated 
in the K\"ahler gauge  $X^0=S_{(0),0}=1$. Note  that $M_\Gamma(\zs)$ is invariant  under K\"ahler gauge transformations.

The motivation for defining the $\widehat \Gamma$-class orginated in the idea of 
identifying the pairing in both categories more naturally.  Both categories have such a pairing 
between the charge classes of objects  and auto-equivalences that leave the pairing invariant. 
In the Fukaya category the pairing is induced  by the intersection pairing coming from ${\bf \Sigma}$ (we abbreviate it as $\Gamma \circ \Gamma'$),   
and the  auto equivalences  can be identifyied with the monodromy group action on  $\Gamma$. 
The natural pairing for objects in $D^b(\text{Coh}(W_n))$,  after mapping  $K_\text{alg}^0$ to $H_\text{vert}$ using the Chern map, 
is the Euler pairing 
\begin{equation} 
{\cal G}\circ {\cal G}'=\int_{W_n} {\rm ch}({\cal G}^*)\,    {\rm ch}({\cal G}' )\, {\rm Td} (TW_n)\ .
\label{eq:pairingG}
\end{equation}  
The \emph{Strominger-Yau-Zaslow conjecture} implies that the sky-scrapper sheaf ${\cal O}_\text{pt}$ and the structure
 sheaf   ${\cal O}_{W_n}$ are mapped  to $\frak{M}({\cal O}_\text{pt})={\bf T}$   and $\frak{M}({\cal O}_{W_n})={\bf S}$, 
where the classes of the two special Lagrangian $n$-cycles  ${\bf T,S}$ have been specified at the end of section \ref{sssec:boundary}. 
In simple cases, $\Pi_{\bf S}$ could be analytically continued  to the MUM-point and some data of  $Z_{{\cal O}_{W_n}}(t)=\Pi_{\bf S}$ 
for  ${\cal O}_{W_n}= \frak{M}({\bf S})$ were known for three-folds, like the famous $\zeta(3)\chi(W_n)$ term \cite{MR1115626} and 
the $\frac{(2 \pi i)^2}{2 4} c_2\cdot D$ terms~\cite{Hosono:1994ax}.  The Todd class ${\rm Td}$ is a multiplicative class generated 
by \cite{MR1335917}
\begin{equation} 
\frac{x}{1-e^{-x}}\ .
\end{equation}  
The $\widehat \Gamma$-class proposal~\cite{MR2282969,MR2683208,MR2483750,MR3536989} is to take  a `square root' of the Todd  class using the following identity 
\begin{equation} 
\Gamma\left(1+\frac{x}{2 \pi i}\right)  \Gamma\left(1-\frac{x}{2 \pi i}\right)=e^{-x/2} \frac{x}{1-e^{-x}} \,,
\end{equation}          
and define the $\widehat \Gamma$-class by 
\begin{equation}
	\widehat \Gamma(TW_n)=\prod_{i} \Gamma\left(1+\frac{\delta_i}{2\pi i}\right)=\exp\left(-\gamma c_1(TW_n)+\sum_{k\ge 2}(-1)^k (k-1)!\, \zeta(k)\, {\rm ch}_k( {\cal G})\right) \, ,   
\end{equation}  
with the Euler-Mascheroni constant $\gamma$. Here $\delta_i$ are the Chern roots of $TW_n$. The transition from the Chern characters to 
the Chern classes  $c_k$  is decribed by Newton's formula 
 \begin{equation}
 	{\rm ch}_k		=	(-1)^{(k+1)} k\, \Bigg[ \log\left( 1+\sum_{i=1}^\infty c_i\, x^i\right) \Bigg]_k \ ,
\end{equation}    
 where $[*]_k$ means to take the $k^{\textrm{th}}$ coefficient (in $x$) of the expansion of the expression $*$.
 On Calabi-Yau spaces one defines ${\cal G}\circ  {\cal G}'=\int_{W_n} {\overline {\psi({\cal G})}} \psi({\cal G}')$  with
  $\Psi({\cal G})={\rm ch}({\cal G}) \cdot \widehat\Gamma (TW_n)$. The operation $\overline {\psi({\cal G})}$ gives a sign $(-1)^k$ on elements in $H^{2k}$, and one
  gets as desired ${\cal G}\circ {\cal G}'=\Gamma\circ \Gamma'$ with $\Gamma=\frak{M}({\cal G})$ and $\Gamma'=\frak{M}({\cal G}')$.  Moreover, 
  the charges  of ${\cal G}$ in the large volume limit of $W_n$, which corresponds to a MUM-point of $M_n$, can be 
  calculated  as~\cite{MR2282969,MR2683208,MR2483750,Bizet:2014uua,Gerhardus:2016iot,MR3536989} 
  \begin{equation} 
  \Pi_{{\cal G}}({\underline{t}}) =\int_{W_n} e^{{\underline \omega}\cdot {\underline t}} \, \widehat \Gamma(TW_n){\rm  ch}({\cal G}) +  {\cal O}(e^{-{\underline t}})\ .
\label{K-theorycharge} 
 \end{equation}     
 If we know the image of the  class of the cycle of a maximal cut, we can use eq.~\eqref{K-theorycharge}
 to compute its precise asymptotic at the MUM-point.  For the banana integral the maximal cut contour related to the imaginary part of the  integral was  identified with ${\bf  S}$ in ref.~\cite{Bonisch:2020qmm}, and it has the dual 
 ${\cal G}={\cal O}_{W_n}$ with $\textrm{ch}({\cal O}_{W_n})=1$. Therefore, it was possible to extract the asymptotic expansion of the Feynman integral involving all the transcendental 
 numbers by identifying
 \begin{equation} 
 \Pi_{\bf S}(\underline{t}(\zs))=\int_{W_n} e^{{\underline \omega}\cdot {\underline t}}\, \widehat \Gamma(TW_n)+ {\cal O}(e^{-{\underline t}})=Z_{{\cal O}_{W_n}}({\underline t}) \, ,
 \label{eq:maxcutasymptotic} 
 \end{equation}  
and comparing the powers of $t^k$ on both sides using  the mirror map in eq.~\eqref{eq:mirrormap}, see also below. 
This uniquely defines the transformation from the Frobenius basis in eq.~\eqref{eq:frobenius_S} to the integer cycle basis and relates the large momentum behavior of the  banana integrals to 
the  topological data of the Calabi-Yau space $W_n$ given in eq.~\eqref{CICY}, where the dimension of the Calabi-Yau space 
$n={\rm dim}(M_n)={\rm dim}(W_n)=l-1$ is determined by the loop order $l$. We also note that the cycle ${\bf T} \sim T^{n}$ 
can be identified with the skyscraper sheaf ${\cal O}_{\textrm{pt}}$ and $\Pi_{{\cal O}_{\textrm{pt}}}=1$~\cite{MR2683208}. Hence, in this case we get no logarithm 
and the corresponding solution is the holomorphic one and ${\cal O}_W\circ  {\cal O}_{pt}={\bf S}\cap {\bf T} =1$. 
It is possible to get the full set of integral $K$-theory classes and specify a complete integral basis of periods  using eq.~\eqref{K-theorycharge}. This 
is reviewed  for three-folds and four-folds in ref.~\cite{Bizet:2014uua}, but should be possible for  all $W_{l-1}$.

The last point to make here goes beyond the case of maximal cuts and should fit into the 
framework of the third generalization of Deligne to define mixed Hodge structures on 
singular manifolds for open cycles. It was  found in ref.~\cite{Bonisch:2020qmm} and proven 
in ref.~\cite{Iritani:2020qyh} that the full banana integral, which is defined over the open cycle $\sigma_l$ in the $n+1=l$-dimensional 
Fano  variety $F$ with $W_n\subset F$ such that the $n$-dimensional $W_n$ embeds as canonical hypersurface, is determined by 
an extended $\widehat \Gamma$-class 
\begin{equation} 
\widehat \Gamma_F(T{F})=\frac{\widehat A(T{F})}{\widehat \Gamma^2(T{F})} = \frac{\Gamma(1-c_1)}{\Gamma(1+c_1)} \cos(\pi c_1) \ , 
\end{equation}            
where $\widehat A$ is the Hirzebruch A-roof genus \cite{MR1335917} and $c_1=c_1(F)\neq 0$. Using this we can get the 
asymptotic behavior of the full Feynman integral by the identification
\begin{equation} 
J_{l,\underline{0}}(\zs,0)=\int_{F} e^{\underline\omega\cdot {\underline t}}\, \widehat\Gamma _F(TF) + {\cal O}(e^{-{\underline t}})\ .
 \label{eq:fullFeynamasymptotic}
 \end{equation}  
The integer symplectic basis element that corresponds to  $J_{l,\underline{0}}(\zs,0)$ can now be determined 
by expanding eq.~\eqref{eq:fullFeynamasymptotic} in the parameters $\underline t(\zs)$. To do this one has to calculate the classical topological intersection data that occur in this expansion. Let $I^{(k)}$ a set of $k$ indices, with $1\le I^{(k)}_p\le h_{11}(F)$ for all $p=1,\ldots , k$. Then typical terms that appear are the  intersections 
of $l$ divisors $D_i$ for $i=1,\ldots, h_{11}(F)$  in $F$, i.e.{}, $\int_{F} \bigwedge_{p=1}^l \omega_{I^{(l)}_p}=\bigcap_{p=1}^l D_{I^{(l)}_p}$ 
or the  intersection of the $k^{\textrm{th}}$ Chern class  $c_k$ with $l-k$ such divisors in $F$, i.e., $\int_{F} c_{k} \bigwedge_{p=1}^{l-k} 
\omega_{I^{l-k}_p}=[c_k]\cap \bigcap_{p=1}^{l-k} D_{I^{(l-k)}_p}$ etc. The evaluation is 
feasible  by simple and fundamental techniques in algebraic geometry and fixes the numerical coefficients of 
a degree $l$ polynomial in $\underline t$ that represents $J_{l,\underline{0}}(\zs,0)$ up to ${\cal O}(e^{-\underline t})$ 
corrections. For the Fano variety $F_l$ in eq.~\eqref{CICY},
the calculation was performed in ref.~\cite{Bonisch:2020qmm} in  detail. Inserting the mirror map in eq.~\eqref{eq:mirrormap},
we can hence get the precise coefficients of the leading logarithmic terms and since $J_{l,\underline{0}}(\zs,0)$
is a solution to a linear differential equation we can uniquely combine the Frobenius solutions in eq.~\eqref{eq:frobenius_S}
to get the exact linear combination that specifies an integer basis element of periods.

\paragraph{Confirmation by the Mellin-Barnes integral representation.}
In section~\ref{sec:mellin-barnes} we will present a Mellin-Barnes integral representation for banana integrals.
Using the contours prescribed in the evaluation of residues of the Mellin-Barnes integral 
in eq.~\eqref{eq:MBready2},  we can confirm both eq.~\eqref{eq:maxcutasymptotic} and eq.~\eqref{eq:fullFeynamasymptotic}, at least for a given loop order. More precisely, we can evaluate the leading behavior of the Mellin-Barnes integral in eq.~\eqref{eq:HypGeomBanana} in the large momentum region $z_i=0$ for all $i=1,\hdots, l+1$.
At the zeroth order in $\epsilon$  we can then determine  the coefficients of the leading powers of logarithms $\log(z_k)$. 
Using again the  identification of $\log(z_k)\sim t^k$ (see eq.~\eqref{eq:mirrormap}),  we confirm, in particular, by the expansion of eq.~\eqref{MBfirstorder}
to zeroth order in $\epsilon$ exactly the predictions of eqs.~\eqref{eq:maxcutasymptotic} and~\eqref{eq:fullFeynamasymptotic} in the equal-mass case. 
From the expansion of eq.~\eqref{eq:asymp_generic} we can confirm similar predictions of the $\widehat \Gamma$-class conjecture for the generic-mass case.

\subsubsection{The four-loop equal-mass banana integral and its  (transcendental) weight}
\label{para:4lequalmass}
In this section we illustrate the abstract concepts introduced above on the example of the banana integrals. 
We also show that there are two ways one can define a notion of weight to the banana integrals. The first one is the weight of the associated cohomology groups, while the second refers the notion of transcendental weight known from the physics literature. As we will see these two notions of weight will differ, but (at least for the banana integrals) they are related in a simple way.

In general, let $X$ be a smooth projective variety defined over a number field 
$K$. We can form its periods by integrating elements of the $K$-vector space 
$H^r_{\text{dR}}(X,K)$ over elements in $H_r(X,\mathbb{Z})$. 
Conjecturally, two periods coming from $H^{r_1}_{\text{dR}}(X_1,K)$ and 
$H^{r_2}_{\text{dR}}(X_2,K)$ can only agree if $r_1=r_2$ (e.g., this would imply that only 
periods of algebraic 0-forms, i.e., of constant rational functions, can be in $K$). In particular, we then have a well-defined weight of periods 
given by the weight of the associated cohomology groups. 

As explained in section \ref{subsection: MirrorSymmetry},
we can construct with the $\widehat \Gamma$-class  in eq.~\eqref{K-theorycharge} 
an integral symplectic basis $\uPi_l^{\text{IS}}(z) = (\varpi_{l,0}^{\textrm{IS}}(z),\ldots,\varpi_{l,l-1}^{\textrm{IS}}(z))^T$ for the periods of the Calabi-Yau manifold $W_{l-1}$. We denote the transition matrix 
from the Frobenius basis $\uPi_l(z)$ in eq.~\eqref{eq:frobenius} (and we use implicitly the notation of eq.~\eqref{eq:Pi_vec}) to the integral symplectic basis $\uPi_l^{\text{IS}}(z)$ by $\frak{T}$, 
i.e.,
\begin{align}
 \uPi_l^{\textrm{IS}}(z) =   \frak{T}\, \uPi_l(z)
\end{align}
is a period vector associated with $H^{l-1}_{\text{dR}}(W_{l-1})$. Therefore we assign to the elements of $ \uPi_l^{\textrm{IS}}(z)$ the \emph{weight} $l-1$. From  eq.~\eqref{eq:fullFeynamasymptotic} (see also ref.~\cite{Bonisch:2020qmm}) we 
 have constants $\lambda^{(l)}_k$ for $k=0,...,l$ so that the banana integral is given by:
\begin{align}
    J_{l,1}(z,0)= \sum_{k=0}^l \lambda^{(l)}_k \varpi_{l,k}(z) \, ,
\end{align}
where $\varpi_{l,k}(z)$ is a special solution to the inhomogeneous equation satisfied by $J_{l,1}(z,0)$.
Now consider the extended basis
\begin{align}
   \widehat{\uPi}_l(z) \coloneqq  \left(\begin{matrix} 2\pi i\,\uPi_l^{\textrm{IS}}(z) \\J_{l,1}(z,0)\end{matrix}\right)  = \left( \begin{array}{cccc}
         &2\pi i \frak{T}&&0  \\
         \lambda^{(l)}_0 &\cdots& \lambda^{(l)}_{l-1} &\lambda^{(l)}_l 
    \end{array} \right)\,
     \left(\begin{matrix} \uPi_l(z) \\\varpi_{l,l}(z)\end{matrix}\right)
    \, .
\end{align}
In this basis the monodromy matrices are unimodular. The K\"unneth theorem implies that the weight is additive, i.e., if two periods have weights $r_1$ and $r_2$, their product has weight $r_1+r_2$. The fact that $2\pi i$ can be realized as a period of $H_{\text{dR}}^2(\mathbb{P}^1)$ (and thus it is assigned weight 2), one finds that the first $l$ elements of $\widehat{\uPi}_l(z)$ can, e.g., be realized as periods of $H_{\text{dR}}^{l+1}(\mathbb{P}^1 \times W_{l-1})$. Hence, we would associate weight $l+1$ to the first $l$ elements of $\widehat{\uPi}_l(z)$. We can now extend this to conjecturally define a weight of the banana integral. We want this weight to be invariant under the monodromy, and since the monodromy adds rational multiples of the first $l$ elements of $\widehat{\uPi}_l(z)$ to the banana integral we conclude that the weight of $J_{l,1}(z,0)$ should also be $l+1$. 

Note that this definition of weight differs from the notion of \emph{transcendental weight} encountered in physics, cf., e.g., refs.~\cite{Kotikov:2001sc,Kotikov:2002ab,Kotikov:2004er,Kotikov:2007cy,Broedel:2018qkq}. For example, the one-loop equal-mass banana integral evaluates to a logarithm in $D=2$ dimensions, which is assigned transcendental weight $l=1$, and there are arguments to assign transcendental weight $l$ to an $l$-loop equal-mass banana integral in $D=2$ dimensions, cf., e.g., refs.~\cite{Broedel:2018qkq,Broedel:2019kmn}. In the following we sketch a way to define a notion of transcendental weight compatible with the physics literature on the Calabi-Yau periods that compute the maximal cuts of $J_{l,1}(z,0)$, and we comment on how this notation is related to the notion of weight defined above. We also show that by an argument similar to the one above, we can conclude that $J_{l,1}(z,0)$ should have transcendental weight $l$, consistent with the (folkloristic) expectations from physics.

We have seen that $z=0$ is a MUM-point. Consequently, the monodromy matrix ${\bf T}_0$ is unipotent, and its logarithm ${\bf N}_0 \coloneqq -\log({\bf T}_0)$ is nilpotent (cf. eq.~\eqref{eq:defNk}). By the same reasoning as in section~\ref{sssec:boundary}, we can define an increasing (monodromy weight) filtration $W_{\bullet}$ on the $\mathbb{C}$-vector space spanned by the periods $\varpi_{l,k}(z)$ (cf. eq.~\eqref{eq:monodromyweightfiltration}), and this filtration satisfies eq.~\eqref{eq:conditions}. We then say that a period has \emph{transcendental weight} $k$ if it lies in $W_{2k}$, but not in $W_{2k-2}$. Loosely speaking, this filtration captures the logarithmic behavior of a period as one approaches the MUM-point $z=0$: elements of transcendental weight 0 are holomorphic at the MUM-point; elements of transcendental weight 1 degenerate with a single power of $\log(z)$ at the MUM-point; etc. In particular, $\varpi_{l,k}(z)$ is assigned transcendental weight $k$ according to this definition. The matrix $\mathfrak{T}$ contains products of zeta values and $i\pi$ so that the entries of $\uPi_l^{\textrm{IS}}(z)$ are then assigned transcendental weight $l-1$ (see the four-loop example below), which is the same as the notion of weight defined above. In general, however, the two notions will not give the same value. 
The same monodromy argument as above then shows that the only transcendental weight consistent with the monodromy that one can assign to $J_{l,1}(z,0)$ is $l$. Note that this is consistent with the fact that $J_{l,1}(z,0)$ degenerates like $\log^l(z)$ close to the MUM-point $z=0$.

Let us briefly compare the two definitions of weight. We see that the two definitions agree up to a factor of 2 (coming from the fact that we had assigned transcendental weight 1 to $\pi$) and an $l$-dependent offset (coming from the fact that we had assigned transcendental weight 0 to $\varpi_{l,0}(z)$). Together this gives the difference between the weight $l+1$ and the transcendental weight $l$ of $J_{l,1}(z,0)$. The definition of weight given above, however, does not apply to the singular point $z=0$. Our notion of transcendental weight, instead, is tightly connected to the existence of a singular point, which is even a MUM-point. Moreover, we stress that this definition of transcendental weight depends on the choice of the MUM-point. Indeed, the unique holomorphic period $\varpi_{l,0}(z)$ at the MUM-point $z=0$ is assigned transcendental weight 0 w.r.t. to this choice of MUM-point. However, $\varpi_{l,0}(z)$ may develop logarithmic behavior close to another MUM-point, and so it would have transcendental weight $w>0$ w.r.t. to this other choice. This is in agreement with the definition of the transcendental weight for elliptic Feynman integrals in ref.~\cite{Broedel:2018qkq}, which requires the choice of a period of the elliptic curve. It may therefore be more appropriate to talk about the \emph{transcendental weight w.r.t. a certain choice of MUM-point}. Clearly, for multiple polylogarithms the transcendental weight obtained in this way is independent of this choice, and we recover the usual definitions from refs.~\cite{Kotikov:2001sc,Kotikov:2002ab,Kotikov:2004er,Kotikov:2007cy}, which do not require such a choice to be made.

\paragraph{The  $l=4$ example.}
\label{para:l=4}

We have already seen that at four loops we have ${\cal M}_{\textrm{4-loop}}\coloneqq\mathbb{P}^1\setminus \{z=0,1/25,1/9,1,\infty\}$. Furthermore, at 
$z=0$ there is a MUM-point, for $z\in \{1/25,1/9,1\}$ there are $C$-points and for
$z=\infty$ there is a $K$-point. From  eq.~\eqref{K-theorycharge}, see refs.~\cite{Bizet:2014uua,Gerhardus:2016iot} for the four-fold case, we get 
\begin{align}
    \frak{T}=(2\pi i)^3 		\begin{pmatrix}
						\frac{-8 \zeta(3)}{(2\pi i)^3} & \frac{12}{24 \cdot 2\pi i} & 0 & \frac{12}{(2\pi i)^3} \\
						\frac{12}{24} & 0 & \frac{-12}{(2\pi i)^2} & 0 \\
						1 & 0 & 0 & 0 \\
						0 & \frac{1}{2\pi i} & 0 & 0 
					\end{pmatrix} \,.
\label{eq:4l_Tfrak}
\end{align}
Moreover, from eq.~\eqref{eq:fullFeynamasymptotic},  see~ref.~\cite{Bonisch:2020qmm}  for more details, we  have
\begin{align}
    \begin{pmatrix}
         \lambda_0^{(4)}  	\\
       	 \vdots		\\
         \lambda_4^{(4)}
    \end{pmatrix}		=	\begin{pmatrix}
           					-450\zeta(4)-80\pi i\zeta(3)\\
           					80\zeta(3)-120\pi i\zeta(2)\\
     					       	180\zeta(2)\\
         					20\pi i\\
 					        -5
 					   \end{pmatrix} \, .
\end{align}
The monodromy group in the $\widehat{\uPi}_l$-basis is then generated by the unimodular matrices
\begin{equation}
\begin{aligned}
	 &{\bf T}_0 = 
		\begin{pmatrix}
			1 & -1 & 3 & 6 & 0	\\
			0 & 1 & -6 & -12 & 0	\\
 			0 & 0 & 1 & 0 & 0	\\
			0 & 0 & 1 & 1 & 0	\\
			-10 & 0 & 0 & 0 & 1
		\end{pmatrix} \,, \qquad 			&&{\bf T}_{1/25} =
										\begin{pmatrix}
											1 & 0 & 0 & 0 & 0	\\
											0 & 1 & 0 & 0 & 0	\\
 											-10 & 0 & 1 & 0 & 0	\\
											0 & 0 & 0 & 1 & 0	\\
											10 & 0 & 0 & 0 & 1
										\end{pmatrix} \,, \\
	&{\bf T}_{1/9} =
		\begin{pmatrix}
 			-9 & -2 & 2 & 0 & 0	\\
 			0 & 1 & 0 & 0 & 0	\\
			-50 & -10 & 11 & 0 & 0\\
			-10 & -2 & 2 & 1 & 0	\\
			0 & 0 & 0 & 0 & 1
		\end{pmatrix} \,,
									&&{\bf T}_1 =
										\begin{pmatrix}
 											-39 & -16 & 16 & -24 & 0	\\
 											60 & 25 & -24 & 36 & 0	\\
 											-100 & -40 & 41 & -60 & 0	\\
 											-40 & -16 & 16 & -23 & 0	\\
											0 & 0 & 0 & 0 & 1
										\end{pmatrix} \,, \\
	& {\bf T}_\infty =
		\begin{pmatrix}
			31 & 17 & -19 & 42 & 0	\\
 			-60 & -35 & 42 & -96 & 0	\\
			60 & 30 & -29 & 60 & 0	\\
 			30 & 16 & -17 & 37 & 0	\\
 			0 & 0 & 0 & 0 & 1
		\end{pmatrix} \,,
\end{aligned}
\label{eq:monodromy_mats_4l}
\end{equation}
which satisfy ${\bf T}_0 \, {\bf T}_{1/25}\,  {\bf T}_{1/9} \, {\bf T}_1 \, {\bf T}_\infty=\mathbbm{1}$.
We notice that the shift at the `conifold'  ${\bf T}_{1/25}$ is $-10$, while for an shrinking $S^3$ 
that is dual to the single logarithmic period at this singular locus one would expect from eq.~\eqref{eq:PicardLefshetz} 
that the shift is by $1$.  The reason is that, as mentioned in section~\ref{ssec:cutsbanana},  the one-parameter 
Calabi-Yau family  for the equal-mass banana integral is obtained by dividing by a symmetry group $G$. In the case 
at hand, $G=\mathbb{Z}/(10\,  \mathbb{Z})$ is divided from the manifold $\hat X^{5,45}$, or rather  from 
eq.~\eqref{eq:singbananatoric}, and one takes care of the projective resolution in a second step, as made explicit
in  ref.~\cite{candelas2019one} (the $k=1$ case). In this way $G$ acts freely on the shrinking $S^3$ and makes it into 
a lense space $L(1,10)$, as explained in ref.~\cite{Gopakumar:1997dv}.

\subsection{Comments on the Relation:  Feynman integrals versus  Calabi-Yau motives}
\label{subs:integralsmotives}

In section~\ref{ssec:cutsbanana} we saw that there are different geometric realization  
of  the periods relevant to the maximals cuts of banana integralsin $D=2$ dimensions. We introduced more 
abstractly the properties of  families of Calabi-Yau motives that can be --
and have been -- studied at least partially in their own right. For example, the Legendre family $\mathcal E_\text{Leg}$ is 
one of the four hypergeometric motives associated to elliptic curves~\cite{MR2500571}; the operator
${\cal L}_{2}$ in eq.~\eqref{eq:diffbanana} corresponds to one of the  Ap\'ery-like families of motives, which were systematically 
studied in ref.~\cite{MR2500571} (listed there as case $C$). The four-loop equal-mass 
banana integral corresponds to a one-parameter Calabi-Yau family of motives, and appears  
in the less complete list described in ref.~\cite{MR3822913} as AESZ 34. However, in 
mathematics there is an even more fundamental perspective on motives which can provide even more detailed information about Feynman 
graphs (see, e.g., refs.~\cite{Bloch:2005bh,Bloch:2008jk,marcollibook,Brown:2015fyf}). We will first give a rough 
definition for this concept and then turn to families of Calabi-Yau motives that are 
defined for higher-dimensional  Fano  varieties and local Calabi-Yau spaces and finally
summarize what the properties of Calabi-Yau motives predict concerning the properties
of Feynman integrals.

\subsubsection{The mathematical perspective on motives}
\label{motives}

The idea of \emph{motives} was proposed by Grothendieck to capture the \emph{cohomological structure of varieties}. We want to briefly explain this idea without going much into detail and giving definitions for all occurring objects. For more details we refer to Milne \cite{milne} and the survey paper by Zagier \cite{ZagierSurvey}. We start by explaining geometric motives. Let $X$ be a smooth projective variety of some dimension $n$. For simplicity, we assume that $X$ is defined over $\mathbb{Q}$ (more generally one could consider any number field). For every integer $0 \leq r \leq 2n$ there are \emph{different cohomology theories} that we can associate with $X$:
\begin{itemize}
	\item Considering the complex points on $X$, we obtain a topological space $X(\mathbb{C})$ which gives rise to the \emph{Betti cohomology group} $H^r(X(\mathbb{C}),\mathbb{Q})$. This is the $\mathbb{Q}$-vector space defined as the homology of the cochain complex tensored with $\mathbb{Q}$.
	\item Using that $X$ is defined over $\mathbb{Q}$ allows us to define the \emph{algebraic de Rham cohomology group} $H^r_{\text{dR}}(X)$. This is the $\mathbb{Q}$-vector space defined as the hypercohomology of the algebraic de Rham complex and was first considered by Grothendieck \cite{grothendieck}. This space has a Hodge filtration
  \begin{align}
    F^rH^r_{\mathrm{dR}}(X) \subseteq F^{r-1}H^r_{\mathrm{dR}}(X) \subseteq ... \subseteq F^0H^r_{\mathrm{dR}}(X) = H^r_{\mathrm{dR}}(X)
  \end{align}
  which, after tensoring with $\mathbb{C}$, gives the familiar Hodge decomposition
  \begin{align}
    H^r_{\text{dR}}(X) \otimes_{\mathbb{Q}} \mathbb{C} = \bigoplus_{p+q=r}H^{p,q}(X(\mathbb{C})) \, .
  \end{align}
	\item Letting $\overline{X}$ be the variety $X$ regarded as a variety over $\overline{\mathbb{Q}}$ one has for any prime $\ell$ the $\ell$-adic cohomology group $H^r(\overline{X},\mathbb{Q}_\ell)$. This is the $\mathbb{Q}_\ell$-vector space defined as the inverse limit of \emph{\'Etale cohomology groups}. The \emph{Galois group} $\text{Gal}(\overline{\mathbb{Q}}/\mathbb{Q})$ naturally acts on $\overline{X}$ and this induces an action on $H^r(\overline{X},\mathbb{Q}_\ell)$.
\end{itemize}

These cohomology groups are related in several ways. For example, the integration of differential forms over closed chains gives a natural isomorphism
\begin{align}
  H^r(X(\mathbb{C}),\mathbb{Q}) \otimes_{\mathbb{Q}} \mathbb{C} \cong H^r_{\mathrm{dR}}(X) \otimes_{\mathbb{Q}} \mathbb{C} \, ,
\end{align}
where the transcendental numbers occuring in this isomorphism are the periods of $X$. Further, there is a natural comparison isomorphism
\begin{align}
  H^r(X(\mathbb{C}),\mathbb{Q}) \otimes_{\mathbb{Q}} \mathbb{Q}_\ell \cong H^r(\overline{X},\mathbb{Q}_\ell) \, .
\end{align}

We now define a \emph{geometric motive} as a rational linear subspace $V \subseteq H^r(X(\mathbb{C}),\mathbb{Q})$ that is compatible with the Hodge decomposition and the action of the Galois group. By this we mean that for the complexification $V_\mathbb{C}= V \otimes_{\mathbb{Q}} \mathbb{C}$,
\begin{align}
  V_{\mathbb{C}} = \bigoplus_{p+q=r} (V_{\mathbb{C}} \cap H^{p,q}(X(\mathbb{C})))\,,
\end{align}
and that for any prime $\ell$ the action of $\text{Gal}(\overline{\mathbb{Q}}/\mathbb{Q})$ on $H^r(X(\mathbb{C}),\mathbb{Q}) \otimes_{\mathbb{Q}} \mathbb{Q}_\ell$ (induced by the comparison isomorphism) can be restricted to $V_{\mathbb{Q}_\ell} = V \otimes_{\mathbb{Q}} \mathbb{Q}_\ell$. We call $r$ the weight and $\text{dim} V$ the rank of the geometric motive $V$. The simplest example of a geometric motive is the complete space $H^r(X(\mathbb{C}),\mathbb{Q})$. Other simple examples are isotypical components of $H^r(X(\mathbb{C}),\mathbb{Q})$ with respect to the natural representation of a suitable group $G$ acting on $X$. \\

More generally, a motive can be thought of as a suitable collection of vector spaces (equipped with a Hodge decomposition and an 
action of $\text{Gal}(\overline{\mathbb{Q}}/\mathbb{Q})$ with additional compatibilities), and it does not have to come from 
any specific variety. Examples of motives that do not refer to specific varieties (although they can be 
realized in this way) are hypergeometric motives and motives associated with modular forms. Considering the motivic structure of Feynman integrals can lead to concrete arithmetic predictions as we summarize next.

\paragraph{Arithmetic Predictions.}

Many standard conjectures from algebraic geometry can be generalized 
to the theory of motives. Examples for this include the \emph{Hodge conjecture} and the \emph{Tate conjecture}. For any rational linear subspace $V \subseteq H^r(X(\mathbb{C}),\mathbb{Q})$ the Hodge conjecture would imply that $V$ is already a geometric motive if it is compatible with the Hodge decomposition while the Tate conjecture would imply that $V$ is already a geometric motive if it is compatible with the Galois action. In particular, it is an interesting question to which extend the equality of Galois representations of motives implies the 
equality of the associated periods, and vice versa. Practically, this relationship can, e.g., be studied by identifying 
Feynman integrals at special points (seen as periods) with special values of $L$-functions 
(which are completely determined by a Galois representation). For Calabi-Yau three-folds such 
relations are studied in refs.~\cite{adek} and \cite{BoenischThesis}. A very explicit and general conjecture in this 
direction is \emph{Deligne's period conjecture} \cite{MR546622}, see also  \cite{MR1852188}. 
This predicts that for a certain class of motives an associated $L$-function value is a rational multiple 
of a specific minor of the period matrix.

\paragraph{Analytic Predictions.}

One  short and simplified message about the relation between the  cut integrals and families of 
Calabi-Yau  manifolds is that the algebraic functions and 
the elliptic periods known to arise from maximal cut computations get 
conjecturally replaced (at least in some cases) by the period integrals of Calabi-Yau $(l-1)$-fold families associated 
to  the (horizontal) middle (co)homology  of rank $h^{l-1}_\text{hor}$,  or more generally, by families of 
Calabi-Yau motives of weight $(l-1)$ and rank $r$. Many of the concepts for controlling the functions analytically 
as well as numerically carry over  from the elliptic integrals to the Calabi-Yau periods defined either by the 
Calabi-Yau-geometry or the Calabi-Yau motive. For the banana graphs  this is not a vague 
conjecture anymore. Thanks to the techniques for Calabi-Yau periods and their extensions, 
the analytic and numerical properties of the cut integrals and  the full banana integral, 
can by now be well controlled at all loop orders in $D=2$ dimensions~\cite{Bonisch:2020qmm}, 
and as explained in sections~\ref{sec:Bananadimreg} and~\ref{subsec:equalmasseps}, also for general $\epsilon$. 

An observation that is trivial from the geometric point of view is the fact that, once one has identified 
the family of Calabi-Yau motives as defined in section \ref{ssec:cutsbanana} and its extensions to the inhomogeous Picard-Fuchs differential ideal (including the general $\epsilon$-dependence, see section \ref{sec:Bananadimreg}), and once one was
able to to fix the boundary conditions, say in the \emph{infrared} (by which we mean the small mass limit, $m_i^2\to 0$) as in eq.~\eqref{eq:fullFeynamasymptotic} 
(or more generall by eq.~\eqref{eq:asymp_generic}),  then the complete analytic properties of the integral 
are fixed. This includes, in  particular,  its behavior in the \emph{ultraviolet} (by which we mean the small momentum limit, $p\to 0$) or in other kinematic  
limits. Some of these physical properties are directly related to prominent geometric properties 
that are related to the global monodromies, e.g.,  the imaginary part of the banana 
integral above threshold is precisely determined by the Seidel-Thomas shift at the conifold that is  
nearest to the MUM point.  

It is also noticable that there is geometrically a prominent real and monodromy invariant quantity 
namely the exponential of the   K\"ahler potential $e^K$ in eq.~\eqref{eq:Kaehlerpotential}, which roughly 
corresponds to combinations of the absolute values of maximal cut integrals.  Remarkably, this 
quantity can be calculated in its own rights by localization techniques as the \emph{sphere partition  
function} of the string propagating in very general Calabi-Yau backgrounds~\cite{Benini:2012ui,Doroud:2012xw,Jockers:2012dk}. 
The latter are described by gauge linear $\sigma$-models, generalizing the  geometries that we have 
encountered here, namely complete intersections in toric varieties, to determinental and more general 
embeddings into Grassmannians, flag manifolds and more general ambient spaces. 
The periods on the other hand have been realized as a splitting of the sphere partition 
function into hemisphere partitions functions  with a complete set of boundary 
conditions~\cite{Hori:2013ika}.  

As it is clear from the points $7, 8$ and $9$ in table~\ref{tabledictionary} that from the geometric point of view the 
extension of the homogeneous Calabi-Yau differential ideal 
to the inhomogeneous differential equations  corresponds to 
performing a chain integral in relative cohomology  of the ambient space 
rather than an integral over closed cycles in the Calabi-Yau manifold. The corresponding 
extension of the Calabi-Yau operator is  conceptually very similar to the 
calculation of open string disk amplitudes ending on special Lagrangians in the Calabi-Yau manifold, 
as it has been pioneered for non-compact Calabi-Yau spaces in refs.~\cite{Aganagic:2000gs,Aganagic:2001nx} 
for non-compact toric special Lagrangians and for compact Calabi-Yau three-folds  for the 
Walcher special Lagrangians in ref.~\cite{Walcher:2006rs}.

In section~\ref{sec:physics_consquences} we will summarize the consequences of the mathematics of Calabi-Yau motives for Feynman integrals.
Before we do this, however, we make a comment of mathematical nature, which is interesting in its own right, but which will also play a role in the discussion in section~\ref{sec:physics_consquences}.
 
\paragraph{More general geometrical realizations of families of Calabi-Yau motives.}
We have seen in section~\ref{subsec:CYmanifolds} that, almost by definition, Calabi-Yau varieties have a vanishing first Chern class, $c_1=0$.
However, Calabi-Yau motives can  appear in the middle cohomology of manifolds 
with positive Chern class, $c_1>0$, i.e., of non-Calabi-Yau spaces in the strict sense of the definition 
given in section \ref{subsec:CYmanifolds}. A first example was discussed in ref.~\cite{Candelas:1993nd} 
in the context of a case that seems to contradict the mirror symmetry hypothesis.  The argument is as follows: There is 
an example  of a  Calabi-Yau three-fold $M_3^0$ with no complex structure  deformations (now many examples 
of these so-called `rigid' Calabi-Yau spaces are known) and $84$ elements in $H^{1,1}(M_3^0,\mathbb{Z})$. 
By the remarks in section \ref{subsec:CYmanifolds},   having no complex structure  deformations  
implies $H^{2,1}(M_3^0,\mathbb{Z})=0$. Inspecting eq.~\eqref{eq:mirror}, one concludes that it cannot have 
a mirror Calabi-Yau space $W_3^0$ because the latter should have $H^{1,1}(W_3^0)=0$, which means  that 
it  would not be K\"ahler. In ref.~\cite{Candelas:1993nd}  a Calabi-Yau motive describing 
the $84$ K\"ahler deformations as a variation of a  Hodge structure in terms of periods was 
nevertheless found  in the middle cohomology  $H^7(\tilde W_7)$, where  the seven-fold $\tilde W_7$  is 
defined  as the vanishing of a \emph{smooth} cubic $P=0$ in $\mathbb{P}^8$. One can show that  $F^7H^7(\tilde W_7)= 
F^6H^7(\tilde W_7)=0$ while  $F^5H^7(\tilde W_7)=\mathbb{C}$. Closer inspections show that $H^7(\tilde W_7)$
has a Hodge decomposition $0,0,1,84,84,1,0,0$. Moreover, the unique $(5,2)$-form can be expressed as $\oint_{P=0} \frac{\mu}{P}$ 
with $\mu$ the standard measure on $\mathbb{P}^8$  defined in eq.~\eqref{measure}. Note that this expression is scale-invariant 
and the Griffiths reduction formula can be applied to get the  Picard-Fuchs differential ideal. This was done 
at least on a symmetric slice in ref.~\cite{Candelas:1993nd}. Similar more involved  realizations of  Calabi-Yau motives (in fact elliptic curve motives) 
were discussed  for the kite-- and the double box Feynman integral in ref.~\cite{Bloch:2021hzs}.  Here the elliptic curve 
motive was identified in manifolds $M_3$ and $M_5$, respectively, which were defined  as the vanishing of \emph{singular cubics} in $\mathbb{P}^4$ and  
$\mathbb{P}^6$, respectively. In their middle cohomology an elliptic curve motive was found since $H^3(M_3)$ had a decomposition $0,1,1,0$ and 
$H^5(M_5)$ had $0,0,1,1,0,0\ $, respectively. It was explained  in ref.~\cite{Bloch:2021hzs} for the latter case 
how to perform  the blow up that resolves the singularities and how to identify the elliptic motives  at an abstract level.

\subsection{Summary of the most important structures relevant for Feynman integrals}
\label{sec:physics_consquences}

We conclude this section with a summary of the mathematical structures related to Calabi-Yau motives reviewed in this section and on how they are related to Feynman integrals. The goal is to provide at one glance the lessons learned from Calabi-Yau motives for banana integrals, and to speculate on how they may generalize to more general Feynman integrals.

\paragraph{Landman's theorem and logarithmic divergences of maximal cuts.}

As a consequence of Landman's theorem (cf.~eq.~\eqref{eq:Landman}), the periods of a family of algebraic varieties of  dimension $n$ cannot develop logarithmic divergences worse than $\log(\Delta)^n$ as $\Delta\to0$. Since the maximal cuts of Feynman integrals are expected to be periods of families of algebraic varieties parametrized by the external kinematics, the logarithmic divergences of the maximal cuts should contain information about the dimension of the family of algebraic variety (rather, the rank of the motive, see below). More precisely, if a maximal cut behaves like $\log(\Delta)^m$ in some kinematic limit $\Delta\to0$, then the dimension of the algebraic variety or the rank of the motive cannot be less than $m$. We stress that this statement is generic and applies to any algebraic variety or motive, independently if it  is of Calabi-Yau type.

\paragraph{Calabi-Yau motives for Feynman integrals.}

By now there is compelling evidence that the geometry of Calabi-Yau manifolds plays an important role for higher-loop Feynman integrals, and we have several infinite families of $l$-loop Feynman integrals associated to Calabi-Yau $(l-1)$-folds, cf., e.g., refs.~\cite{Bloch:2014qca,MR3780269,Bonisch:2020qmm,Klemm:2019dbm,Bourjaily:2018ycu,Bourjaily:2018yfy,Bourjaily:2019hmc}. This begs two immediate questions:
\begin{enumerate}[label=(\roman*)]
	\item Is it possible to assign a unique geometric object to a given Feynman graph?
	\item Are all Feynman integrals associated to Calabi-Yau manifolds (or suitable generalizations thereof), or are there Feynman graphs that lead to geometric objects that require vastly different geometries?
\end{enumerate}
Let us start by commenting on the first question. We have seen in section~\ref{ssec:cutsbanana} that there are two Calabi-Yau $(l-1)$-folds that we can associate to the $l$-loop banana integrals, namely the variety $M_{l-1}^{\textrm{HS}}$ defined by the vanishing of the second Symanzik polynomial (cf.~eq.~\eqref{eq:HSgeom}) and the variety $M_{l-1}^{\textrm{CI}}$ defined as a complete intersection in $\mathbb{P}_{l+1}$ (cf.~eq.~\eqref{MCicy}). It is known that these varieties are in general not diffeomorphic (e.g., for $l=4$ and $l=5$ they have different Euler characterictics). Hence, we can associate (at least) two non-diffeomorphic manifolds to a given banana graph. It thus seems that it is in general not possible to associate a uniquely-defined algebraic variety to a given Feynman graph. Note that for $l=2$ the Calabi-Yau one-fold $M_{1}^{\textrm{HS}}$ and $M_{1}^{\textrm{CI}}$ define the same elliptic curve, so there may be a unique variety attached to the two-loop sunrise integral, along the lines of the findings of ref.~\cite{Frellesvig:2021vdl}. Our arguments indicate a potential breakdown at higher loops of one being able to attach a unique algebraic variety to a given Feynman integral.

Instead of attaching a family of varieties to a given Feynman graph, it seems more appropriate to consider the \emph{family of motives} one can attach to it. Loosely speaking, one can think of a motive as a suitable linear subspace inside the (co)homology groups compatible with the action of the Galois group (see section~\ref{motives}). Two non-diffeomorphic (families of) varieties may define the same (family of) motives. This is indeed the case for the banana graphs: $M_{l-1}^{\textrm{HS}}$ and  $M_{l-1}^{\textrm{CI}}$ contain the same motive as part of their cohomology, even though they are distinct as manifolds, cf.~eq.~\eqref{eq:simplemirror}. We thus conclude that if one can indeed identify a unique geometric object for each Feynman graph, it should be a motive rather than an algebraic variety. 

From this perspective, the second question can now be rephrased in the following way: Are all Feynman integrals associated to (families of) Calabi-Yau motives, or do other motives make their appearance? In refs.~\cite{Huang:2013kh,Hauenstein:2014mda} several families of maximal cuts were analyzed that lead to Riemann surfaces of genus $g>1$. The latter are definitely not Calabi-Yau manifolds in the strict sense of section~\ref{subsec:CYmanifolds} (because $n=1$ and $h^{1,0}=g>1$). However, there are examples of Calabi-Yau {motives} that describe higher-genus surfaces, cf. the discussion earlier in this section, in particular refs.~\cite{MR1467889,Hori:2000kt,MR3636672} mentioned there. 
That discussion implies that it is insufficient to  look at the geometry, say defined by the ${\cal F}$-polynomial of a Feynman 
graph, and conclude  from the fact that it represents a higher-genus geometry with  $c_1<0$ that the Feynman integral   
cannot  be related to a Calabi-Yau motive. In fact, in local mirror symmetry, the  limits of Calabi-Yau period 
motives appear  frequently as motives of higher curves that can be constructed  in this limit~\cite{MR1467889,Hori:2000kt}. 
The main point is, similar to the $c_1>0$ discussed  above, that a non-standard, in this case meromorphic differential, is used to define the motive, 
and this fact cannot be detected by just looking at the curve. Explicit examples of local genus-two Calabi-Yau Picard-Fuchs differential ideals, which  
are limits of Calabi-Yau motives, have been worked out in ref.~\cite{MR3636672}.    
It would be interesting to study the examples in refs.~\cite{Huang:2013kh,Hauenstein:2014mda} in more detail. If they define Calabi-Yau motives, all known examples of Feynman integrals would be associated to Calabi-Yau motives, and it would then be tantalizing to speculate that this may be a general feature.

\paragraph{Maximal cuts and the Frobenius-basis.}

Maximal cuts play an important role in the study of scattering amplitudes and Feynman integrals. They are defined by integrating the differential form defining a Feynman integral over a contour that encircles all the poles of the propagators. Constructing an explicit basis for the maximal cut contours can be a complicated task, see, e.g., refs.~\cite{CaronHuot:2012ab,Primo:2016ebd,Primo:2017ipr,Bosma:2017ens} for concrete examples where a basis of maximal cut contours was constructed. 

However, for applications it may not be needed to construct a basis of cycles explicitly. In particular, Calabi-Yau varieties have generically a MUM-point (though exceptions to this rule are known to exist, see ref.~\cite{MR3951103}), and close to a MUM-point it is possible to construct a basis of periods, called a Frobenius basis, characterized by an increasing hierarchy of logarithms, cf. the discussion around eq.~\eqref{eq:frobenius}. The Frobenius basis provides a basis for the solution space of the Picard-Fuchs differential ideal (or, in physics parlance, the set of homogeneous differential equations for the Feynman integral), and thus also for the maximal cut integrals. However, the Frobenius basis is \emph{not an integral basis}, by which we mean a basis for the periods/maximal cuts computed by integrating over cycles from integral homology, i.e., linear combinations of cycles with integer coefficients. Instead, as we have explained in section~\ref{sssec:boundary} and illustrated on the four-loop example in section~\ref{para:4lequalmass}, the change of basis from the Frobenius basis to the integral basis involves a rotation matrix whose entries are transcendental numbers. In applications to Feynman integrals, however, the concrete form of this rotation may not be required (the important point is to have a basis). The advantage of the Frobenius basis lies in the fact that, if a MUM-point is identified, it can be constructed in a much more straightforward way than the integral basis, whose construction requires a detailed knowledge of a basis of geometric cycles. In section~\ref{subsec:non-max-cuts} we comment on how some of these concepts may generalize to dimensional regularization and to non-maximal cuts.

\paragraph{Quadratic relations among cuts.}

When solving differential equations for Feynman integrals, it is important to know the Wronskian of the system, or, equivalently, a matrix of maximal cuts in integer dimensions, cf. section~\ref{sec:deqs}. Understanding the entries of the Wronskian and their relations is thus very important for applications. In section~\ref{sssect:Griffithstransversality} we have explained that as a consequence of the Griffiths transversality condition in eq.~\eqref{eq:transversality}, the entries of the Wronskian for Calabi-Yau motives (i.e., the maximal cuts) satisfy a collection of quadratic relations, cf. eq.~\eqref{eq:bryantgriffith}. The right-hand side of eq.~\eqref{eq:bryantgriffith} is a rational function that can be calculated explicitly and is known as the Yukawa coupling. In particular, for one-parameter families of Calabi-Yau motives the Yukawa coupling satisfies the simple differential equation given in eq.~\eqref{eq:Yukawa_equation}. We will illustrate these quadratic relations and their application to solve differential equations for the equal-mass banana integrals in section~\ref{sec:banint}.

\paragraph{Identifying one-parameter families of Calabi-Yau motives.}

In many applications one-parameter families of Feynman integrals (which depend on two kinematic scales) play an important role. It is therefore an important question how to determine, whether the homogeneous Picard-Fuchs operator annihilating the maximal cuts, describes a family of Calabi-Yau motives. In section~\ref{sssect:Griffithstransversality} we reviewed a necessary condition, based on self-adjointness, for an operator to describe a family of Calabi-Yau motives, cf. eq.~\eqref{eq:selfadjointness}. This criterion involves the Yukawa coupling, which can itself be determined from the differential operator and the differential equation~\eqref{eq:Yukawa_equation}. If this necessary condition is satisfied, we expect that many of the properties reviewed in this paper apply. Therefore, we believe that this criterion will be an important tool in the future to identify two-scale Feynman integrals that can be solved using the techniques we have described.

\paragraph{Calabi-Yau motives and the transcendental weight of Feynman integrals.}

A mysterious property of Feynman integrals and scattering amplitudes is their transcendental weight. In particular, in certain special quantum field theories like the $\mathcal{N}=4$ supersymmetric Yang-Mills theory in four space-time dimensions, it was observed that, whenever an $l$-loop scattering amplitude can be expressed in terms of multiple polylogarithms, then it only involves multiple polylogarithms of transcendental weight $2l$. Clearly, such a property calls for a deeper theoretical explanation. An important question in this context is how the notion of transcendental weight generalizes to scattering amplitudes that cannot be evaluated in terms of multiple polylogarithms, but require functions associated to more complicated geometries. A proposal for how to generalize the notation of transcendental weight to include elliptic polylogarithms and iterated integrals of modular forms was put forward in ref.~\cite{Broedel:2018qkq}. The proposal of ref.~\cite{Broedel:2018qkq} predicts, in particular, the expected transcendental weight of the two- and three-loop banana integrals in $D=2$ dimensions (which are expected to have transcendental weights two and three respectively; the one-loop banana evaluates to a logarithm and has transcendental weight one), and of the two-loop elliptic double-box integral, that appears in the planar $\mathcal{N}=4$ Super Yang-Mills theory~\cite{Kristensson:2021ani}. It is thus natural to expect that in general $l$-loop banana in $D=2$ dimensions or train-track integrals in $D=4$ dimensions have uniform transcendental weight $l$ and $2l$, respectively. If and how the notion of transcendental weight generalizes concretely to other geometries, in particular of Calabi-Yau type, is still unexplored, mostly due to a lack of explicit results for such integrals.

In section~\ref{sssec:boundary} we have argued that the structures associated with the boundaries of the moduli space of a family of Calabi-Yau varieties, in particular its limiting mixed Hodge structure, allow one to motivate that the $l$-loop banana integral in $d=2$ dimensions has transcendental weight $l$, as expected. In a first step, we see that the monodromy weight filtration allows one to associate a transcendental weight with every period (i.e., maximal cut) of the Calabi-Yau variety. Loosely speaking, at the MUM-point the motive degenerates to a mixed-Tate motive, and the leading behavior of the periods is described by logarithms and zeta values, whose transcendental weight is well understood. The degeneration of the periods at the MUM-point then determines their associated transcendental weight. While it seems that this construction only allows one to define the transcendental weight of the periods/maximal cuts, we have argued at the end of section~\ref{subsection: MirrorSymmetry} how we can uplift the definition of the transcendental weight to include the full $l$-loop banana integrals in $D=2$ dimensions. The idea is that, after changing basis from the Frobenius basis to an integral basis, the monodromy matrices take integer values (cf. eqs.~\eqref{eq:4l_Tfrak} and~\eqref{eq:monodromy_mats_4l}), and we see that under the action of the monodromy group the full Feynman integral gets admixtures from Calabi-Yau periods of transcendental weight $l-1$ multiplied by an additional factor of $\pi$ increasing their total transcendental weight by one. Therefore, the only way to define the transcendental weight of the full $l$-loop banana integral in $D=2$ dimensions that is natural and consistent with the monodromy is therefore to define its transcendental weight as $l$, in agreement with the known results for $l\le 3$. It would be interesting to explore if similar considerations can be applied to the $l$-loop train-track integrals in $D=4$ dimensions, which are expected to have uniform transcendental weight $2l$, since they define scattering amplitudes in the planar $\mathcal{N}=4$ supersymmetric Yang-Mills theory. It was recently confirmed that the two-loop train track graph has the expected uniform transcendental weight~\cite{Kristensson:2021ani}. Understanding how to extend the approach of section~\ref{sssec:boundary} to higher-loop train track integrals may shed light on the mysterious uniform transcendental weight property of  $\mathcal{N}=4$ supersymmetric Yang-Mills beyond multiple polylogarithms.

Let us conclude by mentioning that the concept of uniform transcendental weight is closely related to the concept of \emph{pure functions} that has appeared in physics, cf.~ref.~\cite{ArkaniHamed:2010gh} for the definition for multiple polylogarithms, and ref.~\cite{Broedel:2018qkq} for the extensions for elliptic curves and iterated integrals of modular forms. We will comment on pure functions for higher-loop banana integrals at the end of the next section.

\paragraph{Modularity and families of elliptic curves and one-parameter families of K3 surfaces.}

Most concepts mentioned so far generalize very nicely from elliptic motives to general Calabi-Yau motives. It 
is important to stress that  some concepts are \emph{not expected} to generalize  straightforwardly. 
Most notably,  for families ${\cal E}$ of elliptic curves and families $\mathcal{K}3$ of \emph{algebraic} K3 surfaces, there is a uniformizing map 
to a symmetric space   ${\cal M}_\text{cs}({\cal E}) \cong \Gamma_{\cal E}\backslash \mathrm{SU}(1,1)/\mathrm{U}(1)$  and ${\cal M}_\text{cs}({\cal K}_3)\cong \Gamma_{{\cal K}3}\backslash \mathrm{SO}(2,20-\rho)/(\mathrm{SO}(2)\times \mathrm{SO}(20-\rho))$ respectively, where $\Gamma_{\cal E}$ and $\Gamma_{{\cal K}3}$ are discrete groups determined from the monodromy of the family.   For elliptic curves,  $H=\mathrm{SU}(1,1)/\mathrm{U}(1)$ is the complex upper half plane $H$,  and the monodromy group  $\Gamma_{\cal E}\subset {\rm SL}(2,\mathbb{Z})$ is a congruence subgroup. For 
K3 surfaces, $(1,\rho-1)$ denotes the signature of the Picard-Lattice $\Gamma_{1,{\rho-1}}$ in ${\bf \Sigma}_{\text{K3}}$  
whose rank $\rho\ge 2$ depends on the algebraic embedding  of the  K3 surface, and $\Gamma_{{\cal K}3}\in \mathrm{SO}(2,20-\rho,\mathbb{Z})$. 
The occurance  of modular--  or automorphic symmetries, respectively, has the consequence that 
the elliptic integrals can be written as modular forms,  and the periods of the K3 as more general automorphic forms, 
which sometimes have a simple relation to modular forms. For instance, for complex one-parameter K3 
surfaces $\rho=19$ always. This is related to the fact that the corresponding  third-oder Picard-Fuchs operator 
can be always written as a  symmetric square of the second-order operator of  elliptic curves, cf., e.g., ref.~\cite{Doran:1998hm}. For Calabi-Yau $n$-folds 
with $n\geq 3$, which have the full $\mathrm{SU}(n)$ holonomy, ${\cal M}_\text{cs}$ \emph{cannot} be mapped to a symmetric 
space of the form above. So its is expected that the theory of periods and therefore of Feynman integrals, can become  
qualitatively very different for loop orders $l\ge 4$. As summarized in ref.~\cite{MR3822913}, integral structures like the integrality of the coefficients of  the unique  
holomorphic period at the MUM-point, the integrality in the coefficients of the mirror map and most 
noticeable of the instanton expansion much studied for Calabi-Yau three-folds, but also established 
for Calabi-Yau four-folds, are clearly harbingers of the interesting arithmetic properties of Calabi-Yau motives 
that might become an important  feature of Feynman integrals at higher-loop order, beyond the 
examples of the banana integrals.

\section{Equal-mass banana Feynman integrals in $D=2$ dimensions}
\label{sec:banint}

The goal of this section is to show how one can combine ideas from geometry and Calabi-Yau manifolds with the method of computing Feynman integrals using Gauss-Manin systems. We have already seen in section~\ref{subsec:f-odeqs} that the Wronskian and its inverse play important roles in solving the first-order Gauss-Manin systems in terms of iterated integrals. The Wronskian can be identified with the matrix of maximal cuts in $D=2$ space-time dimensions. We therefore start this section by continuing our analysis of what the geometry of Calabi-Yau manifolds can teach about these maximal cuts.

\subsection{Maximal cuts of $l$-loop banana integrals in $D=2$ dimensions}

In this section we use insights from the theory of Calabi-Yau manifolds to give an explicit description of the maximal cuts of equal-mass banana integrals at an arbitrary number of loops. Our starting point is the observation that the maximal cuts of $J_{l,1}(z;0)$ compute the periods of a Calabi-Yau $(l-1)$-fold, cf.~eq.~\eqref{eq:max_cut_feyn_par} and refs.~\cite{Bloch:2014qca,Bloch:2016izu,Vanhove:2018mto,Klemm:2019dbm,Bonisch:2020qmm}. Hence, while it may be hard to find an explicit basis of maximal cut contours $\Gamma_j$ to evaluate the cut integrals in momentum space (see section~\ref{subsec:f-odeqs}), we can find a basis for the maximal cuts of $J_{l,1}(z;0)$: they form a basis for the periods of the Calabi-Yau $(l-1)$-fold. In section~\ref{sec:CY} we have argued that the moduli space of every Calabi-Yau manifold contains a MUM-point, which was identified with the large momentum limit, $z\to0$. It is thus natural to consider as basis of periods the Frobenius basis $\uPi_l(z)$ at the MUM-point defined in eq.~\eqref{eq:Pi_vec}. As a consequence, all maximal cuts of $J_{l,1}(z;0)$ are of the form
\beq\label{eq:max_cut_decomp}
	J_{l,1}^{\Gamma}(z) 		= 	\ualpha^{\Gamma} \cdot \uPi_l(z) = \sum_{j=1}^{l}\alpha_j^{\Gamma}\,\varpi_{l,j-1}(z) \qquad\text{for } \ualpha^{\Gamma}\in\mathbb{C}^{l}\,.
\eeq
At this point we make a comment: As already mentioned in section~\ref{sssec:boundary}, the elements of the Frobenius basis $\varpi_{l,j}(z)$ are not obtained by integrating the $(l-1,0)$-form over a cycle defined in integral homology. The maximal cut contours, however, are defined with integer coefficients (they are `genuine geometric objects'). In section~\ref{sssec:boundary} we have discussed how to work out the change of basis (see also the four-loop example at the end of section~\ref{subsection: MirrorSymmetry}). The advantage of working with the Frobenius basis is that its structure is well understood and we have efficient methods to determine it (especially close to the MUM-point $z=0$).

The element of the Frobenius basis can be evaluated for all real values of $z$, as explained in section~\ref{subsec:PF} and ref.~\cite{Bonisch:2020qmm}. Note that $J_{l,1}^{\Gamma}(z)$ is always annihilated by the differential operator $\cL_l$ defined in section~\ref{subsec:PF}. In the following it will be useful to write the Picard-Fuchs operator in terms of usual derivatives
\beq
\cL_l =\sum_{k=0}^lB_{l,k}(z)\partial_z^{k}\,,
\eeq
where the $B_{l,k}(z)$ are polynomials. The Wronskian can then be chosen as: 
\beq
{\bf W}_l(z) \coloneqq 
					\begin{pmatrix} 		\varpi_{l,0}(z) 			& \varpi_{l,1}(z) 	 			& \hdots	& \varpi_{l,l-1}(z)	 			\\
									\partial_z\varpi_{l,0}(z) 		& \partial_z\varpi_{l,1}(z) 		& \hdots	& \partial_z\varpi_{l,l-1}(z)		\\
									 \vdots  				& \vdots					& 		& \vdots					\\
									\partial_z^{l-1}\varpi_{l,0}(z) 	& \partial_z^{l-1}\varpi_{l,1}(z) 	& \hdots	& \partial_z^{l-1}\varpi_{l,l-1}	(z)
					\end{pmatrix}\,,
\label{defwronskian}
\eeq
and we have $\uJ_l^{\Gamma}(z) = {\bf W}_l(z) \ualpha^{\Gamma}$, cf.~eq.~\eqref{eq:max_cut_Gamma}. The determinant of the Wronskian is
\beq
\det{\bf W}_l(z) = \left((-1)^lz^{-3}\,\textrm{Disc}(\cL_l)\right)^{-l/2} = \left(z^{l-3}\prod_{k\in \Delta^{(l)}}(1-kz)\right)^{-l/2}\,,
\label{eq:WronskiDet}
\eeq
where $\text{Disc}(\mathcal L_l)$ is given in eq.~\eqref{eq:discriminant}. To prove this identity, first note that from $\partial_z {\bf W}_l(z) = {\bf B}_{l,0}(z){\bf W}_l(z)$ it follows that 
\beq
\partial_z\det{\bf W}_l(z) = \mathrm{Tr}\,{\bf B}_{l,0}(z)\,\det{\bf W}_l(z)\,,
\label{detw}
\eeq
where for our choice of basis 
${\bf B}_{l,0}(z)$ has the form
\begin{equation}
	{\bf B}_{l,0}(z)	=	\begin{pmatrix} 		0 & 1 & 0 &\hdots&0 \\
								0 & 0 & 1 & \hdots & 0\\
								 \vdots  & \vdots & \ddots& &\vdots\\
								0 & 0 &0&\hdots&1 	\\
								-\frac{B_{l,0}(z)}{B_{l,l}(z)} &-\frac{B_{l,1}(z)}{B_{l,l}(z)} &-\frac{B_{l,2}(z)}{B_{l,l}(z)}& \hdots & -\frac{B_{l,l-1}(z)}{B_{l,l}(z)}
				\end{pmatrix}\,.
\label{defB}
\end{equation}
Computing the operator $\mathcal{L}_l$ with the procedure explained in ref.~\cite{Bonisch:2020qmm}, one finds that
\beq\bsp
    B_{l,l}(z)&\,= \text{Disc}(\cL_{l}) = (-z)^l \prod_{k \in \Delta^{(l)}}(1-kz) \,,\\
    B_{l,l-1}(z) &\,= \frac{l}{2}\left(\partial_z B_{l,l}(z)-\frac{3}{z}B_{l,l}(z)\right)\,.
\esp\eeq
This gives 
\begin{align}
   \mathrm{Tr}\,{\bf B}_{l,0}(z)	=	-\frac{B_{l,l-1}(z)}{B_{l,l}(z)}	=	\frac{l}{2}\left( \frac{3-l}{z}+\sum_{k\in\Delta^{(l)}} \frac{k}{1-kz} \right)\,.
   \label{eq:TraceConnection}
\end{align}
It is then easy to see that the determinant is proportional to the right-hand side of eq.~\eqref{eq:WronskiDet}, and the constant of proportionality is fixed by our normalization of the Frobenius basis.

As explained in section~\ref{subsec:f-odeqs}, the Wronskian and its inverse play important roles when solving the Gauss-Manin system satisfied by the Feynman integrals. The elements of the inverse Wronskian are $(l-1)\times(l-1)$ minors of ${\bf W}_l(z)$, i.e., they are homegeneous polynomials of degree $l-1$ in the entries of ${\bf W}_l(z)$. It seems from eq.~\eqref{eq:B_tilde_N_tilde} that the integrand of the iterated integrals will involve polynomials of degree $(l-1)$, and beyond the leading order in $\eps$ even of degree $l$. In the following we show that the quadratic relations from Griffiths transversality from section~\ref{sssect:Griffithstransversality} allow one to reduce this degree considerably.

\subsection{Bilinear relations among maximal cuts from Griffiths transversality}
\label{sec:quad_rel_max_cut}

In section~\ref{sec:CY}, in particular in eq. \eqref{eq:bryantgriffith}, we have explained that the Calabi-Yau periods and their derivatives satisfy bilinear relations as a result of Griffiths transversality. Since the periods and their derivatives are nothing but the maximal cuts for the master integrals in eq.~\eqref{eq:equal_mass_MIs}, Griffiths transversality leads to a set of bilinear relations among maximal cuts. The goal of this section is to describe these relations in the equal-mass case and to write them down explicitly for the first few loop orders, and to highlight some of their consequences for Feynman integrals.

Recall from section~\ref{sec:CY} that there is a bilinear pairing -- the intersection pairing ${\bf \Sigma}$ -- on the entries of ${\bf W}_l(z)$. If we work with the Frobenius basis $\uPi_l(z)$, Griffiths transversality in eq.~\eqref{eq:bryantgriffith} takes the form
\begin{align}
    {\uPi}_{l}(z)^T\, {\bf \Sigma}_l\, \partial_z^k{\uPi}_{l}(z) = \begin{cases} 0 & , \ k<l-1\,, \\ C_{l-1}(z) & ,\ k=l-1\,,\end{cases}
    \label{eq:GriffithsExplicit}
\end{align}
where the intersection matrix ${\bf \Sigma}_l$ is given by
\begin{align}
    {\bf \Sigma}_l &= \left( \begin{array}{cccc}
         & & & 1  \\
         & & -1 & \\
         & 1 & & \\
         \iddots & & &
    \end{array} \right) \, .
\end{align}
We now use this to derive a differential equation satisfied by $C_l(z)$. First note that for any $0 \leq k \leq l-1$ we have
\begin{equation}
\begin{aligned}
    C_l(z)&={\uPi}_{l}(z)^T\, {\bf \Sigma}_l\, \partial_z^{l-1}{\uPi}_{l}(z) = \partial_z (\underbrace{{\uPi}_{l}(z)^T\, {\bf \Sigma}_l\, \partial_z^{l-2}{\uPi}_{l}(z)}_{=0})-\partial_z{\uPi}_{l}(z)^T\, {\bf \Sigma}_l\, \partial_z^{l-2}{\uPi}_{l}(z)
    \\ &= ...= (-1)^k \partial_z^k {\uPi}_{l}(z)^T\, {\bf \Sigma}_l\, \partial_z^{l-1-k}{\uPi}_{l}(z)\,,
    \label{eq:GriffithsAntidiagonal}
\end{aligned}
\end{equation}
and thus 
\begin{align}
    \partial_z^k {\uPi}_{l}(z)^T\, {\bf \Sigma}_l\, \partial_z^{l-1-k}{\uPi}_{l}(z) = (-1)^k C_l(k) \, .
\end{align}
Using this successively, we get
\begin{equation}
\begin{aligned}
    {\uPi}_{l}(z)^T\, {\bf \Sigma}_l\, \partial_z^l{\uPi}_{l}(z) &= \partial_z C_l(z)- \partial_z {\uPi}_{l}(z)^T\, {\bf \Sigma}_l\, \partial_z^{l-1}{\uPi}_{l}(z)\\
    &=2 \partial_z C_l(z)+\partial_z^2 {\uPi}_{l}(z)^T\, {\bf \Sigma}_l\, \partial_z^{l-2}{\uPi}_{l}(z)\\
    &=...=l \partial_z C_l(z) +(-1)^l \partial_z^l{\uPi}_{l}(z)^T\, {\bf \Sigma}_l\, {\uPi}_{l}(z)\,.
\end{aligned}
\end{equation}
The $(-1)^{l+1}$-symmetry of ${\bf \Sigma}_l$ then implies 
\begin{align}
    {\uPi}_{l}(z)^T\, {\bf \Sigma}_l\, \partial_z^l{\uPi}_{l}(z) = \frac{l}{2}\partial_z C_l(z) \, .
\end{align}
Finally using this together with the Picard-Fuchs equation we find that
\begin{equation}
\begin{aligned}
    0&={\uPi}_{l}(z)^T\, {\bf \Sigma}_l\, \mathcal{L}_l{\uPi}_{l}(z)\\
    &=B_{l,l-1}(z){\uPi}_{l}(z)^T\, {\bf \Sigma}_l\, \partial_z^{l-1}{\uPi}_{l}(z)+B_{l,l}(z){\uPi}_{l}(z)^T\, {\bf \Sigma}_l\, \partial_z^l{\uPi}_{l}(z)\\
    &=B_{l,l-1}(z) C_l(z)+\frac{l}{2}B_{l,l}(z) \partial_z C_l(z)\,,
\end{aligned}
\end{equation}
and so we have (cf. eq.~\eqref{eq:Yukawa_equation}):
\begin{align}
    \partial_z C_l(z) +\frac{2}{l} \frac{B_{l,l-1}(z)}{B_{l,l}(z)}C_l(z)=0 \, .
\end{align}
Comparing with eq.~\eqref{eq:TraceConnection} and fixing the constant of proportionality from our normalization of the Frobenius basis (cf. eq.~\eqref{eq:frobenius}), we can find an explicit expression for the Yukawa coupling:
\begin{align}
    C_l(z)=\frac{1}{z^{l-3}\prod_{k\in\Delta^{(l)}}(1-kz)}\,.
\end{align}

Equation~\eqref{eq:GriffithsExplicit} can be interpreted as a collection of bilinear relations between the maximal cuts of $J_{l,1}(z;0)$ and $J_{l,k}(z;0)$ for $k>1$. We obtain more relations by differentiation, e.g.,
\beq\bsp
\partial_zC_l(z) &\,= \partial_z{\uPi}_{l}(z)^T\, {\bf \Sigma}_l\, \partial_z^{l-1}{\uPi}_{l}(z) + {\uPi}_{l}(z)^T\, {\bf \Sigma}_l\, \partial_z^{l}{\uPi}_{l}(z)\\
&\,= \partial_z{\uPi}_{l}(z)^T\, {\bf \Sigma}_l\, \partial_z^{l-1}{\uPi}_{l}(z) - \sum_{j=0}^{l-1}\frac{B_{l,j}(z)}{B_{l,l}(z)} {\uPi}_{l}(z)^T\, {\bf \Sigma}_l\, \partial_z^{j}{\uPi}_{l}(z)\\
&\,= \partial_z{\uPi}_{l}(z)^T\, {\bf \Sigma}_l\, \partial_z^{l-1}{\uPi}_{l}(z) - \frac{B_{l,l-1}(z)}{B_{l,l}(z)} C_l(z)\,,
\esp\eeq
where in the second step we used the fact that $\cL_l\uPi_l(z)=0$ and in the third step we used eq.~\eqref{eq:GriffithsExplicit}. We can proceed in this way to compute all the entries of the matrix
\begin{align}
  \mathbf{Z}_l(z) = \left( \begin{array}{ccc}
       {\uPi}_{l}(z)^T\,{\bf \Sigma}_l\, {\uPi}_{l}(z) & \cdots & {\uPi}_{l}(z)^T\,{\bf \Sigma}_l\, \partial_z^{l-1}{\uPi}_{l}(z) \\
       \vdots & \ddots & \vdots \\
       \partial_z^{l-1}{\uPi}_{l}(z)^T\,{\bf \Sigma}_l\, {\uPi}_{l}(z) & \cdots & \partial_z^{l-1}{\uPi}_{l}(z)^T\,{\bf \Sigma}_l\, \partial_z^{l-1}{\uPi}_{l}(z) \\
  \end{array} \right) \,.
  \label{eq:ZMatrix}
\end{align}
It follows from the previous considerations that all entries are rational functions. Note that ${\bf Z}_l(z)^T=(-1)^{l+1}\,{\bf Z}_l(z)$. 

Let us work out these relations, or equivalently the matrix ${\bf Z}_l(z)$, for the first few loop orders. It turns out that the matrix ${\bf Z}_l(z)^{-1}$ has a more compact form, so we give examples for this matrix for $l \leq 4$:
\begin{align}
{\bf Z}_1(z)^{-1} =& \frac{1}{z^2}-\frac{4}{z} \\
\nonumber{\bf Z}_2(z)^{-1} =& \left(
\begin{array}{cc}
 0 & -\frac{1}{z}+10-9 z \\
 \frac{1}{z}-10+9 z & 0 \\
\end{array}
\right) \\
\nonumber    {\bf Z}_3(z)^{-1} =& \left(
\begin{array}{ccc}
 \frac{1}{z^2}-\frac{8}{z} & -10+64 z & 1-20 z+64 z^2 \\
 -10+64 z & -1+20 z-64 z^2 & 0 \\
 1-20 z+64 z^2 & 0 & 0 \\
\end{array}
\right) \\
\nonumber{\bf Z}_4(z)^{-1} =& \left(
\begin{array}{cc}
 0						&	-\frac{1}{z}+28-285 z+450 z^2 	\\
 \frac{1}{z}-28+285 z-450 z^2 	& 	0						\\
1-70 z+777 z^2-900 z^3 		& 	-z+35 z^2-259 z^3+225 z^4 	\\
z-35 z^2+259 z^3-225 z^4 	& 	0
\end{array}
\right. \\
\nonumber&\hspace{3cm}	\left.
\begin{array}{cc}
 -1+70 z-777 z^2+900 z^3 	& 	-z+35 z^2-259 z^3+225 z^4 	\\
 z-35 z^2+259 z^3-225 z^4 	& 	0						\\
0						& 	0 						\\
0					 	& 	0
\end{array}
\right)~.
\end{align}
For $l=1$ this just shows that the period is an algebraic function; for $l=2$ we get the well-known Legendre relation for the periods of an elliptic curve; for $l=3$ we find that the Picard-Fuchs operator is a symmetric square. Also note that, while for low loop orders the entries of the antidiagonals of ${\bf Z}_l(z)^{-1}$ are equal (up to a possible alternating sign and zeroes on the diagonal), this pattern does not hold in general for higher-loop orders. The only non-trivial exception for this is the middle antidiagonal (see eq.~\eqref{eq:GriffithsAntidiagonal}).

The procedure that we have just described to derive bilinear relations works for arbitrary values of $l$. An important question is if and how many of these relations are independent (at least for generic values of $z$), and if there are bilinear relations that are not captured by Griffiths transversality. The symmetry properties of ${\bf Z}_l(z)$ imply that there are $\frac{l(l+1)}{2}$ independent entries for $l$ odd and $\frac{l(l-1)}{2}$ for $l$ even. For loop orders $l \leq 4$ we have searched for additional relations by evaluating the Wronskian matrix at a generic rational point $z_0$ and using a floating point LLL algorithm to find linear dependencies between bilinear products of elements of the Wronski matrix and $C_l(z_0)$. Within the used precision all of the resulting relations follow from Griffiths transversality. Based on these observations, we conjecture that all global bilinear relations among the maximal cuts for $l$-loop banana integrals follow from Griffiths transversality.

The bilinear relations from Griffiths transversality have important consequences for the inverse Wronskian ${\bf W}_{l}(z)^{-1}$. In terms of the matrix $\mathbf{Z}_l(z)$ defined in eq.~\eqref{eq:ZMatrix}, one can express the inverse Wronskian as
\begin{align}
    {\bf W}_{l}(z)^{-1}={\bf \Sigma}_l{\bf W}_{l}(z)^T\mathbf{Z}_l(z)^{-1} \, .
\label{eq:winverse}
\end{align}
Explicitly, this gives for example, with $1\le k\le l$,
\begin{align}\label{eq:W-1}
    {\bf W}_{l}(z)^{-1}_{k,l}&=\frac{(-1)^{l+k}\varpi_{l,l-k}(z)}{C_l(z)}\,,\\
    {\bf W}_{l}(z)^{-1}_{k,l-1}&=\frac{(-1)^{l+k+1}}{C_l(z)}\left[\partial_z\varpi_{l,l-k}(z)+\left(\frac{l}{2}-1 \right)\frac{\partial_z C_l(z)}{C_l(z)}\varpi_{l,l-k}(z) \right] \, .
\end{align}
Similar relations can be derived for all other entries of the inverse Wronskian. In particular, all entries of the inverse Wronskian are linear in the entries of ${\bf W}_l(z)$. Note that this implies polynomial relations of higher degree between the entries of ${\bf W}_l(z)$, because the entries of the inverse are proportional to $(l-1)\times(l-1)$ minors. Moreover, we see that Griffiths transversality also determines $\det {\bf W}_{l}(z)$. Using eq. \eqref{eq:winverse} one finds
\begin{equation}
	\det {\bf W}_{l}(z) = C_l(z)^{l/2}~,
\end{equation}
which is also a cross check for eq.~\eqref{eq:WronskiDet}.

Let us conclude this section with some comments: First, we mention that this is not the first time that bilinear relations between maximal cuts of banana integrals were considered. In particular, refs.~\cite{Broadhurst:2018tey,Zhou:2019rgc,Zhou:2020glw,Fresan:2020anx} considered quadratic relations between moments of Bessel functions, which are closely related to banana integrals and their cuts. Quadratic relations for maximal cuts for $l=2$ and $l=3$ were also studied in ref.~\cite{Lee:2018jsw}, and such relations also arise from the twisted Riemann bilinear relations~\cite{ChoMatsumoto,Frellesvig:2020qot}. Since all quadratic relations follow from Griffiths transversality (at least conjecturally), it would be interesting to compare those quadratic relations to the ones presented here. Finally, in this section we have only discussed the equal-mass case. Griffiths transversality holds more generally also for the periods, i.e., the maximal cuts, in the case of distinct  propagator masses. It would be interesting to work out these relations explicitly.

\subsection{Banana integrals and integrals of Calabi-Yau periods}
\label{subsec:banintperiods}

So far we have discussed what we can learn about the matrix of maximal cuts for the equal-mass banana integrals, in particular relations among maximal cuts and the inverse of the Wronskian. In this section we apply these results to obtain for the first time a representation of higher-loop banana graphs in terms of (iterated) integrals of Calabi-Yau periods. For simplicity, we focus here on $\eps=0$. Extending our results to include higher orders in $\eps$ is straightforward, assuming the matrices ${\bf B}_{l,0}(z)$ and ${\bf B}_{l,k}(z)$ for $k=2,\hdots,l$ are known. An efficient technique to obtain the matrices ${\bf B}_{l,0}(z)$ for high loop order was presented (implicitly) in ref.~\cite{Bonisch:2020qmm}. In the next section we will show how to extend this method to include higher-order terms in $\eps$.

Our starting point is eq.~\eqref{eq:banana_DEQ} for $\eps=0$:
\beq
\partial_z{\uJ}_{l}(z;0) = {\bf B}_{l,0}(z)\,{\uJ}_{l}(z;0) +  (-1)^{l+1}(l+1)!\, \frac{z}{z^l\prod_{k\in\Delta^{(l)}}(1-kz)}\,\hat{\underline{e}}_l\,,
\eeq
with $\hat{\underline{e}}_l = (0,\ldots,0,1)^T$.
Changing variables like in eq.~\eqref{eq:L_to_J}, we obtain
\beq
\partial_z\uL_{l}^{(0)}(z) = (-1)^{l+1}(l+1)!\, \frac{z}{z^l\prod_{k\in\Delta^{(l)}}(1-kz)}\,{\bf W}_l(z)^{-1}
 \hat{\underline{e}}_l= (l+1)!\,\,\frac{{\bf \Sigma}_l\,\uPi_l(z)}{z^2}\,,
\eeq
where in the last step we used eq.~\eqref{eq:W-1} in vector form, i.e. ${\bf W}_{l}(z)^{-1}\,\hat{\underline{e}}_l=(-1)^{l+1}\frac{{\bf \Sigma}_l\underline\Pi_l(z)}{C_l(z)}$. This equation can easily be solved by quadrature
\beq\bsp\label{eq:L0_sol}
\uL_{l}^{(0)}(z) &\,= \uL_{l}^{(0)}(0) +(l+1)!\,{\bf \Sigma}_l\int_{\vec{1}_0}^z\rd w\,\frac{\uPi_l(w)}{w^2} \,,
\esp\eeq
or equivalently
\begin{align}\label{eq:J0_sol}
\uJ_{l}^{(0)}(z) &\,= {\bf W}_l(z)\uL_{l}^{(0)}(0) +(l+1)!\,{\bf W}_l(z)\,{\bf \Sigma}_l\int_{\vec{1}_0}^z\rd w\,\frac{\uPi_l(w)}{w^2}~.
\end{align}
For the individual master integrals, i.e., the individual components of $\uJ_l^{(0)}$, we find:
\beq\bsp\label{eq:Jk0_sol}
J_{l,k}^{(0)}(z) &\,= \partial_z^{k-1}\uPi_l(z)^T\,{\bf \Sigma}_l\,\uL_{l}^{(0)}(0) +(l+1)!\,\partial_z^{k-1}\uPi_l(z)^T\,{\bf \Sigma}_l\int_{\vec{1}_0}^z\rd w\,\frac{\uPi_l(w)}{w^2}\,.
\esp\eeq
Here ${\vec{1}_0}$ denotes the unit tangent vector at 0. In the limit $z\to 0$, we have
\beq\label{eq:Pi_limit}
\uPi_l(z) = z\,\left(1,\log z,\frac{1}{2}\log^2z,\ldots,\frac{1}{(l-1)!}\log^{l-1}z\right)^T + \ord(z^2)\,,
\eeq
so that the integral in eq.~\eqref{eq:J0_sol} diverges if the lower integration limit is zero. The divergency is regulated by introducing the tangential base point ${\vec{1}_0}$. For a comprehensive review of how this regularization can be practically implemented into the framework of iterated integrals on curves see, e.g., ref.~\cite{Brown:mmv}. In the present case, the tangential base point regularization (often called shuffle regularization) reduces to the prescription:
\beq\label{eq:tangential_base_point}
\int_{\vec{1}_0}^z\rd w\,\frac{\varpi_{l,k}(w)}{w^2}\coloneqq \frac{1}{(k+1)!}\log^{k+1}z + \int_0^z\frac{\rd w}{w}\left(\frac{\varpi_{l,k}(w)}{w} - \frac{1}{k!}\log^{k}w\right)\,.
\eeq
It is easy to check that the remaining integral in eq.~\eqref{eq:tangential_base_point} is absolutely convergent (as long as the range $[0,z]$ does not contain any singular point of $\cL_l$).

Let us now determine the initial condition $\uL_l^{(0)}(0)$. We start from eq.~\eqref{eq:Jk0_sol} for $k=1$.
Since the integral in the second term in eq.~\eqref{eq:tangential_base_point} is convergent, this integral vanishes like a power in the limit $z\to0$, and so this term behaves like $\ord(z^2)$. Using eq.~\eqref{eq:Pi_limit}, we find:
\beq\bsp\label{eq:Jk0_sol0}
J_{l,1}(z;0) &\,=-(-1)^l(l+1)\,z\,\log^{l}z+ z\,\sum_{j=0}^{l-1}L_{l,l-j}^{(0)}(0) \frac{(-\log z)^{j}}{j!}+ \ord(z^2)\,.
\esp\eeq
We see that the vector $\uL_{l}^{(0)}(0)$ is uniquely determined once we know the leading asymptotics of $J_{l,1}(z;0)$ in the limit $z\to 0$. This limit can be related for arbitrary loop order to a novel $\widehat\Gamma$-class~\cite{Bonisch:2020qmm}. In particular, in ref.~\cite{Bonisch:2020qmm} a generating functional for the coefficients of the logarithms in the limit $z\to0$ was obtained (see, for example, Table 2 of ref.~\cite{Bonisch:2020qmm}, for explicit results through $l\le 6$). More precisely, comparing eq.~\eqref{eq:Jk0_sol0} to the asymptotics obtained in ref.~\cite{Bonisch:2020qmm}, we find that
\beq
L^{(0)}_{l,k}(0) = \frac{(l+1)!}{(k+1)!}\,\lambda^{(k)}_0\,,
\eeq
where the $\lambda^{(k)}_0$ are $\mathbb{Q}[i\pi]$-linear combinations of zeta values of uniform transcendental weight $k$ defined by the generating function~\cite{Bonisch:2020qmm}:
\beq\label{eq:BC_D=2}
\sum_{k=0}^{\infty}\frac{x^k}{(k+1)!}\,\lambda^{(k)}_0 = -e^{-i\pi x+\sum_{k=1}^{\infty}\frac{2\zeta_{2k+1}}{2k+1}x^{2k+1}} = -\frac{\Gamma(1-x)}{\Gamma(1+x)}\,e^{-2\gamma x - i\pi x}\,.
\eeq
Equation~\eqref{eq:J0_sol} is one of the main results of this paper. It expresses all master integrals for equal-mass banana integrals at $l$-loop in terms of an integral over the periods of a Calabi-Yau $(l-1)$-fold. The initial condition $\uL_l^{(0)}(0)$ can be given in the form of a generating functional~\cite{Bonisch:2020qmm}, so that all the ingredients needed in eq.~\eqref{eq:J0_sol} are known (the vector of periods $\uPi_l(z)$ and the initial condition $\uL_l^{(0)}(0)$ coming from the novel $\widehat{\Gamma}$-class in ref.~\cite{Bonisch:2020qmm}). We find it remarkable that such a compact formula of geometric origin exists for all loop orders. Note that the relations among maximal cuts from Griffiths transversality play an important role in deriving eq.~\eqref{eq:J0_sol}. Finally, we remark that it is possible to prove eq.~\eqref{eq:Jk0_sol} for $k=1$ by acting directly with the operator $\cL_l$. This operator annihilates the first term in eq.~\eqref{eq:Jk0_sol}. For the second term, we find:
\beq\bsp
    \mathcal{L}_l \left( {\uPi}_{l}(z)\, {\bf \Sigma}_l \int_{\vec{1}_0}^z \frac{\rd w}{w^2}\, {\uPi}_{l}(w) \right) &= A_{l,l}(z)\sum_{k=0}^{l-1} \binom{l}{k}\partial_z^k{\uPi}_{l}(z)\, {\bf\Sigma}_l \,\partial_z^{l-k-1} \frac{{\uPi}_{l}(z)}{z^2}\\
    &=A_{l,l}(z)\sum_{k=0}^{l-1} \binom{l}{k}(-1)^{l-k-1}\partial_z^{l-1}{\uPi}_{l}(z)\, {\bf\Sigma}_l\, \frac{{\uPi}_{l}(z)}{z^2}\\
    &=-\frac{A_{l,l}(z)C_l(z)}{z^2}\sum_{k=0}^{l-1}\binom{l}{k}(-1)^k\\
    &=(-1)^{l+1}\frac{A_{l,l}(z)C_l(z)}{z^2}\\
    &=-z \, ,
\esp\eeq
in agreement with the rescaled inhomogeneity in ref.~\cite{Bonisch:2020qmm} (see also eq. \eqref{inhom} for $\eps=0$).

\subsection{Some considerations about pure functions}
\label{subsec:purefunctions}

Let us conclude this section by commenting on the concept of pure functions of uniform transcendental weight for the banana integrals. At the end of section~\ref{sec:CY}, we have argued that the natural transcendental weight that one can assign to an $l$-loop banana integral in $D=2$ dimensions compatible with the monodromy is $l$. Let us illustrate how this same conclusion can be reached from eq.~\eqref{eq:Jk0_sol}. We discuss the case $k=1$ in detail. 

Based on expectations from $l\le 3$ (cf., e.g., refs.~\cite{Adams:2017ejb,Broedel:2018iwv,Broedel:2018qkq,Broedel:2019kmn,Adams:2018yfj,Bogner:2019lfa}), we expect that 
\beq\label{eq:pure_expectation}
J_{l,1}(z;0) = \varpi_{l,0}(z)\,P_l(z)\,, 
\eeq
where $P_l(z)$ is a so-called \emph{pure function of uniform transcendental weight $l$}~\cite{ArkaniHamed:2010gh,Broedel:2018qkq}. A function $P_l(z)$ is said to be pure of transcendental weight $l$ if, loosely speaking, $P_l(z)$ has only logarithmic singularities and satisfies an inhomogeneous differential equation $\partial_zP_l(z) = f(z)$, where $f(z)$ is a pure function of transcendental weight $l-1$. Examples of pure functions include multiple polylogarithms (which appear for $l=1$), elliptic polylogarithms and iterated integrals of modular forms (which appear for $l=2$ or $3$). The definition of pure functions for other geometries is so far unexplored.

We now argue that we can indeed cast $J_{l,1}(z;0)$ in the expected form in eq.~\eqref{eq:pure_expectation}, and that the resulting function $P_l(z)$ can indeed be qualified as pure. Starting from eq.~\eqref{eq:Jk0_sol}, we obtain:
\beq\bsp\label{eq:Pl_explicit}
P_l(z) &\,= \underline{T}_l(z)^T\,{\bf \Sigma}_l\,\uL_{l}^{(0)}(0) +(l+1)!\,\underline{T}_l(z)^T\,{\bf \Sigma}_l\int_{\vec{1}_0}^z\rd w\,\frac{\varpi_{l,0}(w)}{w}\,\frac{\underline{T}_l(w)}{w}\\
\esp\eeq
where we defined $\underline{T}_l(z) = (T_{l,0}(z),\ldots,T_{l,l-1}(z))^T$ with
\beq
T_{l,0}(z)\eqqcolon1 \quad\text{and}\quad T_{l,k}(z) \eqqcolon \frac{\varpi_{l,k}(z)}{\varpi_{l,0}(z)}  \quad\text{for}\quad 1\le k\le l-1\,.
\eeq
From the generating functional in eq.~\eqref{eq:BC_D=2} it is easy to see that $L_{l,k}^{(0)}(0)$ is a linear combination of zeta values and powers of $i\pi$ of uniform transcendental weight $k$. We now argue that the remaining ingredients to eq.~\eqref{eq:Pl_explicit} can be interpreted as pure functions of uniform transcendental weight, in such a way that these definitions agree with the known definitions for $l\le 3$, and also with the monodromy considerations from section~\ref{sec:CY}.

Let us start by discussing the holomorphic period $\varpi_{l,0}(z) = z+\mathcal{O}(z^2)$. In this normalization (cf. eq.~\eqref{eq:frobenius}), $\varpi_{l,0}(z)$ is assigned transcendental weight 0. This agrees with the known results for $l\le 3$ (cf.~appendix~\ref{app:low_loop}). In particular, for $l=1$, we have an algebraic function:
\beq
\varpi_{1,0}(z) = \frac{z}{\sqrt{1-4z}}\,.
\eeq
Note that $\varpi_{l,0}(z)$ is holomorphic in a neighborhood of the MUM-point $z=0$, but not necessarily close to another singular point.

Next, let us discuss the functions $T_{l,k}(z)$. With the normalization for the functions $\Sigma_{l,k}(z)$ in eq.~\eqref{eq:frobenius}, we have
\beq\label{eq:T_at_0}
\lim_{z\to0}\left[T_{l,k}(z) -\frac{1}{k!}\log(z)^k\right]=0\,.
\eeq
Therefore, it makes sense to qualify $T_{l,k}(z)$ as a pure function of transcendental weight $k$. Note that this is the only choice of transcendental weight that is consistent with the limit $z\to 0$. Moreover, it is also consistent with the transcendental weight defined from the monodromy (cf. section~\ref{sec:CY}). We also note that the normalization $\Sigma_{l,k}(z) = \mathcal{O}(z^2)$ for $k>0$ is the only choice of normalization that leads to eq.~\eqref{eq:T_at_0}, i.e., such that the limit $z\to0$ has uniform transcendental weight. Finally, assigning a uniform transcendental weight $k$ to $T_{l,k}(z)$ agrees with the purity and the transcendental weight assignment for $l=2$ or $3$, where we have $T_{2,1}(z)\sim 2\pi i \tau$, $T_{3,1}(z)\sim 2\pi i \tau$, $T_{3,2}(z)\sim (2\pi i \tau)^2$, where $\tau$ is the modulus of the elliptic curve for the two-loop sunrise, which is pure of transcendental weight 0~\cite{Broedel:2018qkq} (at three loops, the K3 surface is the symmetric square of the elliptic curve appearing at two loops, so the same $\tau$ appears, cf. refs.~\cite{verrill1996,Bloch:2014qca,MR3780269,Broedel:2019kmn}).

Finally, let us discuss the integral of the Calabi-Yau period in eq.~\eqref{eq:Pl_explicit}. From eq.~\eqref{eq:tangential_base_point} we see that $\int_{\vec{1}_0}^z\rd w\,\frac{\varpi_{l,0}(w)}{w^2}\,{{T}_{l,k}(w)}$ degenerates like $\log(z)^{k+1}$ in the limit $z\to0$, so it is natural to define its transcendental weight as $k+1$. This definition is again consistent with the known results for this integral for $l\le 3$. For example, for $l=1$, we have
\beq
\int_{\vec{1}_0}^z\rd w\,\frac{\varpi_{1,0}(w)}{w^2} = \log z + \int_{0}^z\frac{\rd w}{w\,\sqrt{1-4w}} = \log\left(\frac{1-\sqrt{1-4z}}{1+\sqrt{1-4z}}\right) - \log(z)\,.
\eeq
A similar calculation shows that for $l=2$ or $3$, these integrals can be evaluated in terms of integrals of Eisenstein series for $\Gamma_1(6)$, which are pure and have the correct uniform transcendental weight (cf. refs.~\cite{Adams:2017ejb,Broedel:2019kmn}).

To conclude, we see that it is possible to identify functions that can be qualified as \emph{pure functions of uniform transcendental weight} in the sense of refs.~\cite{ArkaniHamed:2010gh,Broedel:2018qkq}. With these definitions, $J_{l,1}(z;0)$ admits a decomposition as in eq.~\eqref{eq:pure_expectation}, with $P_l(z)$ a pure function of uniform transcendental weight $l$, which is the transcendental weight we had obtained also from monodromy considerations in section~\ref{sec:CY}. We stress that our definition of pure functions has passed some highly non-trivial tests: the functions degenerate to pure functions of the expected transcendental weight in the limit $z\to0$, and our definitions agree with the known cases for $l\le 3$, where algebraic functions, elliptic integrals, (elliptic) polylogarithms and iterated integrals of Eisenstein series appear, for which the transcendental weight had previously been defined.

\section{Banana Feynman integrals in dimensional regularization}
\label{sec:Bananadimreg}  

So far we have only considered banana integrals in $D=2$ dimensions. 
The results of section~\ref{sec:banint} beg the question if the compact geometric formula in eq.~\eqref{eq:J0_sol} can be generalized to include higher orders in the dimensional regulator $\eps$. 
In this section we extend the results of ref.~\cite{Bonisch:2020qmm} to include dimensional regularization. In a first step we compute an explicit \emph{hypergeometric series representation} for the generic-mass banana integral in the large momentum limit, $p^2\gg m_i^2>0$. This hypergeometric representation serves a twofold purpose: First, it allows us to obtain the leading asymptotic behavior of all banana integrals in the large momentum limit, which provides the boundary condition for solving the differential equations, similar to eq.~\eqref{eq:BC_D=2}. 
Second, we use this representation to derive a complete set of differential equations satisfied by the banana integrals.

\subsection{A hypergeometric series representation of the banana integral}
\label{sec:mellin-barnes}  

\paragraph{A Mellin-Barnes representation for banana integrals.}

Our starting point is the derivation of a \emph{Mellin-Barnes (MB) integral representation} for the $l$-loop banana integral with arbitrary masses and exponents of the propagators. 
If all propagator masses are zero, the integral is trivial and can be evaluated in terms of gamma functions. We assume from now on that at least one propagator is massive.

To derive the MB representation for  $\tilde I\coloneqq I_{\unu}(p^2,\underline m^2;D)$, we start from the Feynman parameter representation in eq.~\eqref{bananageneralgeneric} and adapt the approach of ref.~\cite[app. A]{Broedel:2019kmn} to the $l$-loop case (the two-loop case was treated in ref.~\cite{Berends:1993ee}; see also ref.~\cite{Adams:2013nia}).
Recall the well-known identity
\begin{equation}\label{eq:mbA3}
\frac{1}{(A+B)^\lambda}=\int_{c-i\infty}^{c+i\infty}\frac{\dd\xi}{2 \pi i}  \  A^\xi B^{-\xi -\lambda} \frac{\Gamma(-\xi) \Gamma(\xi+\lambda)}{\Gamma(\lambda)} ~,
\end{equation}
where for appropriate $c$ the contour runs parallel to the imaginary axis and separates the left poles (due to $\Gamma(\xi+\lambda)$) from the right poles of the integrand (due to $\Gamma(-\xi)$). In the following it will be sufficient to keep the choice of integration contour implicit. We can apply eq.~\eqref{eq:mbA3} to $\mathcal{F}(p^2,\underline m^2)^{-\omega}$ to obtain
\beq\bsp\label{eq:bananaFw}
	\tilde I &= \int\frac{\dd \xi_0}{2 \pi i} \frac{\Gamma(-\xi_0) \Gamma(\xi_0+\omega)}{\Gamma(\omega)} (-p^2)^{\xi_0}  \\ 
&\quad \times
 \int_{[0,\infty)^l}\!\! \dd x_1 \dots \dd x_l 
 \left(\prod_{i=1}^{l+1}x_i^{\xi_0+\delta_i}\right)
\left(\sum_{i=1}^{l+1}\prod_{\substack{
   j=1 \\
   j\neq i
  }}^{l+1} x_j\right)^{-\xi_0-\frac{d}{2}} 
\! \! \! \! \left(\sum_{i=1}^{l+1} m_i^2 x_i\right)^{-\xi_0-\omega} ~,
\esp\eeq
where we have set $x_{l+1}=1$ by going to a set of affine coordinates. Note that eq.~\eqref{eq:bananaFw} is only valid if at least one mass $m_i$ is different from zero. We now repeatedly apply eq.~\eqref{eq:mbA3} to the term of the form $(\sum_{i=i_0}^{l+1}m_{i_0}^2 x_i)^{-\lambda_{i_0-1}}$, giving $m_{i_0}^2x_{i_0}$ the role of $A$ and producing another term of the same form, but with $i_0$ raised by unity. Let us assume that exactly $\bq + 1$ out of $l+1$ (i.e., $1\leq \bq + 1 \leq l+1$) propagator masses are non-zero, and without loss of generality these are the masses $m_1,...,m_{\bq + 1}$. At the expense of introducing $\bq $ additional\footnote{If $\bq = l$, the integral can also be written with one MB parameter less, i.e., without $\xi_l$.} MB parameters $\xi_1, ..., \xi_{\bq}$, we may split off factors of $x_1, ..., x_{\bq}$ with appropriate exponents. It is easy to see that the resulting product of ratios of gamma functions partially telescopes, leaving us with the expression
\begin{align}
	\tilde I& = \int\frac{\dd \xi_0}{2 \pi i}\cdots \int\frac{\dd \xi_{\bq}}{2 \pi i}
\left(\prod_{i=0}^{\bq}\Gamma(-\xi_i)\right)
 \frac{\Gamma\left(\omega+\sum_{i=0}^{\bq}\xi_i\right)}{\Gamma(\omega)} (-p^2)^{\xi_0} \left( \prod_{i=1}^{\bq}(m_i^2)^{\xi_i} \right) m_{\bq + 1}^{-\omega-\sum_{i=0}^{\bq}\xi_i} \nonumber \\
&\times \int_{[0,\infty)^l} \dd x_1 \dots \dd x_l
\left( \prod_{i=1}^{l+1}x_i^{\xi_0+\delta_i} \right) 
\left(\sum_{i=1}^{l+1}\prod_{\substack{
   j=1 \\
   j\neq i
  }}^{l+1} x_j\right)^{-\xi_0-\frac{d}{2}} 
  x_1^{\xi_1}\cdots x_{\bq}^{\xi_{\bq}} \   x_{\bq + 1}^{-\omega-\sum_{i=0}^{\bq}\xi_i} \ .
\end{align}
At this point we may consecutively perform the integration over $x_1, ..., x_{l}$ with the help of the Euler beta type integral:
\begin{equation}
\int_0^\infty \dd x \ x^\alpha \, (A+Bx)^{\beta} = A^{1+\alpha+\beta} B^{-1-\alpha} \frac{\Gamma(1+\alpha)\Gamma(-1-\alpha-\beta)}{\Gamma(-\beta)} \ .
\end{equation}
Careful bookkeeping of the exponents $\alpha$ and $\beta$ again shows that many $\Gamma$-factors cancel, and we may simplify the resulting expression to
\begin{align}
	\tilde I =	\int \frac{\dd \xi_0}{2 \pi i}&\cdots \int\frac{\dd \xi_{\bq}}{2 \pi i}
\left( \prod_{i=0}^{\bq }  \Gamma( -\xi_i ) \right)
\frac{\Gamma ( \omega + \sum_{i=0}^{\bq } \xi_i )}{\Gamma( \omega )}
(-p^2)^{\xi_0}
\left( \prod_{i=1}^{\bq } \left(m_i^2\right)^{\xi_i}  \right)(m_{\bq + 1}^2)^{-\omega-\sum_{i=0}^{\bq } \xi_i } \nonumber \\
 & \times \left( \prod_{i=1}^{\bq} \Gamma(-\xi_i -\delta_i -\epsilon )\right)  \Gamma(\omega + \sum_{i=0}^{\bq} \xi_i -\delta_{\bq + 1} -\epsilon  ) \left( \prod_{i=\bq + 2}^{l+1} \Gamma(-\delta_i -\epsilon  )\right) \\ 
\nonumber & \times \frac{\Gamma(1+(l-1)\epsilon + \sum_{i=1}^{l} \delta_i + \sum_{i=0}^{\bq} \xi_i )}{\Gamma(1-\epsilon + \xi_0)} \ .
\end{align}
The integrand may be further simplified by writing it in terms of $\xi_{0}^{*}\coloneqq-( 1+l \epsilon +\delta +\sum_{i=1}^{\bq} \xi_i)-\xi_0$ with $\delta \coloneqq \sum_{i=1}^{l+1} \delta_i$. Relabelling $\xi_{0}^{*} \rightarrow \xi_0$, and letting $\delta_0\coloneqq\delta_{\bq + 1}$, $m_0 \coloneqq m_{\bq + 1}$ and $\xi \coloneqq \sum_{i=0}^{\bq} \xi_i $, we can write the integral in the following symmetric form:
\beq\bsp\label{eq:MBready}
	\tilde I=\left( \prod_{i=\bq + 2}^{l+1} \Gamma(-\delta_i - \epsilon )\right) \int\frac{\dd \xi_0}{2 \pi i}\cdots \int\frac{\dd \xi_{\bq}}{2 \pi i} 
\ \left(-\frac1{p^2}\right)^{1+l \epsilon + \delta + \xi } \left( \prod_{i=0}^{\bq} \left(m_i^2\right)^{\xi_i} \right)  \\ 
\times 
\left[ \prod_{i=0}^{\bq }\Gamma(-\xi_i) \Gamma\left(-\xi_i -\epsilon -\delta_i \right) \right]
\frac{\Gamma(1+l\epsilon+ \delta + \xi  )}{\Gamma(1+l \epsilon + \delta ) \, \Gamma(-(l+1)\epsilon - \delta - \xi )}
 \ .
\esp\eeq
Equation~\eqref{eq:MBready} is a compact MB representation valid for all banana integrals, for arbitrary values for the number $l$ of loops and the space-time dimension $D$, the propagator masses $\underline{m}^2$ and the exponents $\unu=1+\underline{\delta}$. We find it remarkable that such a simple and compact expression exists. 

\begin{figure}[!t]
	\begin{minipage}{0.4\textwidth}
		\includegraphics[width=\textwidth]{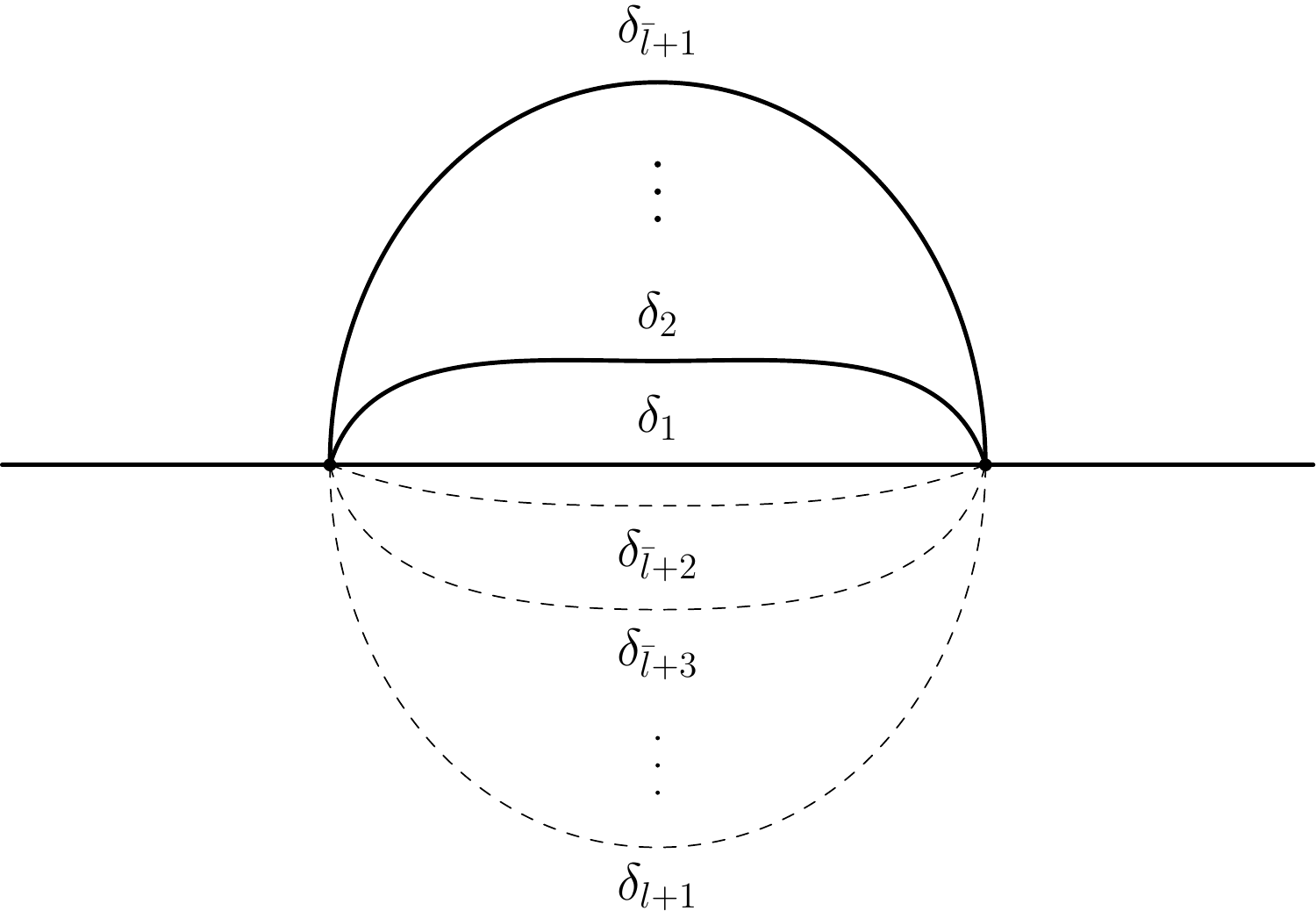}
	\end{minipage}
	\hspace{10pt}
	\begin{minipage}{0.15\textwidth}
		\includegraphics[width=\textwidth]{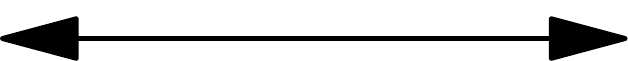}
	\end{minipage}
	\hspace{10pt}
	\begin{minipage}{0.4\textwidth}
		\includegraphics[width=\textwidth]{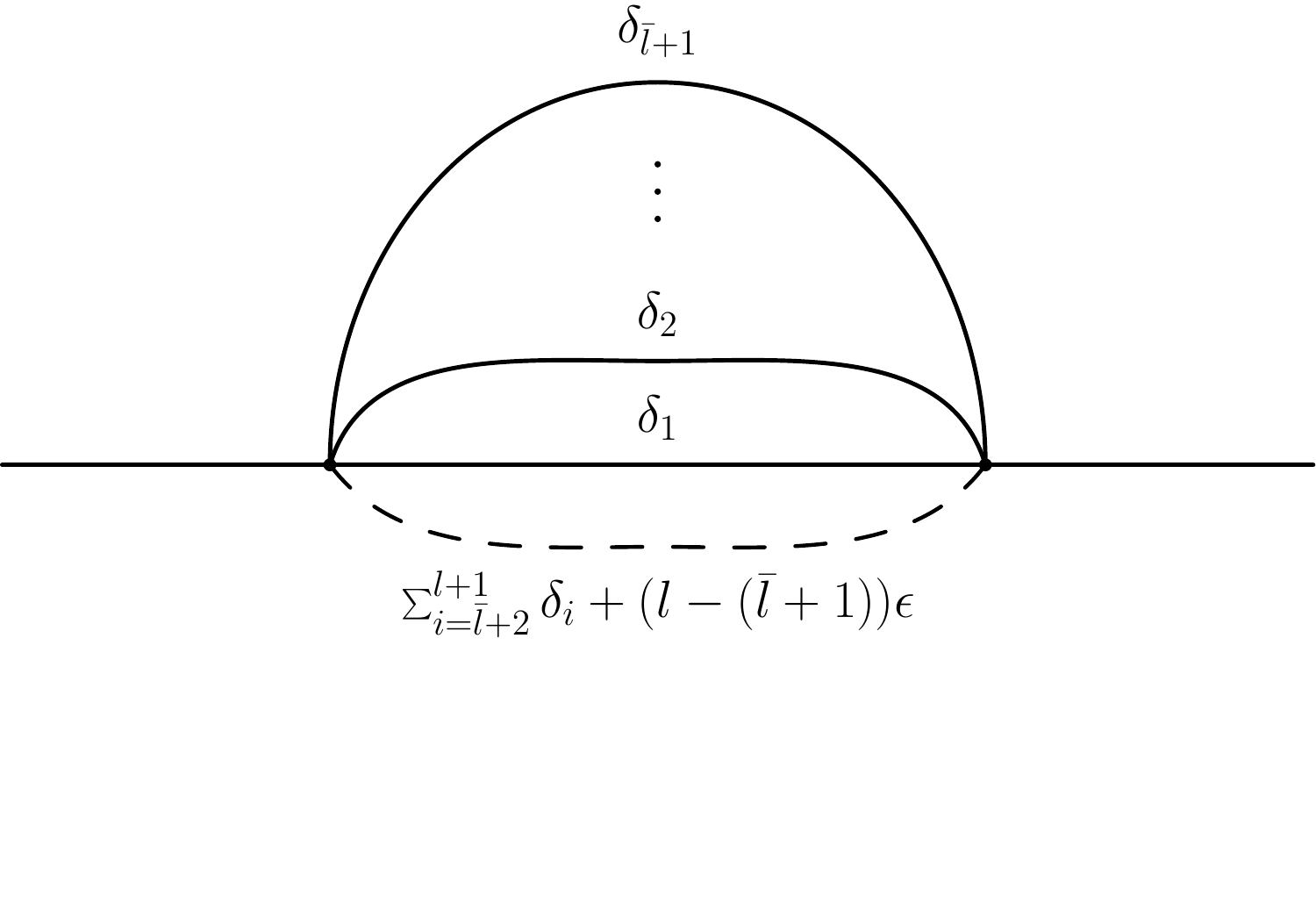}
	\end{minipage}
	\caption{Equivalence between a banana integral with several massless propagators and a single massless propagator with a modified exponent.}
	\label{fig:AdditionMassless}
\end{figure}

Before we evaluate eq.~\eqref{eq:MBready} as a hypergeometric series, let us make a comment about the role of massless propagators.
We observe that, ignoring an overall prefactor, the remaining integral in eq.~\eqref{eq:MBready} almost looks like an integral over the massive part of the banana graph only, where $\bq$ plays the role of the (reduced) loop order. The additional massless propagators only enter via the combination
\begin{equation}
l\epsilon + \delta \ = \ \bq \epsilon + \delta_{(\bq)}\,,
\end{equation}
with $\delta_{(\bq)}$ defined by
\begin{equation}
\delta_{(\bq)} =  \sum_{i=1}^{\bq + 1} \delta_i   + \sum_{i=\bq +2}^{l+1} (\delta_i + \epsilon ) \eqqcolon  \sum_{i=1}^{\bq + 1} \delta_i  + \delta_0\ .
\end{equation}
The terms in $\delta_0$ only concern the massless propagators, $i=\bq+2, \dots, l+1$, and they cannot distinguish between $l-\bq $ massless propagators with integer exponents $1+\delta_i$ and a \emph{single} massless propagator with \emph{modified} exponent $1 + [\sum_{i=\bq + 2}^{l+1}\delta_i + (l-\bq - 1) \epsilon]$ (see figure~\ref{fig:AdditionMassless}). This is consistent with the fact that we could have integrated out pairs of massless propagators iteratively and replace each pair by a single massless propagator with a shifted exponent.
In terms of the shifted quantities 
\begin{equation}
\bar{\epsilon} \coloneqq \epsilon - \delta^{(0)} \ , \qquad \bar{\delta}_i \coloneqq \delta_i + \delta^{(0)} \ , \qquad \bar{\delta} \coloneqq   \sum_{i=1}^{\bq} \bar{\delta}_i \ ,
\end{equation}
which satisfy
\begin{equation}
l\epsilon + \delta \ = \ \bq \bar{\epsilon} + \bar{\delta} \,,
\end{equation}
eq.~\eqref{eq:MBready} (with the  $\Gamma$-prefactors 
removed) becomes:
\begin{align}\label{eq:MBready2}
 \int\frac{\dd \xi_0}{2 \pi i}\cdots &\int\frac{\dd \xi_{\bq}}{2 \pi i} 
\ \left(-\frac1{p^2}\right)^{1+ \bq \bar{\epsilon} + \bar{\delta} + \xi } \left( \prod_{i=0}^{\bq} \left(m_i^2\right)^{\xi_i} \right) \nonumber \\ 
\times 
&\left[ \prod_{i=0}^{\bq }\Gamma(-\xi_i) \Gamma\left(-\xi_i -\bar{\epsilon} -\bar{\delta_i} \right) \right]
\frac{\Gamma(1+ \bq \bar{\epsilon} + \bar{\delta} + \xi  )}{\Gamma(1+ \bq \bar{\epsilon} + \bar{\delta} ) \Gamma(-(\bq+1)\bar{\epsilon} - \bar{\delta} - \xi )}
 \ .
\end{align}
Up to the redefinition of $\epsilon$ and the $\delta_i$, this is indeed a $\bq$-loop massive banana integral with $\bq +1$ non-zero (generic) masses $m_i$ (compare with eq.~\eqref{eq:MBready} for $\bq=l$).

\paragraph{Evaluation in the large momentum regime.}

The MB integral in eq.~\eqref{eq:MBready} can be evaluated in the large momentum regime, $p^2\gg m_i^2>0$.  We close all integration contours to the right and sum up the residues. In order to get only positive exponents of $-\frac1{p^2}$, we only pick up the residues coming from the gamma functions inside the square brackets. For each variable $\xi_i$ and given integer $n_i$, there are two possibilities to pick up a residue from one of the two $\Gamma$-factors. In total this results in $2^{l+1}$ possibilities, which we label by an $(l+1)$-tuple $\underline j=(j_0, \dots, j_{l})$ whose entries are either $0$ or $1$. Without loss of generality, if $\bq = l$ (all propagators are massive) we obtain upon summation of all relevant residues:
\beq\bsp\label{eq:MBresSummed}
	\tilde I &= \sum_{\underline n \in \mathbb{N}_0^{l+1}} \sum_{\underline j \in \{0,1 \}^{l+1}}  \left(-\frac1{p^2}\right)^{1 + n + \underline\delta \cdot \underline j + (j-1)\epsilon } \ (-1)^{n} \, \prod_{i=0}^{l} \left(m_i^2\right)^{n_i + (1-j_i) (-\delta_i -\epsilon )}
 \\
& \qquad \times \left(\prod_{i=0}^l
\frac{
 \Gamma( -n_i + (-1)^{j_i} (\delta_i + \epsilon) )} {n_i!}
 \right)
\frac{\Gamma(1 + n+ \underline\delta\cdot\underline j+ (j-1)\epsilon )}{\Gamma(1+ l\epsilon + \delta ) \Gamma (-j\epsilon - n - \underline\delta \cdot\underline j )} \ .
\esp\eeq
Here we have written $n=\sum_{i=0}^l n_i$, and similarly $j=\sum_{i=0}^l j_i$, while $\underline\delta \cdot\underline j = \sum_{i=0}^l \delta_i j_i$. . 
Note that, once again, it is straightforward to adapt eq.~\eqref{eq:MBresSummed} to the case of massless propagators (i.e., for $\bq   < l$), because the quantities $\delta_i$ and $\epsilon$ enter eq.~\eqref{eq:MBresSummed} only parametrically. Thus, eq.~\eqref{eq:MBready2} simply evaluates to eq.~\eqref{eq:MBresSummed} with $\delta_i$ and $\epsilon$ replaced by  $\bar{\delta_i}$ and~$\bar{\epsilon}$, respectively. Also note that, since we are working with $p^2>0$, we have to interpret the non-integer powers of $(-p^2)$ by assigning a small positive imaginary part to $p^2$, $p^2\to p^2+i0$:
\beq
 \left(-\frac1{p^2}\right)^{(j-1)\epsilon } = (-p^2-i0)^{-(j-1)\epsilon } = e^{(j-1)i\pi\eps}\,\left(\frac{1}{p^2}\right)^{(j-1)\eps}\,.
 \eeq
Using the well-known identity
\begin{equation}\label{eq:gammaids}
\Gamma(-\epsilon-n)\ = \ (-1)^{n-1}\frac{\Gamma(\epsilon)\Gamma(1-\epsilon)}{\Gamma(1+n+\epsilon)} \ = \ (-1)^{n}\frac{\Gamma(-\epsilon)\Gamma(1+\epsilon)}{\Gamma(1+n+\epsilon)} \,,
\end{equation}
valid for integers $n$, we may rewrite eq.~\eqref{eq:bananaFw} as
\begin{align}\label{eq:MBresSummed2}
	\tilde I	&= (-1)^{l+1}
\frac{(\Gamma(-\epsilon)\Gamma(1+\epsilon))^{l+1}}{\Gamma(1+l\epsilon + \delta )}  
\sum_{   \underline j \in \{0,1 \}^{l+1}}
\sum_{\underline n \in \mathbb{N}_0^{l+1}}   
\left(-\frac1{p^2}\right)^{1 + n +(j-1)\epsilon + \underline \delta \cdot\underline j } \,
\prod_{i=0}^{l} \left(m_i^2\right)^{n_i + (j_i-1) (\delta_i + \epsilon )}  
\nonumber \\
&\!\!\! \times \!
\frac{(-1)^{n+j+\delta+\delta\cdot j}}
{\prod_{i=0}^{l} 
\Gamma(1 + n_i +(-1)^{j_i+1}(\delta_i + \epsilon)) \, n_i! }
 \frac{
 \Gamma(1+n+ \underline\delta\cdot\underline j +j\epsilon) \, \Gamma(1+n+\underline\delta\cdot\underline j+(j-1)\epsilon)}
 {\Gamma( -j\epsilon ) \, \Gamma(1+j \epsilon)} ~.
\end{align}
Note that terms with $j=0$ vanish due to a $\Gamma$-function in the denominator. Since we expect the Feynman integral to obey Picard-Fuchs differential equations with polynomial coefficients in the kinematic parameters and the dimensional regulator $\eps$, whose coefficients are rational numbers, it is natural to separate off the non-algebraic $\Gamma$-factors, which are regarded as coefficients of the specific linear combination of solutions that represent the banana integral. For a shorter notation we introduce the rising and falling factorials,\footnote{The rising factorial is also known as the Pochhammer function. Note that there are different notations for rising and falling factorials used in the literature.}
\begin{equation}
\begin{aligned}
	(x)_n		&=	x (x+1) \cdots (x+n-1)\,,	\\
	[x]_n		&=	x (x-1) \cdots (x-n+1)\,,
\label{pochhammer}
\end{aligned}
\end{equation}
with $(x)_0=[x]_0=1$, and we recall the functional equation of the $\Gamma$-function:
\begin{equation}
\Gamma(x+n) = \Gamma(x) \ (x)_{n} ~.
\end{equation}
We then arrive at the final expression:
\begin{align}\label{eq:MBresSummed3}
	&\tilde I	= 
\frac{(-1)^{l+1}}{\Gamma(1+ l \epsilon + \delta )} \left(-\frac1{p^2}\right)^{1+ l \epsilon + \delta } 	\\
&\times\sum_{\underline j \in \{ 0,1 \}^{l+1}}
(-1)^{j+\delta+\underline\delta\cdot\underline j}\frac{\Gamma(-\epsilon)^{l+1}\, \Gamma(1+\epsilon)^{l+1} \ \Gamma(1+j\epsilon + \underline\delta \cdot\underline j) \,  \Gamma(1+(j-1)\epsilon + \underline\delta \cdot\underline j)}{\Gamma(-j\epsilon) \, \Gamma(1+j\epsilon) \ \prod_{i=1}^{l+1} \Gamma(1+(-1)^{j_i+1} (\delta_i + \epsilon))} \nonumber \\ 
& \times  \left[ \prod_{i=1}^{l+1} \left(-\frac{m_i^2}{p^2}\right)^{(j_i-1)(\delta_i + \epsilon)} \sum_{\underline n\in\mathbb{N}_0^{l+1}} \frac{(1+j\epsilon + \underline\delta\cdot\underline j)_{n}\, (1+(j-1) \epsilon + \underline\delta\cdot\underline j)_{n}}{\prod_{i=1}^{l+1} (1+(-1)^{j_i+1}(\delta_i + \epsilon))_{n_i}} \ \prod_{i=1}^{l+1}\frac{1}{n_i !}\left(\frac{m_i^2}{p^2}\right)^{n_i}\right] \, .\nonumber
\end{align}
At this point we note that, for generic non-zero $\epsilon$, the banana integral in the large momentum regime is essentially a special linear combination of (multivariate) hypergeometric series. In particular, for the specific pattern of Pochhammer functions in eq.~\eqref{eq:MBresSummed3}, these hypergeometric series are known as Lauricella $F_C$ functions in $l+1$ variables. For the sake of completeness we have collected the main definitions and properties of the Lauricella $F_C$ functions in appendix \ref{sec:lauricella}. The representation in eq.~\eqref{eq:MBresSummed3}, however, is not valid for arbitrary values of $z_i=m_i^2/p^2$, because the series in eq.~\eqref{eq:MBresSummed3} have a finite radius of convergence. We will address the analytic continuation in section~\ref{sec:diffeqs}, where we will study the differential equations satisfied by the banana integrals.

For convenience we give the important special case $\delta=0$ explicitly, which, using eq.~\eqref{eq:gammaids} and some algebra, simplifies to
\begin{align}\label{eq:HypGeomBanana}
	& I_{1,\ldots,1}(p^2,\underline m^2;2-2\eps) = 
\frac{1}{\Gamma(1+ l \epsilon )} \left(\frac1{-p^2-i0}\right)^{1+ l \epsilon}
\sum_{\underline j \in \{ 0,1 \}^{l+1}}
\frac{\Gamma(-\epsilon)^{j}\, \Gamma(\epsilon)^{l+1-j} \   \Gamma(1+(j-1)\epsilon )}{\Gamma(-j\epsilon)} \nonumber \\ & \times
\left[ \prod_{i=1}^{l+1} \left(\frac{m_i^2}{-p^2-i0}\right)^{(j_i-1) \epsilon} \sum_{\underline n\in\mathbb{N}_0^{l+1}} \frac{(1+j\epsilon)_{n}\, (1+(j-1) \epsilon )_{n}}{\prod_{i=1}^{l+1} (1+(-1)^{j_i+1} \epsilon)_{n_i}}\  \prod_{i=1}^{l+1}\frac{1}{n_i !}\left(\frac{m_i^2}{p^2}\right)^{n_i}\right] \ .
\end{align}
We can see here that, similar to the case $D=2$, the generic-mass banana integral is symmetric in variables $m_i^2/p^2$ for $i=1,\hdots,l+1$. Or the other way round, the equal-mass case is somehow the diagonal of the generic-mass case. This reflects the fact that the banana graphs have a symmetry which interchanges the propagators.

\paragraph{Leading asymptotic behavior at large momentum.}
 
Letting $\underline n=(0,\dots,0)$ in eq.~\eqref{eq:HypGeomBanana}, we can immediately extract the leading behavior of the banana integrals at large momentum. For the generic-mass case one obtains:
\beq\bsp\label{eq:asymp_generic}
 I_{1,\ldots,1}&(p^2,\underline m^2;2-2\epsilon)= 
 -\frac{1}{\Gamma(1+l\epsilon)} e^{i\pi l\eps}\left(\frac1{p^2}\right)^{1+l\epsilon}\,
 \\
&\times\sum_{\underline j \in \{ 0,1 \}^{l+1}}
e^{i\pi(j-1)\eps}
   \frac{\Gamma(-\epsilon)^{j} \, \Gamma(\epsilon)^{l+1-j} \   \Gamma(1+(j-1)\epsilon )}{\Gamma(-j\epsilon)} \prod_{i=1}^{l+1} z_i^{(j_i-1) \epsilon}+ \mathcal{O}\left(z_i^2\right) \,.
\esp\eeq
Equation~\eqref{eq:asymp_generic} gives the leading asymptotics of $I_{1,\ldots,1}(p^2,\underline m^2;2-2\epsilon)$ and can be used as a boundary condition to solve the differential equations for the banana graphs.  
In the equal-mass case, eq.~\eqref{eq:asymp_generic} can be further simplified to
\begin{align}
J_{l,1}(z;\epsilon)&=-\sum_{k=1}^{l+1} \binom{l+1}{k}
\frac{\Gamma(-\epsilon)^{k} \ \Gamma(\epsilon)^{l+1-k} }
{\Gamma(-k\epsilon) } 
\frac{\Gamma(1+(k-1)\epsilon )}
{\Gamma(1+l\epsilon)} 
e^{(k-1)i\pi\eps}\, z^{1+(k-1)\epsilon}
 + \mathcal{O}\left(z^2\right) \,.
\label{MBfirstorder}
\end{align}
Expanding eq.~\eqref{MBfirstorder} around $\epsilon=0$, one obtains
\begin{equation}\label{eq:epsexp}
	J_{l,1}(z;\epsilon) = \sum_{n=0}^{\infty} J_{l,1}^{(n)}(z) \ \epsilon^n \ ,
\end{equation}
 and inspecting the leading order in $\epsilon$, $J_{l,1}^{(0)}$ \emph{precisely} reproduces the logarithmic structure of the $l$-loop banana Feynman integral in $D=2$ spacetime dimensions described in section 3 of ref.~\cite{Bonisch:2020qmm} (see also eq.~\eqref{eq:Jk0_sol0}). Equation~\eqref{MBfirstorder} thus proves the combinatorial pattern of Riemann zeta values in the constants $\lambda_k^{(l)}$ in eq.~\eqref{eq:BC_D=2}, and extends it to higher orders in $\eps$. It was also observed that the transcendental weight of $\lambda_k^{(l)}$, which is $l-k$, and the highest occurring power of logarithms in $\varpi_{l,k}(z)$, which is $k$, always add up to the loop order $l$. 
  
We can now partially generalize these observations to higher orders in the $\eps$-expansion. 
For simplicity we discuss the equal-mass case, while the generic-mass case can be treated similarly.
At order $\epsilon^n$ we find powers of $\log(z)$ up to $L=l+n$ and coefficients with transcendental weight up to $T=l+n$.
Indeed, at any order $\epsilon^n$, the constant of highest occurring transcendental weight $T=l+n$ always multiplies the holomorphic period $\varpi_{l,0}(z)$ of the Calabi-Yau $(l-1)$-fold. 
More generally, at order $\epsilon^n$, consider all terms that multiply constants of a certain transcendental weight $T$ and a certain power $L$ of logarithms, i.e., $\log^L(z)$. We find that there are only non-zero terms for $L+T \leq l+n $. For the maximum value $L+T = l+n$, these terms are again proportional to the holomorphic Calabi-Yau period  $\varpi_{l,0}(z)$ (with a rational number as proportionality constant), independently of the splitting between $L$ and $T$ and the value of $n$.

\subsection{Differential equations for banana Feynman integrals}
\label{sec:diffeqs}  

The goal of this section is to present a method to derive differential equations for banana integrals at arbitrary loop order. In principle, the differential equations can be derived using IBP relations (see section~\ref{sec:deqs}). However, for high numbers of loops $l$ and for many distinct values of the masses $m_i^2$, publicly available computer codes can be rather inefficient, and it can be hard to obtain the explicit form of the differential equations. Here we present an alternative method to derive a set of operators that define inhomogeneous differential equations for $I_{1,\ldots,1}(1,\uz;2-2\eps)$, cf. eq.~\eqref{eq:higher_DEQ}. As explained in section~\ref{sec:diffeqs}, these differential operators form an ideal, and there may be substantial freedom in choosing a generating set for this ideal. As a consequence, our set of operators may look substantially different from the one obtained from first-order differential equations and IBP relations. 

For the generic-mass case we use a purely combinatorial method to find the desired differential operators by analyzing the structure of the gamma functions appearing in the second line of eq.~\eqref{eq:HypGeomBanana}. We work with gamma functions instead of rising factorials, because this will simplify the formulas. This can be achieved simply be rescaling the coefficient in the first line of eq.~\eqref{eq:HypGeomBanana}. We define the $\eps$-Frobenius basis with an explicit $\eps$-dependent indicial by:
\begin{align}
	\mathcal{I}_{(j_1, \hdots, j_{l+1})}(\uz,\eps)	\coloneqq  \sum_{\underline n \in \mathbb{N}_0^{l+1}}       
	 \frac{\Gamma(1+n+j\epsilon)\, \Gamma(1+n+(j-1)\epsilon)}{\prod_{i=1}^{l+1} \Gamma(1+n_i)\, \Gamma(1+n_i+(-1)^{j_i+1}\epsilon)}   \prod_{i=1}^{l+1} z_i^{n_i + j_i \epsilon} \,.
\label{Frobgeneric}
\end{align}
Indeed, it is easy to see that eq.~\eqref{eq:HypGeomBanana} can be written as a linear combination of the $\mathcal{I}_{(j_1, \hdots, j_{l+1})}$. The power series in eq.~\eqref{Frobgeneric}, however, have a finite radius of convergence. Our goal is to derive a set of $1<j\leq l+1$ differential operators that annihilate the elements of the $\eps$-Frobenius basis, where we again use the notation $j=\sum_{i=1}^{l+1}j_i$. and similarly for $n$. The set of differential operators is then extended to a set of inhomogeneous differential equations by including the case $j=1$ (for now, we exclude the case $j=0$, but we briefly comment on it at the very end). These differential equations then serve as a starting point to analytically continue the $\eps$-Frobenius basis to all values of the $z_i$. 

Let us start by analyzing the maximal cuts. We want to find a set of differential operators $\{\mathcal{L}_k\}$ whose solution space is spanned by the $\eps$-Frobenius basis in eq.~\eqref{Frobgeneric} (and only those):
\begin{align}
	\text{Sol}(\{\mathcal L_k\})=\Big\langle \mathcal{I}_{(j_1, \hdots, j_{l+1})}(\uz,\eps) \Big\rangle_{1<j\leq l+1}~.
\label{solspacefrob}
\end{align}
 Clearly, this solution space will then also contain the maximal cuts of the banana integral for a given loop order $l$.

The most general linear differential operator with polynomial coefficients acting on the $\eps$-Frobenius basis can be written in the form 
\begin{align}
	\mathcal L_{\underline\alpha,\underline\beta}	=	\sum_{\underline\alpha,\underline\beta}	a_{\underline\alpha,\underline\beta}	~	\uz^{\underline\alpha}\underline\theta^{\underline\beta}~,
\label{generalop}
\end{align}
where $a_{\underline\alpha,\underline\beta}$ are constants and $\underline\alpha,\underline\beta$ are multi-indices, and $\uz^{\underline\alpha}\coloneqq\prod_{k=1}^{l+1}z_k^{\alpha_k}$ and $\underline\theta^{\underline\beta} \coloneqq \prod_{k=1}^{l+1}\theta_k^{\alpha_k}$. It is therefore sufficient to understand the action of the $\zs^{\underline\alpha}\underline{\theta}^{\underline\beta}$ on the elements of the $\eps$-Frobenius basis. 

Let us start by analyzing the maximally symmetric index $\underline j=(1,\ldots,1)$. With the rising and falling factorial from eq.~\eqref{pochhammer}, we find
\begin{equation}\begin{aligned}
	\uz^{\underline\alpha}\underline\theta^{\underline\beta}~ \mathcal I_{(1, \hdots, 1)}		&=	\sum_{\underline n \in \mathbb{N}_0^{l+1}}       \frac{\prod_{i=1}^{l+1}[n_i]_{\alpha_i}[n_i+\epsilon]_{\alpha_i}}{[n+l\epsilon]_{\alpha}[n+(l+1)\epsilon]_{\alpha}}	 	 \prod_{k=1}^{l+1}(n_k-\alpha_k+\epsilon)^{\beta_k}\\
	&\qquad\qquad\qquad	\times	 \frac{\Gamma(1+n+(l+1)\epsilon)\, \Gamma(1+n+l\epsilon)}{\prod_{i=1}^{l+1} \Gamma(1+n_i)\, \Gamma(1+n_i+\epsilon)}   \prod_{i=1}^{l+1} z_i^{n_i +  \epsilon} ~,
\label{actiongeneralop}
\end{aligned}\end{equation}
where in the left-hand side we again use the simplified notation $\alpha=\sum_i \alpha_i$. Using a computer algebra system we can now solve for the coefficients $a_{\underline\alpha,\underline\beta}$ in eq.~\eqref{generalop} such that the operator $\mathcal L_{\underline\alpha,\underline\beta}$ annihilates the $\eps$-Frobenius basis element with $\underline j=(1,\ldots,1)$,
\begin{align}
	\mathcal L_{\underline\alpha,\underline\beta}~ \mathcal I_{(1, \hdots, 1)}	=0~.
\label{generalop2}
\end{align}
For fixed $\alpha$ and $\beta$, any linear combination that annihilates $\mathcal I_{(1,\hdots,1)}$ can be computed in this way\footnote{A \texttt{Mathematica}-code which generates these operators can be downloaded from \url{http://www.th.physik.uni-bonn.de/Groups/Klemm/data.php} as supplementary data to this paper.}.

As an example, let us consider the two-loop case. Here we can construct the operators
\beq\bsp
	\mathcal L_1	&\,=	(1-z_1)(\theta_2-\epsilon)(\theta_3-\epsilon)-(z_2(\theta_3-\epsilon)+z_3(\theta_2-\epsilon))(2\theta_1+\theta_2+\theta_3+1-\epsilon)\,,	\\
	\mathcal L_4 	&\,=	\theta_1(\theta_1-\epsilon)-z_1(\theta_1+\theta_2+\theta_3+1)(\theta_1+\theta_2+\theta_3+1-\epsilon)	~,
\label{exampleops}
\esp\eeq
and $\mathcal L_2=\mathcal L_1(1\leftrightarrow2), ~\mathcal L_3=\mathcal L_1(1\leftrightarrow3),~ \mathcal L_5=\mathcal L_4(1\leftrightarrow2)$ and $\mathcal L_6=\mathcal L_4(1\leftrightarrow3)$	. For these operators we have chosen $\alpha\leq1$ and $\beta\leq2$. The set of operators 
\beq\label{eq:L2set}
\mathcal L^{(2)} \coloneqq \{\mathcal L_1, \hdots, \mathcal L_6\}
\eeq
 is enough to uniquely determine the $\eps$-Frobenius basis elements with $j=2,3$, and no other solutions.

More generally, for $\alpha$ and $\beta$ chosen large enough, the operators $\mathcal L_{\underline\alpha,\underline\beta}$ will generate the solution space in eq.~\eqref{solspacefrob}, i.e., $\text{Sol}(\{\mathcal L_{\underline\alpha,\underline\beta}\})=\langle \mathcal{I}_{(j_1, \hdots, j_{l+1})}(\uz,\eps) \rangle_{1<j\leq l+1}$. From our computations we actually found that the operators with $\alpha\leq1$ and $\beta\leq l$ are enough to generate the desired solution space:
\begin{align}
	\text{Sol}(\{\mathcal L_{\underline\alpha,\underline\beta}\}_{\alpha\leq1,\beta\leq l})=\Big\langle \mathcal{I}_{(j_1, \hdots, j_{l+1})} (z,\eps)\Big\rangle_{1<j\leq l+1}~.
\label{solspacefrob}
\end{align}
The set $\{\mathcal L_{\underline\alpha,\underline\beta}\}_{\alpha\leq1,\beta\leq l}$ is still overcomplete:
A properly chosen subset of $\{\mathcal L_{\underline\alpha,\underline\beta}\}_{\alpha\leq1,\beta\leq l}$ can be sufficient to generate the whole solution space. For example, the operators from the set $\mathcal L^{(2)}$ are sufficient to generate the desired $\eps$-Frobenius basis elements in the sunset case, although there exist $9$ linearly independent operators for $\alpha\leq1$ and $\beta\leq2$.

As mentioned in subsection \ref{para:MultiParameterPF}, it is in general not clear what is the best way of representing an ideal such that it yields the desired solution space. For instance, naively, the maximal cut of the two-loop banana integral was found to satisfy a fourth order homogeneous differential equation, cf., e.g., ref.~\cite{Caffo:1998du}. Later, it was shown that a second-order differential operator suffices (in two space-time dimensions), cf. refs.~\cite{Muller-Stach:2011qkg,MullerStach:2012mp,Remiddi:2013joa}. Using our method, we would obtain the set $\mathcal{L}^{(2)}$ of six operators, which has still the same solution space. Our set $\mathcal{L}^{(2)}$ consists of more operators, but of simpler type. For example, the polynomials appearing in the operators in eq.~\eqref{exampleops} are of small degree. The question of which representation is more appropriate depends on the concrete application that one has in mind.

So far we have only discussed how to find a set of operators $\{\mathcal L_{\underline\alpha,\underline\beta}\}$ that annihilate the functions $\mathcal{I}_{(j_1, \hdots, j_{l+1})} (z,\eps)$ with $1<j\le l+1$. Equivalently, the solution space $\textrm{Sol}(\{\mathcal L_{\underline\alpha,\underline\beta}\})$ will be generated by all the maximal cuts of the $l$-loop banana graph. 
In order to describe the full uncut Feynman integral, we need to include the corresponding functions with $j=1$. There are $\binom{l+1}{1}=l+1$ different functions of this type (and there is the same number of $l$-loop tadpole graphs for generic masses). These functions, however, are not elements of $\textrm{Sol}(\{\mathcal L_{\underline\alpha,\underline\beta}\})$ (i.e., they are not simultaneously annihilated by all the $\mathcal{L}_{\underline\alpha,\underline\beta}$). Instead, they  are special solutions to certain inhomogeneous differential equations obtained from the $\mathcal L_{\underline\alpha,\underline\beta}$. In analogy with eq.~\eqref{solspacemulti}, we define the solution space of a set of inhomogeneous differential equations:
\begin{align}\label{solspaceinhom}
	\text{SolInhom}&(\{(\mathcal L_k, g_k)\})=\{ f(\uz) |	\mathcal L_i f(\uz)=g_i(\uz)	 \text{ for all }	(\mathcal L_i, g_i) \in \{ (\mathcal L_k,g_k) \}	\}~.
\end{align}
In this language, an element of the $\eps$-Frobenius basis with $j=1$ lies in the solution space $\text{SolInhom}(\{(\mathcal L_{\underline\alpha,\underline\beta}, g_{\underline\alpha,\underline\beta})\})$, for a specific set of inhomogeneities $\{g_{\underline\alpha,\underline\beta}\}$ depending on the particular chosen $j$-vector.  In order to determine these inhomogeneities, we simply apply every generator $\mathcal{L}_{\underline\alpha,\underline\beta}$ to every element of the Frobenius basis with $j=1$. By this procedure one finds for each basis element with $j=1$  inhomogeneities of the form $z_i^\epsilon$. 

Let us illustrate this again on the example of the two-loop banana integral. Acting with the the operators from $\mathcal L^{(2)}$ yields the following inhomogeneities:
\begin{equation}\begin{aligned}
	\mathcal L_1	 \mathcal I_{(1,0,0)}(\uz,\eps) 	&=	\frac1{\Gamma(-\epsilon)^2}~z_1^\epsilon\,,	\\
	\mathcal L_2	 \mathcal I_{(0,1,0)}(\uz,\eps)  	&=	\frac1{\Gamma(-\epsilon)^2}~z_2^\epsilon	\,,\\
	\mathcal L_3	 \mathcal I_{(0,0,1)}(\uz,\eps)  	&=	\frac1{\Gamma(-\epsilon)^2}~z_3^\epsilon	\,.
\label{exampleinhoms}
\end{aligned}\end{equation}
All other operators give zero when applied to the three $\eps$-Frobenius solutions with $j=1$ 
We stress that there is still a freedom in how we choose the set $\{(\mathcal L_{\underline\alpha,\underline\beta}, g_{\underline\alpha,\underline\beta})\}$ and construct the corresponding solution spaces. For example, we could write the vector space $\text{SolInhom}(\{(\mathcal L_{\underline\alpha,\underline\beta}, g_{\underline\alpha,\underline\beta})\})$ as a sum of three solution spaces
\beq\label{eq:SolInhomSum}
\text{SolInhom}(\{\mathcal L_{\underline\alpha,\underline\beta}, g_{\underline\alpha,\underline\beta}\}) = \sum_{p=1}^3 \text{SolInhom}(\mathcal{L}^{(2,p)}) = \Big\langle \mathcal{I}_{(j_1, j_2, j_{3})} (\zs,\eps)\Big\rangle_{1\le j\leq 3}\,,
\eeq
with
\beq
\mathcal{L}^{(2,p)} = \{(\mathcal{L}_i,g_i): g_i(\zs) = \delta_{ip}\,z_p^{\eps}\,\Gamma(-\eps)^{-2}, 1\le i\le 6\}\,.
\eeq
Since $I_{1,1,1}(\uz,D)$ is symmetric under a permutation of the $z_i$, the inhomogeneous term must also have this property. It is therefore sufficient to consider $\{(\mathcal L_{\underline\alpha,\underline\beta}, g_{\underline\alpha,\underline\beta})\}$ such that the solution space contains the sum $\mathcal I_{(1,0,0)}+\mathcal I_{(0,1,0)}+\mathcal I_{(0,0,1)}$, but it does not contain each summand separately. This is achieved by the choice 
\beq
\text{SolInhom}(\{(\mathcal L_{\underline\alpha,\underline\beta}, g_{\underline\alpha,\underline\beta})\}) = \text{SolInhom}(\mathcal{L}^{(2)}_{\textrm{inhom}})\,,
\eeq
with
\beq
\mathcal{L}^{(2)}_{\textrm{inhom}} = \{(\mathcal L_1,{\Gamma(-\epsilon)^{-2}}~z_1^\epsilon), (\mathcal L_2,{\Gamma(-\epsilon)^{-2}}~z_2^\epsilon), (\mathcal L_3,{\Gamma(-\epsilon)^{-2}}~z_3^\epsilon), (\mathcal L_4,0), (\mathcal L_5,0), (\mathcal L_6,0)\}\,.
\eeq
Note that $\text{SolInhom}(\mathcal{L}^{(2)}_{\textrm{inhom}})$ is contained in the sum of vector space in eq.~\eqref{eq:SolInhomSum}, but the converse is not true.
Finally, we can also describe the solution space of a set of homogeneous equations, by multiplying $\mathcal{L}_1$, $\mathcal{L}_2$ and $\mathcal{L}_3$ from left by an operator that annihilates the inhomogeneity. We have
\beq
\text{Sol}(\tilde{\mathcal{L}}^{(2)}) =  \Big\langle \mathcal{I}_{(j_1, j_2, j_{3})} (\zs,\eps)\Big\rangle_{1\le j\leq 3}\,,
\eeq
with 
\beq\label{eq:L2tilde}
\tilde{\mathcal{L}}^{(2)} \coloneqq \{ (\theta_1-\epsilon)\mathcal L_1, (\theta_2-\epsilon)\mathcal L_2, (\theta_3-\epsilon)\mathcal L_3, \mathcal L_4, \mathcal L_5, \mathcal L_6	\}\,.
\eeq
The previous discussion makes it clear that it is a matter of taste whether we consider a set of inhomogeneous differential operators or a set of higher-order differential operators; the resulting solution spaces contain all functions necessary to solve the problem at hand. Our strategy of first constructing combinatorially an ideal which can then be extended to a set of inhomogeneous differential equations guarantees that we generate the correct solution space. 

The method we have just described can easily be implemented into a computer algebra system. We can in this way derive a set of inhomogeneous differential equations satisfied by $J_{l,\underline{0}}(\uz,\eps)$. 
The coefficients of the linear combination in the Frobenius basis in eq.~
\eqref{Frobgeneric} can be read off by comparing to the hypergeometric series representation in eq.~\eqref{eq:HypGeomBanana} (or by using eq.~\eqref{eq:asymp_generic} as a boundary condition in the large momentum limit). At this point we have to make an important comment. Our strategy to obtain the differential equations consisted in starting from the MB representation, which leads to the hypergeometric series representation in eq.~\eqref{eq:HypGeomBanana}. It may thus appear that we did not gain anything, because we have derived the differential equations \emph{after} we knew the solution, cf. eq.~\eqref{eq:HypGeomBanana}. The series representation in eq.~\eqref{eq:HypGeomBanana}, however, does not converge for all values of $\uz$. The differential equations allow us to analytically continue the series in eq.~\eqref{eq:HypGeomBanana}, e.g., by transforming the differential equation to another point and to obtain local power series representations close to that point (see the discussion in section~\ref{sec:deqs}).

Let us make some comments about our differential equations. First, we emphazise that our procedure allows us to derive differential equations for arbitrary values of $\uz$, including zero masses. This follows immediately from the fact that eq.~\eqref{eq:HypGeomBanana} is valid also in the case of massless propagators. 
Second, we point out that from our higher-order inhomogeneous differential equation for $J_{l,\underline{0}}(\uz,\eps)$, we can easily obtain the first-order Gauss-Manin system for the master integrals in eq.~\eqref{eq:banana_MIs}. When extracting the entries of the matrix $\widetilde{{\bf A}}(\uz;\epsilon)$ in eq.~\eqref{eq:firstorder}, we need to divide by the discriminant of the system. This introduces typically a very long and complicated polynomial, especially in the multi-parameter case. Therefore, the matrix $\widetilde{{\bf A}}(\uz;\epsilon)$ is usually very complicated, and we prefer to work with the larger, but simpler, set of differential operators $\{\mathcal L_{\alpha,\beta}\}$ constructed in this section.

\subsubsection{Comments on the number of solutions.}
\label{subsec:commentsnumbersolutions}

We conclude with some comments about the number of elements in the $\eps$-Frobenius basis in eq.~\eqref{Frobgeneric} for a given loop order $l$. We count the number of solutions with the same value of $j$, which itself counts how many $\epsilon$'s appear in the indicials to the differential operators. Thereby, we find for fixed $j$ exactly $\binom{l+1}{j}$ different solutions. In total we obtain the following sequence:
\begin{align}
	\underbrace{\binom{l+1}{1}=l+1}_{j=1}	\quad 	\Bigg| \quad	\underbrace{\binom{l+1}{2}}_{j=2} \quad	\underbrace{\binom{l+1}{3}}_{j=3}\quad	\hdots\quad \underbrace{\binom{l+1}{l+1}=1}_{j=l+1}~.
\label{numbersolutions}
\end{align}
Here the vertical line separates the special solutions of the inhomogeneous equations from the homogeneous ones. For example, for the two-loop case this reduces to 
\begin{align}
	3 \quad | \quad 3\quad 1~,
\end{align}
in agreement with the discussion above. If we compare the number of solutions in the generic-mass case for $\epsilon=0$ and $\epsilon\neq0$, we observe that for $\epsilon\neq0$ we have more solutions. For example, at two-loop order the number of solutions for $\epsilon=0$ is
\begin{align}
	1 \quad | \quad 1\quad 1~.
\end{align}
The fact that the number of master integrals is smaller in exactly $D=2$ dimensions was already observed in refs.~\cite{Muller-Stach:2011qkg,MullerStach:2012mp,Remiddi:2013joa} for the two-loop banana graph. Comparing to the number of solutions for $\eps=0$ for the maximal cut in $D=2$ dimensions in ref.~\cite{Bonisch:2020qmm}, we see that generically the dimension of the solution space, and therefore the number of master integrals, increases by introducing a non-vanishing dimensional regularization parameter $\epsilon$.

Let us make another  comment about the number of solutions for different values of $j$. We can extend the sequence in eq.~\eqref{numbersolutions} by including the function in eq.~\eqref{Frobgeneric} for $j=0$. There is exactly one such function. Note that this solution can also be included into eq.~\eqref{eq:MBresSummed3}, because it would enter the linear combination with a vanishing prefactor. We then obtain the sequence: 
\begin{align}
	\underbrace{1=	\binom{l+1}{0}}_{j=0}	\quad 	\Bigg|\Bigg| \quad		\underbrace{\binom{l+1}{1}=l+1}_{j=1}	\quad 	\Bigg| \quad	\underbrace{\binom{l+1}{2}}_{j=2} \quad	\underbrace{\binom{l+1}{3}}_{j=3}\quad	\hdots\quad \underbrace{\binom{l+1}{l+1}=1}_{j=l+1}~.
\label{numbersolutions2}
\end{align}
The solution space of the functions in eq.~\eqref{Frobgeneric} with $0\le j\le l+1$ can be described in terms of differential equations in different ways. One possibility is that one allows only certain linear combinations of operators determining the solution space corresponding to the sequence in eq.~\eqref{numbersolutions}. One may apply appropriate $\theta$-derivatives to the same operators, similar to the procedure of extending the solution space by the $\eps$-Frobenius solutions with $j=1$. To be precise, let us look again at the two-loop case. Here the $\eps$-Frobenius element with $j=0$ can be included into the solution space if one considers the following ideal $\{ \theta_1(\theta_1-\epsilon)\mathcal D_1, \theta_2(\theta_2-\epsilon)\mathcal D_2, \theta_3(\theta_3-\epsilon)\mathcal D_3, \mathcal D_4, \mathcal D_5, \mathcal D_6\}$. This ideal allows exactly the $8$ desired solutions which can be grouped into
\begin{align}
	1\left.\left.	\right|\right| \quad	\left.	3	\quad 	\right| \quad	3\quad	1~.
\label{numbersolutions2loopcase}
\end{align}

More generally, the pattern of the number of solutions in eq.~\eqref{numbersolutions2} corresponds to the pattern of the dimensions of the (co)homology groups of the 
ambient space $\mathbb{P}_{l+1}=\otimes_{i=1}^{l+1} \mathbb{P}^1_{(i)}$  the ambient space in which the Calabi-Yau space for the critical spacetime dimensions $\epsilon=0$ is 
embedded, as it is explained in eq.~\eqref{CICY}. It is tantalising to speculate that in dimensional regularization the solutions can be 
interpreted as some kind of twisted quantum deformation of the cohomology of the $\mathbb{P}_{l+1}$.

\subsection{Interpretation in terms of cut integrals}
\label{subsec:non-max-cuts}  

In this section we provide an interpretation of the additional special solutions to the inhomogeneous Picard-Fuchs differential ideal in terms of non-maximal cut integrals. This interpretation complements and extends the interpretation of the maximal cut integrals as solutions to the associated homogeneous system and its relationship to the Frobenius basis for the solution space of the Picard-Fuchs differential ideals for the maximal cuts. We start by defining non-maximal cuts in general (not restricted to banana integrals), and then comment on the relationship to the solution space of the Picard-Fuchs differential ideal at the end of this section.

Let us consider the setup and the notation of section~\ref{sec:deqs}. 
A \emph{non-maximal cut contour} $\Gamma$ for the integral $I_{\unu}(\ux;D)$ is a contour  that encircles some of the propagators of $I_{\unu}(\ux;D)$. We assume that $\Gamma$ encircles at least one propagator. If it encircles all of them, then this definition agrees with the definition of maximal cut contours given in section~\ref{sec:deqs}. At this point we have to make an important comment about dimensional regularization. In section~\ref{sec:deqs} we had only defined maximal cut integrals in \emph{integer} dimensions $D$, where the concept of an integration contour has an immediate geometrical interpretation. For the following discussion, it will be useful to extend the definition of cut integrals in dimensional regularization. The corresponding integration contours can still be defined geometrically, but they need to be interpreted as \emph{twisted cycles}~\cite{AomotoKita}, see also refs.~\cite{Mizera:2017rqa,Mastrolia:2018uzb,Mizera:2019gea,Frellesvig:2019kgj,Frellesvig:2019uqt}. The distinction will not be crucial for the discussion that follows, and it is sufficient to apply intuition from ordinary cycles in integer dimensions. It is important, however, to point out that  beyond one loop non-maximal cut integrals may diverge even if the original Feynman integral is finite, cf., e.g., refs.~\cite{Abreu:2014cla,Ananthanarayan:2021cch}. The divergences are of infrared origin and arise from massless particles that are put on-shell when taking the residues. Therefore, it is important to work with an appropriate infrared regulator when discussing non-maximal cut integrals. Dimensional regularization provides such a regulator.

Following eqs.~\eqref{eq:cut_def} and~\eqref{eq:equal-mass_cuts}, if $\Gamma$ is a cut contour and $J_i(\zs;\eps)$ denotes a master integral, then we denote the corresponding cut integral by $J_i^{\Gamma}(\zs;\eps)$. The vector $\uJ^{\Gamma}(\zs;\eps)$ then satisfies the same system of differential equations as the vector of master intergals $J_i(\zs;\eps)$ in eq.~\eqref{eq:JDEQ}, i.e., we have
\beq\label{eq:cut_DEQ}
\rd \uJ^{\Gamma}(\zs;\eps) = \widetilde{A}(\zs;\eps)\uJ^{\Gamma}(\zs;\eps)  \quad \text{for every cut contour $\Gamma$}\,.
\eeq

Let us discuss how we can construct a basis of cut integrals. 
We say that a sector is \emph{reducible} if every integral from this sector can be written as a linear combination of integrals from lower sectors. Let $\Theta_1,\ldots,\Theta_s$ denote the set of irreducible sectors. There is a natural partial order on the $\Theta_r$ (coming from the partial order on sectors, see section~\ref{sec:deqs}). In particular, we choose $\Theta_1=(1,\ldots,1)$. We denote by $\uJ_r(\zs;\eps) = (J_{r,1}(\zs;\eps),\ldots,J_{r,M_r}(\zs;\eps))^T$ the master integrals in the sector $\Theta_r$ (by which we mean that those master integrals cannot be expressed as linear combinations of integrals from lower sectors; cf. eq.~\eqref{eq:sector_masters}).
In each sector $\Theta_r$ we can now choose a basis of $M_r$ maximal cut contours, i.e., a set of $M_r$ independent cut contours that encircle precisely the propagators that define the sector $\Theta_r$. Let us denote the basis of maximal cut contours in the sector $\Theta_r$ by $\Gamma_{r,1},\ldots,\Gamma_{r,M_r}$. Note that for $r=1$ and integer $D$, we recover the maximal cut contours defined in section~\ref{subsec:f-odeqs}. Each contour $\Gamma_{r,i}$ defines a valid non-maximal cut contour for integrals from sectors with more propagators.

Consider the $M\times M$ matrix (where $M=\sum_{r=1}^sM_r$ is the total number of master integrals):
\beq
{\bf J}(\zs;\eps) \eqqcolon \Big(\uJ^{\Gamma_{1,1}}(\zs;\eps),\ldots,\uJ^{\Gamma_{1,M_1}}(\zs;\eps),\uJ^{\Gamma_{2,1}}(\zs;\eps),\ldots, \uJ^{\Gamma_{s,M_s}}(\zs;\eps)\Big)\,.
\eeq
It is easy to check that the columns of ${\bf J}(\zs;\eps)$ are linearly independent (for generic $\zs$), and so ${\bf J}(\zs;\eps)$ is a fundamental solution matrix for the system in eq.~\eqref{eq:cut_DEQ}. As a corollary, we conclude that every master integral can be written as a linear combination of its cut integrals in the basis of cut contours $\Gamma_{r,i}$:
\beq\label{eq:cut_decomposition}
J_k(\zs;\eps) = \sum_{r=1}^s\sum_{i=1}^{M_r}a_{r,i}(\eps)\,J_k^{\Gamma_{r,i}}(\zs;\eps)\,,
\eeq
where the coefficients $a_{r,i}(\eps)$ may depend on $\eps$, but they are independent of $\zs$. Note that this relation is very reminiscient of the celebrated \emph{Feynman Tree Theorem}~\cite{Feynman:1963ax,Feynman:1972mt}. It would be interesting to work out the relationship between the basis decomposition in eq.~\eqref{eq:cut_decomposition} and the Feynman Tree Theorem in the future.

Assume now that ${\cal D}^{(k)}$ generates the Picard-Fuchs differential ideal that annihilates $J_k(\zs;\eps)$, i.e., it is a complete set of differential operators that annihilate $J_k(\zs;\eps)$. Since the solution space of the Picard-Fuchs differential ideal must agree with the general solution obtained from the system in eq.~\eqref{eq:cut_DEQ}, eq.~\eqref{eq:cut_decomposition} implies that
\beq\label{eq:Sol_decomposition}
\textrm{Sol}({\cal D}^{(k)}) =\Big\langle J_k^{\Gamma_{r,i}}(\zs;\eps) \Big\rangle_{{1\le r\le s; 1\le i\le M_r}}
=\sum_{r=1}^s\Big\langle J_k^{\Gamma_{r,i}}(\zs;\eps) \Big\rangle_{{1\le i\le M_r}}\,,
\eeq
where in the second inequality we have made explicit the fact that the solution space can be decomposed into contributions from cut contours from different sectors.
The previous considerations show that we can, at least in principle, obtain a basis of the solution space of the system of differential equations satisfied by the master integrals that consists entirely of cut integrals. Just like in the case of maximal cuts, however, constructing such a basis of cycles explicitly can be a monumental task, and it is in general not possible to follow this route.

In the following, we argue that the special solutions from section~\ref{sec:diffeqs} that extend the solution space for the Picard-Fuchs differential ideal for the maximal cuts can be identified with non-maximal cuts. However, similar to the discussion of the relationship between the Frobenius basis and the maximal cuts defined via cycles from integral homology, the special solutions constructed in the previous section will not be obtained from non-maximal cut contours defined over the integers. 

Let us illustrate this on the example of the two-loop case. In particular, let us discuss the two-loop master integral $J_{1,\underline{0}}(\zs;\eps)$. The Picard-Fuchs differential ideal is generated by the set $\tilde{\cal L}^{(2)}$ in eq.~\eqref{eq:L2tilde}. Its solution space admits the decomposition
\beq\label{eq:Sol2_decomposition}
\textrm{Sol}(\tilde{\cal L}^{(2)}) = \textrm{Sol}({\cal L}^{(2)}) +  \Big\langle\mathcal I_{(1,0,0)}(\uz,\eps)\Big\rangle+\Big\langle\mathcal I_{(0,1,0)}(\uz,\eps) \Big\rangle +\Big\langle\mathcal I_{(0,0,1)}(\uz,\eps)\Big\rangle\,,
\eeq
where ${\cal L}^{(2)}$ was defined in eq.~\eqref{eq:L2set}. Let us interpret eq.~\eqref{eq:Sol2_decomposition} in the light of eq.~\eqref{eq:Sol_decomposition}. We know from section~\ref{sec:deqs} that there are four irreducible sectors for the two-loop banana integral, namely 
\beq
\Theta_1=(1,1,1)\,, \quad \Theta_2=(1,1,0)\,, \quad\Theta_3=(1,0,1)\,,\quad \Theta_4=(0,1,1)\,,
\eeq
and $M_1=4$ and $M_2=M_3=M_4=1$. Let us start by discussing the first term in eq.~\eqref{eq:Sol2_decomposition}. Its interpretation is similar to the discussion in section~\ref{sssec:boundary} and below eq.~\eqref{eq:max_cut_decomp} (which were restricted to $D=2$ dimensions): The cut contours $\Gamma_{1,i}$, $1\le i\le 4$, define a basis for the maximal cut contours of $J_{1,\underline{0}}(\zs;\eps)$.\footnote{Note that these contours define the maximal cuts in $D=2-2\eps$ dimensions. Consequently, there are four maximal cut contours, and not only two, just like there are more master integrals in  $D=2-2\eps$ than in $D=2$ dimensions.} The maximal cut integrals $J_{1,\underline{0}}^{\Gamma_{1,i}}(\zs;\eps)$ are annihilated by the differential operators from ${\cal L}^{(2)}$, and they form an integral basis for the solution space:
\beq
\textrm{Sol}({\cal L}^{(2)}) = \Big\langle J_{1,\underline{0}}^{\Gamma_{1,i}}(\zs;\eps) \Big\rangle_{{1\le i\le 4}}\,.
\eeq
This integral basis may be hard to construct, as it requires a detailed knowledge of the cycles.
We know, however, that the solution space $\textrm{Sol}({\cal L}^{(2)})$ is equally generated by the $\eps$-Frobenius basis in eq.~\eqref{solspacefrob}. The $\eps$-Frobenius basis is not an integral basis, but its advantage is that we can construct it explicitly. 

Next, let us discuss the remaining three terms in eq.~\eqref{eq:Sol2_decomposition}. The non-maximal cut integrals $J_{1,\underline{0}}^{\Gamma_{r,1}}(\zs;\eps)$ for $2\le r\le 4$ are not annihilated by the elements of ${\cal L}^{(2)}$ (because for each sector there is a tadpole integral whose maximal cut is non-zero). As a consequence, $J_{1,\underline{0}}^{\Gamma_{r,1}}(\zs;\eps)$ for $2\le r\le 4$ satisfies an inhomogeneous equation, i.e., it is annihilated by the elements of $\tilde{\cal L}^{(2)}$! Thus, we see that the last three terms in eq.~\eqref{eq:Sol2_decomposition} represent the contributions from the non-maximal cuts in eq.~\eqref{eq:Sol_decomposition}, which we quote here for the two-loop case:
\beq
\textrm{Sol}(\tilde{\cal L}^{(2)}) 
=\Big\langle J_{1,\underline{0}}^{\Gamma_{1,i}}(\zs;\eps) \Big\rangle_{{1\le i\le 4}}+\sum_{r=2}^4\Big\langle J_{1,\underline{0}}^{\Gamma_{r,1}}(\zs;\eps) \Big\rangle=
\textrm{Sol}({\cal L}^{(2)}) +\sum_{r=2}^4\Big\langle J_{1,\underline{0}}^{\Gamma_{r,1}}(\zs;\eps) \Big\rangle\,.
\eeq
Again, constructing explicitly the integer cycles $\Gamma_{r,1}$ with $r>1$ can be extremely complicated, but we can work with the elements of the non-integer $\eps$-Frobenius basis with $j=1$.
Note that it would be wrong to conclude that the elements of the $\eps$-Frobenius basis with $j=1$ are the non-maximal cut integrals $J_{1,\underline{0}}^{\Gamma_{r,1}}(\zs;\eps)$ for $r>1$, defined by integrating over cycles from integral homology. Just like for the maximal cuts, the elements of the $\eps$-Frobenius basis correspond to cut integrals over cycles that are not defined in integral homology, and a given non-maximal cut integral $J_{1,\underline{0}}^{\Gamma_{r,1}}(\zs;\eps)$ is in general for $r>1$ a linear combination with transcendental coefficients of different terms in the $\eps$ -Frobenius, including those with $j>1$.

\section{Equal-mass banana integrals and the Bessel function representation}
\label{subsec:equalmasseps}  

\subsection{Differential operators from Bessel functions}
\label{subsec:diffopbessel}

In this section we focus on the equal-mass banana integrals, and we will derive differential equations for them. One can of course apply the techniques from the previous section, e.g., by considering the equal-mass case as a one-parameter sub-slice in the generic-mass parameter space. However, there are typically several problems if one considers a one-parameter sub-slice in a much larger parameter space, and one cannot simply restrict the partial differential equations obtained in the previous section to the equal-mass case. Moreover, the dimension of the $\eps$-Frobenius basis is smaller, because the equal-mass case is highly symmetric, and there are fewer master integrals. The goal of this section is to present an alternative method to derive the differential equation for equal-masses, based on the Bessel function representation of banana integrals.

We start from the Bessel function representation in $D=2-2\epsilon$ dimensions of the $l$-loop Banana Feynman integral:
\begin{equation}
	J_{l,1}(z=1/t;\epsilon)	=	\frac {2^{l(1-\epsilon)}}{\Gamma(1+l\epsilon)}t^\frac \epsilon2 \int_0^\infty x^{1+l\epsilon}~I_{-\epsilon}(\sqrt tx)K_\epsilon (x)^{l+1}~\mathrm dx~,
\label{besselrep}
\end{equation}
valid for $t\coloneqq 1/z<(l+1)^2$, where $I_\alpha(x)$ and $K_\alpha(x)$ are the modified Bessel functions. A derivation of this representation can be found in ref.~\cite{Berends:1993ee} (see eq. (9) there, which also includes generic masses $m_i$) and in ref.~\cite{Vanhove:2014wqa} for $\epsilon=0$.

Recall from the $\epsilon=0$ case that the maximal cut of the $l$-loop banana integral in $D=2$ dimensions is given by a period integral of an $(l-1)$-dimensional Calabi-Yau variety. The full Feynman integral in turn is given by a linear combination of Calabi-Yau periods, all satisfying the homogeneous differential equation of the maximal cut integral, plus a special solution to an inhomogeneous version of this differential equation.
The differential operator is of degree $l$. Since there are $l$ equal-mass banana integrals at $l$ loops in dimensional regularization (cf.~eq.~\eqref{eq:equal_mass_MIs}), we expect that $J_{l,1}(z;\eps)$ satisfies an inhomogeneous differential equation of degree $l$, whose coefficients are polynomials in $z$ and $\eps$. We now describe a method to determine this operator and the corresponding inhomogeneity.

The modified Bessel functions $I_\alpha(x)$ and $K_\alpha(x)$ satisfy the modified Bessel differential equation:
\begin{equation}
	\mathcal Bf(x)=0\,,	\quad\text{with}\quad \mathcal B=\theta_x^2-x^2-\alpha^2~.
\label{besseleq}
\end{equation}
Using a proposition from refs.~\cite{doi:10.1080/10586458.2008.10129032,MR1809982}, we can construct an operator $\mathcal B_{l+2}$ of degree $l+2$ by the recurrence
\begin{equation}
	\mathcal B_0=1,~\mathcal B_1=\theta_x	\quad\text{and}\quad		\mathcal B_{k+1}=\theta_x\mathcal B_k-(x^2+\epsilon^2)(l+2-k)\mathcal B_{k-1}\,.
\end{equation}
This operator annihilates the $(l+1)^{\textrm{th}}$ power of $K_\epsilon(x)$, i.e.,
\begin{equation}
	\mathcal B_{l+2} K_\epsilon^{l+1}(x)	=	0~.
\end{equation}
Next we use this operator to get the identity
\begin{equation}
	\int_0^\infty x^{1+l\epsilon}~I_{-\epsilon}(\sqrt tx)~\mathcal B_{l+2}K_\epsilon (x)^{l+1}~ \rd x	=	0~.
\label{eqstart}
\end{equation}
Using integration-by-parts, we find\footnote{Here we assume that the boundary terms vanish, as they do for the functions in eq.~\eqref{eqstart}.}
\begin{equation}
	\int_0^\infty f(x)~\theta_x^m g(x)~\mathrm \rd x	=	(-1)^m\int_0^\infty g(x)~(\theta_x+1)^mf(x)~ \rd x\,,
\label{partid}
\end{equation}
as well as the identity
\begin{equation}
	(\theta_x+1)^mx^nf(x)	=	x^n(1+n+\theta_x)^mf(x)\,,
\label{id1}
\end{equation}
we can commute this operator in front of the other Bessel function, and we obtain an operator $\tilde{\mathcal{B}}_{l+2}$ 
 with the property
\begin{equation}
	\int_0^\infty x^{1+l\epsilon}~K_\epsilon (x)^{l+1}~\tilde{\mathcal B}_{l+2}I_{-\epsilon}(\sqrt tx)~\rd x	=	0~.
\label{eq2}
\end{equation}
Note that these operators contain only even powers of $x$.

So far we have still operators in the integration variable $x$. To obtain operators in the kinematic variable $t$, we use the identities:
\begin{align}
	\theta_x^nI_{-\epsilon}(\sqrt tx)	&=	2^n\theta_t^nI_{-\epsilon}(\sqrt tx)\,,	\\
	(x^2)^nI_{-\epsilon}(\sqrt tx)	&=	\left(\frac 1t(4\theta_t^2-\epsilon^2)\right)^nI_{-\epsilon}(\sqrt tx)~.
\label{id2}
\end{align}
At this stage we obtain
\begin{equation}
	\tilde{\mathcal D}_{l+2}(t)\int_0^\infty x^{1+l\epsilon}~K_\epsilon (x)^{l+1}I_{-\epsilon}(\sqrt tx)~\rd x	=	0~.
\end{equation}
Commutation of $t^{\frac\epsilon2}$ into the above expression is simply done by replacing $\theta_t$ by $\theta_t + \frac\epsilon2$ in $\tilde{\mathcal D}_{l+2}$. Thus we obtain an operator $\mathcal D_{l+2}(t)$ such that: 
\begin{equation}
	\mathcal D_{l+2}(t)\left\{ \frac 1{\Gamma(1+l\epsilon)}2^{l(1-\epsilon)}t^\frac \epsilon2 \int_0^\infty x^{1+l\epsilon}~I_{-\epsilon}(\sqrt tx)K_\epsilon (x)^{l+1}~\rd x \right\}	=	0~.
\label{homop}
\end{equation}
The operator $\mathcal D_{l+2}$ is of degree $l+2$. As such, it contains a too high number of derivatives in comparison with our expectations (recall that we expect a differential operator of order $l$). However, it turns out that we can factorize this operator in the following way
\begin{equation}
	\mathcal D_{l+2} = \mathcal N \ \theta_t (\theta_t -\epsilon) \ \mathcal L_{l,\eps}~,
\label{factorizationop}
\end{equation}
where we will fix the normalization $\mathcal N$ later. The operators $\mathcal L_{l,\eps}$ are of degree $l$ and give the homogeneous part of the desired inhomogeneous differential equation for $J_{l,1}(z;\eps)$. Let us denote the inhomogeneity by $S_l(t;\eps)$. We must have
\begin{equation}
	\mathcal L_{l,\eps}  J_{l,1}(z=1/t,\eps) = S_l(t;\eps)\,.
\label{oureq}
\eeq
Since $\theta_t (\theta_t -\epsilon)\mathcal L_{l,\eps}  J_{l,1}(z=1/t,\eps)=0$, the inhomogeneity must satisfy
\begin{equation}
\theta_t(\theta_t-\epsilon)S_l(t;\eps)=0 ~.
\end{equation}
Hence, it is of the form $S_l(t;\eps)=\alpha(\epsilon)+\beta(\epsilon)t^\epsilon$, where $\alpha(\epsilon)$ and $\beta(\epsilon)$ are constants in $t$ but may depend on the dimensional regulator $\epsilon$. 
At the large momentum point $z=1/t=0$, 
 we know the leading $\mathcal O(z)$ terms from eq.~\eqref{MBfirstorder}. We can then fix $\alpha(\epsilon)$ and $\beta(\epsilon)$ by transforming the differential operators from the variable $t$ to $z$. By slight abuse of notation, we use the same symbols $\mathcal L_{l,\eps}$ and $S_l(z;\eps)$ to write $\mathcal L_{l,\eps}  J_{l,1}(z,\eps) = S_l(z;\eps)$. Due to eq.~\eqref{MBfirstorder}, the inhomogeneity reads
\begin{equation}
	S_l(z;\eps)	=	-(l+1)! ~ z ~\frac{\Gamma(1+\epsilon)^l}{\Gamma(1+l\epsilon)}	=	(l+1)!z\, J_{l,0}(z;\epsilon)~,
\label{inhom}
\end{equation}
which generalizes the $\epsilon=0$ case presented in ref.~\cite{Vanhove:2014wqa}. Moreover, we see that the inhomogeneity is related to the tadpole given in eq.~\eqref{eq:equal_mass_MIs}.
Explicitly, the operators in terms of logarithmic derivatives can be found for the first few loop orders in table~\ref{tablediffops} and a \texttt{Mathematica}-code to generate them also for higher loop orders can be found in the supplementary data\footnote{You can download this from \url{http://www.th.physik.uni-bonn.de/Groups/Klemm/data.php}.} to this paper. The normalization constant $\mathcal N$ in eq. \eqref{factorizationop} is fixed such that $\left.\mathcal L_{l,\eps}\right|_{\theta=z=\eps=0}=1$. We stress that our differential operators are exact in $\eps$. 

\begin{table}[h]
\centering
{\small{
\begin{tabular}{@{}cp{12cm}@{}}
\toprule
Loop order $l$ & Differential operator $\mathcal L_{l,\eps}$ \\ \midrule
1 & $1+\epsilon-2 z -(1-4 z) \theta$\\[1ex]
2 & $(1+2 \epsilon ) (1+\epsilon-3 z +z \epsilon )+\left(-2-3 \epsilon+10 z +10 z \epsilon +9 z^2 \epsilon \right) \theta +(1-z) (1-9 z) \theta ^2$\\[1ex]
3 & $(1+2 \epsilon ) (1+3 \epsilon ) (1+\epsilon-4 z +2 z \epsilon )+(-3-12 \epsilon+18 z +60 z \epsilon -11 \epsilon ^2+28 z \epsilon ^2\newline+64 z^2
   \epsilon ^2) \theta 	-3 (-1+10 z) (1+2 \epsilon ) \theta ^2-(1-4 z) (1-16 z) \theta ^3$\\[1ex]
4 & $(1+2 \epsilon ) (1+3 \epsilon ) (1+4 \epsilon ) (1+\epsilon-5 z +3 z \epsilon )  +(-4-30 \epsilon+28 z +189 z \epsilon \newline+26 z^2 \epsilon -225 z^3
   \epsilon -70 \epsilon ^2+343 z \epsilon ^2-225 z^3 \epsilon ^2-50 \epsilon ^3+84 z \epsilon ^3+414 z^2 \epsilon ^3) \theta \newline	+ (6-63 z+26
   z^2-225 z^3+30 \epsilon -315 z \epsilon -675 z^3 \epsilon +35 \epsilon ^2-343 z \epsilon ^2-363 z^2 \epsilon ^2\newline-225 z^3 \epsilon ^2) \theta
   ^2	-2 \left(2-35 z+225 z^3+5 \epsilon -105 z \epsilon +259 z^2 \epsilon +225 z^3 \epsilon \right) \theta ^3\newline+(1-z) (1-9 z) (1-25 z) \theta ^4$\\[1ex]
5 & $(1+2 \epsilon ) (1+3 \epsilon ) (1+4 \epsilon ) (1+5\eps) (1+\eps-6z+4z\eps)+(-5-60 \epsilon+40 z +448 z \epsilon\newline +1152 z^3 \epsilon -255 \epsilon ^2+1664 z \epsilon ^2+472
   z^2 \epsilon ^2-3456 z^3 \epsilon ^2-450 \epsilon ^3+2128 z \epsilon ^3-4608 z^3 \epsilon
   ^3\newline-274 \epsilon ^4+208 z \epsilon ^4+2816 z^2 \epsilon ^4)\theta + (10+90 \epsilon-112 z+1152 z^3 -1008 z \epsilon +236 z^2 \epsilon +1152 z^3 \epsilon +255 \epsilon
   ^2-2772 z \epsilon ^2-10368 z^3 \epsilon ^2+225 \epsilon ^3-2128 z \epsilon ^3-4864 z^2
   \epsilon ^3-4608 z^3 \epsilon ^3)\theta^2 \newline +(-10+168 z-236 z^2+4608 z^3-60 \epsilon +1120 z \epsilon -85 \epsilon ^2+1848 z \epsilon ^2-2976
   z^2 \epsilon ^2\newline-6912 z^3 \epsilon ^2)\theta^3 + (5-140 z+5760 z^3+15 \epsilon -560 z \epsilon +3920 z^2 \epsilon)\theta^4 \newline - (1-4z)(1-16z)(1-36z)\theta^5$\\[1ex]
 \bottomrule
\end{tabular}
}}
\caption{The homogeneous differential operators $\cL_{l,\eps}$ that annihilate the maximal cuts of the banana integrals in $D=2-2\eps$ dimensions.}
\label{tablediffops}
\end{table}

Just like in the generic-mass case, we can rewrite our higher-order differential equation for $J_{l,1}(z;\eps)$ in eq.~\eqref{oureq} in the Gauss-Manin form in eq.~\eqref{eq:banana_DEQ}. For this we write the operator in terms of normal derivatives using eq.~\eqref{eq:d_to_theta}:
\begin{equation}
	\mathcal L_{l,\eps}	=	\sum_{k=0}^l B_{l,k}(z;\eps)\, \partial_z^k\,,
\label{normalderivatives}
\end{equation}
and we define ${\bf B}_{l}(z;\eps)$ similarly as in eq.~\eqref{defB}, but now with the $\eps$-dependent entries of $B_{l,k}(z;\eps)$. From this one can read off the coefficients ${\bf B}_{l,k}(z)$ in the expansion given in eq.~\eqref{1}. We observe that the limit of ${\bf B}_{l}(z;\eps)$ exists for $\eps\to0$. The initial condition for the equation is given by eq.~\eqref{MBfirstorder}.  Hence, we can use eq.~\eqref{eq:B_tilde_N_tilde} and the results for the Wronskian ${\bf W}_l(z)$ of section~\ref{sec:banint} to solve eq.~\eqref{eq:DEG_inhom_L} in terms of iterated integrals involving periods of a Calabi-Yau $(l-1)$-fold. The resulting expressions generalize the closed formula for $\eps=0$ from eq.~\eqref{eq:Jk0_sol} to higher orders in the dimensional regulator. Equation~\eqref{eq:winverse} implies that the matrix $\widetilde{{\bf B}}_{l,k}(z)$ in eq.~\eqref{eq:B_tilde_N_tilde} is at most quadratic in the periods, to all loop orders and for all orders in $\eps$.

\subsection{Discussion of the equal-mass operators in dimensional regularization}
\label{subsec:discussion}

\paragraph{The Riemann $\mathcal{P}$-symbols.}

Let us discuss some of the properties of the differential operators $\mathcal L_{l,\eps}$. 
Due to the $\epsilon$-deformation the indicials of the differential operators $\mathcal L_{l,\eps}$ shift. As in the $\eps=0$ case, one can collect the indicials at the singular points of the differential operator in the Riemann $\mathcal P$-symbol. For example for $l=2,3,4$ we find:
\begin{equation}
\begin{aligned}
	\mathcal P_2	&	\left\{	\begin{matrix*}[l]	%
							0 				& \frac1{9}			& 1			&	\infty			\\ \hline
							1+\eps\phantom{1}	& -2\eps				& -2\eps		&	0			\\
							1+2\eps			& 0					& 0			&	\eps			
						\end{matrix*} \right\}~,	\quad	\mathcal P_3		\left\{	\begin{matrix*}[l]	%
							0 				& \frac1{16}			& \frac14		&	\infty			\\ \hline
							1+\eps\phantom{1}	& 0					& 0			&	-\eps			\\
							1+2\eps			& \frac12	-3\eps		& \frac12-3\eps	&	0			\\
							1+3\eps			& 1					& 1			&	\eps
						\end{matrix*} \right\}\,,	\\[1ex]
	\mathcal P_4	&	\left\{	\begin{matrix*}[l]	%
							0 				& \frac1{25}			& \frac19	& 1		&	\infty			\\ \hline
							1+\eps\phantom{1}	& 0					& 0		& 0		&	0			\\
							1+2\eps			& 1-4\eps				& 1-4\eps	& 1-4\eps	&	\eps			\\
							1+3\eps			& 1					& 1		& 1		&	1			\\
							1+4\eps			& 2					& 2		& 2		&	1+\eps
						\end{matrix*} \right\}~.
\label{Riemanneps34}
\end{aligned}
\end{equation}
We can also give the general Riemann $\mathcal P$-symbol at arbitrary loop order $l$. For even $l$ we have
\begin{equation}
\begin{aligned}
	\mathcal P_\text{even}	&	\begin{Bmatrix*}[l]	%
							0 				& \frac1{(l+1)^2}		& \frac1{(l-1)^2}		& \cdots		& 1				& \infty			\\ \hline
							1+\eps			& 0					& 0				& \cdots		& 0				& 0				\\
							1+2\eps			& 1					& 1				& \cdots		& 1				& 0+\eps			\\
							1+3\eps			& 2					& 2				& \cdots		& 2				& 1				\\
							 				& 					& 				&  			& 				&  1+\eps			\\
							\vdots			& \vdots				& \vdots			& \vdots		& \vdots			& \vdots			\\
							 1+(l-1)\eps		& l-2					& l-2				& \cdots		& l-2				& \frac l2-1		\\
							1+l\eps			& \frac l2-1-l\eps		& \frac l2-1-l\eps	& \cdots		& \frac l2-1-l\eps	& \frac l2-1+\eps	\\		
						\end{Bmatrix*}~,
\label{Riemannepseven}
\end{aligned}
\end{equation}
whereas for odd $l$ we find
\begin{equation}
\begin{aligned}
	\mathcal P_\text{odd}	&	\begin{Bmatrix*}[l]	%
							0 				& \frac1{(l+1)^2}		& \frac1{(l-1)^2}		& \cdots		& \frac14			& \infty			\\ \hline
							1+\eps			& 0					& 0				& \cdots		& 0				& 0				\\
							1+2\eps			& 1					& 1				& \cdots		& 1				& 0+\eps			\\
							1+3\eps			& 2					& 2				& \cdots		& 2				& 1				\\
							 				& 					& 				& 			& 				&  1+\eps			\\
							\vdots			& \vdots				& \vdots			& \vdots		& \vdots			& \vdots			\\
							 				& 					& 				& 			& 				& \frac {l-3}2		\\
							1+(l-1)\eps		& l-2					& l-2				& \cdots		& l-2				& \frac {l-3}2+\eps	\\
							1+l\eps			& \frac l2-1-l\eps		& \frac l2-1-l\eps	& \cdots		& \frac l2-1-l\eps	& \frac {l-3}4-\frac {l-1}2\eps	\\		
						\end{Bmatrix*} ~.
\label{Riemannepsodd}
\end{aligned}
\end{equation}
One can easily prove that all Riemann $\mathcal P$-symbols satisfy the Fuchsian relation in eq.~\eqref{fuchsianrelation}, recalling that the number of singular points is given by $s=\frac l2+3$ for even $l$ and $s=\frac {l-1}2+3$ for odd $l$. Additionally, these Riemann $\mathcal P$-symbols should be compared to the corresponding quantities for $\eps=0$ in eq.~\eqref{Riemann34}.
In general, we observe that the indicials get shifted by $\eps$-contributions such that the degeneracy of the indicials is gone. In particular, this lift of degeneracy turns the MUM-point with indicials $\{1,\hdots,1\}$ into a regular point with indicials $\{1+\epsilon, \hdots, 1+l \epsilon\}$, where all elements of the $\eps$-Frobenius basis are now given by series solutions (free of logarithms). In contrast, the leading behavior of the special solution does not change and is $\mathcal O(z)$, so it does not get an $\eps$-dependent local exponent as the other solutions. 
The special solution is free of logarithms, too. With this indicial structure in mind, eq.~\eqref{MBfirstorder} suffices to fix all the coefficients in the linear combination of Frobenius solutions. As a side note we mention that we can compute the determinant of the Wronskian also for the $\eps$-Frobenius basis. Using the relation in eq. \eqref{detw} we find:
\begin{equation}
	\det{\bf W}_l(z) =	\left(z^{l-3-(l+1)\eps} \prod_{k\in \Delta^{(l)}}(1-kz)^{1+2\eps}\right)^{-l/2}~.
\end{equation}
The other identities computed in section \ref{sec:banint} cannot be generalized (as least not easily) to the $\eps\neq0$ case, because Griffiths transversality does not hold for generic $\eps$. This can for example be checked from the self-adjointness statement after eq. \eqref{eq:selfadjointness}. To be more precise, for $l=1,2$ the operators are for generic values of $\eps$ self-adjoint. Only in $D=1,2,3$, i.e., $\eps=-\frac12,0,\frac12$, also the operators for $l\geq3$ listed in table \ref{tablediffops} are self-adjoint.

\paragraph{Structure of the differential operators expanded in $\eps$.}

Here we give a brief account of the $\eps$-expansion of the differential equations for the equal-mass banana integrals. More precisely, we 
 characterize the differential equations that govern the $J_{l,1}^{(n)}(z)$ in eq.~\eqref{eq:epsexp}. We know from ref.~\cite{Bonisch:2020qmm} for $\eps=0$ that $\mathcal{L}_l J_{l,1}(z,0)=-(l+1)! \, z$, where $\mathcal{L}_l=\mathcal{L}_l^{(0)}=\mathcal{L}_{l,\eps=0}$ is the Picard-Fuchs operator of degree $l$ that annihilates the Calabi-Yau periods, or equivalently the maximal cuts of $J_{l,1}(z,0)=J_{l,1}^{(0)}(z)$. This implies that $\mathcal{L}_l^{ (0, \text{inh} ) } \mathcal{L}_l^{(0)}  J_{l,1}^{(0)}=0$, with $\mathcal{L}_l^{(0, \text{inh} )} =\theta_z -1$.
Similarly, at order $\epsilon^1$ we find a degree-$l$ operator $\mathcal{L}_l^{(1)}$ such that:
\begin{equation}
\mathcal{L}_l^{(1)} \mathcal{L}_l^{ (0, \text{inh} ) } \mathcal{L}_l^{(0)}\ J_{l,1}^{(1)} = P_{l,1}(z) \,,
\end{equation}
where $P_{l,1}(z)$ is a polynomial. 
Hence, there is degree-1 operator $\mathcal{L}_l^{(1,\text{inh})}$ that annihilates this polynomial, which implies:
\begin{equation}
\mathcal{L}_l^{(1,\text{inh})}\mathcal{L}_l^{(1)} \ \mathcal{L}_l^{ (0, \text{inh} ) } \mathcal{L}_l^{(0)} \ J_{l,1}^{(1)} = 0 \ .
\end{equation}
Indeed, the higher-order contributions $J_{l,1}^{(n)}$ belong to iterated extensions of the original (relative) Calabi-Yau system, and satisfy a differential equation of the form:
\begin{equation}\label{eq:higherextension}
\mathcal{L}_l^{(n,\text{inh})}\mathcal{L}_l^{(n)} \cdots \ \mathcal{L}_l^{(1,\text{inh})}\mathcal{L}_l^{(1)} \  \mathcal{L}_l^{ (0, \text{inh} ) } \mathcal{L}_l^{(0)} \ J_{l,1}^{(n)} = 0 \ ,
\end{equation}
where $\mathcal{L}_l^{(n)}$ is differential operator of degree $l$ in $\theta_z$, while $\mathcal{L}_l^{(n,\text{inh})}$ is of degree 1 and annihilates the polynomial inhomogeneity at the respective order $\epsilon^n$. As the explicit expressions for these differential operators quickly become lengthy, we will refrain from displaying them here, but shall mention that, generically, also integer indicials different from unity appear for solutions to the operator on the left hand side of eq.~\eqref{eq:higherextension}. With increasing value of $n$, the vanishing order of the discriminant polynomial multiplying the highest derivative increases as well. Besides that, the discriminant polynomial gets multiplied with further polynomial factors, which, however, describe apparent singularities of the differential system.

\subsection{Numerical results at low loop orders.}
\label{subsec:numerics}

It is possible to solve the differential equation in eq.~\eqref{oureq} in an efficient way even for high values of $l$ and for higher orders in $\eps$. The method to solve such equations was reviewed in section~\ref{subsec:PF}. As an example, we show in figure~\ref{figure:plots} the results for $J_{l,1}^{(n)}(z)$ for $2\le l\le 4$ and $0\le n\le 2$ for $0\le t=1/z\le 40$. The numerical results for $l\le 3$ are in principle not new, and they agree with the previous results obtained in refs.~\cite{Remiddi:2016gno,Bogner:2017vim,Primo:2017ipr,Duhr:2019rrs,Bezuglov:2021jou}. The four-loop results are new and presented for the first time in this paper. Note that, as expected, the imaginary part of $J_{l,1}^{(n)}(z)$ vanishes for $t<(l+1)^2$.  For $l=2$ there is a logarithmic singularity at $t=9$, whereas for $l\geq3$ it is only a non-differentiable point, due to the vanishing of the logarithmic Frobenius basis element at $t=(l+1)^2$, as can be seen from the Riemann $\mathcal P$-symbols in eq.~\eqref{Riemanneps34}.

\begin{figure}[!htbp]\centering
	\includegraphics[scale=0.43]{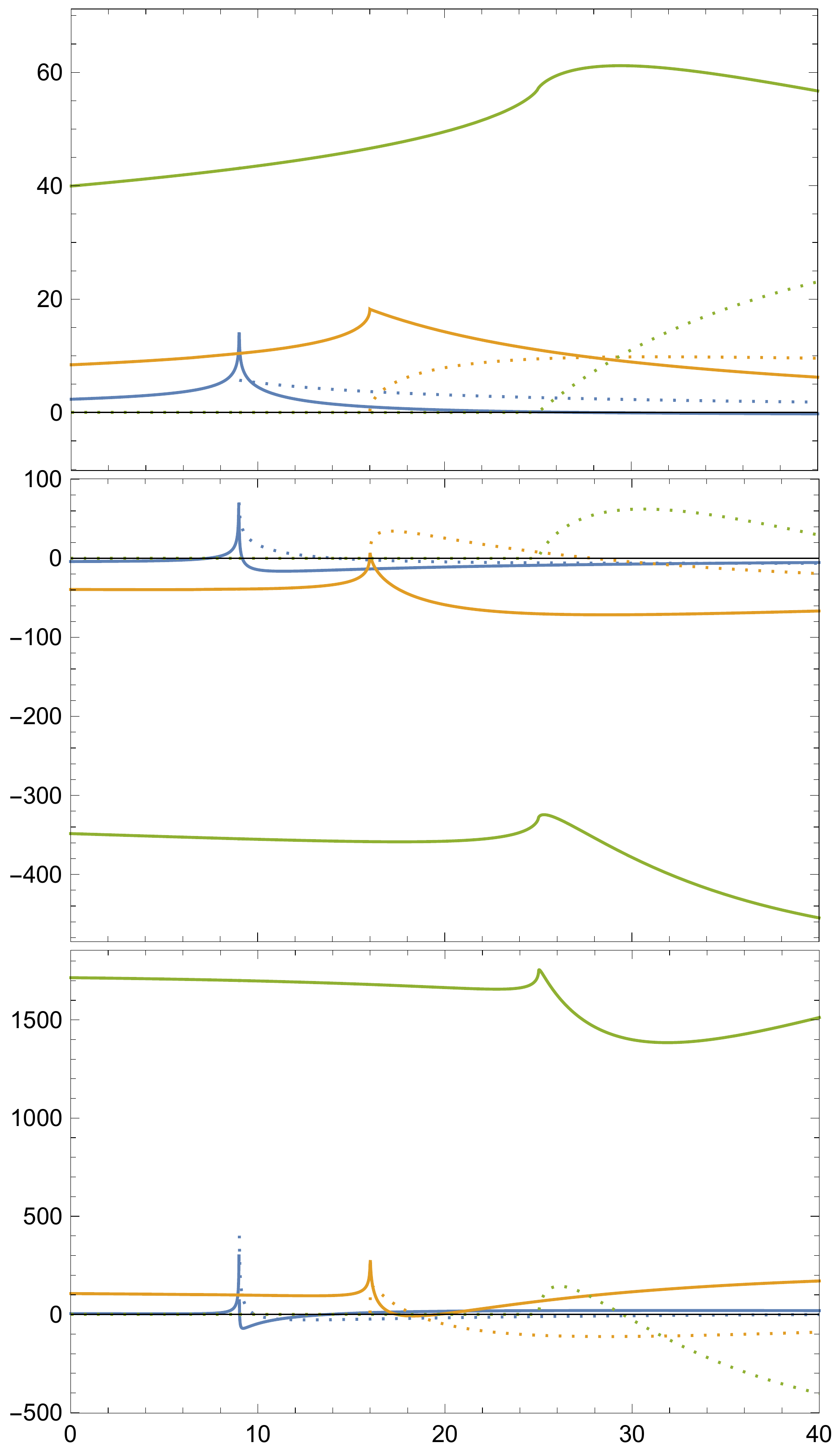}
\caption{The banana integrals $J_{l,1}^{(n)}$ for $l=2,3,4$ (blue, orange, green) and $n=0,1,2$ (upper, middle and lower panels). On the $x$-axis we have drawn the variable $t=p^2/m^2=1/z$. The solid lines correspond to the real part, whereas the dashed lines represent the imaginary part.}
\label{figure:plots}
\end{figure}

\section{Conclusion}
\label{sec:conclusion}  

Over the last decade it has become clear that multi-loop Feynman integrals are closely connected to topics in modern algebraic geometry. 
Understanding the geometry attached to a Feynman integral allows one to infer some of its properties, e.g., its differential equations, the special functions it evaluates to, or the structure of the singularities and monodromies. Very recently, it was shown that several examples of high-loop Feynman integrals are associated to Calabi-Yau geometries, cf., e.g., ref.~\cite{Bloch:2014qca,MR3780269,Bourjaily:2018yfy,Bourjaily:2018ycu,Bourjaily:2019hmc,Klemm:2019dbm,Bonisch:2020qmm}. In particular, in refs.~\cite{Klemm:2019dbm,Bonisch:2020qmm}, some of us have applied methods from Calabi-Yau geometry and mirror symmetry to study the properties of banana graphs with an arbitrary number of loops in $D=2$ dimensions. However, not much is known about Feynman integrals related to Calabi-Yau manifolds, mostly due to a lack of understanding of the relevant mathematics from the particle physics side.

One of the main goals of this paper is to build a bridge between the worlds of Calabi-Yau motives on the one hand, and of multi-loop Feynman integrals on the other. In section~\ref{sec:CY} we have reviewed in detail the geometric concepts relevant to families of Calabi-Yau manifolds, in particular the structure of their moduli spaces and the period integrals. While these mathematical concepts are not new, our goal is to present them in a way that provides an entry point into the subject for physicists working on Feynman integrals. We put a lot of emphasis on illustrating abstract mathematical concepts on concrete examples of families of elliptic curves, like the Legendre and Weierstrass family or the two-loop banana integrals. We expect that the geometric concepts reviewed in section~\ref{sec:CY} will play an increasingly important role in Feynman integral computations in the future. 

The first new and original result of this paper is a compact analytic expression for all master integrals for the equal-mass banana integrals in $D=2$ dimensions for an arbitrary number of loops $l$. In eq.~\eqref{eq:Jk0_sol} we have expressed these master integrals in terms of the periods of a one-parameter family of Calabi-Yau varieties and integrals of these periods (multiplied by a rational function). This representation is the immediate generalization of the representation of the two- and three-loop integrals in terms of iterated integrals of modular forms~\cite{ManinModular,Brown:mmv,Bloch:2014qca,Adams:2017ejb,MR3780269,Broedel:2018iwv,Broedel:2019kmn}. This is the first time that complete analytic results for higher-loop banana integrals in $D=2$ dimensions are presented. We find it remarkable that one can obtain compact analytic results valid for arbitrary loops: so far, only very few (non-trivial) families of Feynman integrals for arbitrary number of loops and general kinematics have been obtained, cf.~e.g., refs.~\cite{Usyukina:1993ch,Usyukina:1992jd,Broadhurst:1995km,Brown:2012ia,Schnetz:2012nt,Drummond:2012bg,Caron-Huot:2018dsv,Basso:2017jwq}. Our results extend this collection by a new infinite family of integrals, and it is the first example of such an infinite family that cannot be expressed in terms of multiple polylogarithms. We emphasize that a solid understanding of the geometry associated to equal-mass banana integrals in $D=2$ dimensions, in particular, the Picard-Fuchs operators, the large momentum behavior described by the $\widehat{\Gamma}$-class~\cite{Bonisch:2020qmm} and the quadratic relations from Griffiths transversality, was crucial in deriving eq.~\eqref{eq:Jk0_sol}. 

The second main result of this paper is the generalization of the results of ref.~\cite{Bonisch:2020qmm} to include dimensional regularization. In section~\ref{sec:Bananadimreg} we have presented a method to describe the Picard-Fuchs differential equations satisfied by $l$-loop banana integrals with arbitrary values for the propagator masses in $D=2-2\eps$ dimensions. We emphasize that, due to the large number of scales, these differential equations may be challenging to obtain with conventional techniques and computer programs. At the same time, we have obtained the leading asymptotic behavior in the large momentum limit, which provides a convenient boundary condition. In section~\ref{subsec:equalmasseps} we have presented an efficient alternative method to derive the Picard-Fuchs differential operator in the equal-mass case in dimensional regularization. The corresponding differential equations can be solved efficiently (at least for the first few loop orders), either using the Frobenius method to obtain fast-converging series representations, or analytically in terms of iterated integrals of Calabi-Yau periods by extending the results of section~\ref{sec:banint}. This opens the way to obtain for the first time complete results for all equal-mass banana integrals, for an arbitrary number of loops, in all kinematics regions, in dimensional regularization.

Let us also give some outlook for future directions of research. First, in section~\ref{sec:banint} we have obtained for the first time analytic results for Feynman integrals that involve (iterated) integrals over a Calabi-Yau period. Since this class of integrals generalizes in a natural way the iterated integrals of modular forms that appear at two- and three-loops, we expect iterated integrals of Calabi-Yau periods to play an increasingly important role for Feynman integral computations in the future. Unlike the case of modular forms, however, not much is known about such iterated integrals even in the mathematics literature, and it would be interesting to study them in more detail from a purely mathematical standpoint. Second, it would be interesting to explore in how far the techniques based on Frobenius bases and Picard-Fuchs differential ideals we have used in this paper can be applied to solve other Feynman integrals. First steps in this directions were taken in refs.~\cite{NasrollahpPeriodsFeynmanDiagrams2016,delaCruz:2019skx,Klausen:2019hrg}, but a complete understanding in how far these techniques can be practically applied is still an open question. Finally, it would be interesting to understand more generally when a Feynman integral is associated to a Calabi-Yau geometry or motive, or if (and which) more complicated geometric structures play a role. In refs.~\cite{Bourjaily:2018ycu,Bourjaily:2018yfy,Bourjaily:2019hmc,Vergu:2020uur} several infinite classes of multi-loop integrals associated to Calabi-Yau varieties have been identified, though starting from two loops in the non-planar sector~\cite{Huang:2013kh,Georgoudis:2015hca} and three loops in the planar sector~\cite{Hauenstein:2014mda}, geometries associated to Riemann surfaces of genus $g>1$ have been identified. For the future it would be interesting to have a clear understanding, and possibly a classification, of the geometries relevant to quantum field theory computations.

\section*{Acknowledgment}

We like to thank the members of the international \emph{Group de Travail on differential equations} in Paris  
for  stimulating talks  and in particular Spencer Bloch, Vasily Golyshev, Matt Kerr, Fernando Rodriguez 
Villegas, Emre Sert\"oz and  Duco van Straten  for discussions on Feynman integrals and Calabi-Yau motives.      

K.B. is supported by the International Max Planck Research School on Moduli Spaces of the Max Planck Institute for Mathematics in Bonn.


\appendix

\section{Ideals of differential operators in one variable}
\label{app:pid_one_var}

Let $\mathbb{Q}(z)$ be the field of rational functions in one complex variable $z$, and $R$ is the ring of differential operators of the form $\cL=\sum_{i=0}^pa_i(z)\partial_z^i$, $a_i(z)\in\mathbb{Q}(z)$, $1\le i\le p$ and $a_p(z)\neq 0$. We call $p\coloneqq\textrm{ord}(\cL)$ the order of $\cL$, and $a_p(z)\coloneqq\textrm{Disc}(\cL)$ its discriminant.

\begin{prop} R is a principal left-ideal domain, i.e., every ideal in $R$ is of the form $R\cL_0$ for some $\cL_0\in R$.
\end{prop}

\begin{proof}
Let $\cI\subseteq R$ be an ideal. We need to show that there is $\cL_0\in \cI$ such that $\cI=R\cL_0$. If $\cI$ is the zero or trivial ideal, $\cI=0$ or $\cI=R$, then the claim is obviously true. We therefore assume from now on that $\cI$ is neither zero nor trivial.

Define
\beq
d_{\cI} \eqqcolon \min\{\textrm{ord}(\cL) : \cL\in \cI\textrm{ and } \cL\neq 0\} >0 \,,
\eeq
where the inequality follows from the fact that $\cI$ is neither zero nor trivial. It is easy to see that there is $\cL_0\in\cI$ such that $\textrm{ord}(\cL_0)=d_{\cI}$ and $\textrm{Disc}(\cL_0) = 1$. We now show that this $\cL_0$ generates $\cI$. 

Let $\cL =\sum_{i=0}^pa_i(z)\partial_z^i \in \cI$, $a_p(z)\neq0$. We need to show that there is $\widetilde{\cL}\in R$ such that $\cL=\widetilde{\cL}\cL_0$. We proceed by induction in $\textrm{ord}(\cL)=p\ge d_{\cI}$. 
\begin{itemize}
\item If $p=d_{\cI}$, it is easy to see that $\textrm{ord}(\cL-a_p(z)\cL_0)<d_{\cI}$. Therefore we must have $\cL-a_p(z)\cL_0=0$.
\item If $p>d_{\cI}$, we assume that the claim is true for all operators in $\cI$ of order up to $p-1$. We have $\textrm{ord}(\cL-a_p(z)\partial_z^{p-d_{\cI}}\cL_0)<p$, and so by induction hypothesis there is $\cL_1\in R$ such that $\cL-a_p(z)\partial_z^{p-d_{\cI}}\cL_0=\cL_1\cL_0$.
\end{itemize}

\end{proof}

\section{Maximal cuts of equal-mass banana integrals up to three loops}
\label{app:low_loop}

In this appendix we present closed formulas for the periods $\varpi_{l,k}(z)$ for $l\le 3$ in terms of algebraic functions and complete elliptic integrals of the first kind. We focus on the large-momentum region $0<z<1/(l+1)^2$. The expressions in other regions require careful analytic continuation, cf. refs.~\cite{Remiddi:2016gno,Primo:2017ipr,Bogner:2017vim}.

At one loop, there is only one period, which is an algebraic function:
\beq
\varpi_{1,0}(z) = \frac{z}{\sqrt{1-4z}}\,.
\eeq
Under analytic continuation, the square can change sign. The monodromy group is thus $\Gamma_{\mathcal{E}_{\textrm{ban}_1}} = \mathbb{Z}_2$. 

At two loop order, the maximal cuts of $J_{2,1}(z;\eps)$ can be expressed in terms of complete elliptic integrals of the first kind, cf. ref.~\cite{Laporta:2004rb}. We introduce the shorthand notation:
\beq
\lambda = \frac{\left(\sqrt{z}+1\right) \left(3 \sqrt{z}-1\right)^3}{\left(\sqrt{z}-1\right) \left(3 \sqrt{z}+1\right)^3}\,.
\eeq
We then find:
\beq\bsp\label{eq:omega_2_loop_K}
\varpi_{2,0}(z) &\,= \frac{2 z}{\pi  \sqrt{1-27 z^2+18 z+8 \sqrt{z}}}\,\textrm{K}\!\left(1-\lambda\right)\,,\\
\varpi_{2,1}(z) &\,= -\frac{4 z}{\sqrt{1-27 z^2+18 z+8 \sqrt{z}}}\, \textrm{K}\!\left(\lambda\right)\,.
\esp\eeq
The monodromy group is $\Gamma_{\mathcal{E}_{\textrm{ban}_2}} = \Gamma_1(6)$~\cite{Bloch:2013tra,Adams:2017ejb,Frellesvig:2021vdl}. The periods in eq.~\eqref{eq:omega_2_loop_K} can be expressed in terms of Eisenstein series of weight one for $\Gamma_1(6)$~\cite{Bloch:2013tra,Adams:2017ejb}.

At three loops the Picard-Fuchs operator is a symmetric square, and the periods can be written as products of complete elliptic integrals of the first kind~\cite{Joyce,Primo:2017ipr}. If we define:
\beq\bsp
\lambda_1 &\,= \frac{32 z^{3/2}}{i \sqrt{1-16 z}-8 \left(i \sqrt{1-4 z}-2 \sqrt{z}\right) z+i \sqrt{1-4 z}}\,,\\
\lambda_2 &\,=-\frac{32 z^{3/2}}{i \sqrt{1-16 z}+8 \left(i \sqrt{1-4 z}-2 \sqrt{z}\right) z-i \sqrt{1-4 z}}\,,
\esp\eeq
we have:
\beq\bsp\label{eq:omega_3_loop_K}
\varpi_{3,0}(z) &\,= \frac{\sqrt{\lambda_1 \lambda_2}}{4 \pi ^2}\,\textrm{K}(\lambda_1) \,\textrm{K}(\lambda_2)\,,\\
\varpi_{3,1}(z) &\,= -\frac{\sqrt{\lambda_1 \lambda_2}}{2 \pi }\,\textrm{K}(\lambda_1) \,\textrm{K}(1-\lambda_2) -\frac{ \sqrt{\lambda_1\lambda_2}}{4 \pi }\,i \textrm{K}(\lambda_1) \textrm{K}(\lambda_2)\,,\\
\varpi_{3,2}(z) &\,=\frac{1}{6} \sqrt{\lambda_1 \lambda_2}\,\textrm{K}\!(1-\lambda_1) \,\textrm{K}(1-\lambda_2) +\frac{i}{6} \,\sqrt{\lambda_1 \lambda_2}\,\textrm{K}(\lambda_1)\, \textrm{K}(1-\lambda_2)\\
&\, -\frac{1}{8}\,\sqrt{\lambda_1 \lambda_2}\, \textrm{K}(\lambda_1)\, \textrm{K}(\lambda_2)\,.
\esp\eeq
The monodromy group is $\Gamma_{\mathcal{E}_{\textrm{ban}_2}} = \Gamma_0(6)^{+3}$~\cite{verrill1996}. $\Gamma_0(6)^{+3}$ contains $\Gamma_1(6)$ as a subgroup, and the  periods in eq.~\eqref{eq:omega_3_loop_K} can be expressed in terms of Eisenstein series of weight two for $\Gamma_1(6)$~\cite{Bloch:2014qca,MR3780269,Broedel:2019kmn}.

\section{The Lauricella hypergeometric series $F_C$}
\label{sec:lauricella}

As mentioned in section \ref{sec:Bananadimreg}, the $l$-loop banana integral in dimensional regularization is a linear combination of a specific kind of multivariate hypergeometric series, which were first described by Lauricella \cite{lauricella1893}. For completeness, we shall briefly elaborate on this facts in the present appendix.

 Consider the series in square brackets in eq.~\eqref{eq:MBresSummed3}, first for $j=(1,\dots,1)$. It can be recognized as a Lauricella hypergeometric series in $l+1$ variables $z_i$~\cite{lauricella1893} (see also refs.~\cite{Exton:1976yx,MR3616338}):
\begin{equation}\label{eq:defFC}
F_C^{(l+1)}(a,b;c_1,\dots,c_{l+1}; z_1,\dots, z_{l+1}) \coloneqq \sum_{n \in \mathbb{N}_0^{l+1}} \frac{(a)_{|n|} \ (b)_{|n|}}{\prod_{i=1}^{l+1} (c_i)_{n_i}} \prod_{i=1}^{l+1} \ \frac{z_i^{n_i}}{n_i !} \ .
\end{equation}
The case at hand is described by
\begin{equation} \label{eq:lauricparam}
a=1+(l+1)\epsilon+\delta, \qquad  b=1+l\epsilon+\delta, \qquad c_i=1+\delta_i +\epsilon \ .
\end{equation}
Up to the overall normalization, eq.~\eqref{eq:HypGeomBanana} reduces in the $l=2$ case to eq.~(123) of ref.~\cite{Adams:2013nia}, which in turn corrects the result presented in ref.~\cite{Berends:1993ee}. Recurrence (or contiguity) relations with respect to the $a,b$ and $c$ parameters (see, e.g., eq.~(6.4.17) of ref.~\cite{Exton:1976yx}) then imply algebraic relations between banana integrals with different propagator exponents. A systematic account of this is given in ref.~\cite{Bytev:2016ibi}. 

Lauricella reports a system of second-order linear partial differential equations sastisfied by eq.~\eqref{eq:defFC}~\cite{lauricella1893}.
Focussing on $\delta=0$, we give an independent and simple construction of a family of equivalent or sub-systems in section~\ref{sec:diffeqs}, and discuss special features of the specific hypergeometric system obtained for eq.~\eqref{eq:lauricparam} such as the relation to the (relative) Calabi-Yau period system in $D=2$ dimensions. Note that, according to ref.~\cite{lauricella1893},  
the general solution of Lauricella's differential system
\begin{equation}\label{eq:LauricellaDE}
\left[ \theta_k (\theta_k + c_k - 1)- z_k (\theta + a)(\theta + b )\right] F = 0 \ ,
\end{equation}
where $k=1,\dots, l+1$ and $\theta = \sum_{k=0}^{l+1}\theta_k$, depends on $2^{l+1}$ integration constants. In the generic case, the different elements of a fundamental system can be labelled by
$\underline j \in \{ 0,1 \}^{l+1} $ and read explicitly
\begin{align}\label{eq:LauricellaSols}
F_j = \prod_{i=1}^{l+1} z_i^{j_i' c_i'} & \cdot F_C^{(l+1)}\!\left(
a + \sum_{i=1}^{l+1} j_i' c_i' , \
b + \sum_{i=1}^{l+1} j_i' c_i' \, ; \
j_1 c_1 + j_1' (1+c_1')\, , \dots ,\, j_n c_n + j_n' (1+c_n')\right) \, , \nonumber \\
 &\text{where} \qquad j_i' \coloneqq 1-j_i  \qquad \text{and} \qquad \qquad c_i' \coloneqq 1-c_i \ .
\end{align}
The $z_i$ arguments suppressed in eq.~\eqref{eq:LauricellaSols} are left unchanged with respect to eq.~\eqref{eq:defFC}. We observe that the terms in square brackets in eq.~\eqref{eq:MBresSummed3}, including the monomial prefactor and the $j=0$ solution, precisely correspond to the $2^{l+1}$ fundamental solutions in eq.~\eqref{eq:LauricellaSols} identified by Lauricella, once the appropriate identification of parameters in eq.~\eqref{eq:lauricparam} is made. This observation is true for a generic choice of $\delta$, i.e., generic propagator exponents $1+\delta_i$, and by the considerations following eq.~\eqref{eq:MBready} this applies mutatis mutandis also in the presence of massless extra propagators. 

The coincidence of all series solutions in the $\epsilon=\delta=0$ limit does not mean a drop in the dimension of the solution space of eq.~\eqref{eq:LauricellaDE}, but rather indicates the need for considering logarithmic solutions. For $a=b=c_1=\dots=c_{l+1}=1$ the solution space of eq.~\eqref{eq:LauricellaDE} yet extends the space of period integrals of the holomorphic $(l-1)$-form on the associated Calabi-Yau $(l-1)$-fold and the relative period extension associated with the chain integral over the simplex $\sigma_l$.

We shall also mention that the series in eq.~\eqref{eq:defFC} converges for
 \begin{equation}\label{eq:convdomain}
 |\sqrt{z_1}| \, + \, \dots \, + \, |\sqrt{z_{l+1}}| \ <1 \ ,
\end{equation} 
independently of the values of $a,b$ and $c_1, \dots, c_n$, i.e, this also defines the domain of convergence of the series obtained for other values of $j$ in eq.~\eqref{eq:MBresSummed3}. Partial results for the analytical continuation to other domains can be found in refs.~\cite{Exton:1976yx,Berends:1993ee,Ananthanarayan:2019icl,Ananthanarayan:2020xut}. In general, the singular locus of the Lauricella system of ref.~\cite{MR3238325} is determined by the zeroes of
\begin{equation}
\left(\prod_{k=1}^{l+1} z_k\right)  \prod_{\epsilon_1,\dots,\epsilon_{l+1}=\pm 1 } \left(1+\sum_{k=1}^{l+1} \epsilon_k \sqrt{z_k}\right) \ .
\end{equation}
The fundamental group of the complement of this singular locus was recently studied in ref.~\cite{MR3857198}.  For more (mathematically rigorous) studies of the Lauricella system see also refs.~\cite{MR3616338,MR4031249,MR4143725,MR3152203}.

\paragraph{A Laplace type representation.}

Before closing this appendix, we note that there is a Laplace type integral for the Lauricella function $F_C^{l+1}$~\cite{Exton:1976yx} , reading
\begin{align}
 & \ F_C^{(l+1)}(a,b;c_1,\dots,c_{l+1}; z_1,\dots, z_{l+1})\nonumber \\ &= \frac{1}{\Gamma(a) \Gamma(b)}\int_{0}^{\infty}  \int_{0}^{\infty} \mathrm{e}^{-s-t} s^{a-1} t^{b-1}  \ {}_{0}F_{1}(-;c_1;z_1 \, st)  \ \cdots \ {}_{0}F_{1}(-;c_{l+1};z_{l+1} \, st) \ \mathrm{d}s \, \mathrm{d}t \, .
\end{align}
Here $a$ and $b$ must have positive real parts. Also
\begin{equation}
{}_{0}F_{1}(-; c; z) =  \Gamma(c) \ z^{\frac{1-c}{2}} \ I_{c-1}(2 \sqrt{z}) 
\end{equation}
can be identified with a modified Bessel function of the first kind. For the case in eq.~\eqref{eq:lauricparam}, setting $\epsilon=\delta=0$, we recover the statement of ref.~\cite{Bonisch:2020qmm} that the holomorphic Calabi-Yau period around the MUM-point is (up to an overall factor) given by the double Borel sum  of the $(l+1)^{\textrm{th}}$ symmetric power of the series associated with $I_0$.


\bibliographystyle{JHEP}
\bibliography{References}

\end{document}